\documentclass[cmp,draft]{svjour} 
\usepackage{amsmath,amsfonts,amssymb}
\usepackage{mathrsfs}
\usepackage{stmaryrd}
\usepackage[all,cmtip]{xy}
\usepackage{graphicx}
\usepackage[final,bookmarksdepth=subsection]{hyperref}

\allowdisplaybreaks[1] 

  \newcommand{\oo}{\infty}
  \newcommand{\del}{\partial}
\renewcommand{\d}{\mathrm{d}}
\renewcommand{\dh}{\mathrm{d}_{\mathsf{h}}}
  \newcommand{\dv}{\mathrm{d}_{\mathsf{v}}}
  \newcommand{\eps}{\varepsilon}
         \def\oast{\circledast} 
  \newcommand{\ot}{\leftarrow}
  \newcommand{\sso}{\subset}
  \newcommand{\sse}{\subseteq}
  \newcommand{\dsse}{\rotatebox[origin=c]{-90}{$\subseteq$}}
  \newcommand{\usse}{\rotatebox[origin=c]{90}{$\subseteq$}}

  \newcommand{\base}{\mathrm{\Pi}}
  \newcommand{\id}{\mathrm{id}}
  \newcommand{\im}{\operatorname{im}}
  \newcommand{\coker}{\operatorname{coker}}
  \newcommand{\supp}{\operatorname{supp}}
  \newcommand{\tr}{\operatorname{tr}}
  \newcommand{\Secs}{\mathrm{\Gamma}}
  \newcommand{\Forms}{\mathrm{\Omega}}
  \newcommand{\A}{\mathcal{A}}
\renewcommand{\AA}{\mathscr{A}}
  \newcommand{\B}{\mathcal{B}}
  \newcommand{\C}{\mathcal{C}}
  
  \newcommand{\CatC}{\mathfrak{C}}
  \newcommand{\CatI}{\mathfrak{I}}
  \newcommand{\DD}{\mathscr{D}}
  \newcommand{\E}{\mathcal{E}}
  \newcommand{\EE}{\mathrm{E}}
  \newcommand{\EL}{\mathrm{EL}}
  \newcommand{\F}{\mathcal{F}}
  \newcommand{\FF}{\mathscr{F}}
  
  \newcommand{\G}{\mathrm{G}}
\renewcommand{\H}{\mathrm{H}}
  \newcommand{\J}{\mathrm{J}}
  \newcommand{\K}{\mathrm{K}}
\renewcommand{\L}{\mathcal{L}}
  \newcommand{\M}{\mathcal{M}}
  \newcommand{\MM}{\mathscr{M}}
  
\renewcommand{\S}{\mathcal{S}}
\renewcommand{\SS}{\mathscr{S}}
\renewcommand{\P}{\mathrm{P}}
  \newcommand{\PP}{\mathcal{P}}
  \newcommand{\PPP}{\mathscr{P}}
  \newcommand{\R}{\mathbb{R}}
  \newcommand{\Lie}{\mathcal{L}}

  \newcommand{\Loc}{\operatorname{Loc}}

  \newcommand{\PolyLoc}{\operatorname{PolyLoc}}
  \newcommand{\POISS}{\operatorname{POISS}}
  \newcommand{\Bkgr}{\mathfrak{Bkgr}}
  \newcommand{\GlobHyp}{\mathfrak{GlobHyp}}
  \newcommand{\Poiss}{\mathfrak{Poiss}}
  \newcommand{\Cat}{\mathfrak{Cat}}
  \newcommand{\CAlg}{\mathfrak{CAlg}}
  \newcommand{\stAlg}{{*}\mbox{-}\mathfrak{Alg}}
  
  \newcommand{\Man}{\mathfrak{Man}}
  \newcommand{\Bndl}{\mathfrak{Bndl}}
  \newcommand{\VBndl}{\mathfrak{VBndl}}
  \newcommand{\CBndl}{\mathfrak{CBndl}}
  \newcommand{\ChrBndl}{\mathfrak{ChrBndl}}
  \newcommand{\SpBndl}{\mathfrak{SpBndl}}
  \newcommand{\SpBkgr}{\mathfrak{SpBkgr}}
  \newcommand{\Symp}{\mathfrak{Symp}}

\spnewtheorem{hypothesis}{Hypothesis}{\bf}{\it}

\begin{document}

\title{{Characteristics, Conal Geometry and Causality in Locally Covariant Field Theory}}

\author{Igor Khavkine\inst{1}}
\institute{Institute for Theoretical Physics, Utrecht, Leuvenlaan 4,
NL-3584 CE Utrecht, The Netherlands \email{i.khavkine@uu.nl}}

\date{\today}

\maketitle
\begin{abstract}
The goal of this work, motivated by the desire to understand causality
in classical and quantum gravity, is an in depth investigation of
causality in classical field theories with quasilinear equations of
motion, of which General Relativity is a prominent example. Several
modern geometric tools (jet bundle formulation of partial differential
equations (PDEs), the theory of symmetric hyperbolic PDE systems,
covariant constructions of symplectic and Poisson structures) and
applies them to the construction of the phase space and the algebra of
observables of quasilinear classical field theories. This construction
is shown to be diffeomorphism covariant (using auxiliary background
fields if necessary) using categorical tools in a strong parallel with
the locally covariant field theory (LCFT) formulation of quantum field
theory (QFT) on curved spacetimes. In this context, generalized versions
of LCFT axioms become theorems of classical field theory, which includes
a generalized Causality property. Considering deformation quantization
as the connection to QFT, a plausible conjecture is made about the
Causal structure of quantum gravity. In the process, conal manifolds are
identified as the generalization of the causal structure of Lorentzian
geometry to quasilinear PDEs. Several important concepts and results are
generalized from Lorentzian to conal geometry. Also, the proof of
compatibility of the Peierls formula for Poisson brackets and the
covariant phase space symplectic structure for hyperbolic systems is
generalized to now encompass systems with constraints and gauge
invariance.
\end{abstract}

\section{Introduction}\label{sec:intro}
Great strides have been made in the past decade in terms of the
definition and perturbative construction quantum field theories (QFTs)
in the algebraic framework~\cite{bf-lcqft}. Notable successes include
the extension of perturbative renormalization of non-linear field
theories to arbitrary, curved (though non-dynamical), globally
hyperbolic spacetimes and the treatment of gauge
theories~\cite{hollands-ym,fr-bv,rejzner-thesis}. These ingredients are,
for instance, sufficient to study perturbative general relativity (GR)
in this framework.  Unfortunately, deeply rooted in the formalism of QFT
on curved spacetimes is the reliance on a non-dynamical background
metric to supply the causal structure that defines the singularity and
support structures of various $n$-point functions, time ordered
products, and retarded products.

The objects in classical field theory that are most closely related to
these QFT objects are the retarded and advanced Green functions of the
dynamical field equations of motion linearized at some particular
solution (the dynamical linearization point). But we know that, in
classical gravity, the singularity and support structures of these
classical Green functions change as a function of the metric that
provides the dynamical linearization point.  In perturbative GR, the
dynamical metric can be expressed as $g=\bar{g}+\kappa h$, where
$\bar{g}$ is a fixed background metric and $h$ is the dynamical
perturbation, $\kappa^2$ is proportional to Newton's gravitational
constant, which is to be treated as the formal perturbation parameter.
As already mentioned, the singularity structure of the Green function of
the linearized gravitational wave equation depends on the full,
dynamical linearization point $g=\bar{g}+\kappa h$. But, to order
$O(\kappa^0)$, this singularity structure is fully determined by the
non-dynamical background metric $\bar{g}$.

Clearly, the necessary dependence of this singularity structure on the
dynamical perturbations $h$ will appear at higher orders in $\kappa$.
However, at least in the case of QFT $n$-point functions, it is not
\emph{a priori} clear what form this modifications would take or how
they are to be recognized in the result of some higher order
perturbative calculation. Given the tight connection between the
singularity structure of Green functions and what is commonly referred
to as \emph{causal structure}, it is fair to say that the causal
structure of gravity is \emph{dynamical}. Therefore, one might say that
in the current state of affairs one does not know how to recognize a
dynamical causal structure in perturbative QFT nor even how to interpret
the meaning of causal structure in a non-perturbative QFT, whose
classical version has dynamical causal structure.

Unfortunately, currently, QFT of GR is only accessible to us
perturbatively. Eventually, though, we would like to move to a
non-perturbative formulation. However, at the non-perturbative level,
the axiomatic framework of locally covariant field theory (LCFT), both
classical and quantum, encompasses only theories whose equations of
motion are systems of PDEs with principal symbols that only depend on a
non-dynamical background metric. This explicitly excludes so-called
\emph{quasilinear} equations like GR, hydrodynamics, elasticity and
non-linear $\sigma$-models. Quasilinear systems are defined by the
property that their principal symbols (roughly, the coefficients of the
highest order derivatives, which play a major role in the properties of
the previously mentioned Green functions) depend on the dynamical
fields. Fortunately, at the classical level, GR and other quasilinear
systems can be studied non-perturbatively using classical PDE theory.
Using these tools, the causal structure can be described and understood
in fairly direct terms, closely related to the geometry of \emph{domain
of dependence} theorems.

This work has multiple goals. In broad outline, it is concerned with
obtaining a deep, non-perturbative understanding of an algebraic
formulation of causality in classical field theories with
dynamical causal structures, in a way that parallels the well understood
geometric formulation of well-posedness in hyperbolic PDE systems. Once
such an understanding has been obtained another aim is to formulate a
conjecture as to how it would translate to quantum field theory. Of
course, at this point, one can only expect to make conjectures in this
direction, owing to the difficulties in dealing with quantum field
theories non-perturbatively. Once formulated, these notions could be
adapted to a formal perturbative context, where they could be used in
conjunction with practical calculations, for instance in perturbatively
quantized GR.

More specifically, another goal is to present in a unified way all the
necessary geometric aspects of the construction of a classical field
theory. These aspects include an intrinsic geometric formulation of PDE
systems using jet bundles, an application of the theory of symmetric
hyperbolic PDE systems to the construction of the phase space, and the
use of the covariant phase space method and Peierls formula to endow it
with symplectic and Poisson structures. Each of these aspects is well
known in some segments of the mathematical analysis and mathematical
physics communities, but certainly is not familiar to the broad
high energy theory and theoretical relativity communities. Some of these
aspects are presented here also in potentially novel ways, which
highlight how they are used together. In a sense, this paper can be seen
as a follow that supplements with symplectic geometric and Poisson
algebraic aspects the older work of Geroch~\cite{geroch-pde} who
essentially applied the theory of symmetric hyperbolic systems to a
unified construction of the spaces of solutions of relativistic
classical field theories.

Finally, this paper also aims to apply the category-theoretic methods that
have proven so useful in formulating the algebraic structure of QFT on
curved spacetimes (now known as LCQFT) to both the algebraic and dual
geometric formulation of classical field theory. Essentially, these
tools are used to draw a strong parallel or duality between the notion
of well-posedness from classical PDE theory and the algebraic properties
of classical field theory, so that both of which could be formulated as
mutually dual, succinct theorems. These theorems could later be used as
a starting point for an improved axiomatization of classical and quantum
field theories.

\subsection{Outline of the work}
Sect.~\ref{sec:freelcft} recalls the axiomatic framework of LCFT and
emphasizes that it currently excludes quasilinear systems like GR.

Sect.~\ref{sec:hypersys} introduces the geometric formulation of PDE
systems in terms of jet bundles, Sect.~\ref{sec:jets-pdes}, with
reference to Apdx.~\ref{sec:jets} for background on jets. It then
focuses on the notion of a quasilinear, hyperbolic system of PDEs, with
special emphasis on first order, symmetric hyperbolic systems,
Sect.~\ref{sec:symhyper}. The precise formulation of the condition of
symmetric hyperbolicity given there differs slightly from the standard
one in order to fit better in the subsequent geometric context. It
concludes with a discussion of prolongation and equivalence of PDE
systems, needed for the reduction of variational equations of motion to
hyperbolic form, Sect.~\ref{sec:integrability}.

Sect.~\ref{sec:chargeom} examines the role of characteristics of a PDE
system in its causal structure. Sect.~\ref{sec:geom-cones} identifies
chronal and spacelike cone bundles as essential abstractions of the
resulting causal order. Sect.~\ref{sec:slow-sec} groups solutions
together into slow patches by the property of being slower than a
certain fixed chronal or spacelike cone bundle.
Sect.~\ref{sec:pde-theory} uses the thus far developed notions of causal
structure to state some standard theorems of PDE theory, which will be
used as the basic tools in the construction of classical field theories.
Sect.~\ref{sec:lin-inhom} presents some more specialized results on
linear systems, including global existence and properties of Green
functions. The presented material is standard, but may be formulated in
a slightly novel way to fit better with the geometric context of the
Peierls formula for the Poisson bracket appearing later.

Sect.~\ref{sec:classquant} constitutes roughly half of the technical
bulk of this work. It collects the mathematical tools presented thus far
and applies them in a unified way to construct the phase space of a
classical field theory and its algebra of observables.
Sect.~\ref{sec:var-sys} reviews the covariant phase space method for
constructing the symplectic form of a variational PDE system. Since the
focus of this work is more geometrical than functional analytical, we
avoid most details of the treatment of infinite dimensional manifolds,
of which the space of solutions of a PDE system, to be identified with
the phase space of a classical field theory, is an example.  The purpose
of Sect.~\ref{sec:formal-dg} is to introduce enough technical machinery
to formally with the tangent and cotangent spaces of the space of
solutions, while keeping with the spirit of the preceding remark.
Sect.~\ref{sec:symp-pois} identifies the space of solutions with the
phase space, by endowing it with (formal) symplectic and Poisson
structures. The Poisson structure is given by a generalized Peierls
formula. This Peierls formula is shown to be equivalent to the Poisson
bracket obtained from the covariant phase space symplectic form, a
result that, compared to previous forms of this equivalence, now also
encompasses hyperbolic PDE systems with constraints. As an application of
this construction, a refined version of classical microcausality is
proven both for algebras of observables on a given slow patch as well as
on the global phase space, Sects.~\ref{sec:obsv-local}
and~\ref{sec:glob-phsp}. Also of note is a discussion in
Sect.~\ref{sec:obsv-local} of the relation between the on-shell
formalism used in this work and the off-shell formalism used in other
related literature.

Sect.~\ref{sec:natural} makes a connection between the constructions
thus far and the covariant formalism of LCFT via the notion of natural
bundles and PDE systems, relying on some background material on category
theory from Apdx.~\ref{sec:limits}. For convenience,
Sect.~\ref{sec:funct} systematically summarizes the notation for various
concepts and constructions appearing in this work, while also remarking
on their categorical and functorial properties.

Sect.~\ref{sec:classcaus} constitutes the other half of the technical
bulk of this work. It generalizes the axioms of LCFT to quasilinear
field theories taking into account their dynamical causal structures,
and proves them as theorems of classical field theory. Unfortunately, at
this point the proofs rely on a number of sufficient technical
hypotheses. The necessity of these hypotheses remain to be examined on a
case by case basis or in the context of deeper investigation of the
relevant functional analytical details. In particular, the
generalization of the Causality axiom provides a sought deeper
understanding of causality in quasilinear classical field theories and
its algebraic formulation. An important check,
Sect.~\ref{sec:semilin-lcft}, is that these theorems reduce to the
standard axioms of LCFT in the case of a well-posed semilinear classical
field theory, as is to be expected.

Sect.~\ref{sec:quantcaus} recalls the notion of deformation quantization and
remarks that it is the leading candidate for the modern formulation of
what it means to quantize a classical mechanical system. This notion of
quantization is then used to make a conjecture about the translation of
the algebraic formulation of Causality from classical to quantum field
theory. In particular, this conjecture answers an old question about the
structure of the commutator of two quantum metric field operators in the
quantum field theory of GR.

Finally, Sec.~\ref{sec:discuss} concludes with a discussion of the
results presented in this work and lists several ideas, conjectures and
hypotheses that had arisen in the process that would be fruitful to
investigate further.

\section{Free, Interacting and Perturbative Locally Covariant Field Theories}
\label{sec:freelcft}
The currently leading axiomatic framework in the algebraic framework for quantum
field theory is called \emph{locally covariant field theory} (LCFT) and
rests on the axioms originally proposed by Brunetti, Fredenhagen and
Verch~\cite{bfv}. These axioms rest on simple physical principles, namely
\emph{locality}, \emph{causality}, and the existence of a
\emph{dynamical law}, and have a convenient mathematical formulation.
Moreover, they are flexible in that they can and have been adapted to
the classical (\emph{locally covariant classical field theory}
abbreviated again LCFT), quantum (\emph{locally covariant quantum field
theory} or LCQFT), and both classical and quantum perturbative contexts
(respectively abbreviated pLCFT or pLCQFT)~\cite{bf-lcqft}.

\subsection{Brunetti-Fredenhagen-Verch axioms for LCFT}
To state the axioms in a succinct way, we need to appeal to some notions
of category theory. Some basic information on categories and functors
can  be found in~\cite{cattheory,borceux}. A classical mechanical system, at a
bare minimum is described by a Poisson algebra over $\R$, say
$(F,\{\})$, called the \emph{algebra of observables}. Actually, we are
likely to be interested only the algebras $F$ that correspond to smooth
functions on some phase space manifold, with the bracket $\{\}$ defined
by a Poisson tensor.  However, as we will not need it, we do not
consider a detailed characterization of these algebras.  Such Poisson
algebras form the category $\Poiss$, with Poisson homomorphisms as
morphisms. On the other hand, it will be important below to restrict
possible homomorphisms to injective ones. Thus we define the following
categories.
\begin{definition}
Let $\Poiss$ be the category of Poisson algebras over $\R$ as objects
and Poisson homomorphisms as morphisms. Let $\Poiss_i$ be the
subcategory where the morphisms are restricted to injective
homomorphisms.
\end{definition}

A field theory assigns a classical mechanical system to a region
of spacetime, that is, a Poisson algebra of observables whose spacetime
supports%
	\footnote{The relevant notion of \emph{spacetime support} will be
	discussed in detail in Sect.~\ref{sec:obsv-local}.} %
overlap with that region. A field theory is locally covariant, if
these assignments are coherent with respect to embedding smaller
spacetimes into larger ones, in a sense to be specified below. First,
however, we must define what constitutes a spacetime. The notion of a
LCFT was developed in the context of quantum and classical field theory
on curved spacetimes. Thus it is natural to consider Lorentzian
manifolds as spacetimes. Overwhelming
experience from the physics literature suggests that it is sensible to
restrict our attention to globally hyperbolic%
	\footnote{A \emph{globally hyperbolic} Lorentzian manifold can be
	identified in several equivalent ways. Perhaps the simplest definition
	is the existence of a Cauchy surface, which is a surface intersected
	exactly once by every inextensible timelike curve~\cite{wald-gr,geroch-gh}.} %
spacetimes. The spacetime morphisms are restricted to ensure
compatibility between the metric and causal structures on the source and
target manifolds.
\begin{definition}\label{def:globhyp}
Given two oriented and time oriented Lorentzian $n$-manifolds $(M,g)$
and $(M',g')$, a smooth map $\chi\colon M\to M'$ is called a
\emph{causal isometric embedding} if (i) $\chi$ is an open embedding
preserving orientation, (ii) it is an isometry, $\chi^* g' = g$,
preserving time orientation, and (iii) if any two points $x,y\in\chi(M)$
can be joined by a causal curve in $M'$ then they can be joined by a
causal curve in $\chi(M)$ (causal compatibility).

Let $\GlobHyp$ denote the category of oriented and time oriented
globally hyperbolic Lorentzian spacetimes as objects and isometric
embeddings as morphisms. Let $\GlobHyp_c$ denote the subcategory where
the morphisms are restricted to causal isometric embeddings.
\end{definition}

To express the causality property, it is convenient to introduce a
tensor product on Poisson algebras and the notion of independent
subsystems. First, though, a few words about tensor products in
categories. The necessary formal setting is that of a \emph{symmetric
monoidal category}~\cite{moncat}.
Such a category is equipped with a
self-bifunctor $(A,B)\colon A\otimes B$ called \emph{tensor product}
(also \emph{monoidal structure}), as well as an \emph{identity} object.
This product is required to satisfy some identities, expressed as
commutative diagrams, which guarantee that associativity holds, that the
identity object as is an identity for the product, and that there
existence of a canonical isomorphism $A\otimes B\cong B\otimes A$.
Categorical products and coproducts (Sect.~\ref{sec:limits}) are well
known examples of tensor products. In fact, we do not give the detailed
list of the above axioms because the tensor products that will be used
in this paper are all constructed by putting some extra structure on an
underlying (co)product in a way that makes the needed axioms manifest.
Like the underlying (co)products, our tensor products will be equipped
with canonical inclusion, $A,B\to A\otimes B$, or projection, $A\otimes
B \to A,B$, morphisms. An important distinction of a tensor product from
a (co)product is the lack of universality. That is, the existence of a
pair of morphisms $A\to C$ and $B\to C$ (or $C\to A$ and $C\to B$) does
not guarantee the existence of a canonical morphism $A\otimes B \to C$
($C\to A\otimes B$), unlike for a coproduct (or product). However, since
our tensor products are based on underlying (co)products, which such a
morphism exists, it is canonical. A functor between two tensor
categories that the tensor product structure is called a \emph{tensor
functor} (also a \emph{symmetric monoidal functor}).
\begin{definition}\label{def:indep-subsys}
Let $(F,\{\}_F)$ and $(G,\{\}_G)$ be Poisson algebras, their
\emph{independent subsystems (tensor) product} is defined by
\begin{equation}
	(F,\{\}_F) \otimes (G,\{\}_G) \cong (F\otimes G, \{\}) ,
\end{equation}
where $F\otimes G$ is the tensor product (coproduct) of commutative
algebras to which the Poisson bracket is extended by the rule
$\{F\otimes 1, 1\otimes G\} = 0$. The product Poisson algebra is
equipped with the canonical inclusion morphisms
\begin{equation}
\vcenter{\xymatrix{
	(F,\{\}_F) \ar[r] & (F,\{\}_F)\otimes(G,\{\}_G) & (G,\{\}_G) \ar[l]
}} ,
\end{equation}
which are given respectively by $f_1\mapsto f_1\otimes 1$ and
$f_2\mapsto 1\otimes f_2$.
\end{definition}
\begin{definition}
Two subalgebras $(F_i,\{\}_i)\sse (F,\{\})$ $i=1,2$, of a larger Poisson
algebra are said to be \emph{independent subsystems} if the inclusion
morphisms factor through the independent subsystems according to the
following commutative diagram:
\begin{equation}
\xymatrix{
	(F_1,\{\}_1) \ar[r] \ar[rd]
		& (F_1,\{\}_1)\otimes(F_2,\{\}_2) \ar@{-->}[d]
		& (F_2,\{\}_2) \ar[l] \ar[ld] \\
		& (F,\{\})
} ,
\end{equation}
where the horizontal morphisms are the canonical inclusions, the
diagonal arrows are the given subalgebra inclusion morphisms, and the
vertical dotted line morphism is $f_1\otimes f_2 \mapsto f_1 f_2$ and is
a canonical injective Poisson homomorphism.
\end{definition}

The axioms classical field theory can now be succinctly stated as
follows.
\begin{definition}[\cite{bf-lcqft}]\label{def:axlcft}
A \emph{locally covariant classical field theory} $\F$
satisfies the following axioms:
\begin{enumerate}
\item[\emph{Isotony}]
	It is a covariant functor, $\F\colon \GlobHyp_c \to \Poiss_i$.
\item[\emph{Time Slice}]
	The image $\F(\chi)$ of $\chi\colon M\to M'$ is a Poisson isomorphism
	whenever $\chi(M)$ contains a Cauchy surface of $M'$.
\item[\emph{Causality}]
	The images of morphisms $\F(\chi_i)\colon \F(M_i)\to \F(M)$ are
	independent subsystems of $\F(M)$, whenever the images of morphisms
	$\chi_i\colon M_i\to M$, $i=1,2$, are spacelike separated in $M$.
\end{enumerate}
\end{definition}
The Causality axiom may also be rephrased in terms of tensor products.
We can introduce the disjoint union $(M_1,g_1)\sqcup (M_2,g_2) =
(M_1\sqcup M_2, g_1\oplus g_2)$ as a tensor product in $\GlobHyp_c$. It
is based on the disjoint union (coproduct) or manifolds, such that
disconnected components are considered spacelike separated from each
other. As for the independent subsystems product on Poisson algebras,
the canonical dotted line morphism in the following diagram does not
always exist:
\begin{equation}
\xymatrix{
	(M_1,g_1) \ar[dr] \ar[r] & (M_1,g_1)\sqcup (M_2,g_2) \ar@{-->}[d]
		& \ar[l] \ar[dl] (M_2,g_2)\\
	& (M,g)
} .
\end{equation}
It exists only if the images of $M_i$ are spacelike separated in $M$.
Thus, the Causality axiom simply states that the LCFT functor $\F\colon
\GlobHyp_c \to \Poiss_i$ is a tensor functor.

It should also be clear by now that the above formalism explicitly ties
the causal structure of a field theory to a fixed (non-dynamical)
Lorentzian metric. A Lorentzian metric provides a notion of a causal
relation, which appears crucially in each of the named axioms. What
about field theories that have a dynamical Lorentzian metric or have no
Lorentzian metric naturally associated to them at all? It is worth
taking a step back and examining why a Lorentzian metric is important in
the first place.

\subsection{Linear and semilinear PDEs}
We start by considering examples of field theories that can be shown to
actually satisfy the LCFT axioms.  There are two prominent families of
examples: free waves field and globally well-posed, interacting wave
fields without derivative couplings. Prototypical representatives of
each family are the free scalar field, obeying the Klein-Gordon equation
\begin{equation}
	\square \phi - m^2\phi = 0 ,
\end{equation}
and the $\phi^4$-interacting scalar field,
obeying the semilinear equation
\begin{equation}
	\square \phi - m^2\phi - \lambda \phi^3 = 0 .
\end{equation}

By \emph{wave field}, we mean a field theory, defined on a Lorentzian
manifold, whose equations of motion have the corresponding d'Alambertian
wave operator $\square$ as the principal symbol. The \emph{principal
symbol} is roughly the differential operator consisting of the highest
derivative order terms of the PDE (cf.~Eq.~\ref{eq:prsym}). For the wave
operator $\square$, the principal symbol is essentially $g^{\mu\nu}
\del_\mu \del_\nu$, where $g^{\mu\nu}$ is the inverse Lorentzian metric,
which is contracted with two coordinate derivatives. As is usual, the
term \emph{free} means that the corresponding equations of motion are
\emph{linear}. The absence of \emph{derivative couplings} is taken to
mean the absence in the equations of motion of non-linear terms with
derivatives. Such equations are termed \emph{semilinear}%
	\footnote{Actually, semilinear equations allow non-linear terms with
	derivatives, as long as they are not of the highest present order.}. %
Finally, \emph{well-posedness} is taken to mean, in particular, that
solutions corresponding to arbitrary smooth initial data (obeying
appropriate boundary conditions, of course) do not form singularities in
finite time.  Well-posedness of linear equations can be proven under
quite general assumptions, which is not the case for semilinear ones.

These field theories depend on the background Lorentzian metric
$g_{\mu\nu}$ through the principal symbol of their equations of motion.
The exclusion of interactions with derivative couplings, ensures that
this principal symbol is always the same as that of the wave operator
$\square$. Since the metric is Lorentzian, these partial differential
equations are special cases of hyperbolic PDE systems. Roughly,
\emph{hyperbolicity} is an algebraic and geometric property of the
principal symbol, which is exploited by standard PDE theory to prove
analytical results like (a) local-in-time well-posedness of the Cauchy
initial value problem (this entails the existence of solutions for
arbitrary smooth initial data, uniqueness of solutions, continuous
dependence of the solution on the initial data), (b) finiteness of the
speed of propagation of disturbances. In particular, as a generic
result, the maximal speed of propagation in any direction happens to be
given by the null directions of the background Lorentzian metric. Global
well-posedness can also be proven under generic circumstances for linear
PDEs and under some additional hypotheses for semilinear ones.

These results are important because, roughly, global well-posedness and
finite propagation speed jointly imply the validity of the Isotony, Time
Slice and Causality axioms for the corresponding classical field theory.
The Causality axiom, and hence the notion of causal structure, is
particularly closely related to finite propagation speed. We shall see
in Sect.~\ref{sec:semilin-lcft} in more detail how these results from PDE
theory can be used to construct LCFTs and verify these axioms.

Following the logic of the last paragraph, it makes sense to look for a
generalization of the notion of causal structure, and corresponding
generalizations of LCFT axioms, by studying the way PDE theory
establishes the properties of well-posedness and finite propagation
speed in more general hyperbolic systems.  Of particular interest are
PDE systems where these properties are not directly tied to a
background, non-dynamical Lorentzian metric. As will be shown in the
following sections, such generalizations can be found in the PDE theory
of quasilinear hyperbolic systems, of which GR is a special case.

\section{Quasilinear Hyperbolic Systems}
\label{sec:hypersys}
In this section we formulate several aspects the theory of systems of
partial differential equations (PDEs) in a geometrical way, namely in
terms of jets. This approach is not completely standard in the analysis
or mathematical physics literatures, but it does have some advantages.

Eventually, we would like to make contact with locally covariant field
theory, which assigns algebras to spacetime regions in a functorial way
(a dif\-fe\-o\-morph\-ism-covariant way). We would like to construct these
algebras as algebras of functions on spaces of solutions of some PDEs.
Thus we would like to assign PDE systems to manifolds in a functorial
way as well. It turns out that this is more conveniently done using
natural bundles (to be introduced in Sect.~\ref{sec:natural}) and jet
bundles. Moreover, it is also more convenient to discuss integrability
conditions and equivalence of PDE systems in the language of jets.
Finally, the variational bicomplex on jet space is the natural setting
for conservation laws and the covariant symplectic structure (to be
discussed in Sect.~\ref{sec:classquant}).

Within this framework, we define PDE systems of quasilinear hyperbolic
type. We rely mostly on the notion of \emph{symmetric hyperbolicity} of
first order systems (following~\cite{geroch-pde}), but also briefly comment on
the related notion of \emph{regular hyperbolicity}.

From this point, we will be discussion finite dimensional spacetime
manifolds, bundle manifolds and infinite dimensional manifolds of
sections and solutions, as well as morphisms between them. For later
convenience and to fix some notations it is helpful to define the
following categories.
\begin{definition}
Let $\Man$ denote the category of smooth manifolds (of possibly infinite
dimension) and $\Man_e$ the subcategory where morphisms are restricted
to open embeddings. The morphism are always $C^\oo$ maps. Denote also by
$\Man^n$ and $\Man^n_e$ the respective subcategories of manifolds of
fixed dimension $n=0,1,2,\ldots,\oo$. The finite dimensional manifolds
are taken to be oriented.

Let $\Bndl$ denote the category of smooth, finite dimensional bundles
over manifolds and $\Bndl_e$ be the subcategory where morphisms are
restricted to open embeddings of the total space. These categories are
fibered over manifolds, $\base\colon \Bndl \to \Man$, with $\base\colon
(E\to M) \mapsto M$. Let $\Bndl(M)$ denote the subcategory of $\Bndl$ of
bundles over the base manifold $M$ with \emph{base fixing} morphisms
(projection to the base is always $\id:M\to M$). Thus, in an obvious
way, we have defined a functor from manifolds to categories,
$\Bndl\colon \Man\to \Cat$.

Let $\VBndl$ denote the category of vector bundles, a subcategory of
$\Bndl$, and $\VBndl_e = \VBndl\cap \Bndl_e$. We similarly introduce the
notations $\VBndl(M)$ and $\VBndl_e(M)$.
\end{definition}
We shall adopt $\base$ as the generic notation for forgetful functors.
More generally, from time to time, we shall reuse functor names as long
as they can be distinguished by the domain category and are similar in
purpose.

We shall not be specific about the precise definition of the class of
infinite dimensional manifolds in question. We presume that such details
can be added when applying the current formalism to specific situations.
Such functional analytical questions have been recently tackled in
earnest in~\cite{fr-bv,bfr} (cf.~also~\cite{bsf}).

A prominent example of an infinite dimensional space that we will treat
as a manifold is the space $\Secs(E)$ of all smooth sections of the
vector bundle $E\to M$, as well as some infinite dimensional subspaces
thereof, like the space of all solutions of a PDE system.

\subsection{Jet bundles and systems of PDEs}\label{sec:jets-pdes}
This section outlines the description of PDEs as submanifolds of the jet
bundle. Jet bundles are briefly introduced in Sect.~\ref{sec:jets},
where also notation is fixed (not all of it being completely standard)
and standard literature references are given.  Such a description of
PDEs is more intrinsic than than the usual one in terms of equations,
but is essentially equivalent. This approach is well known in the
geometric and formal theory of differential
systems~\cite{seiler-inv,bcggg}. It will be later combined with the
notion of natural bundles to define locally covariant field theories.
The discussion of natural PDE systems in connection with natural bundles
will be postponed to Sect.~\ref{sec:natural}. Below we consider only
fixed base manifolds.

From now on, fix $M$ to be finite dimensional manifold and let $n=\dim
M$. Also fix a vector bundle $F\to M$. We refer to $M$ as the
\emph{spacetime} manifold and to $F$ as the \emph{field bundle}.

We restrict our attention to \emph{regular} PDEs in the following sense.
\begin{definition}
A \emph{PDE system $\E$ of order $k$} is a smooth, closed sub-bundle (in
the $\Bndl$ sense) of $J^kF\to M$, $\E\sso J^kF$.
\end{definition}
Note that $\E$ need not be a vector sub-bundle of $J^kF$. The above
definition may seem unfamiliar to some, but can be cast in more
recognizable form using the following
\begin{proposition}
Given a PDE system $\E$ of order $k$, there exists a vector bundle $E\to
M$, a smooth sub-bundle $E'\sse E$ containing the zero section of $E\to
M$, and a smooth base fixing bundle morphism $f\colon J^kF \to E$ (in
the $\Bndl$ sense) such that the image of $f$ is contained in $E'$, the
image of $f$ is transverse in $E'$ to the zero section of $E$ and $\E$
is precisely the preimage of the zero section, that is, $\E$ satisfies
$f=0$.
\end{proposition}
The proof follows from basic differential topology, up to a global
topological obstruction~\cite[\textsection 7]{goldschmidt}. Clearly, the
equation form is not unique. For instance, applying any invertible
transformation to the equations $f=0$ gives another equation form $f'=0$,
which describes exactly the same PDE system. 

We refer to $E\to M$ as the \emph{equation bundle} and to $f$ or the
pair $(f,E)$ as the \emph{equation form} of the PDE system $\E$. A
section $\phi\colon M \to F$, also referred to as a \emph{field
configuration}, is said to satisfy the PDE system $\E$ if the $k$-jet
prolongation of $\phi$ is contained in $\E$, $j^k\phi(x)\in \E_x\sso
J^k_x(F,M)$. Equivalently, $j^k\phi$ is a section of $\E\to M$. We
denote the space of all solution sections by $\S(F)\sso \Secs(F)$ or
$\S_\E(F)$ when the PDE system needs to be mentioned explicitly. Using
the above proposition, we can equivalently say that $\phi$ is a solution of
the PDE system $\E$ if
\begin{equation}
	f[\phi] = f(j^k\phi) = 0.
\end{equation}
Expressing the $k$-jet in local coordinates, $j^k\phi(x) = (x,\phi^a(x),
\del_i \phi^a(x), \ldots)$, it is clear that $f(x,\phi^a(x),
\del_i\phi^a(x), \ldots)=0$ is a system of partial differential
equations in the usual sense of the term. Starting with a PDE system in
the usual sense, its geometric form as a sub-bundle of the jet bundle
can be obtained by a converse of the above lemma. At this point, the
regularity assumptions on both $\E$ and $f$ become important. Namely,
the transversality properties of $f$ ensure that the zero set of $f=0$
is a submanifold of $J^kF$ and vice versa.

The linear and affine structures on $J^kF$ give us the possibility of
defining the notion of \emph{linear} and \emph{quasilinear} PDE systems.
\begin{definition}
A PDE system $\E\sso J^kF$ is called \emph{linear} if $\E\to M$ is a
vector sub-bundle of the vector bundle $J^kF\to M$. The PDE system is
called \emph{quasilinear} if $\E\to J^{k-1}F$ is an affine sub-bundle of
the affine bundle $J^kF\to J^{k-1}F$.
\end{definition}
The connection to the usual meanings of these terms can be seen through
adapted equation forms.
\begin{lemma}
The PDE system $\E\sso J^kF$ is linear iff it has an equation form
$(f,E)$, where $f\colon J^kF\to E$ is a morphism of vector bundles over
$M$.

The PDE system $\E\sso J^kF$ is quasilinear iff it has an equation form
$(f,E)$, where $f\colon J^kF\to E$ is a morphism of affine bundles,
which fits into the commutative diagram
\begin{equation}
\vcenter{\xymatrix{
	J^kF \ar[d] \ar[r]^f & E \ar[d] \\
	J^{k-1}F    \ar[r]   & M ~ ,
}}
\end{equation}
where the vertical maps define the affine bundles, with the vector bundle
$E\to M$ naturally considered an affine one.
\end{lemma}
The proof is immediate. Alternatively, the quasilinear case can be cast
into the form of a base fixing affine bundle morphism $f\colon J^kF\to
(E)^{k-1}$, where both bundles are over $J^{k-1}F$. Such equation forms
are called \emph{adapted}.

In the more common language of adapted local coordinates, the conditions
of linearity and quasilinearity are expressed as follows. Consider
adapted local coordinates $(x^i,v_A)$ on the equation bundle $E$,
$(x^i,u^a)$ on the field bundle $F$, and the corresponding $(x^i,u^a_I)$
on the $k$-jet bundle $J^kF$.
Let
$\phi\colon M\to F$ be a field configuration, then its $k$-jet in local
coordinates is $j^k\phi(x) = (x^i,\del_I\phi^a(x))$. The
above lemma asserts the existence of an equation form that looks like
\begin{equation}\label{lin-eqform}
	f_{Aa}^I(x) \del_I\phi^a(x) = 0 .
\end{equation}
Note that this equation is linear in $\phi(x)$ and its derivatives and
that the coefficients $f_{Aa}^I(x)$, with multi-indices $I$, depend only
on the base space coordinates $x$. On the other hand, for a quasilinear
equation, the lemma asserts the existence of an equation form that looks
like (with $|I|=k$)
\begin{equation}\label{qlin-eqform}
	f_{Aa}^I(x,j^{k-1}\phi(x)) \del_I \phi^a(x) + f_A(x,j^{k-1}\phi(x)) = 0 .
\end{equation}
Note that, in the linear case, the fact that the coefficients of
$f\colon J^k(F,M)\to E$ only depend on the base space coordinates $x$ is
captured by the requirement that it is a morphism of vector bundles over
$M$. In the quasilinear case, the coefficients of $f$ can obviously
depend on both $x$, $\phi(x)$ as well as all derivatives $\del_I\phi(x)$
up to order $|I|=k-1$, which is captured by allowing $f\colon J^kF\to
(E)^{k-1}$ to be a (base fixing) bundle morphism over $J^{k-1}F$. It is
worth remarking that any linear PDE system is also naturally
quasilinear.

Recall that the affine bundle $J^kF\to J^{k-1}F$ is modeled on
the vector bundle $(S^kT^*M\otimes_M F)^{k-1}\to J^{k-1}F$.
Therefore, an adapted equation form $(f,E)$ of a quasilinear PDE system
$\E\sso J^kF$ naturally singles out a section
\begin{equation}\label{eq:prsym}
	\bar{f}\colon J^{k-1}F
		\to (E\otimes_M F^* \otimes_M S^kTM)^{k-1} .
\end{equation}
In local coordinates, $\bar{f}$ corresponds to the coefficient
$f_{Aa}^{I}$ of the highest derivative term $\del_{I}\phi(x)$ with
$|I|=k$ in Eq.~\eqref{qlin-eqform}. This section $\bar{f}$ is called the
\emph{principal symbol} of the given equation form of $\E$.

Below, we will be mostly concerned with first or second order PDE
systems.

\subsection{Symmetric and regular hyperbolicity}\label{sec:symhyper}
The notion of hyperbolicity for a PDE system is strongly linked to the
ability to formulate it as an initial value problem. Locally such a
formulation is constrained by the existence of so called
\emph{characteristic surfaces} or \emph{characteristic covectors}, which
are defined in detail below. Various notions of hyperbolicity are then
stated either in terms of geometric conditions on the locus of
characteristic covectors or in terms of equivalent algebraic conditions
on the principal symbol of a special equation form of the system.

An initial value formulation consists of converting the PDE system into
an infinite dimensional ODE as follows. Suppose we are given an
$n$-dimensional manifold $M$ that can be smoothly factored as $M\cong
\R\times S$, where $S$ is a manifold of dimension $n-1$, and the
projection on the first factor is denoted by $t\colon M\to \R$, the
\emph{time function}. Then the space of restrictions of field
configurations to a level set of $t$ (always a codimension-$1$ surface
diffeomorphic to $S$) forms an infinite dimensional space, which can be
considered as a fiber of a smooth, infinite dimensional bundle over
$\R$. If the PDE system can be turned into an ODE system on this bundle,
one can use Picard iteration and a Gronwall-type lemma to prove
local-in-$t$ existence and uniqueness of solutions. Various notions of
hyperbolicity essentially correspond to sufficient conditions under
which the above construction can be carried out. Some of these
conditions are local (referred to simply as \emph{hyperbolicity}) and
some global (referred to as \emph{global hyperbolicity}). Local
conditions put restrictions on the principal symbol of a PDE system and
on the germs of the leaves of the $S$-foliations of $M$ to which the
above construction would be applicable. The global conditions require
$M$ to factor in a way similar to above, with the $S$-leaves satisfying
the corresponding local conditions. Strictly speaking, when the
$S$-leaves are non-compact or have a boundary, boundary conditions may
have to be supplied, on top of those already mentioned, to ensure a
well-posed initial value problem. Such boundary conditions will not be
discussed in this work.

To obtain some necessary conditions on the geometry of initial value
surfaces, we need to first look at the linear and affine geometry of
$k$-jet space in the presence of a preferred codimension-$1$ subspace
$\tau_x \sso T_xM$ at a fixed $x\in M$. This subspace can be thought of
as tangent to a putative initial value surface $S$ passing through $x$.
For convenience of notation, we promote $\tau_x$ to be the fiber of a
codim-$1$ vector sub-bundle $\tau \sso TM$ (also called a \emph{tangent
plane distribution}). However, all subsequent calculations will be
purely local and go through equally well if $\tau$ is only defined in a
neighborhood of $x$, or just at $x$ itself. In particular, the
distribution $\tau$ need not be integrable.

Let $\nu\sso T^*M$ be the $1$-dimensional vector sub-bundle of covectors
annihilating $\tau$ (the \emph{conormals} to the putative surface $S$).
In turn, the $1$-dimensional vector bundle $\nu$ singles out the
$1$-dimensional vector sub-bundle $\nu^{\otimes k} =
\nu\otimes_M\nu\otimes_M \cdots \sso S^kT^*M$, and finally the vector
sub-bundle $N=\nu^{\otimes k}\otimes_M F \sso S^kT^* M\otimes_M F$.
Recall that the affine bundle $J^kF\to J^{k-1}F$ is modeled on the
vector bundle $(S^kT^*M \otimes_M F)^{k-1}$. Note that the fibers of
$J^kF\to J^{k-1}F$ are foliated by affine planes parallel to the fibers
of the vector bundle $(N)^{k-1}$. We define $J^{k-1,\perp}$ to be the
leaf space of this foliation. In other words, the short exact sequence
of vector bundles
\begin{equation}
	0 \to N \to S^kT^*M\otimes_M F \to (S^kT^*M\otimes_M F)/N \to 0
\end{equation}
induces the following affine bundle projections
\begin{equation}
	J^kF \to J^{k-1,\perp}F \to J^{k-1}F,
\end{equation}
where the second projection defines an affine bundle with fibers modeled
on the vector bundle $(S^kT^*M \otimes_M F/N)^{k-1}$ and the first
projection defines an affine bundle with fiber modeled on the vector
bundle $(N)^{k-1,\perp}$, with the latter notation meaning the pullback
of the bundle $N\to M$ along the canonical projection $J^{k-1,\perp}\to
M$. A more in depth discussion of the bundle $J^{k,\perp}F$ can be found
in~\cite[Ch.6]{spring-convex}.

The above constructions are easily illustrated in local coordinates.
Consider local coordinates $x^j = (t,s^i)$ on a neighborhood of $x\in M$ and
local adapted coordinates $(t,s^i,u^a)$ on the corresponding
neighborhood of $F_x\sso F$. Suppose that the $t$ coordinate is chosen
such that the level set $t=t(x)$ is tangent to the plane $\tau_x\sso
T_xM$, in other words $dt(x)\in\nu_x\sso T^*_xM$. The $k$- and
$(k-1,\perp)$-jets of a section $\phi\colon M\to F$ can be represented
by
\begin{align}
	j^k\phi(x)
	&= (x^j, \del_J \phi^a(x)), ~~ |J|\le k, \\
	&= (t,s^i, \del_{I}\del_t^l \phi^a(x)), ~~ |I|+l \le k , \\
	&= (x^j, \del_{I_0}\phi^a(x), \del_t \del_{I_1} \phi^a(x),
			\ldots, \del_t^{k-1} \del_{I_{k-1}} \phi^a(x),
			\del_t^k \phi^a(x)) , \\
\notag
	&\qquad \qquad |I_l| \le k-l ,
\end{align}
where the multi-indices in $J$ range over all indices of $x^j$, while
the multi-indices of $I$ and $I_l$ range only over the indices of $s^i$
and not $t$, and by
\begin{align}
	j^{k-1,\perp}\phi^a(x)
	&= (t,s^i, \del_{I_0}\phi^a(x), \del_t\del_{I_1} \phi^a(x),
			\ldots, \del_t^{k-1} \del_{I_{k-1}} \phi^a(x)) , \\
\notag
	&\qquad \qquad ~~ |I_l| \le k-l .
\end{align}
In other words, the affine bundle projection $j^k\phi(x)\mapsto
j^{k-1,\perp}\phi(x)$ simply discards the last component, which is the
highest order derivative $\del^k_t \phi^a(t,s)$ along the $t$-direction,
which is transverse to the level set of $t$ passing through $x$.

In order to convert the PDE system into an ODE with respect to the $t$
coordinate, still working locally at $x\in M$, 
$S$, given a section $\phi\colon M\to F$ that satisfies the PDE system,
$j^k\phi(x)\in\E_x$, we need to uniquely determine the highest order
$t$-derivative of the dynamical part of $\phi$ at $x$ as a function of
all other derivatives of equal or lower order%
	\footnote{In the classical PDE literature this is known as putting the
	equation in \emph{Cauchy-Kovalevskaya form}~\cite{evans-pde,ch2}.}. %
So, geometrically, each fiber of the affine bundle $(J^kF)^{k-1}\to
(J^{k-1,\perp}F)^{k-1}$ needs to intersect $\E$ exactly once.  On the
other hand, to express this condition algebraically, we need to pick an
adapted equation form $(f,E)$ for this quasilinear PDE system, with
principal symbol $\bar{f}$. Pick sections $u\colon M\to J^kF$,
$u^\perp\colon M\to J^{k-1,\perp}F$ and $\bar{u}\colon M\to J^{k-1}F$
such that $u\mapsto u^\perp$ and $u\mapsto \bar{u}$ under the affine
bundle maps $J^kF\to J^{k-1,\perp}F$ and $J^kF\to J^{k-1}F$
respectively.  Then by the affine structure of the $k$-jet space, any
other $k$-jet section $v$ that projects to $u^\perp$ is given by
$v=u+p^{\otimes k}\otimes\psi$, with sections $\psi\colon M\to D$ and
$p\colon M\to \nu$.  If we demand that $v$ is actually a section of
$\E\sso J^kF$, then a simple calculation shows that
\begin{align}
	f(v) &= \bar{f}_{\bar{u}}(p^{\otimes k}\otimes \psi) + f(u) = 0 \\
	\implies ~ (\bar{f}_{\bar{u}}\cdot p^{\otimes k})\psi &= -f(u).
\end{align}
The last equality is a linear equation for $\psi$. For it to have a
unique solution, as we demanded above, the linear map
$\bar{f}_{\bar{u}}\cdot p^{\otimes k}\colon F \to E$ needs to be invertible.
This discussion motivates the following definition.
\begin{definition}
Consider a vector bundle $F\to M$ and a quasilinear PDE system $\E\sso
J^kF$, with adapted equation form $(f,E)$ and principal symbol
$\bar{f}$. Given $\bar{u}\in J^{k-1}F$, a covector $p\in (T^*M)^{k-1}_{\bar{u}}$
is called \emph{$\bar{u}$-non-characteristic} (or just
\emph{non-characteristic}) if the contraction of the principal symbol
with $p^{\otimes k}$, $(\bar{f}_{\bar{u}}\cdot p^{\otimes k})\colon
(F)^{k-1}_{\bar{u}} \to (E)^{k-1}_{\bar{u}}$, is invertible. Otherwise, $p$ is
said to be \emph{$\bar{u}$-characteristic} (or just \emph{characteristic}).
\end{definition}
Recall that in local coordinates $(x^i,u^a)$ on $F$, a $(k-1)$-jet
$\bar{u}\in J^{k-1}F$ is represented as $\bar{u} = (x^i,u^a_I)$, $|I|\le k-1$.
For linear and semilinear systems, the principal symbol depends only on
the base space coordinates $x^i$. This implies that for such systems a
covector $p\in T^*_xM$ can be decided to be characteristic without
looking at the field value $u^a$ or higher jet components $u^a_I$. This
is the usual situation in relativistic field theory with a fixed
background metric, where characteristic covectors coincide with
\emph{null} covectors of the background metric.  On the other hand, for
quasilinear PDE systems, the principal symbol $\bar{f}_{\bar{u}}$ may depend
on $u^a$ as well as higher jet components $u^a_I$ up to order $k-1$. In this
sense, the notion of a characteristic covector becomes \emph{field
dependent}.

As we shall see later on, the geometry of the locus of characteristic
covectors of a PDE system is closely related to the causal structure of
the corresponding classical field theory, in particular to the domain of
dependence and finite propagation speed results.  Compare now the
equations of motion of relativistic field theory (say the Standard
Model) with a fixed background metric and coupled with GR, which
provides a dynamical metric. As remarked above, characteristic vectors
and hence the causal structure of the theory on a fixed background is
field independent, since the equations of motion constitute at most a
semilinear PDE system. On the other hand, coupled to GR, relativistic
field theory becomes quasilinear, since all principal symbols depend on
the metric, which is dynamical. Thus, GR with any of the Standard Model
matter coupled to it constitutes a system with \emph{field dependent
causal structure}. This statement will be made more precise in the next
section. At the very least, we expect the causal structure of GR to be
significantly different (and more complicated) than that of other
relativistic field theories.

Characteristic covectors are obstructions to converting a PDE system
into ODE or Cauchy-Kovalevskaya form. An initial value formulation with
data on a codim-$1$ surface $\iota\colon S\sso M$ can be achieved if the
highest $S$-transverse derivatives of the unknown section could be
solved for in terms of data $\varphi\colon S\to \iota^*J^{k-1}F$. For
this it is necessary that a non-vanishing conormal section $p\colon S\to
\iota^*T^*M$ is everywhere non-characteristic with respect to $\varphi$.
Such a pair $(S,\varphi)$ is referred to as \emph{non-characteristic
initial data}.

However, being non-characteristic is not a sufficient property to set up
a well-posed initial value problem for a given initial data set
$(S,\varphi)$. This is where the notion of \emph{hyperbolicity} comes
in. There are several different notions of hyperbolicity. We will only
discuss two, which are sufficient for our purposes. First, we define
\emph{symmetric hyperbolicity} and then make some comments about
\emph{regular hyperbolicity}.

For a first order quasilinear system, the geometry of the principal
symbol simplifies. It can be expressed as a morphism
\begin{equation}
	\bar{f}\colon (F \otimes_M T^*M)^0 \to (E)^0
\end{equation}
of vector bundles over $J^0F\cong F$.  Given a section $p\colon F\to
(T^*M)^0$ and $\psi\colon F\to (F)^0$, we get a section
$\xi=\bar{f}(p\otimes \psi) = (f\cdot p)\psi$ of $(E)^0\to F$.  In local
coordinates $(x^i,u^a)$ on $F$ and $(x^i,v_A)$ on $E$, we have
\begin{equation}
	\xi_A(x,u) = \bar{f}^{i}_{Aa}(x,u) \, p^i(x,u) \psi^a(x,u) .
\end{equation}

For the next definition, we introduce some new notation. Recall that
$\Lambda^n M\to M$ is the bundle of densities on $M$. Given any vector
bundle $V\to M$, we denote the bundle of \emph{$V$-valued densities} by
$\tilde{V} = \Lambda^nM\otimes_M V$. The bundle of \emph{dual densities}
is denoted by $\tilde{V}^* = \Lambda^nM\otimes_M V^*$. \emph{Densitized}
symmetric powers will be denoted by $\tilde{S}^k V = \Lambda^n
M\otimes_m S^k V$. The defining property of an orientable manifold is
the existence of nowhere vanishing sections of $\Lambda^n M\to M$. An
oriented manifold chooses a privileged class of \emph{positive
densities}, which are nowhere vanishing and such that any two positive
densities are related through multiplication by an everywhere positive
scalar function.  \begin{definition}\label{def:symhyp} Consider a first
order, quasilinear PDE system on $F\to M$, with $M$ oriented and with
adapted equation form $(f,\tilde{F}^*)$ and principal symbol $\bar{f}$.
This PDE system is said to be \emph{symmetric hyperbolic} if for each
$\bar{u}\in F$ there
exists a covector $p\in (T^*M)^0_{\bar{u}}$ such that
\begin{equation}
	\bar{f}_{\bar{u}}\cdot p \in (\tilde{S}^2F^*)^0_{\bar{u}}
		\sso (\Lambda^nM\otimes_M F^*\otimes_M F^*)^0_{\bar{u}} ,
\end{equation}
that is, $\bar{f}_{\bar{u}} \cdot p$ is a \emph{symmetric} bilinear form
on the fiber $(F)^0_{\bar{u}}$, and moreover that $\bar{f}_u\cdot p$ is
\emph{positive definite} with respect to the orientation on $M$.  Such a
covector $p$ called \emph{$\bar{u}$-spacelike} (or just
\emph{spacelike}) and \emph{future oriented}, while $-p$ is
\emph{$\bar{u}$-spacelike} but \emph{past oriented}.
\end{definition}
Decoding the above definition in local coordinates $(x^i,u^a)$ on $F$,
we have that $p_i \bar{f}^i_{ab}(x,u) \psi^a \psi^b > 0$ (as an element
of $(\Lambda^n_x M)^0_{(x,u)}$, with respect to the orientation on $M$)
with $p\in (T^*M)^0_{(x,u)}$ and for all $\psi\in (F)^0_{(x,u)}$.

Note that positive definiteness is stronger than invertibility. So every
spacelike covector is also non-characteristic, but not every
non-characteristic covector is necessarily spacelike. An \emph{initial
data set} $(S,\varphi)$, with $\iota\colon S\sso M$, is called
\emph{spacelike} if a non-vanishing conormal section $p\colon S\to
\iota^*T^*M$ is everywhere $\varphi$-spacelike. It is for spacelike
initial data that the theory of symmetric hyperbolic PDEs establishes
local well-posedness.

\begin{remark}
Our definition of a symmetric hyperbolic system is slightly different
from the standard one, where the principal symbol is valued in $S^2 F^*$
rather than the densitized $\tilde{S}^2 F^*$. For a manifold with a
fixed orientation, as we are considering, the difference consists of
tensoring with a positive density. Densitizing the equation form turns
out to be more convenient in the current setting, in particular in
Sect.~\ref{sec:classquant}. However, densitization is also natural in
the original context where symmetric hyperbolicity is used, the
construction of an energy norm by integrating over a future oriented,
spacelike surface $\iota\colon\Sigma\sso M$, $\|\psi\|^2 = \int_\Sigma
\iota^*\tr \bar{f}(\psi,\psi)$, where the trace naturally converts the
principal symbol into a current density (the single contravariant index
of $\bar{f}$ is contracted with one of its $n$ antisymmetric covariant
ones), which can be naturally integrated when pulled back to the
codim-$1$ surface $\Sigma$.
\end{remark}
\begin{remark}
Note also that our usage of the word \emph{spacelike} differs from the
standard one one in Lorentzian geometry, where our spacelike covectors
would be referred to as ``timelike''.  However, this terminology is
consistent with the literature on hyperbolic PDEs. Moreover, outside
pseudo-Riemannian geometry, there is no natural identification between
vectors and covectors. On the other hand, covectors are still naturally
identified with codimension-$1$ tangent planes. As expected a spacelike
plane consists of spacelike vectors we merely extend this terminology to
a corresponding covector. The terminology introduced for vectors below
is standard in either literature.
\end{remark}

It is convenient to introduce the following
geometric notions as well.
\begin{definition}
Consider $\bar{u}\in F$. A tangent vector $v\in (TM)^0_{\bar{u}}$ is
called \emph{causal} and \emph{future directed} if $p\cdot v\ge 0$ for
every spacelike and future oriented $p\in (T^*M)^0_{\bar{u}}$. If the
inequality is reversed, $p\cdot v\le 0$, then $v$ is called
\emph{causal} and \emph{past directed}.
\end{definition}

As noted above, the notion of symmetric hyperbolicity applies only to
first order, quasilinear systems. This is not too strong of a
restriction in practice, as often higher order PDE systems can be
reduced to first order ones by introducing extra fields (increasing the
dimension of the field bundle fibers). In particular, this exercise was
carried out for GR, all Standard Model fields, as well as relativistic
hydrodynamics in~\cite{geroch-pde}. On the other hand, many equations obtained as
the Euler-Lagrange equations of an action are not directly in symmetric
hyperbolic form. Such equations can actually be conveniently treated
directly, without reduction to first order form. The relevant notion
which comes with a well-posedness theory is \emph{regular
hyperbolicity}, which was developed much more recently by Christodoulou
in~\cite{chr-pde}. It would be ideal if the
class of regularly hyperbolic systems, upon reduction to first order
form, were included in the class of symmetric hyperbolic ones.
Unfortunately, this does not appear to be the case due to subtle
differences between the two definitions~\cite{bs-hyper}. On the other
hand, the similarities between symmetric and regular hyperbolicity
include similar definitions of characteristic and spacelike covectors
and similar conal properties of the geometry of their loci, which are
explored in the next section. These properties are all we need for the
purposes of this paper. Therefore, in the sequel, we refer only to
results for symmetric hyperbolic systems, even though similar results
can be obtained for regular hyperbolic ones.

At this point it is worth making a few comments on the so-called
\emph{energy methods} used to establish well-posedness for symmetric and
regular hyperbolic systems of order $k=1,2$. A central role is played by a
family of local, horizontally conserved, positive definite, coercive
\emph{energy current densities} $\eps^\tau$, parametrized by $\tau>0$.
Locality means that it is a section of the horizontal $(n-1)$-form
bundle over $J^kF$, $\eps^\tau\colon J^{k-1}F\to
(\Lambda^{n-1}M)^{k-1}$. Equivalently, $\eps^\tau \in \Forms^{n-1,0}(F)$
and is projectable to $J^{k-1}F$. Horizontally conserved means that $\d
(j^k\phi)^*\eps^\tau = (j^k\phi)^*\dh \eps^\tau = 0$, when $\phi\colon
M\to F$ is a solution of the given PDE system. This conservation
criterion can actually be relaxed up to lower order terms, that is, we
allow $(j^k\phi)^*\dh \eps^\tau = \gamma(j^{k-1}\phi) +
\tau\beta(j^k\phi)$ to be a function of $\phi$ and its derivatives.
Here, both terms on the right hand side are of lower order: $\gamma$
because it depends only derivatives up to order $k-1$ and $\beta$
because it is proportional to $\tau$, which can be made arbitrarily
small. Positive definite means that the pullback
$\iota^*[(j^k\phi)^*\eps^\tau]$ onto a future oriented, spacelike
codim-$1$ surface $\iota\colon\Sigma \sso M$ gives a positive density,
so that the \emph{energy integral} $E^\tau_\Sigma[\phi] = \int_\Sigma
\iota^*[(j^k\phi)^*\eps^\tau]$ is also positive. Coercivity means that
the algebraic structure of $\eps^\tau$ is such that it allows
$E^\tau_\Sigma[\phi]$ to dominate $\|j^{k-1}\phi\|_\Sigma$, some
$L^2$-norm on the restrictions of $(k-1)$-jets to $\Sigma$, in order
to prove an inequality of Gronwall type,
\begin{equation}
	\|j^{k-1}\phi\|_{\Sigma_t} \le C e^{Ct} \|j^{k-1}\phi\|_{\Sigma_0},
	~ C > 0 , 
\end{equation}
where the surfaces $\Sigma_t$ are spacelike surfaces with common
boundary, whose interiors foliate a domain in $M$. Such domains are
called \emph{lens-shaped} and, as we shall see later on, they
essentially determine domains of dependence.  This inequality is crucial
in the proofs of well-posedness.

The main difference between the symmetric and regular hyperbolic
equations is in how the energy current density $\eps^\tau$ is obtained.
In the symmetric hyperbolic case, $\eps^\tau$ is constructed directly
from the principal symbol, namely $\eps^\tau[\phi] = e^{-t/\tau} \tr
\bar{f}_\phi(\phi,\phi)$, where $t$ is a time function whose spacelike
level sets foliate a lens-shaped domain. In the regular hyperbolic case,
Christodoulou extends N\"other's theorem to convert a timelike vector
field into a corresponding conserved (up to lower order terms)
``symmetry current density'' $\eps$, which satisfies all the desired
properties of an energy current density. See Sect.~5.0
of~\cite{chr-pde}, as well as~\cite{bs-hyper}, for a more
detailed comparison of the two cases.

From the above discussion, it follows that it is not so much the
structure of a special equation form that defines the PDE system that
matters directly. Rather what is important to establish well-posedness
using energy methods is the presence of (almost) conserved, local energy
current densities. The algebraic structure of the PDE system is then
important in so far as it gives rise to such energy current densities.
In the geometric theory of PDE systems (or equivalently of so-called
\emph{differential systems}) the study of such (almost) conserved
currents (or \emph{conservation laws}) has been named
\emph{characteristic cohomology}~\cite{bg-cohom,at,bbh}
(this meaning of the overloaded term
\emph{characteristic} is distinct from its use in \emph{characteristic
covector}). For some PDE systems, these conservation laws can be
classified exhaustively, in an intrinsic manner (that is, independent of
the equation form used to define the PDE system)~\cite{vk}.
This observation leads one to speculate that a more general notion of
hyperbolicity can be defined in terms of intrinsic invariants of a PDE
system given by its characteristic cohomology, such that both symmetric
and regular hyperbolicity become special cases thereof.

\subsection{Prolongation, integrability, hyperbolization}
\label{sec:integrability}
One reason to discuss PDE systems as submanifolds of a jet bundle is
independence of a particular equation form. Any two equation forms are
equivalent if they define the same PDE system manifold. We should
specify our notion of equivalence.
\begin{definition}\label{def:pde-equiv}
Consider two field bundles $F_i\to M$, $i=1,2$, and two PDE systems
$\E_i \sse J^{k_i}F_i$. Denote the corresponding spaces of smooth
solution sections by $\S_i(F_i)$. The PDE systems $\E_1$ and $\E_2$ are
said to be equivalent if there exist bundle morphisms $e_{ij}\colon
J^{l_i}F_i\to F_j$, $i\ne j$, such that
\begin{equation}
	\phi_i \in \S_i(F_i) ~~\text{and}~~
	\phi_j = e_{ij}\circ j^{l_i}\phi_i ~~\text{implies}~~
	\phi_j \in \S_j(F_j) ,
\end{equation}
as well as that $e_{12}\circ j^{l_1}$ and $e_{21}\circ j^{l_2}$ are
mutual inverses when restricted to the solution spaces $\S_1(F_1)$ and
$\S_2(F_2)$.
\end{definition}
We can easily extend the notion of equivalence to equation forms of PDE
systems. In that case two different equations forms that define the same
PDE system manifold are trivially equivalent. Note that neither the
field bundles nor the orders of the PDE systems need to be same for
equivalence to hold.

Let us restrict to the case that will be of importance in a later
section, namely of $F_1=F_2=F$ and $e_{12}$ and $e_{21}$ respectively
equal to the canonical projections $J^{l_1}F\to F$ and $J^{l_2}F\to F$,
which are in a sense trivial. In this case, it is certainly sufficient
that $\E_1 = \E_2$ for equivalence to hold, but it is not necessary. In
fact $\E_1$ and $\E_2$ could be of different orders. To obtain necessary
conditions for equivalence, we need to consider \emph{prolongation} of
PDE systems and the possible resulting \emph{integrability conditions}.

A discussion of these notions in the setting of the jet bundle
description of PDE systems can be rather technical. On the other hand,
the theory of equivalence of PDE systems formulated in these terms has
become quite mature and has yielded some important results. The
technical details of this theory can be found
elsewhere~\cite{bcggg,seiler-inv}.
Below we give a brief non-technical introduction to this theory and
state some simplified results relevant for hyperbolic systems.

The step by step derivation and inclusion of integrability conditions
into a PDE is called prolongation. It is easiest to define prolongation
in equation form and in local coordinates. Consider an equation form
$(f,E)$ of a PDE system $\E\sse J^kF$, as well as local coordinates
$(x^i,u^a)$ on $F$ and $(x^i,v_A)$ on $E$. If the section $\phi\colon
M\to F$ satisfies the PDE system, we have the following system of
equations holding in local coordinates
\begin{equation}
	f_A(x^j,\del_J\phi^a) = 0 .
\end{equation}
These equations hold for each point $x\in M$, therefore when both sides
are differentiated with respect to the coordinates on $M$, the resulting
equations are still satisfied,
\begin{equation}
	\del_i f_A(x^j,\del_J\phi^a)
	= (\hat{\del}_i f_A)(x^j,\del_J\phi^a) = 0 ,
\end{equation}
where $\hat{\del}_i f_A$ are functions on $J^{k+1}F$ obtained by pulling
back the functions $f_A$ from $J^kF$ to $J^{k+1}F$ and applying the
horizontal vector field $\hat{\del}_i$. These new functions $f_{iA} =
\hat{\del}_i f_A$, together with the old $f_A$ ones, constitute the
local coordinate expression for the equation form $(p^1 f,J^1E)$, where
$p^1$ is the $1$-prolongation defined in Sect.~\ref{sec:jets}.
We call the corresponding PDE system $\E^1 = \E_{p^1 f} \sso J^{k+1}F$
the \emph{first prolongation} of $\E$ or also its \emph{prolongation to
order $k+1$}. Prolongations to any higher order, $(p^l f,J^lE)$ and
$\E^{l} \sso J^{k+l}F$, are defined iteratively.

Let $p_l\colon J^{k+l}F\to J^kF$ be the canonical jet projection, which
restricts to $p_l\colon \E^l\to \E$.  Notice that we necessarily have
$p_l(\E^{l}) \sse \E$, since the prolonged system contains the original
one as a subsystem. We have just shown that sections satisfying $\E$
automatically satisfy $\E^{l}$, and vice versa. In other words,
$\S_{\E}(F) = \S_{\E^{l}}(F)$ and the two PDE systems are equivalent.
However, the inclusion $p_l\E^{l} \sse \E$ may be strict, which would
mean that there exist non-trivial \emph{integrability conditions}. There
exists an equation form $(f\oplus g,E\oplus G)$ for $p_l(\E^{l})\sso J^k
F$, where $(g,G)$ is an equation form for the integrability conditions.
These observations provide another sufficient condition for the
equivalence of two PDE systems, namely that there exists an order $l\ge
k_1,k_2$ such that $\E_1^{l-k_1} = \E_2^{l-k_2}$ as subsets of $J^lF$.

Prolongation can be iterated indefinitely. Taking this process to its
limit, we obtain the infinite order prolongation $\E^\oo \sso J^\oo F$
from the equation form $(p^\oo f,J^\oo E)$, which takes all possible
integrability conditions into account. One can then show that the
equality $\E_1^{\oo} = \E_2^{\oo}$, as subsets of $J^\oo F$, is both a
necessary and a sufficient condition for the equivalence of two PDE
systems. It is a deep theorem of the geometric theory of PDE
systems~\cite{seiler-inv,bcggg,goldschmidt}
that, for any given PDE system, there exists a finite order $l$ such
that prolongations above that order introduce no new integrability
conditions. Therefore, this restricted version of the equivalence
problem can be decided in finitely many steps. Finally, it can also be
shown that any PDE system is equivalent to one of first order, though
usually defined on a different field bundle.

It is a well known fact that many PDE systems of mathematical physics
are not given directly in symmetric hyperbolic equation form, though for
somewhat different reasons. The Klein-Gordon equation, though regularly
hyperbolic, in its variational form is not first order. The Dirac
equation, though first order in its variational form, does not have a
symmetric principal symbol. Maxwell equations, in terms of the vector
potential, cannot be hyperbolic because of local gauge invariance (which
spoils uniqueness in the Cauchy problem). The Proca equation, again in
its variational form, Maxwell equations, in terms of the field strength
or in terms of a gauge fixed vector potential, are only equivalent to a
hyperbolic system with additional constraints.

Consider a quasilinear PDE system $\E\sso J^kF$ with an equation form
$(f\oplus c,\tilde{F}^*\oplus E)$, where $(f,\tilde{F}^*)$ is an
adapted, first order, quasilinear, symmetric hyperbolic equation form.
We refer to $(f,\tilde{F}^*)$ as the \emph{hyperbolic subsystem} and to
$(c,E)$ as the \emph{constraints subsystem}. The constraints are said to
be \emph{(symmetric) hyperbolically integrable} if there exists a first
order, linear PDE system $\E'\sso J^1E$ with adapted (symmetric)
hyperbolic equation form $(h,\tilde{E}^*)$, the \emph{consistency
subsystem}, such that the following identity is satisfied: $h\circ c =
q\circ h$, with some differential operator $q\colon J^k\tilde{F}^*\to
E^*$ satisfying $q(0) = 0$. In other words, for any section $\phi\colon
M\to F$ we have
\begin{equation}
	h[c[\phi]] = q[f[\phi]] ,
\end{equation}
and $h[c[\phi]] = 0$ if $f[\phi] = 0$. The consistency subsystem
$(h,\tilde{E}^*)$ is linear in the sense that $h[c[\phi]]$ depends
linearly on $c[\phi]$, but it may depend non-linearly on $\phi$, such
that the \emph{compound} system with equation form $(f\oplus h,
\tilde{F}^*\oplus \tilde{E}^*)$ is first order, quasilinear, symmetric
hyperbolic. A PDE system with such an adapted equation form $(f\oplus c,
\tilde{F}^*\oplus E)$ is called \emph{(symmetric) hyperbolic with
constraints}.  When the constraints are hyperbolically integrable, by
studying the properties of the compound hyperbolic PDE system on
$F\oplus_M E$, we shall see in Sect.~\ref{sec:pde-theory} that all the
well-posedness results for symmetric hyperbolic systems also apply in
the presence of constraints.

We call a PDE system \emph{(symmetric) hyperbolizable} if it is
equivalent to one that is (symmetric) hyperbolic with constraints. We
call the equivalence map that brings a PDE system into such a form a
\emph{(symmetric) hyperbolization}.

\subsection{Examples}
Many examples of reductions of relativistic field theories are given in
Appendix A of~\cite{geroch-pde}. Another source of examples of field theories
described by hyperbolic systems with constraints is~\cite{hs-gauge}. In
the latter reference a slightly different notion of hyperbolicity is
used, but the given examples still fit into our framework provide they
are first reduced to symmetric hyperbolic form.

\section{Chracteristic Geometry, Causality, Domain of Dependence}
\label{sec:chargeom}
The term \emph{Lorentzian geometry} refers to the study of structures
induced on spacetime manifolds by the presence of a Lorentzian metric.
One example of this kind of structure are the cones of \emph{null
vectors}. In particular, it is these cones that determine causal
relationships between points in Lorentzian spacetimes. Below we will be
similarly interested in the geometry of \emph{characteristic covectors}
of a hyperbolic PDE system, which also form cones and also induce a
causal order on the points of a manifold. We refer to the investigation
of the geometry of cones of characteristic covectors and related
structures as \emph{characteristic geometry}. Even more generally, we
are interested in \emph{conal geometry}, which is concerned with the
differential topology of \emph{conal manifolds}~\cite{lawson,neeb}. Each
point of a conal manifold is smoothly assigned an open cone of tangent
or cotangent vectors. The study of conal manifolds, referred to below as
\emph{cone bundles}, though still an immature field, has the potential
to capitalize on and then
subsume much of the earlier work on causal order on spacetime
manifolds~\cite{penrose,geroch-gh,he,gps}. The
generalization from Lorentzian cones to more general ones, for the
purposes of describing causality in quantum field theory has been
considered before~\cite{bannier,rainer1,rainer2}, but not in a concrete way.

Fix a vector bundle $F\to M$ and a first order, quasilinear, symmetric
hyperbolic PDE system $\E\sso J^1F$ on it with adapted equation form
$(f,\tilde{F}^*)$ and a symmetric adapted principal symbol $\bar{f}$. We mostly
follow~\cite{beig-hyper} with the regards characteristic
geometry terminology.

\subsection{Geometry of cone bundles}\label{sec:geom-cones}
\begin{definition}
For $\bar{u}\in F$, denote by $\Gamma_{\bar{u}}^\oast \sso
(T^*M)^0_{\bar{u}}$ the set of \emph{$\bar{u}$-spacelike, future
oriented covectors}. Similarly, denote by $\hat{\C}_{\bar{u}}^*\sso
(T^*M)^0_{\bar{u}}$ the set of \emph{$\bar{u}$-characteristic
covectors}. Let a covector be called \emph{cocausal, future oriented} if
it belongs to $\Gamma_{\bar{u}}^* = \bar{\Gamma}_{\bar{u}}^\oast$, where
the bar denotes closure. Finally, let a covector be called \emph{inner
$\bar{u}$-characteristic} and \emph{future oriented} if it belongs to
$\C_{\bar{u}}^* = \hat{\C}^*_{\bar{u}} \cap \Gamma^*_{\bar{u}}$.
Also, denote
\begin{equation}
	\Gamma^\oast = \bigcup_{\bar{u}\in F} \Gamma_{\bar{u}}^\oast ,
	\quad
	\Gamma^* = \bigcup_{\bar{u}\in F} \Gamma_{\bar{u}}^* ,
	\quad
	\hat{\C}^* = \bigcup_{\bar{u}\in F} \hat{\C}_{\bar{u}}^* ,
	\quad\text{and}\quad
	\C^* = \bigcup_{\bar{u}\in F} \C_{\bar{u}}^* .
\end{equation}
\end{definition}
It is easy to show that $\Gamma^\oast\sso (T^*M)^0$ is an open subset
and hence a submanifold, as well as that $\del\Gamma^* =
\del\Gamma^\oast= \C^*$.  Moreover, each $\Gamma^\oast_{\bar{u}}$ is a
convex \emph{cone} (invariant under multiplication by positive scalars).
Note that, since the interior of an open convex cone is diffeomorphic to
a ball, $\Gamma^\oast\to M$ is actually a smooth bundle, in the sense of
$\Bndl$. On the other hand, $\Gamma^\oast$ has much more structure than
a generic smooth bundle. To take this structure into account we
introduce another category.
\begin{definition}
Let the \emph{category of cone bundles} $\CBndl$ be the subcategory of
$\Bndl$ described as follows. An object $C\to M$ of $\CBndl$, termed a
\emph{cone bundle}, is a smooth bundle for which there exists a vector
bundle $E\to M$ and an inclusion bundle morphism $\iota\colon C \sso
E$, such that each fiber $C_x$, $x\in M$, is an open convex cone in the
corresponding fiber $E_x$ (where we have implicitly identified $C$ with
its image $\iota(C)\sso E$). Given two cone bundles $C\to M$ and $C'\to
M'$, with corresponding enveloping vector bundles $E\to M$ and $E'\to
M'$, for every $\CBndl$ morphism $\chi\colon C\to C'$ there exists a
vector bundle morphism $\psi\colon E\to E'$ such that $\chi = \psi|_C$,
namely the following diagram commutes
\begin{equation}
\vcenter{\xymatrix{
	C \ar[d]^\chi \ar[r]^\sso & E \ar[d]^\psi \\
	C'            \ar[r]^\sso & E' ~ .
}}
\end{equation}
\end{definition}
Thus, $\Gamma^\oast\to F$ is obviously a cone bundle enveloped by
$(T^*M)^0\to F$. We refer to $\Gamma^\oast\to F$ as the \emph{cone
bundle of future oriented, spacelike covectors}. Unfortunately, the
closure $\Gamma^* = \bar{\Gamma}^\oast$ is in general not expected to be
a manifold, since its boundary, consisting of the inner characteristic
covectors $\C^*$, is fiberwise a piecewise algebraic variety and hence
can have corners and other singularities. However, even though
$\C^*,\Gamma^*\sso (T^*M)^0$ are not objects of $\CBndl$ (by lack of
convexity or by presence of a boundary) we refer to them as cone bundles
anyway, namely the \emph{cone bundle of future oriented, inner
characteristic covectors} and the \emph{cone bundle of future oriented,
cocausal covectors}, respectively.

Next we turn to causal and related vectors. These are most easily
described using the following geometric notion of duality that is often
used in convex geometry~\cite{rockafellar}.
\begin{definition}
Given a finite dimensional vector space $V$ and a convex cone $C\sso V$,
denote its closure by $\bar{C}$ and its open interior by $\mathring{C}$.
We define the \emph{convex dual} $C^*\sso V^*$ as the set
\begin{equation}
	C^* = \{ u\in V^* \mid u\cdot v \ge 0 \quad\text{for all}~v\in C \} .
\end{equation}
We define the \emph{strict convex dual} $C^\oast\sso V^*$ as the set
\begin{equation}
	C^\oast = \{ u\in V^* \mid u\cdot v > 0 \quad\text{for all}~v\in
		\bar{C}\setminus\{0\} \} .
\end{equation}
\end{definition}
To attribute \emph{strict} may be dropped from the description of
$C^\oast$ when it is clear from context. It is easy to check the
following
\begin{proposition}
Consider a convex cone $C$.
\begin{enumerate}
\item[(i)]
	The convex dual $C^*$ is always closed and also convex. In addition,
	$C^{**}=\bar{C}$. The strict convex dual $C^\oast$ is always open and
	convex.
\item[(ii)]
	$C^*\setminus\{0\}$ is non-empty iff $C$ is contained in a closed half
	space. $C^\oast$ is non-empty iff $C$ contains no affine line (it is
	\emph{salient}).
\item[(iii)]
	If $C$ is open and salient, then $C^{\oast\oast} = \mathring{C}$.
\item[(iv)]
	The inclusion of cones $C_1\sse C_2$ implies the reverse inclusion of
	their duals, $C_1^*\supseteq C_2^*$ and $C_1^\oast\supseteq C_2^\oast$.
\item[(v)]
	The convex dual of the intersection of closures of cones $C_1$ and
	$C_2$ is the convex union (convex hull of the union) of their duals,
	$(\bar{C}_1\cap \bar{C}_2)^* = C_1^* + C_2^*$, where the right hand
	side is written as a Minkowski sum, which for cones coincides with the
	convex hull of the union. The converse identity holds as well,
	$(\bar{C}_1+\bar{C}_2)^* = C_1^*\cap C_2^*$. Similarly, if $C_1$ and
	$C_2$ are open and salient, then $(C_1\cap C_2)^\oast = C_1^\oast +
	C_2^\oast$ and $(C_1+C_2)^\oast = C_1^\oast \cap C_2^\oast$.
\end{enumerate}
\end{proposition}
For convenience, we extend this duality to cone bundles. That is, if
$C\to M$ is a cone sub-bundle of a vector bundle $E\to M$, then the
\emph{convex dual cone bundle} $C^*\to M$ is the cone sub-bundle of
$E^*\to M$ such that each fiber $C^*_x$ is the convex dual of the
corresponding fiber $C_x$, for $x\in M$. Similarly, we can define the
\emph{strict convex dual cone bundle} $C^\oast\to M$. The operations of
intersection $(\cap)$ and convex union $(+)$ are extended to cone
bundles in the same way. Clearly both $C\to M$ and $C^\oast\to M$ are
cone bundles, that is objects of $\CBndl$. On the other hand, $C^*\to M$
may not be, though again we shall refer to it as a cone bundle anyway.

\begin{definition}
The \emph{cone bundle of future directed, timelike vectors} $\Gamma\sso
(TM)^0$ is defined to be strict convex dual of the cone bundle of future
directed, spacelike covectors, $\Gamma = (\Gamma^\oast)^\oast$. Let
$\bar{\Gamma} \sso (TM)^0$ be called the \emph{cone bundle of future
directed, causal vectors}. Incidentally, $\bar{\Gamma} =
(\bar{\Gamma}^\oast)^*$. Let $\C = \del\Gamma = \del\bar{\Gamma} \sso
(TM)^0$ be called the \emph{cone bundle of future directed, outer ray
vectors}\cite{perlick}. The \emph{past directed} versions are obtained by
reflecting the cones through the origin.
\end{definition}
While the definition of $\Gamma^\oast$ is primitive, it is also
straightforward to show that in fact it is the strict convex dual of
$\Gamma$, that is, $\Gamma^\oast = (\Gamma)^\oast$.

These definitions extend straightforwardly to hyperbolic systems with
constraints. Suppose the adapted equation form $(f\oplus h,
\tilde{F}^*\oplus_M \tilde{E}^*)$ defines the corresponding compound
system, with $(f,\tilde{F}^*)$ the hyperbolic and $(h,\tilde{E}^*)$ the
consistency subsystems. The compound system is quasilinear, but the
consistency subsystem is linear in the constraints (sections of $E\to
M$). Therefore the cones of inner characteristic and spacelike
covectors, as well as outer ray and timelike vectors depend only on the
field bundle $F$ and not on the constraint bundle $E$. Therefore, for a
hyperbolic system on $F\to M$ with hyperbolically integrable
constraints, we define the cone bundles $\C^*$, $\Gamma^\oast$ and
$\Gamma$ over $F$ as those associated to the corresponding compound
system.

Note that these cones are defined intrinsically (independently of any
background metric) by the geometry of a first order, quasilinear,
symmetric hyperbolic PDE system, or equivalently algebraic and geometric
properties of its principal symbol. Of course, the dependence on a
background metric (if one is present) may appear implicitly through the
principal symbol, as it does for linear and semi-linear wave equations.
This remark and the significance of the geometry of $\Gamma$ and
$\Gamma^\oast$ in the domain of dependence theorem to be described later is
sufficient to motivate the fact that such cones are the central objects
of the causal structure of classical field theory. The following
definition leads to the study of conal manifolds and conal geometry
alluded to in the introduction to this section. Many of the results in
conal geometry do not depend on the provenance of these cones from the
characteristic or spacelike covectors of a hyperbolic PDE system.
However, the later results concerning causal relations in PDE and
classical field theories do use that connection.
\begin{definition}
Given a manifold $M$, a \emph{chronal cone bundle} on $M$ is a cone
bundle $C\to M$ (in the sense of $\CBndl$) enveloped by the tangent
bundle $TM\to M$ such that each fiber $C_x$, $x\in M$, is a \emph{proper
cone} (non-empty, open, convex, salient). The elements of $C$ are
\emph{future directed, timelike vectors}.  The strict convex dual
$C^\oast\to M$, enveloped by the cotangent bundle $T^*M\to M$, is the
corresponding \emph{spacelike cone bundle}. The elements of $C^\oast$
are \emph{future oriented, spacelike covectors}. The corresponding cone
bundles $\bar{C}\to M$ and $\bar{C}^\oast\to M$ are referred to,
respectively, as the \emph{causal} and \emph{cocausal} cone bundles.
\end{definition}
Note that the open convex dual of a proper cone is again a proper cone.
Prototypical examples of chronal and spacelike cone bundles are the
pullbacks of the cone bundles of timelike vectors and spacelike
covectors along a section $\phi\colon M \to F$, $C = \phi^*\Gamma$ and
$C^\oast = \phi^*\Gamma^\oast$.
\begin{definition}
Let the \emph{category of chronal cone bundles} $\ChrBndl$ be the
subcategory of $\CBndl_e$ where each object $C\to M$ must be enveloped by
the tangent bundle $TM\to M$ and the morphisms must be restrictions of
the natural maps between tangent bundles induced by base space maps.
Similarly, let the \emph{category of spacelike cone bundles} $\SpBndl$
be the subcategory of $\CBndl_e$ where each object $C^\oast\to M$ must be
enveloped by the cotangent bundle $T^*M\to M$ and the morphisms must be
restrictions of the natural maps between cotangent bundles induced by
base space maps.
\end{definition}
Note that the morphisms are restricted to be open embeddings between
base spaces to make sure that their natural extensions to the tangent
and cotangent bundles are defined unproblematically.

At this point, one may recall some standard notions of Lorentzian
geometry, as long as they are defined only in terms of timelike cones, and
apply them to the geometry cone bundles. Many of the standard theorems
translate as well, some directly and others with some extra effort. We
restrict ourselves to those that are relevant to the issues at hand.

Fix a chronal cone bundle $C\to M$ on a spacetime manifold $M$. Note
that any chronal cone bundle is \emph{time oriented}, since by
definition the fibers consist of single cones rather than double cones
like in the Lorentzian case. By assumptions the cones are directed into
the future.
\begin{definition}
A smooth curve $\gamma$ in $M$ is called \emph{future directed,
timelike} if the tangent to $\gamma$ is everywhere contained in $C$.
The \emph{chronal precedence relation} $I^+\sse M\times M$ (also
$I^+_C$) is defined as
\begin{equation}
	I^+ = \{ (x,y)\in M\times M \mid \exists \gamma,
		~\text{future directed, timelike curve from $x$ to $y$} \} .
\end{equation}
When $(x,y)\in I^+$, we say that $x$ \emph{chronologically precedes} $y$
and also write $x\ll y$.

If we replace $C$ by $\bar{C}$ in the above definitions, we obtain
\emph{causal curves} and the \emph{causal precedence relation} $J^+\sse
M\times M$, denoted $x<y$. The inverse relations are written $I^-$ and
$J^-$.
\end{definition}
It is immediate that $I^+$ is an open, transitive relation ($x\ll y$ and
$y\ll z$ implies $x\ll z$). Using this relation, we can define the usual
causal hierarchy.
\begin{definition}
\begin{enumerate}
\item[(i)]
	If $I^+$ is irreflexive ($x\not\ll x$ or, equivalently, no closed
	timelike curves exist), then $C$ is \emph{chronological}.
\item[(ii)]
	Given an open $N\sse M$, $I^+_{C|_N}$ and $I^+_{C}\cap (N\times N)$ are
	both relations on $N\times N$. We say that $N$ is
	\emph{chronologically compatible} if both of these relations coincide.
	More conventionally, this means that any two points of $N$ that can be
	joined by a timelike curve in $M$ can also be joined by a timelike
	curve in $N$.
\item[(iii)]
	An open $N\sse M$ is called \emph{chronologically convex} if any timelike
	curve that joins any two points of $N$ must also lie in $N$, which is
	a stronger condition than chronological compatibility.
\item[(iv)]
	The chronal cone bundle $C$ is said to be \emph{strongly chronological} if
	it is chronological and for every $x\in M$ and every neighborhood $N\sse M$
	of $x$, there exists a smaller open neighborhood $L\sse N$ that is
	chronologically convex.
\item[(v)]
	The chronal cone bundle $C$ is said to be \emph{stably chronological}
	if there exists another chronal cone bundle $C'$ that is itself
	chronological and an open neighborhood of the closure of $C$, that is,
	$\bar{C}\setminus\{0\} \sso C'$. For spacelike cone bundles, stable
	chronology is equivalent to the reverse inclusion
	$\bar{C}^{\prime\oast}\setminus\{0\} \sso C^\oast$. The full
	subcategories $\ChrBndl_{sc}$ and $\SpBndl_{sc}$ of $\ChrBndl$ and
	$\SpBndl$ contain only stably chronological cone bundles as objects.
\end{enumerate}
Each of these definitions has an obvious analog when the adjectives
\emph{chronological} or \emph{timelike} are replaced by \emph{causal}.
\end{definition}
Note that \emph{stably chronological} is equivalent to \emph{stably
causal}, so these terms will be used interchangeably.  Moreover, the
chronological chronal cone bundle $C'$ such that contains a stably
chronological chronal cone bundle $C$ can itself be chosen to be stably
chronological.
Next we turn from curves to surfaces.
\begin{definition}
Each of the following concepts may be prefaced with $C$- or $C^\oast$-
to be more specific.
\begin{enumerate}
\item[(i)]
	An oriented codim-$1$ surface $S\sso M$ is called \emph{future
	oriented, spacelike} if its oriented conormals are everywhere
	contained in $C^\oast$.
\item[(ii)]
	A codim-$1$ surface $S\sso M$ is called \emph{achronal} if it
	$(S\times S)\cap I^+ = \varnothing$, that is, no two points of $S$ are
	connected by a timelike curve. Similarly, $S$ is \emph{acausal} when
	$(S\times S)\cap J^+ = \varnothing$.
\item[(iii)]
	A codim-$1$ surface $S\sso M$ is called \emph{Cauchy} if it is acausal
	and every inextensible causal curve intersects $S$ exactly once.
\item[(iv)]
	A chronal cone bundle $C\to M$ is called \emph{globally hyperbolic} if
	there exists a Cauchy surface $S\sse M$. A spacelike cone bundle
	$C^\oast\to M$ is called \emph{globally hyperbolic} if $C\to M$
	is. The full subcategories $\ChrBndl_H$ and $\SpBndl_H$ of $\ChrBndl$
	and $\SpBndl$ contain only globally hyperbolic cone bundles as
	objects.
\end{enumerate}
\end{definition}
Next we define some commonly used domains. Fix $S\sso M$ to be a
$C$-acausal codim-$1$ submanifold, such that either $S$ is closed or
$\bar{S}\sso M$ is a submanifold with boundary.
\begin{definition}
\begin{enumerate}
\item[(i)]
	The \emph{future/past domain of influence} $I^\pm(N)$ of a subset
	$N\sse M$ is the set of points $y\in M$ such that there exists $x\in
	N$ with either $x\ll y$ ($+$) or $y\ll x$ ($-$). Let $I(N) = I^+(N)
	\cup I^-(N)$.
\item[(ii)]
	The \emph{domain of dependence} $D(S)$ is the largest open subset of $M$ for
	which $S$ is a Cauchy surface. More commonly, $D(S)$ is the set of points
	$y\in M$ such that every inextensible timelike curve through $y$
	intersects $S$. Let $D^\pm(S) = D(S) \cap I^\pm(S)$.
\item[(iii)]
	An open subset $L\sse M$ is \emph{lens-shaped} with respect to $S$ if
	it can be smoothly factored as $L\cong (-1,1)\times S$, with $t\colon
	L\to (-1,1)$ denoting the projection onto the first factor (the
	\emph{time function}), such that the level set $t=0$ is $S$ and all
	other level sets are spacelike as well as share the same boundary as
	$S$ in $M$ (which may be empty).
\end{enumerate}
\end{definition}

The literature in relativity and Lorentzian geometry mostly makes use of
the notion of \emph{global hyperbolicity} as given above, but
specialized to Lorentzian cone bundles. On the other hand, the
literature on symmetric (and regular) hyperbolic PDE systems mostly
makes use of lens-shaped domains. It is a non-trivial fact that these
two notions coincide. The argument is essentially that the time function
of a lens-shaped domain foliates it with Cauchy surfaces. Conversely, a
globally hyperbolic cone bundle admits a time function and a smooth
factorization that turns it into a lens-shaped domain (glossing over
some details related to spatial compactness). The original argument
establishing the converse link in Lorentzian geometry is due to
Geroch~\cite{geroch-gh}.  However, his argument only established the
existence of a continuous time function. The details necessary to
establish the smooth version of the result are due to more recent work
of Bernal and Sanchez~\cite{bs-smooth1,bs-smooth2}. The very recent
result by Fathi and Siconolfi~\cite{fs}, using completely different
methods, established the existence of a smooth time function and
factorization for a class of cone bundles sufficient for the current
discussion.
\begin{proposition}\label{prp:globhyp-lens}
An open subset $D\sse M$ is $C$-lens-shaped with respect to $S\sse D$
iff it is globally $C|_D$-hyperbolic, with $S$ being $C|_D$-Cauchy.
\end{proposition}

The last point we address in this section is the kind of maps we allow
between chronal cone bundles. Let $\ChrBndl(M)$ and $\SpBndl(M)$ denote
the respective subcategories of $\ChrBndl$ and $\SpBndl$ where the
objects are restricted to have the same base manifold $M$ and all
morphisms must restrict to the identity on the base $M$. Essentially,
the only allowed morphisms between the cone bundles in these categories
are inclusions. These inclusions, and hence the morphisms, of the
corresponding categories, implement a partial order on chronal and
spacelike cone bundles, ordering the bundles by \emph{speed}. Consider
chronal cone bundles $C_1\to M$ and $C_2\to M$, as well as the
corresponding spacelike cone bundles $C_1^\oast\to M$ and
$C_2^\oast\to M$. The morphism
\begin{equation}
\vcenter{\xymatrix{
	C_1 \ar[d] \ar[r] & C_2 \ar[d] \\
	M \ar[r]^{\id}    & M
}} \quad \text{should read as} \quad
\vcenter{\xymatrix{
	\text{\textit{slow}} \ar[d] \ar[r]^{\sse} & \text{\textit{fast}} \ar[d] \\
	M \ar[r]^{\id}                   & M ~ .
}}
\end{equation}
\emph{Fast} and \emph{slow} are used in the conventional sense:
$C_2$-timelike curves can propagate signals faster than $C_1$-timelike
ones. This ordering is reversed for spacelike cone bundles. The morphism
\begin{equation}
\vcenter{\xymatrix{
	C_2^\oast \ar[d] \ar[r] & C_2^\oast \ar[d] \\
	M \ar[r]^{\id}    & M
}} \quad \text{should read as} \quad
\vcenter{\xymatrix{
	\text{\textit{fast}} \ar[d] \ar[r]^{\sse} & \text{\textit{slow}} \ar[d] \\
	M \ar[r]^{\id}                   & M ~ .
}}
\end{equation}
\begin{definition}
Given a chronal cone bundle $C\to M$, a section $\phi\colon M \to F$ is
called \emph{$C$-slow} (also \emph{slower than $C$}) if $\phi^*\Gamma
\sse C$, that is, signals propagated by $C$-timelike curves can travel
at least as fast as signals propagated by $\phi$-timelike curves. Also,
$\phi$ is said to be \emph{strictly slower than $C$} if
$\phi^*\bar{\Gamma}\sse C$. Given a convex dual spacelike cone
bundle $C^\oast\to M$, the section $\phi$ is \emph{(strictly)
$C^\oast$-slow} (also \emph{(strictly) slower than $C^\oast$}) if it is
(strictly) $C$-slow.
\end{definition}
In terms of spacelike cone bundles, $C^\oast$-slowness is equivalent to
$C^\oast\sse \phi^*\Gamma^\oast$, while strict $C^\oast$-slowness is
equivalent to $\bar{C}^\oast\sse \phi^*\Gamma^\oast$.

Consider two spacelike cone bundles $C^\oast_1\sse C^\oast_2$ and two
sections $\phi_i\colon M\to F$, $i=1,2$, such that $\phi_i$ is
$C^\oast_i$-slow. Then clearly $\phi_2$ is also $C^\oast_1$-slow, while
$\phi_1$ need not be $C_2^\oast$-slow. This observation shows that slow
sections can be pulled back to slow sections along morphisms in the
$\SpBndl(M)$ category. It is elaborated on in the next section.

Let us expand our attention to the larger categories $\ChrBndl$ and
$\SpBndl$, where the base manifold is no longer fixed, however only open
embeddings are allowed as morphisms between base spaces (which naturally
extend to morphisms of the total spaces). Consider a pair of objects and
a morphism between them in the $\ChrBndl$ category.
\begin{equation}
\vcenter{\xymatrix{
	C_1 \ar[r]^{T\chi} \ar[d] & C_2 \ar[d] \\
	M_1 \ar[r]^\chi           & M_2 ~ .
}}
\end{equation}
It is interesting to note that the same base space morphism $\chi\colon
M_1 \to M_2$ extends to a morphism between the convex dual cone bundles in the $\SpBndl$ category,
\begin{equation}
\vcenter{\xymatrix{
	C^\oast_1 \ar[rr]^{(T^*\chi)^{-1}} \ar[d] & & C^\oast_2 \ar[d]\\
	M_1 \ar[rr]^\chi           & & M_2 ~ ,
}}
\end{equation}
only if the two cone bundles agree exactly, that is $\chi_* C_1 = C_2$,
or equivalently $C^\oast_1 = \chi^*C^\oast_2$. In this case, it is clear
that the relations $I^+_{C_1}$ and $I^+_{\chi_*^{-1}C_2}$ on the source
manifold $M_1$ are the same. However, it may not be true that the
chronological precedence relation $I^+_{C_2|_{\chi(M_1)}} =
\chi\times\chi(I^+_{C_1})$ is necessarily the same as the restriction
$I^+_{C_2}\cap (\chi(M_1)\times\chi(M_1))$ on the target manifold $M_2$.
The two are only equal when $\chi(M_1)$ is a $C_2$-chronologically
compatible subset of $M_2$.
\begin{definition}\label{def:chr-compat}
A $\ChrBndl$ morphism $T\chi\colon C_1\to C_2$ in is called
\emph{crhonologically compatible} if $\chi(M_1)$ is a
$C_2$-chronologically compatible subset of $M_2$. Similarly, the
morphism is called \emph{chronologically convex} if $\chi(M_1)$ is a
$C_2$-chronologically convex subset of $M_2$.

A $\SpBndl$ morphism $T^*\chi\colon C^\oast_1 \to C^\oast_2$ is called
\emph{chronologically compatible} or \emph{convex} if the corresponding
$\ChrBndl$ morphism $T\chi\colon C_1\to C_2$ exists and is respectively
chronologically compatible or convex.
\end{definition}
When $C_1\to M$ is a globally hyperbolic chronal cone bundle, there is
no distinction between and chronally convex $\ChrBndl$ morphisms. Thus,
in this context, these terms will sometimes be used interchangeably.
\begin{lemma}\label{lem:chr-compat-convex}
Consider chronal cone bundles $C_i\to M_i$, $i=1,2$, and an open
embedding $\chi\colon M_1\to M_2$ such that the induced morphism
$T\chi\colon C_1\to C_2$ in $\ChrBndl$ is chronally compatible. Then
$T\chi$ is also chronally convex.
\end{lemma}
\begin{proof}
Consider any two points $p,q\in M_1$ that are $C_1$-chronally related.
Then, by chronal compatibility, $\chi(p),\chi(q)\in M_2$ are also
$C_2$-chronally related. By a simple application of Zorn's lemma, among
all timelike curves passing through $\chi(p)$ and $\chi(q)$ there must
exist some maximal (inextensible) curves $\gamma_2$. For each such
curve, the preimage $\gamma_1 = \chi^{-1}(\gamma_2\cap\chi(M_1))$ is
also inextensible. Assume for the moment that chronal convexity fails.
Then we can chose $\gamma_2$ such that it has at least one point not
contained in $\chi(M)$, which implies that $\gamma_1$ consists of more
than one connected component, each of which is inextensible. On the
other hand, global hyperbolicity implies that each component of
$\gamma_1$ intersects a $C_1$-Cauchy surface $\Sigma_1\sse M$ and so the
same can be said about $\gamma_2$ and $\Sigma_2 = \chi(\Sigma_1)$. In
other words, $\gamma_2$ intersects $\Sigma_2$ more than once, which is
impossible, since from chronal compatibility $M_2$ must be achronal.
Therefore, $T\chi$ must also be chronally convex. \qed
\end{proof}

These conditions on cone bundle morphisms will become important in the
formulation of the generalized classical or quantum LCFT functors.  In
particular, it would be extremely convenient for the construction of an
LCFT, namely verifying the Isotony property
(cf.~Sect.~\ref{sec:gen-isot}), if it were possible, in a spacetime with
a globally hyperbolic chronal cone bundle, to be able to construct a
Cauchy surface by extending an existing acausal surface, which may not
be intersected by all inextensible causal curves. However, this question
does not seem to have received much attention, even in the more
restricted setting of Lorentzian geometry. Therefore, we simply
formulate this property as a plausible conjecture, whose validity would
have to be explored in future work.
\begin{conjecture}\label{cnj:cauchy-ext}
Consider globally hyperbolic spacelike cone bundles $C^\oast\to M$ and
$C^{\prime\oast}\to M'$, together with an embedding $\chi\colon M\to M'$
that induces a chronologically compatible morphism $T^*\chi\colon
C^\oast \to C^{\prime\oast}$. If $\Sigma\sso M$ is a $C^\oast$-Cauchy
surface, then for any compact subset $K\sse \Sigma$, there exists a
$C^{\prime\oast}$-Cauchy surface $\Sigma'\sso M'$ such that $\Sigma'$
agrees with $\chi(\Sigma)$ on $\chi(U)$, where $U$ is a neighborhood of
$K$ in $\Sigma$.
\end{conjecture}
We say that the $C^{\prime\oast}$-Cauchy surface $\Sigma'$
\emph{extends} $K$. If there exists a single $C^{\prime\oast}$-Cauchy
surface $\Sigma'$ that extends every compact $K\sse \Sigma$, we say that
$\Sigma'$ \emph{extends} $\Sigma$. It is possible that there exists no
$\Sigma'$ that extends a given $\Sigma$. Think of a lens-shaped domain
in $M'$ whose spacelike boundary is not smooth.


\subsection{Slow sections}\label{sec:slow-sec}
As we have noticed above, a distinctive feature of quasilinear PDE
systems, compared to linear ones, is the field dependence of their
causal structure, as embodied in the geometry of the cone bundles of
timelike vectors and spacelike covectors $\Gamma$ and $\Gamma^\oast$,
whose base space is the field bundle $F$ instead of just the manifold
$M$. This means that any two sections $\phi,\psi\colon M \to F$ of the
field bundle (solution or not) define by pullback the cone bundles
$\phi^*\Gamma\to M$ and $\psi^*\Gamma\to M$. A priori, these cone
bundles need not be related to each other in any way. This complicates
the analysis of the causal structure in contexts where multiple
solutions or sections have to be considered simultaneously. In
particular, such contexts arise when considering nontrivial subsets of
the space of solution sections in classical field theory, probability
distributions on the space of solution sections in classical statistical
mechanics, and deformation quantization of the space of solution sections
in quantum field theory.

It is fortunate that multiple cone bundles can be easily combined via
intersection or convex union. Thus, if we have two sections
$\phi,\psi\colon M\to F$, we can form a chronal cone bundle from their
convex union, $C_{\phi,\psi} = (\phi^*\Gamma + \psi^*\Gamma)$. Then
clearly both $\phi$ and $\psi$ are $C_{\phi,\psi}$-slow, or equivalently
$C_{\phi,\psi}$ is faster than both $\phi$ and $\psi$. Similarly, in
terms of spacelike cone bundles, both $\phi$ and $\psi$ are
$C^\oast_{\phi,\psi}$-slow, where $C^\oast_{\phi,\psi} =
\phi^*\Gamma^\oast \cap \psi^*\Gamma^\oast$.  The only caveat is that
the combined cones may fail to be proper, $C_{\phi,\psi}$ could fail to
be salient or $C^\oast_{\phi,\psi}$ could be empty. In this case, we
call the sections $\phi$ and $\psi$ \emph{chronologically incomparable}.
We can generalize this construction to any collection of sections
$\{\phi_i\}_{i\in I}$. Define $C_I = \sum_{i\in I} \phi_i^*\Gamma$ and
$C^\oast_I = \left(\bigcap_{i\in I} \phi_i^* \Gamma^\oast
\right)^{\circ}$. If either of these cones fails to be proper, we again
call the collection of sections $\{\phi_i\}_{i\in I}$ chronologically
incomparable.

This correspondence between sections and collectively faster cones can
be easily inverted.
\begin{definition}
Given a chronal cone bundle $C\to M$ or the corresponding spacelike cone
bundle $C^\oast\to M$, we denote the set of all sections \emph{strictly
slower than $C\to M$} by $\Secs(F,C)\sse \Secs(F)$, or also
$\Secs(F,C^\oast)$. We can similarly restrict the space of solutions,
$\S(F,C) = \S(F,C^\oast) = \S(F)\cap \Secs(F,C)$. It follows immediately
from the definition that given two chronal cone bundles $C_1\sse C_2$
(equivalently $C^\oast_2 \sse C^\oast_1$ or $C_1$ slower than $C_2$) we
have the inclusions
\begin{equation}\label{eq:slow-incl}
	\Secs(F,C_1) \sse \Secs(F,C_2)
	\quad\text{and}\quad
	\S(F,C_1) \sse \S(F,C_2) .
\end{equation}
Incidentally, we denote by $\Secs_{sc}(F)\sse \Secs(F)$ the space of
stably causal sections and by $\Secs_H(F)\sse \Secs_{sc}(F)$ the space
of all globally hyperbolic sections.  The same notation also extends to
solution spaces: $\S_{sc}(F) = \S(F)\cap \Secs_{sc}(F)$, and $\S_H(F) =
\S(F)\cap \Secs_H(F)$.
\end{definition}
If $C\to M$ is stably causal, then all sections in $\Secs(F,C)$ are a
fortiori also stably causal. It then follows immediately from basic
properties of stable causality that
\begin{equation}\label{eq:sc-cover}
	\Secs_{sc}(F) = \bigcup_{\text{stably causal}~C} \Secs(F,C),
\end{equation}
where the union is over all stably causal cone bundles $C\to M$.
It would be desirable to make a similar statement about $\Secs_H(F)$.
However, that would entail showing that global hyperbolicity is stable
in a certain technical sense. This is a known result for Lorentzian cone
bundles~\cite{bnm}, but remains to be proven in the general case.
However, since $\Secs_H(F)\sse \Secs_{sc}(F)$ (and both are topologized
by their inclusion in $\Secs(F)$), the sets $\Secs_H(F,C) = \Secs_H(F)
\cap \Secs(F,C)$, where $C$ is stably causal, do furnish a cover of
$\Secs_H(F)$. We show below that each $\Secs(F,C)$ is open in $\Secs(F)$
in the Whitney fine topology. Thus, in this topology, $\Secs(F,C)$ and
$\Secs_H(F,C)$, with $C$ stably causal, furnish open covers of
$\Secs_{sc}(F)$ and $\Secs_H(F)$, respectively. An observation that is
important for later developments in Sect.~\ref{sec:classcaus}. While we
could continue to work with an open cover by the sets $\Secs_H(F,C)$, it
is simpler and likely not to make a large difference to instead assume
the following
\begin{conjecture}[Stability of global hyperbolicity]\label{cnj:gh-stab}
The space of globally hyperbolic sections $\Secs_H(F)$ is open in the
space of all sections $\Secs(F)$, provided with the Whitney fine $C^0$
topology. Moreover, the following identity holds
\begin{equation}\label{eq:gh-cover}
	\Secs_H(F) = \bigcup_{\text{globally hyperbolic}~C} \Secs(F,C).
\end{equation}
\end{conjecture}
This conjecture also implies the stability of global hyperbolicity in
a common Fr\'echet topology on $\Secs(F)$, which coincides with the
Whitney fine $C^\oo$ topology, as well, as it is finer (has more open
sets) than the Whitney fine $C^0$ topology.

Recall that \emph{Whitney fine $C^0$ topology} (also known as the
\emph{Whitney strong $C^0$ topology} or the \emph{wholly open $C^0$
topology}) on the space of sections of the bundle $F\to M$ is defined as
follows~\cite{hirsch,km}. Consider sets of the form $U(\phi,V)$, with
$\phi\colon M\to F$ a section and $V\sse F$ an open neighborhood of the
image $\phi(M)$, consisting of all sections $\psi\colon M\to F$ such
that $\psi(M)\sse V$. These sets form a sub-base for the open sets of
the \emph{Whitney fine topology} on $\Secs(F)$. This topology should be
contrasted with the \emph{compact open $C^0$ topology} (also known as
the \emph{Whitney weak $C^0$ topology}), where a sub-base for the open
sets consists of sets of the form $U(\phi,K,V)$, with $\phi\colon M\to
F$ a section, $K\sse M$ compact and $V\sse F|_K$ an open neighborhood of
$\phi(K)$, consisting of all sections $\psi\colon M\to F$ such that
$\psi(K)\sse V$. These two topologies coincide only if $M$ is compact.

The $C^\oo$ versions of these topologies are easily defined in terms of
jet prolongations. As before, consider the sets $U_k(\phi,V)$, but now
with $V\sse J^kF$ an open neighborhood of $j^k\phi(M)$, and
$U_k(\phi,K,V)$, but now with $V\sse J^kF|_K$ an open neighborhood of
$j^k\phi(K)$. They form sub-bases for the open sets of the \emph{Whitney
fine $C^k$} and \emph{compact open $C^k$ topologies}, respectively. The
index $k$ may also take on the value $\oo$. With these topologies, of
which the compact open one is somewhat more commonly used, $\Secs(F)$ is
a Fr\'echet space.

\begin{theorem}\label{thm:slow-open}
Given a chronal cone bundle $C\to M$, the subset of smooth sections
strictly slower than $C$, $\Secs(F,C)\sse \Secs(F)$, is open in the
Whitney fine $C^0$ topology.
\end{theorem}
\begin{proof}
Pick a section $\phi\in\Secs(F,C)$. We will construct a Whitney fine
$C^0$ neighborhood of it that is contained entirely in $\Secs(F,C)$.
Every element $\psi$ of such a neighborhood would satisfy $\bar{C}^\oast
\sse \psi^*\Gamma^\oast$ (recall that convex duality reverses
inclusion).

The fibers of the cone bundle $\bar{C}^\oast$ are cones of compact
section. We can collect these sections into a continuous fibration $K\to
M$ with compact fibers as follows. Pick an affine sub-bundle $\tau^*\sso
T^*M$ whose fibers are of fiber codim-$1$ and intersect the fibers of
$\bar{C}^\oast$ on compact sets with non-empty interior. Such an affine
sub-bundle is easily constructed by modeling it on the vector sub-bundle
of $T^*M$ that is annihilated by a section of $C\to M$, which always
exists, since each fiber is contractible. Let $K$ be the union of the
boundaries of these fiberwise cross sections, $K=\del(\tau^*\cap
\bar{C}^\oast)$. Essentially, $K\to M$ is a topological sphere bundle
due to convexity of the fibers of $\bar{C}^\oast$. From the strict
$C$-slowness hypotheses, we can conclude that $\phi^*\Gamma^\oast$ is an
open neighborhood of $K$.

Next we pick special open neighborhoods for each point $(x,u,p)\in
F\times_M T^*M$ with $u=\phi(x)$ and $(x,p)\in K$. Let $\bar{f}\colon M
\to \tilde{S}^2F^*\otimes_M TM$ be the adapted principal symbol of the
symmetric hyperbolic PDE system under consideration.  Consider the
following composition of continuous bundle maps
\begin{equation}\label{eq:prsym-comp}
\begin{array}{c@{\quad}c@{\quad}c@{\quad}c@{\quad}cl}
	F\times_M T^*M & \to & (\tilde{S}^2F^*\otimes_M TM)\times_M T^*M & \to
		& \tilde{S}^2D^* & , \\
	(y,v,q) & \mapsto & (y,\bar{f}_{(y,v)},q) & \mapsto
		& (y,\bar{f}_{(y,v)}\cdot q) & .
\end{array}
\end{equation}
We note the following facts: (i) the subset of $\tilde{S}^2F^*$
consisting of orientation positive definite bilinear forms in the fibers
is open, (ii) the above composition of continuous maps sends $(x,u,p)$
in this open set, (iii) $F\times_M (\phi^*\Gamma^\oast)$ is an open
neighborhood of $(x,u,p)$. These observation allow us to conclude that
there exists a neighborhood $\tilde{U}$ of $(x,u,p)$ that maps via the
composition~\eqref{eq:prsym-comp} to orientation positive definite
bilinear forms in $\tilde{S}^2F^*$, such that $\tilde{U}\sse F\times_M
(\phi^*\Gamma^\oast)$. Choosing a local bundle trivialization $\iota$
contained within the base projection of $\tilde{U}$ shows the existence
of an open neighborhood of $(x,u,p)$ of the form $\iota(U\times V\times
W) \sse \tilde{U}$, where $U$ is open in $M$, $V$ is adapted to the
fibers of $F$ and $W$ is adapted to the fibers of $T^*M$. This means
that for any $(y,v,q)\in \iota(U\times V\times W)$ we have that $q$ is
$(y,v)$-spacelike and future oriented.

Next, for each $x\in M$, we let $u=\phi(x)$ and we build special
neighborhoods of $(x,u)\in F$ and $(x,u)\times_M K_x \in F\times_M
T^*M$. For each $(x,u,p) \in (x,u)\times_M K_x$, let
$\iota_{x,p}(U_{x,p}\times V_{x,p} \times W_{x,p})$ be the special
neighborhood described in the previous paragraph. By compactness of the
fiber $K_x$, we can choose finitely many $p_j$ such that the
corresponding neighborhoods cover $(x,u)\times_M K_x$. Denote by
$\pi\colon F\times_M T^*M\to F$ the canonical projection. Then, let
\begin{align}
	\tilde{V}
	&= \bigcap_j \pi\circ \iota_{x,p_j}(U_{x,p_j}
		\times V_{x,p_j} \times W_{x,p_j}) , \\
	\tilde{W}
	&= \pi^{-1}(\tilde{V}) \cap \bigcup_j \iota_{x,p_j}(U_{x,p_j}
		\times V_{x,p_j} \times W_{x,p_j}) .
\end{align}
That is, $\tilde{V}$ is an open neighborhood of $(x,u)\in F$ and
$\tilde{W}$ is an open neighborhood of $(x,u)\times_M K_x\sse
F\times_M T^*M$. Moreover, by construction, for each $(x,v)\in
\tilde{V}$ and each $(x,p)\in K_x$, the covector $p$ is
$(x,v)$-spacelike, since $(x,v,p)\in \tilde{W}$.

Next, by continuity of the section $\phi\colon M\to F$ and the bundle
$K\to M$, for each $x\in M$, we can find an open neighborhood $U_x$ of
$x$, trivializations $\iota_x$ of $F$ and $\iota'_x$ of $T^*M$ on $U_x$
(with $\iota_x = \pi\circ \iota'_x$), an open neighborhood
$\iota_x(U_x\times V_x) \sse \tilde{V}$ of $(x,u)\in F$, and an open
neighborhood $\iota'_x(U_x\times V_x \times W_x) \sse \tilde{W}$ of
$(x,u)\times_M K_x \sso T^*M$, such that $\phi(U_x) \sse
\iota_x(U_x\times V_x)$ and $K|_{U_x} \sse \iota'_x(U_x \times V_x
\times W_x)$. By construction, for every $(y,v,q)\in \iota'_x(U_x\times
V_x \times W_x)$, and a fortiori for every $(y,q)\in K|_{U_x}$, we have
that $q$ is $(y,v)$-spacelike.

Finally, let
\begin{equation}
	\mathcal{N} = \bigcup_{x\in M} \iota_x(U_x \times V_x) ,
\end{equation}
which is an open neighborhood of $\phi(M)$ in $F$ such that, for any
other section $\psi\colon M\to F$ with $\psi(M)\sso \mathcal{N}$, each
$(x,p)\in K$ is $\psi$-spacelike. By properties of convex cones, this
implies that $\bar{C}^\oast \sso \psi^*\Gamma^\oast$. The set
$U(\phi,\mathcal{N})$ of sections whose images are contained in
$\mathcal{N}$ is then a Whitney fine $C^0$ neighborhood of $\phi$ that
is contained in $\Secs(F,C)$. \qed
\end{proof}
This theorem will be used in the following section when discussing the
geometry and topology of the phase space of classical field theory.

An analogous result for Lorentzian cone bundles was first given
in~\cite{lerner}. The above proof is significantly more complicated
because the fibers of both $C^\oast$ and $\phi^*\Gamma^\oast$ cone
bundles need not have elliptical cross sections, like they do in the
Lorentzian case.

\subsection{PDE theory}
\label{sec:pde-theory}
As we shall see in Sect~\ref{sec:classquant}, PDE theory can be seen as
a tool for a non-perturbative construction of classical theory. The
following theorems are the main workhorses in this construction: local
existence, possible global existence, as well as uniqueness, domain of
dependence and finite propagation speed. The details of these results can
be found in the standard literature on hyperbolic
PDEs~\cite{john,lax,hoermander-lect}. See also more in depth discussion
and bibliography in~\cite{bfr}.

\begin{theorem}[Local Existence]\label{thm:locext}
Consider a spacelike initial data set $(S,\varphi)$. There exists an
open neighborhood $U$ of $S$ in $M$, $S\sso U\sse M$, and a solution
section $\phi\colon U\to F|_U$ that agrees with $(S,\varphi)$, that is
$\phi|_S = \varphi$ on $U$.
\end{theorem}

Recall that, for a linear PDE system, cone bundle $\Gamma^\oast$ does
not depend on the dynamical fields.  Therefore, it can be seen as a cone
bundle over the spacetime manifold $M$, $\Gamma^\oast\to M$. For linear
equations, the existence result can be strengthened to a global
one~\cite[Ch.XIII]{hoermander-II}, \cite{lax}, \cite[Ch.IV]{taylor},
\cite[Ch.7]{ringstroem}, \cite{bgp}, \cite{waldmann-pde}.
\begin{theorem}[Global Existence]\label{thm:globext}
Consider a spacelike initial data set $(S,\varphi)$ and suppose that the
PDE system is linear. If the spacelike cone bundle $\Gamma^\oast$ is
globally hyperbolic and the surface $S$ is $\Gamma^\oast$-Cauchy, then,
under generic conditions, there exists a global solution section
$\phi\colon M\to F$ that agrees with $(S,\varphi)$, that is $\phi|_S =
\varphi$.
\end{theorem}

\begin{theorem}[Uniqueness]\label{thm:uniq}
Suppose that $\phi\colon M\to F$ is a solution section, that the
spacelike cone bundle $C^\oast=\phi^*\Gamma^\oast$ is globally
hyperbolic and that $S\sso M$ is a $C^\oast$-Cauchy surface. Then $\phi$
is the unique solution section on $M$ that agrees with the initial data
$(S,\phi|_S)$.
\end{theorem}
There is a more local, though equivalent, version of the uniqueness
theorem, that usually goes under a different name.
\begin{corollary}[Domain of Dependence]\label{cor:dod}
Let $(S,\varphi)$ be an initial data set and $\phi\colon U \to F|_U$ be
solution section defined on an open neighborhood $U\supset S$ such that
$\phi|_S = \varphi$. If $V\sse U$  is a lens-shaped domain with respect
to $S$ and the spacelike cone bundle $\phi^*\Gamma^\oast$, then $\phi$
is the unique solution section on $V$ that agrees with $(S,\varphi)$.
\end{corollary}
\begin{proof}
By Prop.~\ref{prp:globhyp-lens}, the interior of $V$ is globally
hyperbolic and $S$ is Cauchy with respect to $\phi^*\Gamma^\oast$. The
result then follows directly from Thm.~\ref{thm:uniq}. \qed
\end{proof}
The same result can also be interpreted as showing finite speed of
propagation of disturbances.
\begin{corollary}[Finite Propagation Speed]\label{cor:finspeed}
Suppose $\phi\colon M\to F$ is a solution section such that the
spacelike cone bundle $C^\oast=\phi^*\Gamma^\oast$ is globally
hyperbolic and the surface $S\sso M$ is $C^\oast$-Cauchy. If $\phi'$ is
another solution section whose restriction to $S$ differs from $\phi$
only on an open submanifold $S'\sse S$, then $\phi$ and $\phi'$ agree in
the interior of the complement of the domain of influence $I_C(S')$.
\end{corollary}
In other words, if a disturbance is confined to $S'$, it does not
propagate faster than allowed by the chronal cones $C=\phi^*\Gamma$ of
the undisturbed solution.
\begin{proof}
It is a consequence of the definitions that the interior of the
complement $M\setminus I(S')$ is the domain of dependence $D(S'')$,
where $S'' = (S\setminus S')^\circ$ and the interior is taken with the
respect to the submanifold topology on $S$. Since $D(S'')$ is globally
hyperbolic, the result follows directly from Thm.~\ref{thm:uniq}. \qed
\end{proof}

What the preceding theorems allow us to do is parametrize globally
hyperbolic solution sections in terms of initial data sets. The
parametrization is unfortunately not bijective. Any solution section
$\phi$ on $M$ corresponds to many different initial data, one for each
$\phi$-Cauchy surface $S$, namely $(S,\phi|_S)$. On the other hand, not
all initial data sets correspond to global solutions. There may exist
data sets $(S,\varphi)$ for which the corresponding solution, which is
guaranteed to exist in a neighborhood of $S$ develops singularities
within $\phi$ and so does not extend to a global solution.  It is also
possible that there exists a solution $\phi$ that agrees with the given
initial data, but is not globally hyperbolic on $M$. Let us refer as
\emph{globally hyperbolic initial data sets} to those that come from the
restriction of a globally hyperbolic solution section $\phi$ to a
$\phi$-Cauchy surface $S\sso M$. Two globally hyperbolic initial data
sets that come from the same solution section are said to be equivalent,
which defines an equivalence relation. Thus, by the above uniqueness
theorem, globally hyperbolic solution sections are in bijective
correspondence with equivalence classes of globally hyperbolic initial
data.

The last paragraph refers to solutions of hyperbolic PDE systems. What
about those with hyperbolically integrable constraints? Consider the
adapted equation form $(f\oplus c,\tilde{F}^*\oplus_M E)$, with
$(f,\tilde{F}^*)$ and $(c,E)$ respectively the hyperbolic and
constraints subsystems, and the corresponding hyperbolic compound system
$(f\oplus h, \tilde{F}^*\oplus_M \tilde{E}^*)$ with identity $h\circ c =
q\circ f$. Certainly, if a section $\phi\colon M\to F$ satisfies both
$f[\phi]=0$ and $c[\phi]=0$, so does its restriction (or the restriction
of $j^k\phi$ for sufficiently high $k$) to any $\phi$-Cauchy surface. On
the other hand, consider initial data $(S,\varphi)$ such
that any extension of $\varphi$ to a solution $\phi$ of $f[\phi]=0$ on
an open neighborhood of $S$ satisfies the constraints $c[\phi]|_S=0$. Is
it necessarily true that an extension of $\varphi$ to a solution $\phi$
of $f[\phi]=0$ on all of $M$ also satisfies $c[\phi]=0$ on all of $M$?
That is indeed the case if the section $\phi$ is globally hyperbolic on
$M$ with respect to the compound system, which includes the consistency
subsystem $(h,\tilde{E}^*)$. To see this, notice that $\phi\oplus
c[\phi]$ gives a solution of the compound system on $M$, since the
consistency identity and $q[0] = 0$ yield
\begin{equation}
	h[c[\phi]] = q[f[\phi]] = q[0] = 0.
\end{equation}
On the other hand, we know that the section $\phi\oplus c[\phi]$
restricts to the initial data set $(S,\varphi\oplus0)$ for the compound
system. But the consistency subsystem is linear and zero is always a
solution, implying that $\phi\oplus 0$ is also a solution with the same
initial data. Finally, if the section $\phi$ is in fact globally
hyperbolic with respect to the compound system, the uniqueness
Thm.~\ref{thm:uniq} shows that the two solutions are identical, that is
that $c[\phi] = 0$ on $M$. So, for symmetric hyperbolic systems with
symmetric hyperbolically integrable constraints, globally hyperbolic
solution sections $\phi\colon M\to F$ are in bijective correspondence
with equivalence classes of globally hyperbolic initial data sets
$(S,\varphi)$ satisfying the constraints. Note that the notion of global
hyperbolicity is in both cases with respect to the hyperbolic structure
of the compound system $(f\oplus h, \tilde{F}^*\oplus_M \tilde{E}^*)$.

\subsection{Linear inhomogeneous problems}\label{sec:lin-inhom}
In the preceding section we have discussed the Cauchy problem for
quasilinear systems. If the system under consideration is linear, we can
also discuss the linear algebra of the corresponding inhomogeneous
problem
\begin{equation}
	f[\phi] = \tilde{\alpha}^* ,
\end{equation}
where $\tilde{\alpha}^*$ is a compactly supported dual density, that is,
a section of the bundle $\tilde{F}^*\to M$.

The field independent spacelike cone bundle $\Gamma^\oast$ defined by
the principal symbol $\bar{f}^\mu_{ab}$ endows $M$ with causal
structure. Let us assume that $M$ is globally hyperbolic with respect to
$\Gamma^\oast$.  It is convenient to introduce spaces of sections with
restricted supports, in particular in ways related to the causal
structure.
\begin{definition}
Consider a vector bundle $V\to M$. We define the following subspaces of
the space of sections $\Secs(V)$:
\begin{align}
	\Secs_0(V) &= \{ \phi\in\Secs(V)
		\mid \text{$\supp\phi$ is compact} \} , \\
	\Secs_+(V) &= \{ \phi\in\Secs(V)
		\mid \text{$\supp\phi$ is retarded} \} , \\
	\Secs_-(V) &= \{ \phi\in\Secs(V)
		\mid \text{$\supp\phi$ is advanced} \} , \\
	\Secs_{SC}(V) &= \{ \phi\in\Secs(V) \mid
		\text{$\supp\phi$ is spacelike compact} \} ,
\end{align}
where \emph{retarded support}, \emph{advanced support}, or
\emph{spacelike compact support} means, respectively, that $\supp\phi
\sso I^+(K)$, $\supp\phi\sso I^-(K)$, or $\supp\phi\sso I(K)$ for some
compact $K\sso M$.

The corresponding subspaces of the solution space $\S(F)$ are defined as
\begin{equation}
	\S_{0,\pm,SC}(F) = \S(F) \cap \Secs_{0,\pm,SC}(F) .
\end{equation}
\end{definition}

\subsubsection{Duhamel's principle}
It is a well known fact of classical PDE theory that for some equations
(like the heat and wave equations) the solutions of Cauchy and
inhomogeneous problems are equivalent. That is, a solution of one
problem can be obtained from the other. This is usually known as
\emph{Duhamel's Principle}~\cite{ch2,evans-pde}. For completeness,
we present it here in a form appropriate for our geometric formulation
of symmetric hyperbolic systems. In addition, once we have the retarded
and advanced Green functions for the inhomogeneous problem, we can
construct an exact sequence that conveniently parametrizes the solution
space of the linear PDE. This parametrization will be useful in
Sects.~\ref{sec:formal-symp} and~\ref{sec:formal-pois} for the
construction and analysis of the symplectic and Poisson structures of
classical field theory.

Consider a Cauchy surface $\Sigma\sso M$ (where the Cauchy property
holds with respect to $\Gamma^\oast\to M$). Let $t\colon M\to \R$ be a
time function, such that $t|_\Sigma = 0$ and $\d{t}$ is future oriented.
The ability to solve the Cauchy problem on $\Sigma$ means that we have
access to
\begin{definition}
The \emph{Cauchy Green function} $\tilde{\G}_\Sigma\colon
\Secs_0(F|_\Sigma) \to \Secs_{SC}(F)$ is a linear map uniquely defined
by the requirement
\begin{equation}
	\phi = \tilde{\G}_\Sigma[\varphi]
	~~ \implies ~~
	f[\phi] = 0 ~~\text{and}~~ \phi|_\Sigma = \varphi.
\end{equation}
\end{definition}
More explicitly, its action is given by the contraction of a bitensor
distribution on $M\times \Sigma$ with the initial data on $\Sigma$:
\begin{equation}
	\tilde{\G}_\Sigma[\varphi]^a(x)
		= \int_\Sigma \d\tilde{s}\, {\G_\Sigma}^a_b(x,s) \varphi^b(s) ,
\end{equation}
where we have used local coordinates $(x^i,u^a)$ on $F$ and $(s^j,v^b)$
on $F|_\Sigma$, and $\d\tilde{s}$ is the coordinate volume $(n-1)$-form on
$\Sigma$.

Consider a compactly supported dual density $\tilde{\alpha}^*\in
\Secs_0(\tilde{F}^*)$. The ability to solve the inhomogeneous problem
$f[\phi] = \tilde{\alpha}^*$ means that we have access to
\begin{definition}
The \emph{retarded/advanced Green function} $\G_\pm\colon
\Secs_0(\tilde{F}^*) \to \Secs_{SC}(F)$ is a linear map defined uniquely
by the requirement
\begin{equation}
	\phi_\pm = \G_\pm[\tilde{\alpha}^*]
	~~\implies~~
	f[\phi_\pm] = \tilde{\alpha}^*
	~~\text{and}~~
	\supp \phi_\pm \sse I^\pm(\supp \tilde{\alpha}^*) ,
\end{equation}
where $+$ denotes the retarded and $-$ the advanced the boundary
condition.
\end{definition}
More explicitly, its action is given by the contraction of a bitensor
distribution on $M\times M$ with the source term,
\begin{equation}
	\G_\pm[\tilde{\alpha}^*]^a(x) = \int_M \G_\pm^{ab}(x,y)
		\alpha^*_b(y)\, \d\tilde{y} ,
\end{equation}
where we have used the local coordinates $(x^i,u^a)$ and $(y^j,v^b)$ on
the two copies of $F$, and $\alpha^*_b(y)\,\d\tilde{y}$ are the
components of $\tilde{\alpha}^*$ in local coordinates, with
$\d\tilde{y}$ the coordinate volume form on $M$. We also use the
notation $\G(x|y) = \G(x,y)$ for the Cauchy or retarded/advanced Green
functions, to highlight the different roles of the arguments.

Consider any section of the field bundle $\phi\colon M\to F$ and define
$\theta_\pm = \Theta(\pm t)$, with $\Theta(t)$ the Heaviside step
function, which are the characteristic functions of $S^\pm =
I^\pm(\Sigma)$, the future and past of $\Sigma$. The main identity that
allows us to relate the Cauchy and retarded/advanced Green functions is
\begin{equation}\label{eq:f-dtheta}
	f[\theta_\pm \phi] = (\bar{f}\cdot\d\theta_\pm) \phi + \theta_\pm f[\phi] .
\end{equation}
A calculation in local coordinates $(x^i,u^a)$ reveals
$\d\theta_\pm = \pm \delta(t) \, \d{t}$ and
\begin{equation}\label{eq:ft-def}
	[(\bar{f}\cdot\d\theta_\pm) \phi]_a
	= \pm\delta(t)\, \d{t}_i \bar{f}^i_{ab} \phi^b
	= \pm\delta(t)\, \bar{f}^t_{ab} \phi^b ,
\end{equation}
where for convenience, we have defined $\bar{f}^t_{ab}$ by the equation
$\d{t}\wedge \bar{f}^t_{ab} = \d{t}_i \bar{f}^i_{ab}$.

We are now ready to prove what is known in classical PDE theory as
\begin{lemma}[Converse Duhamel's Principle]
	In local coordinates $(x^i,u^a)$ on $(t,s^j,v^b)$ on the two copies of
	$F$, with the second restricting to $(s^j,v^b)$ on $F|_\Sigma$, the
	Cauchy Green function can be expressed in terms of the
	retarded/advanced Green functions as follows
	\begin{equation}
		{\G_\Sigma}^a_b(x,s)\,\d\tilde{s}
		= \iota^*\sum_{\pm} \pm \G_{\pm}^{ac}(x|0,s) \bar{f}^t_{cb}(0,s) ,
	\end{equation}
	where the pullback is along the inclusion $\iota\colon \Sigma \sso M$,
	corresponding to $t=0$.
\end{lemma}
\begin{proof}
Consider $\varphi \in \Secs_0(F|_\Sigma)$. We will construct a solution
$\phi$ agreeing with this initial data, $\phi|_\Sigma = \varphi$, using
retarded/advanced Green functions. If such a solution $\phi$ existed,
then from the identity~\eqref{eq:f-dtheta} it would satisfy
\begin{equation}\label{eq:phi-initd}
	f[\theta_\pm \phi]_a
	= \pm \delta(t)\,\d{t}\wedge \bar{f}^t_{ab} \phi^b
	= \pm \delta(t)\,\d{t}\wedge \bar{f}^t_{ab} \varphi^b ,
\end{equation}
where the $\delta(t)$ prefactor allows us to consider only the values of
$\phi|_\Sigma = \varphi^b$. Conversely, knowing only the initial data on
$\Sigma$, we can define the following solutions to distributional
inhomogeneous problems:
\begin{equation}
	f[\phi_\pm] = \delta(t)\,\d{t}\wedge \bar{f}^t\cdot\phi
	\quad \implies \quad
	\phi_\pm = \G_\pm[\delta(t)\,\d{t}\wedge \bar{f}^t\cdot \varphi] .
\end{equation}
Clearly, the section $\phi = \phi_+ - \phi_-$ is then a solution of
$f[\phi] = 0$ that satisfies the identities $\theta_\pm \phi = \pm\phi_\pm$
and hence, by Eq.~\eqref{eq:phi-initd}, agrees with initial data
$\phi|_\Sigma = \varphi$. Global hyperbolicity of $M$ implies that $\phi$ is
unique. Hence we can define the Cauchy Green function by
$\tilde{G}_\Sigma[\varphi] = \phi$. More explicitly, in local coordinates
$(x^i,u^a)$ and $(t,s^j,v^b)$ on the two copies of $F$, we find
\begin{align}
	\phi^a(x)
	&= \sum_{\pm}\pm \int_M \G^{ac}_\pm(x|t,s) \delta(t)\, \d{t}\wedge
		\bar{f}^t_{cb}(t,s) \phi^b(t,s) \\
	&= \int_\Sigma \iota^* \left[
		\sum_{\pm}\pm \G_\pm^{ac}(x|0,s) \bar{f}^t_{cb}(0,s) \right] \varphi^b(s)
	= \int_\Sigma \d\tilde{s}\, {\G_\Sigma}^a_b(x,s) \varphi^b(s) ,
\end{align}
wherefrom the desired result follows immediately. \qed
\end{proof}
To obtain the retarded/advanced Green functions from the Cauchy ones, we
need only set up an initial value problem whose solution determines the
retarded and advanced responses to the point source: $f[\psi_\pm] =
\id_F\tilde{\delta}(x,y)$. In local coordinates $(x^i,u^a)$ and
$(y^i,v^b)$ on the two copies of $F$, $(\id_F)_a^b = \delta_a^b$
(Kronecker delta) and $\tilde{\delta}(x,y) = \prod_i
\delta(x^i-y^i)\,\d\tilde{x}$ (Dirac delta densitized on the first
argument). Without loss of generality, we introduce local coordinates
$(t,r^j,v^b)$ on $F$ such that the point $y$ lies on the surface
$\iota\colon \Sigma\sso M$, $y=(0,r)$ with $\Sigma$ a Cauchy surface
defined by $t=0$. We define also $\underline{f}_t^{ab} =
(\iota^*\bar{f}^t_{ab})^{-1}$, an $(n-1)$-multivector field on $\Sigma$,
which means that
\begin{equation}
	(\iota^*\bar{f}^t_{ab})\cdot \underline{f}_t^{bc} =
	\bar{f}^t_{ab}\cdot (\iota_*\underline{f}_t^{bc}) =
	\delta_a^c ,
\end{equation}
where $\iota^*$ denotes the pullback of $(n-1)$-forms from $M$ to
$\Sigma$, $\iota_*$ denotes the pushforward of $(n-1)$-multivectors from
$\Sigma$ to its image in $M$, and the $\cdot$ denotes the contraction of
the $(n-1)$-form indices with the corresponding $(n-1$)-multivector
indices. We then have the direct form of
\begin{lemma}[Duhamel's Principle]
In local coordinates $(x^i,u^a) = (t,s^i,u^a)$ and $(y^j,v^b) =
(t,r^j,v^b)$ on the two copies of $F$, restricting to $(s^i,u^a)$ and
$(r^j,v^b)$ on the corresponding copies of $F|_\Sigma$, the
retarded/advanced Green function can be expressed in terms of the Cauchy
Green function as follows:
\begin{equation}
	\G_\pm^{ab}(x,y)
		= \pm\theta_\pm(x) {\G_\Sigma}^a_c(x,r)\,\d\tilde{r} \cdot
			\underline{f}_t^{cb}(r) .
\end{equation}
\end{lemma}
\begin{proof}
The result can be verified by direct calculation. First, note that when
$x=(0,s)$ the Cauchy Green function satisfies
${\G_{\Sigma}}^a_b(0,s|r)\,\d\tilde{r} = \delta^a_b \prod_j
\delta(s^j-r^j)\,\d\tilde{r}$. Let
\begin{equation}
	\phi^{ab}(x,r)
	= {\G_\Sigma}^a_c(x,r)\,\d\tilde{r}\cdot \underline{f}_t^{cb}(r).
\end{equation}
It is clear that $\phi$ is a solution of $f[\phi] = 0$ with initial data
\begin{equation}
	(\phi|_\Sigma)^{ab}(s,r)
	= \prod_j \delta(s^j-r^j)\,\d\tilde{r}
		\cdot \underline{f}_t^{ab}(r) .
\end{equation}
If $\phi_\pm^{ab}(x,r) = \pm \theta_\pm(x) \phi^{ab}(x,r)$,
then Eq.~\eqref{eq:f-dtheta} implies
\begin{align}
	f[\phi_\pm]_a^c(x,r)
	&= \delta(t)\,\d{t}\wedge \bar{f}^t_{ab}(x) \phi^{bc}(t,s|r) \\
	&= \delta(t)\,\d{t}\wedge \bar{f}^t_{ab}(x)
		\prod_j \delta(s^j-r^j)\,\d\tilde{r}\cdot \underline{f}_t^{bc}(r) \\
	&= \delta(t)\prod_j \delta(s^j-r^j)\,\d{t}\wedge\d\tilde{s}\,
		(\iota^*\bar{f}^t_{ab})(r)\cdot \underline{f}_t^{ac}(r) \\
	&= \delta_a^c \delta(x|0,r) .
\end{align}
Since $y=(0,r)$, we find the desired identity from $\G_\pm^{ab}(x,y) =
\G_\pm^{ab}(x|0,r) = \phi_\pm^{ab}(x,r)$. For other values of $t$ we can
find $\G_\pm^{ab}(x|t,r)$ in a similar way. \qed
\end{proof}

\subsubsection{Causal Green function (without constraints)}
Now that we are sure to have access to the retarded/advanced green
functions $\G_\pm$ for the linear, symmetric hyperbolic system $f[\phi]
= 0$, we can define the so-called \emph{causal Green function}
\begin{equation}\label{eq:caus-green}
	\G = \G_+ - \G_- .
\end{equation}
This new Green function helps to conveniently parametrize the space of
solutions $\S_{SC}(F) \cong \ker f \sso \Secs_{SC}(F)$ by featuring in
the following
\begin{proposition}\label{prp:exact}
The sequence
\begin{equation}\label{eq:hyp-seq}
\vcenter{\xymatrix{
	0
		\ar[r] &
	\Secs_0(F)
		\ar[r]^{f} &
	\Secs_0(\tilde{F}^*)
		\ar[r]^{\G} &
	\Secs_{SC}(F)
		\ar[r]^{f} &
	\Secs_{SC}(\tilde{F}^*)
		\ar[r] &
	0 ,
}}
\end{equation}
is exact (in the sense of linear algebra).
\end{proposition}
The proof given in~\cite[Thm.3.4.7]{bgp} and~\cite[Lem.3.2.1]{wald-qft}
(which excludes the final surjection) for wave equations directly
carries through to the symmetric hyperbolic case. The final surjection
is covered by the proof in~\cite[Ch.3,Cor.5]{bf-lcqft}.

We can interpret the above proposition in the following way. Since
$\S_{SC}(F) \cong \im \G$, we can express any solution to the
homogeneous problem as $\phi = \G[\tilde{\alpha}^*]$, where $\alpha\in
\Secs_0(\tilde{F}^*)$ is some smooth dual density of compact support.
Also, since $\Secs_{SC}(F) \cong \im f$, for any dual density
$\tilde{\alpha}^*$ with spatially compact support, there exists a
solution $\phi$ with spatially compact support of the inhomogeneous
problem $f[\phi] = \tilde{\alpha}^*$.

\begin{definition}\label{def:adapt-pu}
Consider one Cauchy surface $\Sigma\sso M$ and two more Cauchy surfaces
$\Sigma^\pm\sso M$ to the past and future of $\Sigma$, where
$\Sigma^\pm\sso I^\pm(\Sigma)$, and let $S^\pm = I^\pm(\Sigma^\mp)$. Let
$\{\chi_+,\chi_-\}$ be a partition of unity adapted to the open cover
$\{S^+,S^-\}$ of $M$, that is, $\chi_+ + \chi_- = 1$ and $\supp \chi_\pm
\sso S^\pm$. We call $\{\chi_+, \chi_-\}$ a \emph{partition of unity
adapted to the Cauchy surface $\Sigma$}.
\end{definition}

\begin{lemma}\label{lem:exsplit}
The exact sequence of Prop.~\ref{prp:exact} splits at
\begin{equation}
	\Secs_0(\tilde{F}^*) \cong \Secs_0(F) \oplus \S_{SC}(F)
	\quad\text{and}\quad
	\Secs_{SC}(F) \cong \S_{SC}(F) \oplus \Secs_{SC}(\tilde{F}^*) .
\end{equation}
Given a partition of unity $\{\chi_+,\chi_-\}$ adapted to a Cauchy
surface $\Sigma$, there exist (noncanonical) splitting maps
\begin{align}
	f_\chi \colon & \im\G \to\Secs_0(\tilde{F}^*) , &
		f_\chi[\phi] &= \pm f^\pm_\chi[\phi] = \pm f[\chi_\pm\phi] , \\
	\G_\chi \colon & \Secs_{SC}(\tilde{F}^*)\to \Secs_{SC}(F) , &
		\G_\chi[\tilde{\alpha}^*]
		&= \G_+[\chi_+\tilde{\alpha}^*]
			+ \G_-[\chi_-\tilde{\alpha}^*] .
\end{align}
\end{lemma}
\begin{proof}
We demonstrate that the maps defined in Def.~\ref{def:adapt-pu} are
precisely the splitting maps needed for this Lemma. Note that the
splitting maps are not canonical, as they depend on the choice of a
Cauchy surface and a partition of unity adapted to it.

When $\phi\in\S_{SC}(F) \cong \im\G$, the identity~\ref{eq:f-dtheta}
shows that
\begin{equation}
	f^\pm_\chi[\phi] = f[\chi_\pm\phi] = (\bar{f}\cdot\d\chi_\pm) \phi
\end{equation}
does in fact have compact support, as $\supp\phi$ is spacelike compact
while $\supp\d\chi_\pm \linebreak \sso S^+\cap S^-$ is timelike compact.
Also, since $\d(\chi_+ + \chi_-) = 0$, we have $f^+_\chi[\phi] +
f^-_\chi[\phi] = 0$, which means that the map $f_\chi = \pm f^\pm_\chi$
is well defined. On the other hand, we have $\G_\pm[f[\chi_\pm\phi]] =
\chi_\pm\phi$ from the uniqueness of solutions to the inhomogeneous
problem with retarded/advanced support. The definition of the causal
Green function then immediately implies that $\G\circ f_\chi = \pm\id$
on $\S_{SC}(F)$. Also, a direct calculation shows that $f\circ \G_\chi =
\id$ on $\Secs_{SC}(\tilde{F}^*)$:
\begin{equation}
	f\circ \G_\chi [\tilde{\alpha}^*]
	= \chi_+\tilde{\alpha}^* + \chi_-\tilde{\alpha}^*
	= \tilde{\alpha}^* .
\end{equation}
This concludes the proof. \qed
\end{proof}

We conclude this section by noting a simple but important fact.
\begin{lemma}[Covariance]\label{lem:green-covar}
Consider two manifolds $M$ and $M'$ with globally hyperbolic, linear,
symmetric hyperbolic PDE systems on them, with respective equation forms
$(f,\tilde{F}^*)$ and $(f',\tilde{F}^{\prime*})$. Suppose that the open
embedding $\chi\colon M\to M'$ is such that $(f,\tilde{F}^*) = (\chi^*
f', \chi^*\tilde{F}^{\prime*})$ and, moreover, the induced morphism of
chronal cone bundles is chronally compatible. Then the Causal Green
function of the two systems agree: $\G = (\chi\times \chi)^* \G'$.
\end{lemma}
\begin{proof}
Since $\G = \G_+ - \G_-$, the result follows trivially if we know that
$\G_\pm = (\chi\times \chi)^*\G'_\pm$. The latter identity follows from
the uniqueness of the retarded/advanced Green functions on $M$ and the
fact that the pullbacks $(\chi\times\chi)^* \G'_\pm$ do in fact satisfy
the desired Green function identity with appropriate boundary
conditions, which is implied by chronal compatibility. \qed
\end{proof}

\subsubsection{Causal Green function (with constraints)}\label{sec:caus-green}
The presence of constraints complicates the parametrization of solution
spaces to the homogeneous and inhomogeneous problems. Recall that a
symmetric hyperbolic system with hyperbolically integrable constraints
consists of the hyperbolic subsystem $(f,\tilde{F}^*)$, a constraints
subsystem $(c,E)$, and a consistency subsystem $(h,\tilde{E}^*)$
satisfying the identity $h\circ c = q\circ h$ for some differential
operator $q$. Since in this section we are concerned with linear
systems, we take all differential operators to be linear. Moreover, the
causal structure is presumed to be determined by the symbol of the
linear symmetric hyperbolic compound system $(f\oplus h, \tilde{F}^*
\oplus \tilde{E}^*)$.

We will not discuss the most general kind of constraints and restrict
our attention only to \emph{parametrizable} ones. By the term
parametrizable, we mean that there exist an additional vector bundle
$E'\to M$ and additional differential operators $h'$, $c'$ and $q'$,
which fit into the following commutative diagram
\begin{equation}
\vcenter{\xymatrix{
	\Secs(E') \ar[d]^{h'} \ar[r]^{c'} &
		\Secs(F) \ar[d]^{f} \ar[r]^{c} &
		\Secs(E) \ar[d]^{h} \\
	\Secs(\tilde{E}^{\prime*}) \ar[r]^{q'} &
		\Secs(\tilde{F}^*) \ar[r]^{q} &
		\Secs(\tilde{E}^*)
}}
\end{equation}
such that $(h',\tilde{E}^{\prime*})$ is symmetric hyperbolic, and that
the horizontal complexes of differential operators are formally exact
(meaning that $c\circ c' = 0$, $q\circ q' = 0$ and that the
corresponding principal symbols form an exact sequence of vector bundle
maps), that is, they form a \emph{elliptic complex}~\cite[\textsection
XIX.4]{hoermander-III}. Since both $(h,\tilde{E}^*)$ and
$(h',\tilde{E}^{\prime*})$ are symmetric hyperbolic, we can define their
retarded/advanced Green functions, $\H_\pm$ and $\H'_\pm$, as well as
their causal Green functions, $\H = \H_+ - \H_-$ and $\H' = \H'_+ -
\H'_-$. All these operators then fit into the following commutative
diagram:
\begin{equation}\label{eq:chyp-seq}
\vcenter{\xymatrix{
	0
		\ar[r] &
	\Secs_0(E')
		\ar[d]^{c'} \ar[r]^{h'} &
	\Secs_0(\tilde{E}^{\prime*})
		\ar[d]^{q'} \ar[r]^{\H'} &
	\Secs_{SC}(E')
		\ar[d]^{c'} \ar[r]^{h'} &
	\Secs_{SC}(\tilde{E}^{\prime*})
		\ar[d]^{q'} \ar[r] &
	0  \\
	0
		\ar[r] &
	\Secs_0(F)
		\ar[d]^{c} \ar[r]^{f} &
	\Secs_0(\tilde{F}^*)
		\ar[d]^{q} \ar[r]^{\G} &
	\Secs_{SC}(F)
		\ar[d]^{c} \ar[r]^{f} &
	\Secs_{SC}(\tilde{F}^*)
		\ar[d]^{q} \ar[r] &
	0  \\
	0
		\ar[r] &
	\Secs_0(E)
		\ar[r]^{h} &
	\Secs_0(\tilde{E}^*)
		\ar[r]^{\H} &
	\Secs_{SC}(E)
		\ar[r]^{h} &
	\Secs_{SC}(\tilde{E}^*)
		\ar[r] &
	0
}}
\end{equation}
We call the constraints subsystem \emph{globally parametrizable} if the
elliptic complexes constituting the columns of the above diagram
are all exact (their cohomologies are trivial).

\begin{lemma}\label{lem:inhom-constr}
The retarded/advanced inhomogeneous problem
\begin{align}
	f[\phi] &= \tilde{\beta}^* , \\
	c[\phi] &= \gamma ,
\end{align}
with $\tilde{\beta}^*\in \Secs_0(\tilde{F}^*)$ and $\gamma \in
\Secs_\pm(E)$, is solvable for $\phi\in \Secs_\pm(F)$ iff $h[\gamma] =
q[\tilde{\beta}^*]$.
\end{lemma}
\begin{proof}
Let $\phi = \G_\pm[\tilde{\beta^*}]$. We obviously have $f[\phi] =
\tilde{\beta}^*$. It remains to check
\begin{equation}
	c[\phi]
	= c[\G_+[\tilde{\beta}]]
	= \H_+[q[\tilde{\beta}^*]]
	= \H_+[h[\gamma]]
	= \gamma .
\end{equation}
This concludes the proof. \qed
\end{proof}

\subsubsection{Green functions and adjoints}\label{sec:green-adj}
We conclude this section by remarking the identities
\begin{equation}\label{eq:adj-id}
	(\G_\pm)^* = \G^*_\mp ,
\end{equation}
where on the left hand side $(G_\pm)^*$ denotes the adjoint of the
retarded/advanced Green function $\G_\pm$ of the equation $f[\phi]=0$,
and on the right hand side $\G^*_\mp$ denotes the advanced/retarded
Green function of the adjoint equation $f^*[\phi] = 0$. Note that taking
the adjoint flips the support between retarded and advanced.
\begin{definition}\label{def:green-form}
Given two differential operators $f,f^*\colon \Secs(F) \to
\Secs(\tilde{F}^*)$ are said to be mutually \emph{adjoint} if there
exists a bilinear differential operator $G\colon \Secs(F)\times \Secs(F)
\to \Forms^{n-1}(M)$ such that
\begin{equation}\label{eq:green-form}
	f[\phi]\cdot \psi - \phi\cdot f^*[\psi]
	= \d G(\phi,\psi).
\end{equation}
for any sections $\phi,\psi\colon M\to F$. The $(n-1)$-form valued
bilinear differential operator $\G(\phi,\psi)$ is called a \emph{Green
form} associated to $f$ and $f^*$~\cite[\textsection IV.5]{palais},
\cite[\textsection V.1.3]{ch1}. Note that Eq.~\eqref{eq:green-form}
defines $G$ up to the addition of an exact
form~\cite{anderson-big,gms}, $G\sim G+\d{H}$.
Denote by $[G]$ the uniquely defined equivalence class modulo exact
local bilinear forms $\d{H}(-,-)$.
\end{definition}
Notice that the principal symbols of
mutually adjoint, first order, linear, symmetric hyperbolic%
	\footnote{It is at this point that it be comes convenient to have
	chosen the adapted equation form of symmetric hyperbolic systems to be
	densitized. With the standard, non-densitized definition, the adjoint
	operator $f^*$ is in any case densitized. So it is only in symmetric
	hyperbolic form after contraction with a nowhere vanishing degree $n$
	multivector field.} %
differential operators $f$ and $f^*$ are negatives of each other
\begin{equation}
	\bar{f}^* = - \bar{f}.
\end{equation}
Also, the Green form of $f$ and $f^*$, given local coordinates
$(x^i,u^a)$ on $F$, has the following representative
\begin{equation}\label{eq:gf-rep}
	G(\phi,\psi) = \phi\cdot (\tr \bar{f})\cdot \psi
		= (\tr \bar{f})_{ab} \phi^a \psi^b ,
\end{equation}
where $(\tr\bar{f})_{ab}$ are $(n-1)$-forms obtained by contracting the
single contravariant index of the symbol $\bar{f}$ with one of its
covariant $n$-form indices. When pulled back to a codim-$1$ surface
$\iota\colon\Sigma \sso M$, the Green form forms a density that can be
integrated over $\Sigma$. Let $\Sigma$ be defined as the zero set $t=0$
of a smooth function $t$. Then a straight forward coordinate calculation
shows that
\begin{equation}\label{eq:pb-gf-rep}
	\iota^* G(\phi,\psi)
	= \iota^* (\tr \bar{f})_{ab} \phi^a \psi^b
	= \iota^* \bar{f}^t_{ab} \phi^a \psi^b,
\end{equation}
where $\bar{f}^t_{ab}$ was defined in Eq.~\eqref{eq:ft-def}. In
particular, when $\Sigma$ is spacelike and future oriented, this shows
that for a symmetric hyperbolic differential operator $f$ the pulled
back Green form $\iota^* G(-,-)$ defines an orientation positive
definite, symmetric bilinear form on the fibers of the restricted field
bundle $F|_\Sigma\to \Sigma$. This fact will be used in
Sect.~\ref{sec:tt*-sols}.

If we keep the orientation on $M$ fixed, from the relation between their
principal symbols, the symmetric hyperbolic equation forms
$(f,\tilde{F}^*)$ and $(f^*,\tilde{F}^*)$ define the same spacelike
covectors, except for the future/past orientation, which gets flipped.
This means that the causal relations defined by $f^*$ are simply the
reverse of those defined by $f$. In this section, when using $\pm$
indices to denote retarded or advanced support, we always refer to the
notions of past and future defined by $f$.

Recall that we may introduce a natural pairing between elements $\phi
\in \Secs(F)$ and $\tilde{\alpha} \in \Secs(\tilde{F}^*)$ given by
\begin{equation}
	\langle \phi, \tilde{\alpha} \rangle =
	\langle \tilde{\alpha}, \phi \rangle =
	\int \phi\cdot \tilde{\alpha}^* .
\end{equation}
The pairing is only partially defined, that is, for simplicity, only for
those pairs of sections for which the integrand $\phi\cdot
\tilde{\alpha}^*$ has compact support. Its properties, including
non-degeneracy, will discussed in more detail in
Sects.~\ref{sec:tt*-conf}--\ref{sec:tt*-sols-constr}. It is easy to show
that the adjoint $f^*$ coincides with the adjoint of $f$ with respect to
this natural pairing: $\langle f[\phi], \psi \rangle = \langle \phi,
f^*[\psi] \rangle$. This natural pairing allows us to define adjoints
for integral operators like Green functions, namely $\langle
\G_\pm[\tilde{\alpha}^*], \tilde{\beta}^* \rangle = \langle
\tilde{\alpha}^*, (G_\pm)^*[\tilde{\beta}^*] \rangle$.

It is straight forward to check that the identities $f\circ \G_\pm =
\G_\pm \circ f = \id$ hold on $\Secs_{\pm}(F)$ as well as characterize
the retarded/advanced Green functions, and also that the natural pairing
$\langle -, -\rangle$ is non-degenerate on the spaces
$\Secs_\pm(F)\times \Secs_\mp(\tilde{F}^*)$. It is now easy to verify
the adjoint identities~\eqref{eq:adj-id} since
\begin{align}
	\int_M \phi_\mp \cdot f\circ \G_\pm[\tilde{\alpha}^*_\pm]
	&= \langle \phi_\mp, \tilde{\alpha}^*_\pm \rangle
		= \int_M (\G_\pm)^*\circ f^*[\phi_\mp] \cdot \tilde{\alpha}^*_\pm \\
	\int_M \G_\pm\circ f[\phi_\pm] \cdot \tilde{\alpha}^*_\mp ,
	&= \langle \phi_\pm, \tilde{\alpha}^*_\mp \rangle
		= \int_M \phi_\mp \cdot f^*\circ (\G_\pm)^*[\tilde{\alpha}^*_\mp] ,
\end{align}
for any $\phi_\pm \in \Secs_\pm(F)$ and $\tilde{\alpha}^*_\pm \in
\Secs_\pm(F)$. The causal Green functions then satisfy $(\G)^* = -\G^*$,
where $\G^*$ is the causal Green function for $f^*$.

\section{Construction of the Classical Field Theory}
\label{sec:classquant}
Classical mechanical systems, and field theories in particular, have
three standard levels of description: \emph{spacetime}, \emph{phase
space}, and \emph{observables}. At the spacetime level, the mechanical
system is specified by the underlying \emph{spacetime} as a manifold $M$, by
the dynamical degrees of freedom as a field bundle $F\to M$, and by the
dynamics as a PDE system $\E\sso J^k(F,M)$. At the phase space level,
the set of all possible solutions forms a (possibly infinite
dimensional) symplectic manifold, referred to as the \emph{phase space}. At the
level of observables, the mechanical system is associated with the
Poisson algebra of smooth functions on the phase space, known as the
\emph{algebra of observables}. For locally covariant field theories, all of
these descriptions should be associated functorially to the given
spacetime manifold $M$. A discussion of these functorial aspects is
delayed until Sect.~\ref{sec:natural}. For now, we discuss individual
PDE systems.

So far, we have only discussed field theories at the spacetime level, as
PDE systems. To move on to the phase space level, we must consider the
space of solutions of the PDE system. To become a viable phase space, it
must be equipped with symplectic structure, preferably in a
spacetime-local way (which is precised later on). It is well known that
this is possible when the PDE system has an equation form of
Euler-Lagrange equations of a local variational principle (or
\emph{local action principle}). The corresponding spacetime-local
symplectic structure is given by a horizontally conserved symplectic
current density on the jet space $J^\oo F$. If this symplectic current
density is provided along with the PDE system, its space of solutions
can be directly treated as the phase space of a classical field theory,
without the need to introduce a local action. What is less well
known~\cite{kh-inv}, and perhaps a little surprising, is that a
local action can be recovered from a symplectic current density. In
other words, the two ways of specifying a classical field theory are
essentially equivalent.

The construction of the phase space of the classical field theory is
broken up below into three sections. Sect.~\ref{sec:var-sys} starts with
a local Lagrangian and extracts from it a local presymplectic form,
which is used the construct the symplectic and Poisson tensors on the
space of solutions in Sect.~\ref{sec:symp-pois}. Because the space of
a PDE system is in general infinite dimensional, we first establish some
formal properties of its tangent and cotangent spaces necessary for the
discussion of these phase space structures in Sect.~\ref{sec:formal-dg},
which constitutes the bulk of this Section. Making the discussion the
relevant infinite dimensional geometry non-formal would require
substantially more functional analytical detail, which would detract
from the geometric focus of this paper. More details in this direction
can be found in~\cite{bsf,fr-bv,rejzner-thesis,bfr}.

Before proceeding, it is worth remarking that the construction of the
symplectic and Poisson structures on the classical phase space more
commonly carried out in a Hamiltonian framework. However, such an
approach usually requires a non-canonical 3+1 decomposition of an
underlying 4 dimensional spacetime (and similarly in higher dimensions),
which destroys manifest 4-dimensional spacetime covariance. On the other
hand, we hold spacetime covariance as an important guiding principle
underlying the construction of locally covariant field theories,
cf.~Sects.~\ref{sec:freelcft} and~\ref{sec:natural}. Fortunately, it has
been known for a long time, that the symplectic and Poisson structures
on the phase space can be built directly from the Lagrangian without
giving up spacetime covariance~\cite{lw,dewitt-qft,peierls}.  In
fact, it is known from general principles that the usual Hamiltonian
formalism is subsumed as a special case of this Lagrangian
formalism~\cite{bhs,henneaux-elim,fr-pois}.

\subsection{Variational systems}\label{sec:var-sys}
Consider a field vector bundle $F\to M$ over an $n$-dimensional manifold
$M$. A local action functional of order $k$ on $F\to M$ is a function
$S[\phi]$ of sections $\phi\colon M \to F$,
\begin{equation}
	S[\phi] = \int_M (j^k\phi)^*\L,
\end{equation}
where $\L$, the Lagrangian density, is a section of the bundle
$(\Lambda^n M)^k\to J^kF$ densities, which could depend on jet
coordinates of order up to $k$. The Lagrangian density is called local
because, given a section $\phi$ and local coordinates $(x^i,u^a_I)$ on
$J^kF$, the pullback at $x\in M$ can be written as 
\begin{equation}
	(j^k\phi)^*\L(x) = \L(x^i,\del_I\phi^a(x)) ,
\end{equation}
which depends only on $x$ and on the derivatives of $\phi$ at $x$ up to
order $k$. For the most part, the integral over $M$ can be considered
formal, since all the necessary properties will be derived from $\L$. On
the other hand, the finiteness of $S[\phi]$ or related quantities may be
important while discussing boundary conditions in spacetimes with
non-compact spatial extent. However, we will no discuss these issues
below.

Recall that Sect.~\ref{sec:jets} introduces the variational
bicomplex $\Forms^{h,v}(F)$ of vertically and horizontally graded
differential forms on $J^\oo F$. Below, we use the notation introduced
in that section. A Lagrangian density is then an element
$\L\in\Forms^{n,0}(F)$ that can be projected to $J^kF$. Incidentally the
usual variational derivative of variational calculus can be put into
direct correspondence with the vertical differential $\dv$ on this
complex, which is how the name \emph{variational bicomplex} was
established~\cite{anderson-big,anderson-small}.

Let $(x^i,u^a_I)$ be a set of adapted coordinates on the $\oo$-jet
bundle $J^\oo F$, where all the following calculations can be lifted. Any
result that depends only on jets of finite order can then be projected
on to the appropriate finite dimensional jet bundle. Using the
integration by parts identity~\eqref{eq:byparts} if necessary, we can
always write the first vertical variation of the Lagrangian density as
\begin{equation}\label{eq:dvL}
	\dv \L = \EL_a\wedge\dv{u^a} - \dh\theta.
\end{equation}
All terms proportional to $\dv{u^a_I}$, $|I|>0$, have been absorbed into
$\dh\theta$. In the course of the performing the integrations by parts,
$\EL_a$ can acquire dependence on jets up to order $2k$, and $\theta$ on
jets up to order $2k-1$. Note that $\EL_a=0$ are the
\emph{Euler-Lagrange equations} associated with the action functional
$S[\phi]$ or the Lagrangian density $\L$. We can identify the form
$\EL_a\wedge\dv u^a$ with a bundle morphism $\EL\colon J^{2k}F \to
\tilde{F}^*$. Therefore, $(\EL,\tilde{F}^*)$ is an equation form
of a PDE system $\E_\EL\sso J^{2k}F$ on $F$ of order $2k$. A PDE system
with an equation form given by Euler-Lagrange equations of a Lagrangian
density is said to be \emph{variational}. Also, the form $\theta$ is an
element of $\Forms^{n-1,1}(F)$, projectable to $J^{2k-1}F$. It is
referred to as the \emph{presymplectic potential current density}.
Applying the vertical exterior differential to $\theta$ we obtain the
\emph{presymplectic current density} (or the \emph{presymplectic current
density defined by $\L$} if the extra precision is necessary).
\begin{equation}\label{eq:omega-def}
	\omega = \dv \theta,
\end{equation}
with $\omega\in \Forms^{n-1,2}(F)$. This terminology implies that
$\omega$ can be integrated over a codim-$1$ spacetime surface to
construct a presymplectic form, Sect.~\ref{sec:formal-symp}, which is
then necessarily local. This method of construction a symplectic form on
the phase space of classical field theory is sometimes referred to as
the \emph{covariant phase space method}~\cite{lw,cw,abr}.

The following lemma is an easy consequence of the definition of
$\omega$.
\begin{lemma}
The form $\omega\in\Forms^{n-1,2}(F)$ defined in
Eq.~\eqref{eq:omega-def} is both horizontally and vertically closed when
pulled back to $\iota_\oo\colon \E^\oo_\EL\sse J^{2k} F$:
\begin{align}
	\dh \iota_\oo^* \omega &= 0 , \\
	\dv \iota_\oo^* \omega &= 0 .
\end{align}
\end{lemma}
\begin{proof}
The horizontal and vertical differentials on $\E_\EL$ are defined by
pullback along $\iota_\oo$, that is, $\dh \iota_\oo^* = \iota_\oo^* \dh$
and $\dv \iota_\oo^* = \iota_\oo^* \dv$. Since $\omega=\dv\theta$ is
already vertically closed on $J^{2k-1}F$, it is a fortiori vertically
closed on $\E$. The rest is a consequence of the nilpotence and
anti-commutativity of $\dh$ and $\dv$:
\begin{align}
	0 = \dv^2\L
		&= \dv\EL_a\wedge\dv{u^a} - \dv\dh\theta, \\
	\dh\omega
		&= \dh\dv\theta
		= -\dv \EL_a\wedge\dv{u^a} , \\
	\dh\iota_\oo^*\omega
		&= \iota_\oo^*\dh\omega
		= -\iota_\oo^* \dv\EL_a \wedge \dv u^a
		= 0 ,
\end{align}
since $\EL_a$ and $\dv\EL_a$ generate the differential ideal in
$\Forms^*(J^\oo F)$ annihilated by the pullback $\iota_\oo^*$. \qed
\end{proof}
In fact, we will promote the name \emph{presymplectic current density}
to any form satisfying these properties.
\begin{definition}
Given a PDE system $\iota\colon \E \sso J^kF$ we call a form $\omega$ a
\emph{presymplectic current density compatible with $\E$} if $\omega\in
\Forms^{n-1,2}(F)$ and it is both horizontally and
vertically closed on solutions:
\begin{align}
	\dh \iota_\oo^* \omega &= 0 , \\
	\dv \iota_\oo^* \omega &= 0 .
\end{align}
\end{definition}
The particular form $\omega$ defined by Eq.~\eqref{eq:omega-def} will be
referred to as the presymplectic current density associated to or
obtained from the Lagrangian density $\L$, if there is any potential
confusion.

\subsection{Formal differential geometry of solution spaces}
\label{sec:formal-dg}
Before describing the symplectic and Poisson structures on the space of
solutions, we should say something about the differential geometry of
the manifold of solutions of a PDE system as well as its tangent and
cotangent spaces. As usual for infinite dimensional manifolds, there are
some subtleties.

The main goal of this section is to describe the \emph{formal} tangent
and formal cotangent spaces of the manifold arbitrary field sections and
the manifold of solution sections. The adjective formal, in the last
sentence, alludes to the fact that we avoid most technical issues of
infinite dimensional analysis and concentrate on what would be dense
subspaces of the true tangent and cotangent spaces with a reasonable for
their topologies. Results are algebraic and (finite dimensional)
geometric identities that would form the core of an earnest functional
analytical formulation of their non-formal versions. The formal tangent
and cotangent spaces have a natural dual pairing, which we prove to be
non-degenerate, as a substitute for the absence of true topological
duality between them. In the presence of constraints, the proof is
carried out under some additional sufficient conditions.

We start with Sect.~\ref{sec:top-choice}, which discusses the choices of
topology on the infinite dimensional spaces of field configurations and
solutions. Sects.~\ref{sec:constr} and Sect.~\ref{sec:gauge} discuss
sufficient conditions on the constraints and gauge transformations
needed for later results. Sects.~\ref{sec:tt*-conf}, \ref{sec:tt*-sols}
and~\ref{sec:tt*-sols-constr} define the formal tangent and cotangent
spaces in the progressively more complicated cases of the space of field
configurations, the space of solutions (without constraints), and the
space of solutions (with constraints). Finally,
Sect.~\ref{sec:formal-forms} uses the preceding discussion to define
formal local differential forms, of which the symplectic form will be an
example.

\subsubsection{Choice of topology}\label{sec:top-choice}
What we really want to do is describe the space of solutions $\S_H(F)$.
However, we first start with the space $\Secs_H(F)$ or arbitrary field
sections, because it has a simpler structure. In fact, since $F\to M$ is
a vector bundle, $\Secs(F)$ is a vector space. It can be turned into a
topological vector space for several reasonable choices of topology. One
such choice is the \emph{compact open topology} (or more precisely the
\emph{$C^\oo$ compact open topology}). When the base ma\-ni\-fold $M$ is not
compact, the \emph{Whitney topology} (also known as the \emph{Whitney
$C^\oo$ topology}, or the \emph{wholly open topology}) is another
natural choice.  Unfortunately, $\Secs(F)$ ceases to be a topological
vector space%
	\footnote{With this choice of topology, $C^\oo(M)$ is still a
	topological ring, but not a topological algebra over $\R$. The space
	of sections $\Secs(F)$ is then a topological module over $C^\oo(M)$,
	where both have the Whitney topology.}, %
because multiplication by scalars fails to be continuous. These two
topologies coincide iff the base manifold $M$ is compact. The Whitney
topology is naturally singled out by Thm.~\ref{thm:slow-open}, which
shows that the sets of slow sections $\Secs(F,C)$ are Whitney-open%
	\footnote{They are proved to be open in the Whitney $C^0$ topology.
	They are also open in the Whitney $C^\oo$ topology, since the latter
	is finer (it has more open sets).} %
in $\Secs_H(F)$. For the purposes of the discussion in
Sect.~\ref{sec:limits}, it is advantageous to consider a cover of
the space of solutions by open sets of slow sections. On the other hand,
the differential geometry of infinite dimensional manifolds not modeled
on a topological vector space becomes significantly more complicated.
Also, it becomes more difficult to make contact with the approach
of~\cite{fr-bv,rejzner-thesis,bfr}, who use the compact open topology.

One can think of different ways out of this impasse. One could bite the
bullet and consider infinite dimensional manifolds modeled on
topological $C^\oo(M)$-modules. One could also restrict the spacetime
manifold $M$ to be the interior of a compact manifold with boundary and
restrict $\Secs(F)$ to the set of sections that extend continuously in
some way to the boundary. At the moment, we remain agnostic about these
or other possibilities, as they do not affect the results presented
below, and simply assume that some choice has been made so that the sets
$\Secs(F,C)$ are open in $\Secs_H(F)$, or could be effectively treated
as such. For example, the choice of the functor $C^\oo(-)$ that assigns
the algebra of smooth functions to our infinite dimensional manifolds is
insensitive to the difference between the compact open and Whitney
topologies. Therefore, we presume (assuming also that
Conj.~\ref{cnj:gh-stab} holds) the following hypothesis
\begin{hypothesis}\label{hyp:opencover}
The space $\Secs_H(F)$ of globally hyperbolic sections is open in
$\Secs(F)$, and hence a manifold modeled on $\Secs(F)$, with an open
cover by $\Secs(F,C)$, with globally hyperbolic $C\to M$, as in
Eq.~\eqref{eq:gh-cover}.
\end{hypothesis}

Now, the space of solutions can be topologized as a subset $\S(M)\sso
\Secs(F)$. We are not interested in all solutions, only in the open
subset of globally hyperbolic ones, $\S_H(M) = \S(M)\cap \Secs_H(M)$.
This will be our classical phase space.

\subsubsection{Constraints}\label{sec:constr}
Recall that we are considering a symmetric hyperbolic PDE system with
constraints defined on the field bundle $F\to M$, whose equation form is
given by the hyperbolic subsystem $(f,\tilde{F}^*)$, the constraint
subsystem $(c,E)$ and the consistency subsystem $(h,\tilde{E}^*)$,
satisfying the identity $h\circ c = q\circ f$ for some differential
operator $q$. The formal tangent vectors on the space of solutions will
essentially consist of solutions of the linearized version of this PDE
system. Supposing that some dynamical linearization point
$\phi\in\S_H(M)$ is held fixed, the PDE subsystems linearized about
$\phi$ will consist of $(\dot{f},\tilde{F}^*)$, $(\dot{c},E)$ and
$(\dot{h},\tilde{E}^*)$, satisfying $\dot{h}\circ \dot{c} = \dot{q}
\circ \dot{f}$, respectively. All of $\dot{f}$, $\dot{c}$, $\dot{h}$,
$\dot{q}$ are now linear differential operators. The essential
properties of linear, symmetric hyperbolic systems with constraints that
we will refer to are laid out in Sect.~\ref{sec:lin-inhom}.  Dually, the
formal cotangent space will be defined using the adjoint operators
$\dot{f}^*\colon \Secs(F)\to \Secs(\tilde{F}^*)$, $\dot{c}^*\colon
\Secs(\tilde{E}^*)\to \Secs(\tilde{F}^*)$, $\dot{h}^*\colon \Secs(E)\to
\Secs(\tilde{E}^*)$, $\dot{q}^*\colon \Secs(E)\to \Secs(F)$, which
satisfy $\dot{c}^*\circ \dot{h}^* = \dot{f}^* \circ \dot{q}^*$. Note
that $(\dot{f}^*, \tilde{F}^*)$ and $(\dot{h}^*, \tilde{E}^*)$ are also
symmetric hyperbolic.

When dealing with constrained systems, some results covered in this
section will require the further sufficient condition that the
constraints be parametrizable so that we can extend both the linearized
system and its adjoint to the following commutative diagrams:
\begin{equation}\label{eq:glpar}
\vcenter{\xymatrix{
	0
		\ar[r] &
	\Secs_0(E')
		\ar[d]^{\dot{c}'} \ar[r]^{\dot{h}'} &
	\Secs_0(\tilde{E}^{\prime*})
		\ar[d]^{\dot{q}'} \ar[r]^{\H'} &
	\Secs_{SC}(E')
		\ar[d]^{\dot{c}'} \ar[r]^{\dot{h}'} &
	\Secs_{SC}(\tilde{E}^{\prime*})
		\ar[d]^{\dot{q}'} \ar[r] &
	0  \\
	0
		\ar[r] &
	\Secs_0(F)
		\ar[d]^{\dot{c}} \ar[r]^{\dot{f}} &
	\Secs_0(\tilde{F}^*)
		\ar[d]^{\dot{q}} \ar[r]^{\G} &
	\Secs_{SC}(F)
		\ar[d]^{\dot{c}} \ar[r]^{\dot{f}} &
	\Secs_{SC}(\tilde{F}^*)
		\ar[d]^{\dot{q}} \ar[r] &
	0  \\
	0
		\ar[r] &
	\Secs_0(E)
		\ar[r]^{\dot{h}} &
	\Secs_0(\tilde{E}^*)
		\ar[r]^{\H} &
	\Secs_{SC}(E)
		\ar[r]^{\dot{h}} &
	\Secs_{SC}(\tilde{E}^*)
		\ar[r] &
	0
}}
\end{equation}
and
\begin{equation}\label{eq:glpar*}
\vcenter{\xymatrix{
	0
		\ar@{<-}[r] &
	\Secs_{SC}(\tilde{E}^{\prime*})
		\ar@{<-}[d]^{\dot{c}^{\prime*}} \ar@{<-}[r]^{\dot{h}^{\prime*}} &
	\Secs_{SC}(E')
		\ar@{<-}[d]^{\dot{q}^{\prime*}} \ar@{<-}[r]^{\H^{\prime*}} &
	\Secs_{0}(\tilde{E}^{\prime*})
		\ar@{<-}[d]^{\dot{c}^{\prime*}} \ar@{<-}[r]^{\dot{h}^{\prime*}} &
	\Secs_{0}(E')
		\ar@{<-}[d]^{\dot{q}^{\prime*}} \ar@{<-}[r] &
	0  \\
	0
		\ar@{<-}[r] &
	\Secs_{SC}(\tilde{F}^*)
		\ar@{<-}[d]^{\dot{c}^*} \ar@{<-}[r]^{\dot{f}^*} &
	\Secs_{SC}(F)
		\ar@{<-}[d]^{\dot{q}^*} \ar@{<-}[r]^{\G^*} &
	\Secs_{0}(\tilde{F}^*)
		\ar@{<-}[d]^{\dot{c}^*} \ar@{<-}[r]^{\dot{f}^*} &
	\Secs_{0}(F)
		\ar@{<-}[d]^{\dot{q}^*} \ar@{<-}[r] &
	0  \\
	0
		\ar@{<-}[r] &
	\Secs_{SC}(\tilde{E}^*)
		\ar@{<-}[r]^{\dot{h}^*} &
	\Secs_{SC}(E)
		\ar@{<-}[r]^{\H^*} &
	\Secs_{0}(\tilde{E}^*)
		\ar@{<-}[r]^{\dot{h}^*} &
	\Secs_{0}(E)
		\ar@{<-}[r] &
	0
}}
\end{equation}
The rows form exact sequences, while the columns form elliptic
complexes, as described in Sect.~\ref{sec:caus-green}. Note that the
adjoint diagram also describes a symmetric hyperbolic system with
hyperbolically integrable constraints, except that the role of the
constraint subsystem is now played by $(\dot{q}^{\prime*},E')$ and the
consistency subsystem is $(\dot{h}^{\prime*},\tilde{E}^{\prime*})$,
which satisfies the consistency identity $\dot{h}^{\prime*}\circ
\dot{q}^{\prime*} = \dot{c}^{\prime*}\circ \dot{f}^*$.

Since these systems are linear, their causal structures are field
independent. Since we will be discussing both the linearized system and
its adjoint system, we will refer to the causal structure defined by the
extended compound systems $(\dot{h}'\oplus \dot{f}\oplus \dot{h},
\tilde{E}^{\prime*}\oplus_M \tilde{F}^*\oplus_M \tilde{E}^*)$ and
$(\dot{h}^*\oplus \dot{f}^*\oplus \dot{h}^{\prime*}, \tilde{E}^*\oplus_M
\tilde{F}^*\oplus_M \tilde{E}^{\prime*})$.  In fact, it is easy to check
that their causal structures coincide and in turn coincide with that of
the dynamical linearization point $\phi\in \S_H(M)$ and hence are
globally hyperbolic.  Recall that adjoint symmetric hyperbolic systems
have opposite notions of future and past. We shall always take future
and past to be defined by $\dot{h}'\oplus \dot{f}\oplus \dot{h}$, rather
than its adjoint.

\subsubsection{Gauge transformations}\label{sec:gauge}
Many important classical field theories exhibit gauge invariance, like
Maxwell theory, Yang-Mills theory, and GR. A \emph{gauge transformation}
is a family of maps $g_\eps\colon \Secs(F) \to \Secs(F)$, parametrized
by sections $\eps\in\Secs(P)$ of the \emph{gauge parameter bundle} $P\to
M$, that take solutions to solutions, while not changing modifying a
field section outside the support of $\delta$, $g_\delta[\phi](x) =
\phi(x)$ if $x\not\in \supp \eps$, which may be compact. If we linearize
about some pair of background section $\delta \to \delta + \eps$, we
obtain a linearized gauge transformation $g_\eps[\phi] \to g_\eps[\phi]
+ \dot{g}[\eps]$. It is another requirement on gauge transformations
that the \emph{generator of linearized gauge transformations} $\dot{g}\colon \Secs(P)
\to \Secs(F)$ is a differential operator, which may depend on the
background sections $\delta$ and $\phi$.

Equivalence classes of sections under gauge transformations are
considered physically equivalent. Therefore, physical observables will
consist only of those functions on phase space that are gauge invariant
(constant on orbits of gauge transformations). Equivalently, observables
are annihilated by the action of linearized gauge transformations.
Another way to look at it, is to consider observables as functions on
the space of gauge orbits, denoted $\bar{\Secs}_H(F)$ for the space of
field configurations and $\bar{\S}_H(F)$ for the space of solutions. We
call the space of globally hyperbolic solution sections modulo gauge
transformations, $\bar{S}_H(F) = \S_H(F)/{\sim}$, the \emph{physical
phase space}.

Often it is convenient to impose subsidiary conditions on field sections,
called \emph{gauge fixing}, that restrict the choice of representatives
of gauge equivalence classes. The gauge fixing is called \emph{full} if
they only allow a unique representative from each equivalence class, and
otherwise called partial. The gauge transformations that are compatible
with a partial gauge fixing are called \emph{residual}.

Unfortunately, PDE systems with gauge invariance cannot have a
well-posed initial value problem, and hence cannot be hyperbolized.
However, the addition of subsidiary conditions on field sections can
make the new PDE system equivalent to a hyperbolic one, usually with
constraints.  In practice, many hyperbolic systems with constraints
arise after adding such gauge fixing conditions to a non-hyperbolic
system with gauge invariance.  However, there may remain non-trivial
residual gauge freedom. For later convenience, as we did with
constraints, we restrict our attention to what we call
\emph{recognizable gauge transformations}. That is, given linearized
gauge transformations of the form $\dot{g}[\eps]$ and a partially gauge
fixed symmetric hyperbolic $(\dot{f},\tilde{F}^*)$, we can fit them into
the following commutative diagram, whose columns form elliptic
complexes:
\begin{equation}\label{eq:glrec}
\vcenter{\xymatrix{
	0
		\ar[r] &
	\Secs_0(P)
		\ar[d]^{\dot{g}} \ar[r]^{\dot{k}} &
	\Secs_0(\tilde{P}^*)
		\ar[d]^{\dot{s}} \ar[r]^{\K} &
	\Secs_{SC}(P)
		\ar[d]^{\dot{g}} \ar[r]^{\dot{k}} &
	\Secs_{SC}(\tilde{P}^*)
		\ar[d]^{\dot{s}} \ar[r] &
	0 \\
	0
		\ar[r] &
	\Secs_0(F)
		\ar[d]^{\dot{g}'} \ar[r]^{\dot{f}} &
	\Secs_0(\tilde{F}^*)
		\ar[d]^{\dot{s}'} \ar[r]^{\G} &
	\Secs_{SC}(F)
		\ar[d]^{\dot{g}'} \ar[r]^{\dot{f}} &
	\Secs_{SC}(\tilde{F}^*)
		\ar[d]^{\dot{s}'} \ar[r] &
	0  \\
	0
		\ar[r] &
	\Secs_0(P')
		\ar[r]^{\dot{k}'} &
	\Secs_0(\tilde{P}^{\prime*})
		\ar[r]^{\K'} &
	\Secs_{SC}(P')
		\ar[r]^{\dot{k}'} &
	\Secs_{SC}(\tilde{P}^{\prime*})
		\ar[r] &
	0
}}
\end{equation}
Their adjoints fit into the adjoint diagram whose columns are also
elliptic complexes:
\begin{equation}\label{eq:glrec*}
\vcenter{\xymatrix{
	0
		\ar@{<-}[r] &
	\Secs_{SC}(\tilde{P}^*)
		\ar@{<-}[d]^{\dot{g}^*} \ar@{<-}[r]^{\dot{k}^*} &
	\Secs_{SC}(P)
		\ar@{<-}[d]^{\dot{s}^*} \ar@{<-}[r]^{\K^*} &
	\Secs_{0}(\tilde{P}^*)
		\ar@{<-}[d]^{\dot{g}^*} \ar@{<-}[r]^{\dot{k}^*} &
	\Secs_{0}(P)
		\ar@{<-}[d]^{\dot{s}^*} \ar@{<-}[r] &
	0 \\
	0
		\ar@{<-}[r] &
	\Secs_{SC}(\tilde{F}^*)
		\ar@{<-}[d]^{\dot{g}^{\prime*}} \ar@{<-}[r]^{\dot{f}^*} &
	\Secs_{SC}(F)
		\ar@{<-}[d]^{\dot{s}^{\prime*}} \ar@{<-}[r]^{\G^*} &
	\Secs_{0}(\tilde{F}^*)
		\ar@{<-}[d]^{\dot{g}^{\prime*}} \ar@{<-}[r]^{\dot{f}^*} &
	\Secs_{0}(F)
		\ar@{<-}[d]^{\dot{s}^{\prime*}} \ar@{<-}[r] &
	0  \\
	0
		\ar@{<-}[r] &
	\Secs_{SC}(\tilde{P}^{\prime*})
		\ar@{<-}[r]^{\dot{k}^{\prime*}} &
	\Secs_{SC}(P')
		\ar@{<-}[r]^{\K^{\prime*}} &
	\Secs_{0}(\tilde{P}^{\prime*})
		\ar@{<-}[r]^{\dot{k}^{\prime*}} &
	\Secs_{0}(P')
		\ar@{<-}[r] &
	0
}}
\end{equation}
The systems $(\dot{k},\tilde{P}^*)$ and $(\dot{k}^{\prime*},
\tilde{P}^{\prime*})$ are required to be symmetric hyperbolic and $P'\to
M$ is called the \emph{gauge invariant field bundle}, while $\dot{g}'$
is called the \emph{operator of gauge invariant field combinations}.

The above commutative diagrams are formally similar to those of
symmetric hyperbolic systems with parametrizable constrains, as
described in Sect.~\ref{sec:caus-green}. By analogy, we say that the
gauge transformations are \emph{globally recognizable} if all the
vertical elliptic complexes in the two diagrams above are exact.

\subsubsection{Formal $T$ and $T^*$ for configurations}
\label{sec:tt*-conf}
Here we consider a section $\phi\in \Secs_H(F)$ and examine the formal
tangent and cotangent spaces at $\phi$, $T_\phi\Secs = T_\phi\Secs_H(F)$
and $T^*_\phi\Secs = T^*_\phi \Secs_H(F)$.
\begin{definition}
We define the \emph{formal full tangent space at $\phi$} as the set of
spacelike compact sections We define and the \emph{formal full
cotangent space at $\phi$} as the set
\begin{equation}
	T_\phi\Secs \cong \Secs_{SC}(F)
	\quad\text{and}\quad
	T^*_\phi\Secs \cong \Secs_0(\tilde{F}^*) .
\end{equation}
The natural pairing $\langle-,-\rangle \colon T_\phi\Secs \times
T^*_\phi\Secs\to \R$ is
\begin{equation}
	\langle \psi, \tilde{\alpha}^* \rangle
	= \int_M \psi\cdot \tilde{\alpha}^* .
\end{equation}
\end{definition}
\begin{lemma}
The natural pairing between $T_\phi\Secs$ and $T^*_\phi\Secs$ is
non-degenerate.
\end{lemma}
This is essentially the fundamental lemma of the calculus of variations
and the proof is standard~\cite[\textsection IV.3.1]{ch1}.

Since the physical phase space will be identified with the space of
gauge orbits $\bar{\S}_H(F)$ in the solution space $\S_H(F)$, given a
solution section $\phi\in\S_H(F)$, the formal tangent space
$T_\phi\bar{\S} = T^*_\phi\bar{\S}_H(F)$ at the corresponding
equivalence class $[\phi]\in\bar{\S}_H(F)$ in the space of gauge orbits
consists of equivalence classes of linearized solutions up to linearized
gauge transformations. Dually, the formal cotangent space
$T^*_\phi\bar{S}=T^*_\phi\bar{\S}_H(F)$ will consist of dual densities
annihilated by the adjoint of infinitesimal gauge transformation
generator.

Since gauge transformations act on field configurations and not just
solutions, it makes sense to consider all field configurations related
by gauge transformations as physically equivalent. Thus, we also
introduce $\bar{\Secs}_H(F)$ as the space of gauge orbits of field
configurations as well as its formal tangent and cotangent spaces. The
natural pairing between them is shown to be non-degenerate under the
condition of global recognizability, that is, the vertical elliptic
complexes in diagrams~\eqref{eq:glrec} and~\eqref{eq:glrec*} are exact.
We deal with field configurations first and delay the discussion of
solutions to the next section.

The exactness of the composition $\dot{g}'\circ\dot{g}=0$ ensures that
we can recognize pure gauge field configurations, which are of the form
$\psi=\dot{g}[\eps]$ for some $\phi$-spacelike compact section
$\eps\colon M\to P$, precisely as those $\phi$-spacelike compact field
sections $\psi\colon M\to F$ that give vanishing gauge invariant field
combinations $\dot{g}'[\psi] = 0$.  On the other hand, the exactness of
the dual composition $\dot{g}^*\circ \dot{g}^{\prime*} = 0$ ensures that
we can parametrize gauge invariant, compactly supported dual densities
$\tilde{\alpha}^*\colon M\to \tilde{F}^*$, those satisfying
$\dot{g}^*[\alpha] = 0$, precisely as the image of the differential
operator $\dot{g}^{\prime*}$ acting on compactly supported sections of
$\tilde{P}^{\prime*}\to M$.

Since gauge transformations act on arbitrary sections, not just
solutions, we introduce the formal tangent and cotangent spaces for
arbitrary and solution sections at the same time. As before, fix a
section $\phi$ in $\Gamma_H(F)$ or $\S_H(F)$, as needed.
\begin{definition}
The \emph{formal gauge invariant full tangent space} at $\phi$ is the
set of gauge equivalence classes of $\phi$-spacelike compact sections,
\begin{align}
	T_\phi\bar{\Secs} = T_\phi\bar{\Secs}_H(F)
		&= \{ [\psi] \mid \psi\in \Secs_{SC}(F) \} , \\
	[\psi] &\sim \psi + \dot{g}[\eps], ~
			\text{with}~ \eps\in \Secs_{SC}(P) .
\end{align}
The \emph{formal gauge invariant full cotangent space} at $\phi$ is the
set of compactly supported gauge invariant dual densities,
\begin{equation}
	T_\phi\bar{\Secs} = T^*_\phi\bar{\Secs}_H(F)
		= \{ \tilde{\alpha}^* \in \Secs_0(\tilde{F}^*) \mid
			\dot{g}^*[\alpha]=0 \} .
\end{equation}
The natural pairing $\langle-,-\rangle\colon T_\phi\bar{\Secs} \times
T^*_\phi\bar{\Secs} \to \R$ is 
\begin{equation}
	\langle [\psi], \tilde{\alpha}^* \rangle
	= \int_M \psi\cdot\tilde{\alpha}^* .
\end{equation}
\end{definition}
\begin{lemma}\label{lem:ginv-full-nondegen}
If the gauge transformations are globally recognizable, the natural
pairing between the gauge invariant spaces $T_\phi\bar{\Secs}$ and
$T^*_\phi\bar{\Secs}$ is non-degenerate.
\end{lemma}
\begin{proof}
Non-degeneracy in the second argument follows once again from the
fundamental lemma of the calculus of variations: $\langle [\psi], \alpha
\rangle = \langle \psi, \alpha \rangle = 0$ for all $\psi\in
T_\phi\Secs$ implies that $\alpha = 0$.

Non-degeneracy in the first argument requires an appeal to the global
recognizability of the gauge transformations. Suppose that $\langle
[\psi], \alpha \rangle = 0$ for all $\alpha\in T^*_\phi\bar{\Secs}$.  We
need to show that this implies $\psi = \dot{g}[\eps]$ is pure gauge, for
some $\phi$-spacelike compactly supported $\eps\colon M\to P$. From the
elliptic complex property of the first column of~\eqref{eq:glrec*}, we
are free to use gauge invariant dual densities of the form $\alpha =
\dot{g}^{\prime*}[\tilde{\eps}^{\prime*}]$ for arbitrary
$\tilde{\eps}^{\prime*}\in \Secs_0(\tilde{P}^{\prime*})$ and find
\begin{equation}
	\langle [\psi] , \tilde{\alpha}^* \rangle
	= \langle \psi, \tilde{\alpha}^* \rangle
	= \langle \psi, \dot{g}^{\prime*}[\tilde{\eps}^{\prime*}] \rangle
	= \langle \dot{g}'[\psi], \tilde{\eps}^{\prime*} \rangle .
\end{equation}
Since $\tilde{\eps}^{\prime*}$ could be arbitrary, the vanishing of
$\langle \dot{g}'[\psi], \tilde{\eps}^{\prime*} \rangle$ implies that
$\dot{g}'[\psi] = 0$. On the other hand, from the exactness of the third
column~\eqref{eq:glrec}, we can then conclude that $\psi =
\dot{g}[\eps]$ for some $\eps\in \Secs_{SC}(P)$. This shows that $\psi$
is pure gauge and completes the proof. \qed
\end{proof}

\subsubsection{Formal $T$ and $T^*$ for solutions (without constraints)}%
\label{sec:tt*-sols}
The formal tangent space $T_\phi\S$ will consist of linearized
solutions, that is solutions of the linearized constrained hyperbolic
system $\dot{f}[\psi] = 0$ and $\dot{c}[\psi] = 0$. The formal cotangent
space will naturally consist of equivalence classes of dual densities up
to the images of the adjoints of $\dot{f}$ and $\dot{c}$. After giving
the precise definitions below, we prove that that the natural pairing
between these formal tangent and cotangent spaces is non-degenerate,
provided the constraints $\dot{c}[\psi] = 0$ are trivial. The discussion
of the case with non-trivial constraints is deferred to the next
section.

\begin{definition}
We define the \emph{formal solutions tangent space at $\phi$} as the set of
spacelike compact linearized solution sections,
\begin{equation}
	T_\phi\S = T_\phi\S_H(F) = \{ \psi \in \Secs_{SC}(F) \mid 
		\dot{f}[\psi] = \dot{c}[\psi] = 0 \} .
\end{equation}
We define the \emph{formal solutions cotangent space at $\phi$} as the set
of equivalence classes of compactly supported dual densities,
\begin{align}
	T^*_\phi\S = T^*_\phi\S_H(F) &= \{ [\tilde{\alpha}^*] \mid
		\tilde{\alpha}^* \in \Secs_0(\tilde{F}^*) \} , \\
	[\tilde{\alpha}^*] &\sim \tilde{\alpha}^*
		+ \dot{f}^*[\xi] + \dot{c}^*[\tilde{\eps}^*],  ~~\text{with}~~
		\xi \in \Secs_0(F), ~ \tilde{\eps}^*\in \Secs_0(\tilde{E}^*) .
\end{align}
The natural pairing $\langle-,-\rangle \colon T_\phi\S \times
T^*_\phi\S\to \R$ is
\begin{equation}
	\langle \psi, [\tilde{\alpha}^*] \rangle
		= \int_M \psi\cdot \tilde{\alpha}^* .
\end{equation}
\end{definition}

Before proving the non-degeneracy of the above pairing, we introduce an
auxiliary pairing and prove its non-degeneracy first. Recall that the
Green form (Def.~\ref{def:green-form}) associated to $\dot{f}$ and
$\dot{f}^*$ is orientation positive definite when pulled back to a
future oriented, spacelike surface, or at least the representative in
Eq.~\eqref{eq:pb-gf-rep} is.
\begin{definition}\label{def:green-pairing}
Let $\iota\colon \Sigma\sso M$ be a future oriented, $\phi$-Cauchy
surface and $G$ the Green form associated to $\dot{f}$. The Green
pairing on $\Sigma$ is a bilinear form $\langle -,- \rangle_{G,\Sigma}
\colon \Secs_0(F|_\Sigma)^2 \to \R$ given by
\begin{equation}
	\langle \varphi_1, \varphi_2 \rangle_{G,\Sigma}
	= \int_\Sigma \iota^* G(\varphi_1,\varphi_2)
	= \int_\Sigma \iota^* \bar{f}^t_{ab} \varphi_1^a \varphi_2^b .
\end{equation}
\end{definition}
\begin{lemma}\label{lem:green-posdef}
The Green pairing $\langle -,- \rangle_{G,\Sigma}$ depends only on the
equivalence class $[G]$ and is positive definite (hence non-degenerate).
\end{lemma}
\begin{proof}
Any two representatives $G_1$ and $G_2$ of $[G]$ will differ by an exact
term $\d{H}$, with $H(-,-)$ a bilinear bidifferential operator.
Therefore, the integrands $\iota^* G_i(\varphi_1,\varphi_2)$ will differ
by the exact term $\d\iota^* H(\varphi_1,\varphi_2)$, with necessarily
compact support. Therefore, since $\Sigma$ has no boundary, we can use
any representative of $[G]$ to evaluate the pairing.

As already mentioned above, the representative given by
Eq.~\eqref{eq:pb-gf-rep} is orientation positive,
$\iota^*G(\varphi,\varphi) > 0$ at $s\in \Sigma$, whenever $\varphi(s)
\ne 0$. It then an elementary conclusion that
\begin{equation}
	\langle \varphi, \varphi \rangle_{G,\Sigma}
	= \int_\Sigma \iota^* G(\varphi,\varphi) > 0
\end{equation}
whenever $0\ne \varphi\in \Secs_0(F|_\Sigma)$. Therefore, the Green
paring on $\Sigma$ is positive definite and non-degenerate, which
concludes the proof. \qed
\end{proof}

\begin{lemma}\label{lem:sols-nondegen}
If the constraints $\dot{c}[\phi] = 0$ are trivial, then the natural
pairing between $T_\phi\S$ and $T^*_\phi\S$ is non-degenerate.
\end{lemma}
\begin{proof}
Non-degeneracy in the first argument follows again from the fundamental
lemma of the calculus of variations: $\langle \psi, [\tilde{\alpha}^*] \rangle =
0$ for all $\tilde{\alpha}^*\in \Secs_0(\tilde{F}^*)$ implies that $\psi = 0$.

Non-degeneracy in the second argument is more tricky. Suppose we have
$\langle \psi, [\tilde{\alpha}^*] \rangle = 0$ for all $\phi$-spatially
compact linearized solutions $\psi\in T_\phi\S$. From this we need to
deduce that $[\tilde{\alpha}^*] = [0]$, which means $\tilde{\alpha}^* =
\dot{f}^*[\xi]$ for some compactly supported $\xi\in \Secs_0(F)$.  Let
$K=\supp \tilde{\alpha}^*$ and $\iota_+\colon\Sigma^+ \sso M$ a future
oriented, $\phi$-Cauchy surface in the future of $K$ and not
intersecting it. Let $\theta_-$ be the characteristic function of
$I^-(\Sigma^+)$, the past of $\Sigma^+$. If $\Sigma^+$ is defined as a
regular zero set $t=0$ of smooth function $t$ on $M$, then $\d\theta_- =
-\delta(t)\,\d{t}$.

Also, let
$\xi_+ = \G^*_+[\tilde{\alpha}^*]$ be the retarded solution of
$\dot{f}^*[\xi_+] = \tilde{\alpha}^*$, so then
$\supp\xi_+ \sse I^+(K)$. Note that the intersection $I^-(\Sigma^+)
\cap I^+(K)$ is compact and contains $K$. In particular, we have $\theta_-
\tilde{\alpha}^* = \tilde{\alpha}^*$, from which follows
\begin{align}
	\langle \psi, [\tilde{\alpha}^*] \rangle
		&= \langle \psi, \theta_- \tilde{\alpha}^* \rangle
		= \int_M \theta_-\psi \cdot \dot{f}^*[\xi_+] \\
		&= -\int_M \theta_-
				(\dot{f}[\psi]\cdot\xi_+ - \psi\cdot\dot{f}^*[\xi_+])
		= -\int_M \theta_- \d G(\psi,\xi_+) \\
		&= \int_M \d\theta_-\wedge G(\psi,\xi_+)
		= -\int_M \delta(t)\,\d{t}\wedge G(\psi,\xi_+)	\\
		&= -\int_{\Sigma^+} \iota_+^* G(\psi,\xi_+)
		= -\langle \psi|_{\Sigma^+}, \xi_+|_{\Sigma^+} \rangle_{G,\Sigma^+} .
\end{align}
Since $\psi$ is allowed to be any $\phi$-spatially compactly supported
linearized solution, its restriction $\psi|_{\Sigma^+}$ to the Cauchy
surface $\Sigma^+$ could be any element of $\Secs_0(F|_{\Sigma^+})$. The
non-degeneracy of the Green pairing, Lem.~\ref{lem:green-posdef}, then
implies that the restriction $\xi_+|_{\Sigma^+}$ is identically zero.

On the other hand, by global hyperbolicity, we can foliate the future of
$\Sigma_+$ with other $\phi$-Cauchy surfaces and apply the same argument
to each of them. Then $\xi_+$ must vanish identically in the future of
$\Sigma_+$, so its support is contained in the intersection
$I^-(\Sigma_+)\cap I^+(K)$, which we have already noted is compact.
Finally, we have the desired conclusion that $\tilde{\alpha}^* =
\dot{f}^*[\xi]$ with $\xi = \xi_+\in \Secs_0(F)$. \qed
\end{proof}

In the presence of gauge symmetries, the formal tangent space consists
of equivalence classes of linearized solutions up to gauge
transformations. On the other hand, the formal cotangent space is
restricted to equivalence represented by gauge invariant dual densities.
After giving the precise definitions below, we prove that the natural
pairing between these formal tangent and cotangent spaces is
non-degenerate, provided the constraints $\dot{c}[\phi] = 0$ are trivial
and the gauge transformation are globally recognizable. The discussion
of the case with non-trivial constraints is deferred to the next
section.
\begin{definition}\label{def:tt*-sols-gauge}
We define the \emph{formal gauge invariant solutions tangent space} at
$\phi$ as the set of gauge equivalence classes of $\phi$-spacelike
compact linearized solution sections,
\begin{align}
\notag
	T_\phi\bar{\S} = T_\phi\bar{\S}_H(F)
		&= \{ [\psi] \mid \psi\in \Secs_{SC}(F), \dot{f}[\psi] = 0,
			\dot{c}[\phi] = 0 \} , \\
	[\psi] &\sim \psi + \dot{g}[\eps], ~\text{with}~ 
		\eps\in \Secs_{SC}(P) .
\end{align}
The \emph{formal gauge invariant solutions cotangent space} at $\phi$ is
the set of equivalence classes of compactly supported gauge invariant
dual densities,
\begin{align}
\notag
	T^*_\phi\bar{\S} = T^*_\phi\bar{\S}_H(F)
		&= \{ [\tilde{\alpha}^*] \mid \tilde{\alpha}^* \in \Secs_0(\tilde{F}^*) ,
			\dot{g}^*[\tilde{\alpha}^*] = 0 \} , \\
	[\tilde{\alpha}^*] &\sim \tilde{\alpha}^*
		+ \dot{f}^*[\xi] + \dot{c}^*[\tilde{\eps}^*], \\
		&{} \quad \text{with}~
			\xi\in \Secs_0(F), \tilde{\eps}^* \in \Secs_0(\tilde{E}^*) \\
		&{} \quad \text{and}~
			\dot{g}^*[\dot{f}^*[\xi] + \dot{c}^*[\tilde{\eps}^*]] = 0 .
\end{align}
The natural pairing $\langle-,-\rangle\colon T_\phi\bar{\S}\times
T^*_\phi\bar{\S} \to \R$ is
\begin{equation}
	\langle [\psi], [\tilde{\alpha}^*] \rangle
		= \int_M \psi\cdot\tilde{\alpha}^* .
\end{equation}
\end{definition}

\begin{lemma}\label{lem:ginv-sols-nondegen}
If the constraints $\dot{c}[\phi] = 0$ are trivial and the gauge
transformations are globally recognizable (cf.~diagrams~\eqref{eq:glrec}
and~\eqref{eq:glrec*}), then the natural pairing between $T_\phi\bar\S$
and $T^*_\phi\bar\S$ is non-degenerate.
\end{lemma}
\begin{proof}
Unfortunately, we cannot directly rely on the fundamental lemma of the
calculus of variations to prove non-degeneracy in the second argument of
$\langle [\psi], [\tilde{\alpha}^*] \rangle = \langle \psi,
\tilde{\alpha}^* \rangle$, since we are only allowed to use equivalence
classes represented by gauge invariant dual densities,
$\dot{g}^*[\tilde{\alpha}^*] = 0$. Fortunately, any dual density of the
form $\tilde{\alpha}^* = \dot{g}^{\prime*}[\tilde{\beta}^{\prime*}]$,
with unrestricted choice of $\tilde{\beta}^{\prime*} \in
\Secs_0(\tilde{P}^{\prime*})$, is an allowed representative. Then the
fundamental lemma of the calculus of variations and
\begin{equation}
	\langle [\psi], [\tilde{\alpha}^*] \rangle
	= \langle \psi, \dot{g}^{\prime*}[\tilde{\beta}^{\prime*}] \rangle
	= \langle \dot{g}'[\psi], \tilde{\beta}^{\prime*} \rangle = 0 ,
\end{equation}
for arbitrary $\tilde{\beta}^{\prime*}$, imply that $\dot{g}'[\psi] =
0$. By global recognizability of gauge transformations, or more
precisely the exactness of the third column of
diagram~\eqref{eq:glrec}, we have $\psi = \dot{g}[\eps]$ for some
$\eps\in \Secs_{SC}(P)$. That is, $\psi$ must be pure gauge, $[\psi] =
[0]$, which proves non-degeneracy of the natural pairing in the second
argument.

To prove non-degeneracy in the first argument, we proceed as in the
proof of Lem.~\ref{lem:sols-nondegen}. Namely, assume that we have a
gauge invariant dual density $\dot{g}^*[\tilde{\alpha}^*]=0$, such that $\langle
[\psi], [\tilde{\alpha}^*] \rangle$ for all $[\psi] \in T_\phi\bar\S$,
and let $\xi_+ = \G^*_+[\tilde{\alpha}^*]$, so that $\dot{f}^*[\xi_+] =
\tilde{\alpha}^*$.  Also, let $K = \supp\tilde{\beta}^{\prime*}$,
$\Sigma^+ \sso M$ a future oriented, $\phi$-Cauchy surface that does not
intersect $K$ itself but does intersect $I^+(K)$, and $\theta_-$ the
characteristic function of $I^-(\Sigma^+)$. Then
\begin{equation}
	\langle [\psi], [\tilde{\alpha}^*] \rangle
	= \langle \psi, \theta_- \dot{f}^*[\xi_+] \rangle
	= - \langle \psi|_{\Sigma^+}, \xi_+|_{\Sigma^+} \rangle_{G,\Sigma^+} ,
\end{equation}
which implies that $\supp \xi_+ \sso I^-(\Sigma^+)$ and hence is
compact. In other words, with $\xi = \xi_+ \in \Secs_0(F)$, we have
$[\tilde{\alpha}^*] = [\dot{f}^*[\xi]] = [0]$ and the that natural
pairing is non-degenerate in the first argument as well. \qed
\end{proof}

\subsubsection{Formal $T$ and $T^*$ for solutions (with constraints)}
\label{sec:tt*-sols-constr}
When the constraints are non-trivial, the proof that
the natural pairing between formal tangent and cotangent spaces is
non-degenerate complicates a bit. In fact, it requires the sufficient
condition that the constraints be globally pa\-ra\-metr\-izable. We
first handle the case without gauge invariance.

\begin{lemma}\label{lem:sols-constr-nondegen}
If the constraints $\dot{c}[\phi] = 0$ are globally parametrizable
(cf.~diagrams~\eqref{eq:glpar} and~\eqref{eq:glpar*}), then the natural
pairing between $T_\phi\S$ and $T^*_\phi\S$ is non-degenerate.
\end{lemma}
\begin{proof}
Non-degeneracy in the first argument follows again from the fundamental
lemma of the calculus of variations: $\langle \psi, [\tilde{\alpha}^*]
\rangle = \langle \psi, \tilde{\alpha}^* \rangle = 0$ for all $\alpha\in
\Secs_0(\tilde{F}^*)$ implies that $\psi=0$.

Non-degeneracy in the second argument means that $\langle \psi,
[\tilde{\alpha}^*] \rangle = 0$ for all $\psi\in T_\phi\S$ implies
$[\tilde{\alpha}^*] = [0]$, which means that
$\dot{g}^*[\tilde{\alpha}^*] = 0$ and $\tilde{\alpha}^* = \dot{f}^*[\xi]
+ \dot{c}^*[\tilde{\eps}^*]$ for some $\xi\in \Secs_0(F)$ and
$\tilde{\eps}^*\in \Secs_0(\tilde{E}^*)$.  Using the same approach as in
the unconstrained case, define $\xi_+ = \G^*_+[\tilde{\alpha}^*]$, so
that $\dot{f}^*[\xi_+] = \tilde{\alpha}^*$.  Also, let $K=\supp \alpha$
and $\Sigma^+\sso M$ a future oriented, $\phi$-Cauchy surface that does
not intersect $K$ itself but does intersect $I^+(K)$.  Unfortunately, we
will not be able to show that $\xi_+$ vanishes outside a compact set,
but instead we will show that $\dot{q}^{\prime*}[\xi_+]$ vanishes
outside a compact set.

To do that, we first need the identity
\begin{equation}
	\langle \dot{c}'[\psi']|_{\Sigma^+}, \xi_+|_{\Sigma^+} \rangle_{G,\Sigma^+}
	= \langle \psi'|_{\Sigma^+}, \dot{q}^{\prime*}[\xi_+]|_{\Sigma^+}
		\rangle_{H',\Sigma^+} ,
\end{equation}
where $H'$ is the Green form associated to the symmetric hyperbolic
operator $\dot{h}'$ and $\psi'\in \Secs_{SC}(E')$ satisfies
$\dot{h}'[\psi'] = 0$, hence also $\dot{f}\circ\dot{c}'[\psi] = 0$. Let
$C'$ denote the Green forms associated to the operator $\dot{c}'$. This
identity then follows from
\begin{align}
	\d G(\dot{c}'[\psi'],\xi_+)
	&= -\dot{c}'[\psi']\cdot \dot{f}^*[\xi_+] \\
	&= -\psi'\cdot \dot{c}^{\prime*}[\dot{f}^*[\xi_+]]
			-\d C'(\psi',\dot{f}^*[\xi_+]) \\
	&= -\psi'\cdot \dot{h}^{\prime*}[\dot{q}^{\prime*}[\xi_+]]
			-\d C'(\psi',\tilde{\alpha}^*) \\
	&= \d H'(\psi',\dot{q}^{\prime*}[\xi_+])
			-\d C'(\psi',\tilde{\alpha}^*)
\end{align}
and the fact that $\supp C'(\psi',\tilde{\alpha}^*) \sse \supp
\tilde{\alpha}^*$ does not intersect $\Sigma^+$. Since we are certainly
allowed to use $\psi = \dot{c}'[\psi']$, using the same argument as in
Lem.~\ref{lem:sols-nondegen}, we have
\begin{equation}
	\langle \dot{c}'[\psi], [\tilde{\alpha}^*] \rangle
	= -\langle \dot{c}'[\psi']|_{\Sigma^+}, \xi_+|_{\Sigma^+} \rangle_{G,\Sigma^+}
	= -\langle \psi'|_{\Sigma^+}, \dot{q}^{\prime*}[\xi_+]|_{\Sigma^+}
		\rangle_{H',\Sigma^+} ,
\end{equation}
which forces $\dot{q}^{\prime*}[\xi_+] = 0$ on $\Sigma^+$. In fact, we
can also conclude that the same equality holds on a neighborhood $S^+$
of $\Sigma^+$ that contains $I^+(\Sigma^+)$, on which we have
$\dot{f}^*[\xi_+] = 0$ as well.

Now, let $\eta$ denote the unique solution of $\dot{f}^*[\eta] = 0$ on
$M$ such that $\eta = \xi_+$ on $S^+$. Since, by the global
parametrizability assumption (cf.~diagram~\eqref{eq:glpar*}) the
constraint $\dot{q}^{\prime*}[\eta] = 0$ is hyperbolically integrable,
it is satisfied everywhere in $M$ as well. Therefore, from the results
of Sect.~\ref{sec:caus-green}, it follows that there exists
$\tilde{\eps}^*\in \Secs_0(\tilde{E}^*)$ such that $\eta =
\dot{q}^*[\H^*[\tilde{\eps}^*]]$.

Let $\eta_+ = \dot{q}^*[\H_+^*[\tilde{\eps}^*]]$ and $\xi = \xi_+ -
\eta_+$. Then, by construction, $\xi$ has compact support and
\begin{align}
	\tilde{\alpha}^*
	&= \dot{f}^*[\xi_+] = \dot{f}^*[\xi] + \dot{f}^*[\eta_+] \\
	&= \dot{f}^*[\xi] + \dot{f}^*[\dot{q}^*[\H_+^*[\tilde{\eps}^*]]]
	= \dot{f}^*[\xi] + \dot{c}^*[\dot{h}^*[\H_+^*[\tilde{\eps}^*]]] \\
	&= \dot{f}^*[\xi] + \dot{c}^*[\tilde{\eps}^*] ,
\end{align}
as desired, implying $[\tilde{\alpha}^*] = [0]$, which completes the
argument that the natural pairing is non-degenerate in the second
argument. \qed
\end{proof}

The final non-degeneracy result is for the case where both non-trivial
constraints and gauge transformations are present.

\begin{lemma}\label{lem:ginv-constr-nondegen}
If the constraints are globally parametrizable
(cf.~diagrams~\eqref{eq:glpar} and~\eqref{eq:glpar*}) and the gauge
transformations are globally recognizable (cf.~diagrams~\eqref{eq:glrec}
and~\eqref{eq:glrec*}), then the natural pairing between $T_\phi\bar\S$
and $T^*_\phi\bar\S$ is non-degenerate.
\end{lemma}
The proof follows straight forwardly by combining the methods of the
preceding non-degeneracy lemmas, namely
Lem.~\ref{lem:ginv-sols-nondegen} and
Lem.~\ref{lem:sols-constr-nondegen}.

\subsubsection{Formal local differential forms}\label{sec:formal-forms}
In this section, we use differential forms on $J^\oo F$ to construct
formal differential forms on $\Secs_H(F)$ or $\S_H(F)$.

Consider any form $\alpha\in\Forms^{h,v}(F)$ and suppose that
$\phi\colon M\to F$ is any section. We already know that the pullback
along the jet prolongation $j^\oo\phi$ intertwines the horizontal
differential $\dh$ with the de~Rham differential on $M$,
$\d(j^\oo\phi)^*\alpha = (j^\oo\phi)^* \dh\alpha$. However, if $\alpha$
has non-zero vertical degree, $v\ne 0$, then automatically
$(j^\oo\phi)^*\alpha = 0$ from the defining properties of vertical
forms. Though, if $\alpha$ is first contracted with $v$ evolutionary
vector fields $\hat{\psi}_j$ before effecting the pullback, the result
\begin{equation}
	(j^\oo\phi)^*\iota_{\hat{\psi}_{1}}\cdots\iota_{\hat{\psi}_{v}} \alpha ,
\end{equation}
being the pullback of vertical degree $0$ form, is not-necessarily zero.
Recall that, a section $\psi\colon M\to F$ defines a corresponding
evolutionary vector field $\hat{\psi}$ that has the action
$\hat{\psi}(x^i) = 0$, $\hat{\psi}(u^a_I) = \del_I \psi^a$ in local
coordinates $(x^i,u^a_I)$.

In fact, if we were concerned with calculus on $\Secs_H(F)$ or
$\S_H(F)$, the infinite dimensional manifolds of sections of $F\to M$
and of solution sections of a PDE system, we could use vertical
differential forms on $J^\oo F$ to define \emph{local} differential
forms on the manifold of sections. A \emph{point local} form $A$ depends
only on a specific point $x\in M$. That is, evaluated at
$\phi\in\Secs_H(F)$, $A$ depends only on $\phi(x)$ and finitely many of
its derivatives at the same point. A \emph{local form} is an integral
over an $M$-family of point local forms, with respect to some smooth
density on $M$. A degree $v$ form on $\Secs_H(F)$ or $\S_H(F)$ would be
a functional of $v$ tangent vectors $\psi_j$, which would correspond to
sections of $F\to M$ or linearized solution sections. Since we are not
delving into the details of how these infinite dimensional objects are
defined, we are content to formally associate a local differential form
$A_N[\phi]$ of degree $v$ on $\Secs_H(F)$ to each differential form
$\alpha$ of degree $(h,v)$ on $J^\oo F$ as follows:
\begin{equation}\label{eq:vform-def}
	A_N(\psi_1,\ldots,\psi_v)[\phi]
	= \int_N (j^\oo\phi)^*[\iota_{\hat{\psi}_v}\cdots\iota_{\hat{\psi}_1}
		\alpha] ,
\end{equation}
where $N\sse M$ is a submanifold of dimension $h$.

Recall that, given local coordinates $(x^i,u^a)$ on $F$, an evaluation
functional is a function of $\Secs_H(F)$ of the form $\Phi_x^a[\phi] =
\phi^a(x)$. A vector field $\Psi$ on $\Secs_H(F)$ with $(\Lie_\Psi
\Phi_x^a)[\phi]$ depending only on $\phi(x)$ and finitely many
derivatives at the same point is called a \emph{local vector field}.  We
can formally associate a local vector $\Psi$ field to any evolutionary
vector field $\hat{\psi}$ on $J^\oo F$, by defining its action on formal
local $0$-forms on $\Secs_H(F)$ by
\begin{equation}
	(\Lie_\Psi A_N)[\phi] = \int_N (j^\oo\phi)^* \Lie_{\hat{\psi}} \alpha .
\end{equation}
This is enough to define a formal exterior derivative $\delta$ on
$\Secs_H(F)$, whose action on formal local forms can be formally given
by the standard expressions with contractions and Lie derivatives. It
boils down to
\begin{align}
	(\delta A_N)(\Psi_0,\cdots,\Psi_v)[\phi]
	&= \sum_{j=0}^v (-)^j [\Lie_{\Psi_j}
			(\iota_{\Psi_v}\cdots \check{\iota}_{\Psi_j}
				\cdots \iota_{\Psi_0} A_N)[\phi] \\
\notag & {} \qquad
			- \Lie_{\Psi_j}(\iota_{\Psi_v}\cdots \iota_{\Psi_{j+1}})
				\iota_{\Psi_{j-1}} \cdots \iota_{\Psi_0} A_N[\phi] ] \\
	&= \sum_{j=0}^v (-)^j \int_N (j^\oo\phi)^*
		[\Lie_{\hat{\psi}_j}
			(\iota_{\hat{\psi}_v}\cdots \check{\iota}_{\hat{\psi}_{j}}
				\cdots\iota_{\hat{\psi}_0}\alpha) \\
\notag
	&\qquad {}
		- \Lie_{\hat{\psi}_j}(\iota_{\hat{\psi}_v}\cdots\iota_{\hat{\psi}_{j+1}})
				\iota_{\hat{\psi}_{j-1}}\cdots\iota_{\hat{\psi}_0} \alpha] \\
	&= (j^\oo\phi)^*\iota_{\hat{\psi}_v}\cdots\iota_{\hat{\psi}_0}
			\dv \alpha ,
\end{align}
where $\check{}$ means that the marked symbol is omitted. The last
formula gives the real significance of the vertical differential. To
derive the first equality, one must make use of the formal
identification~\eqref{eq:vform-def} and the identity~\eqref{eq:jet-Lie}.
The second equality results from simplifications in the usual
differential calculus on $J^\oo F$.

\subsection{Symplectic and Poisson structure}\label{sec:symp-pois}
In this section, we endow the space of solutions $\S_H(F)$ of a
variational PDE system with the structures of both a symplectic and a
Poisson manifold (or rather formal versions of these structures), really
turning it into the phase space of classical field theory. First, a
variational system has to be put into symmetric hyperbolic forms,
possibly with constraints. If the constraints are globally
parametrizable, the gauge transformations globally recognizable and the
hyperbolization of the variational Euler-Lagrange system satisfies an
extra sufficient condition (which identifies part of the constraints
with gauge fixing conditions), then we can use a generalized Peierls
formula to construct the (formal) Poisson bivector and hence show that
the covariant presymplectic structure is invertible and hence actually
symplectic.

\subsubsection{Constraints and gauge fixing}\label{sec:constr-gf}
Consider a Lagrangian density $\L\in \Omega^{n,0}(F)$ of order $k$ and the
corresponding Euler-Lagrange PDE system $\E_\EL\sso J^{2k}F$, with equation
form $(\EL,\tilde{F}^*)$. More often then not, the EL equation form is
not directly an adapted equation form for a symmetric hyperbolic system.
That could be true for several reasons. It is possible that $\E_\EL$ is
in fact symmetric hyperbolic and that the EL equations need only be
hyperbolized. However, if the EL equations are not of first order, then
would need to be first reduced to first order form (this may be done
already at the level of the Lagrangian, by introducing auxiliary
fields). On the other hand, the EL equations may only admit a
hyperbolization after several prolongations that identify all the
relevant integrability conditions. Finally, if the Lagrangian is
singular (it has non-trivial gauge invariance), the EL system cannot be
hyperbolized at all, since systems with gauge invariance do not have a
well-posed initial value problem. In that case, one would first have to
introduce gauge fixing conditions, whose conjunction with the EL system
is hyperbolizable (Sect.~\ref{sec:integrability}).

Below, we postulate a condition on the hyperbolization of a gauge fixed
variational EL system that in the next section will be shown to be
sufficient to construct certain Green functions needed for the Peierls
formula for the Poisson bracket. By a hyperbolization we mean an
equivalence between the gauge fixed EL equations for $(\EL\oplus
c_g,\tilde{F}^*\oplus E_g)$ and a constrained symmetric hyperbolic
system $(f\oplus c, \tilde{F}^*\oplus E)$. Here, $c_g\colon \Secs(F)\to
\Secs(E_g)$ is a differential operator constituting the gauge fixing
condition. The equivalence is presumed to have the following form
\begin{equation}
	\left\{
		\begin{aligned}
			\EL &= R\circ(f\oplus c)  \\
			c_g &= R_g\circ c  \\
		\end{aligned}
	\right.
	\quad \iff \quad
	\left\{
		\begin{aligned}
			f &= \bar{R}\circ(\EL\oplus c_g)  \\
			c &= \bar{R}_g \circ (\EL\oplus c_g) 
		\end{aligned}
	\right. ,
\end{equation}
where the $R$, $\bar{R}$, $R_g$ and $\bar{R}_g$ are (possibly
non-linear) differential operators.

Since we will be concerned with symplectic and Poisson tensors, we can
work with individual (formal) tangent and cotangent spaces at a fixed
dynamical linearization point $\phi\in \S_H(F)$, which is a globally
hyperbolic solution of the PDE system defined by the equation form
$(f\oplus c, \tilde{F}^*\oplus E)$. In other words, we need to work with
the linearized versions of each of the hyperbolic system, the
constraints, the Euler-Lagrange system, the gauge fixing conditions, the
gauge transformations, as well as the hyperbolization. As before, we
denote the adapted equation form of the linearized hyperbolic system
$(\dot{f},\tilde{F}^*)$. The linearized constraints are presumed to be
globally parametrizable and fit into the commutative
diagrams~\eqref{eq:glpar} and~\eqref{eq:glpar*}. The linearized gauge
transformations are presumed to be globally recognizable and fit into
the commutative diagrams~\eqref{eq:glrec} and~\eqref{eq:glrec*}. The
linearized EL equations are denoted $(\J,\tilde{F}^*)$ and are also
called the \emph{Jacobi system}, with $\J$ the Jacobi
operator~\cite{dewitt-qft}, while the linearized gauge fixing
conditions are denoted by
the equation form $(\dot{c}_g,E_g)$. In local coordinates $(x^i,u^a)$ on
$F$, the components of the Jacobi operator satisfy the identity
\begin{equation}
	\J^I_{ab} \wedge \dv u^b_I = \dv \EL_a .
\end{equation}
The equivalence of the linearized systems takes the following form:
\begin{equation}\label{eq:Jgf-equiv}
	\left\{
		\begin{aligned}
			\J &= r\circ \dot{f} + r_c\circ \dot{c}  \\
			\dot{c}_g &= r_g\circ \dot{c}  \\
		\end{aligned}
	\right.
	\quad \iff \quad
	\left\{
		\begin{aligned}
			\dot{f} &= \bar{r}\circ \J + \bar{r}_c\circ \dot{c}_g \\
			\dot{c} &= \bar{r}_\J\circ \J + \bar{r}_g\circ \dot{c}_g
		\end{aligned}
	\right.
\end{equation}
If the operator $\bar{r}_\J$ is non-vanishing, it means that part of the
constraints consist of integrability conditions of the Jacobi system.

Note that, strictly speaking, the $r$- and $\bar{r}$- differential
operators effecting the equivalence are not inverses of each other.
Their compositions may differ from the identity by some differential
operator that factors through a differential identity, that is,
$\dot{q}\circ \dot{f} - \dot{h}\circ \dot{c} = 0$ or $\dot{g}^*\circ \J
= 0$. In other words, we must have
\begin{align}
\label{eq:rinv1}
	r\circ \bar{r} + r_c\circ\bar{r}_\J &= \id + p_\J \circ \dot{g}^* ,
		& \bar{r} \circ r &= \id + p_f \circ \dot{q} , \\
\label{eq:rinv2}
	r\circ \bar{r}_c + r_c\circ \bar{r}_g &= 0 ,
		& \bar{r}\circ r_c + \bar{r}_c\circ r_g &= -p_f\circ \dot{h} , \\
\label{eq:rinv3}
	r_g\circ \bar{r}_\J &= p_g \circ \dot{g}^* ,
		& \bar{r}_\J \circ r &= p_c\circ \dot{q} , \\
\label{eq:rinv4}
	r_g\circ \bar{r}_g &= \id ,
		& \bar{r}_g\circ r_g + \bar{r}_\J \circ r_c &= \id - p_c\circ \dot{h} ,
\end{align}
for some differential operators $p_\J$, $p_f$, $p_g$ and $p_c$. Also,
the identity $\dot{q}\circ \dot{f} - \dot{h}\circ \dot{c} = 0$, when
expressed in terms of the $\J$ and $\dot{c}_g$ operators, is identically
satisfied when
\begin{align}
\label{eq:consist1}
	\dot{q}\circ \bar{r} - \dot{h}\circ \bar{r}_\J &= q_\J\circ \dot{g}^*, \\
\label{eq:consist2}
	\dot{q}\circ \bar{r}_c - \dot{h}\circ \bar{r}_g &= 0 .
\end{align}
It is worth noting that the above relations involving the $r$- and
$\bar{r}$-operators follow from the equivalence~\eqref{eq:Jgf-equiv}
only when $\J$ and $\dot{c}_g$ satisfy no additional differential
identities. However, we will simply presume that they hold as needed
sufficient conditions for the derivation of the Peierls formulas later
in Sect.~\ref{sec:formal-pois}.

Finally, to make sure that the condition $\dot{c}_g[\psi] = 0$ in fact
constitutes a gauge fixing condition, we require the following
compatibility between the gauge transformation operator and the
constraints that we shall refer to as the \emph{gauge fixing
compatibility} condition:
\begin{equation}\label{eq:gf}
	\dot{s}' \circ \bar{r}_c\circ\dot{c}_g = 0
\end{equation}
This condition connects constraints (represented by $\dot{c}_g$) and
gauge transformations (represented by $\dot{s}'$). Essentially, this
condition says that the part of the constraints $\dot{c}[\psi] = 0$ that
comes from $\dot{c}_g[\psi] = 0$ is sufficient, when adjoined to
$\J[\psi] = 0$ to make the gauge fixed system hyperbolizable.

To summarize, we require that the PDE system of interest satisfies the
following conditions: (a) variationality, (b) hyperbolizability, (c)
global parametrizability of constraints, (d) global recognizability of
gauge transformations, and (e) gauge fixing compatibility. The
consequences of imposing these conditions are explored in the following
section. We stress that these conditions are sufficient for our purposes
and can in fact satisfied by most fundamental physical field theories,
but some of the same results could also hold under weaker conditions.

\subsubsection{Causal Green functions}\label{sec:Jgreen}
The goal of this section is to use the gauge fixing compatibility
condition~\eqref{eq:gf} to construct a causal Green function for the
gauge fixed Jacobi system.

First, recall that the Jacobi system, due to its variational character
is easily shown to be self-adjoint:
\begin{equation}
	\J^* = \J .
\end{equation}
Also, gauge invariance and N\"other's second theorem imply the identities
\begin{equation}
	\J \circ \dot{g} = 0
	\quad\text{and}\quad
	\dot{g}^* \circ \J = 0 .
\end{equation}
The equivalence of the gauge fixed Jacobi system and a constrained
hyperbolic system postulated in~\eqref{eq:Jgf-equiv} then gives
\begin{equation}
	\dot{s}\circ \dot{k}
	= \dot{f}\circ \dot{g}
	= \bar{r}_c \circ \dot{c}_g \circ \dot{g} .
\end{equation}
Suppose that $\psi \in \Secs_{SC}(F)$ such that $\dot{c}_g[\psi] = 0$
and $\psi = \dot{g}[\eps']$ for some $\eps'\in\Secs_{SC}(P)$. Then
$\dot{s}\circ \dot{k}[\eps'] = \bar{r}_c \circ \dot{c}_g \circ
\dot{g}[\eps'] = 0$, or $\dot{k}[\eps'] = \tilde{\beta}^*$, where
$\tilde{\beta}^* \in \ker \dot{s} \sso \Secs_{SC}(\tilde{P}^*)$.  Let
$\{\chi_\pm\}$ be a partition of unity adapted to a Cauchy surface
(Def.~\ref{def:adapt-pu}) and recall the associated splitting maps
(Lem.~\ref{lem:exsplit}), which lead to the identity
\begin{equation}
	\dot{g}[\K_\chi[\tilde{\beta}^*]]
	= \sum_\pm \G_\pm[\dot{s}[\chi_\pm \tilde{\beta}^*]]
	= \G[\dot{s}_\chi[\tilde{\beta}^*]] ,
\end{equation}
where $\dot{s}_\chi[\tilde{\beta}^*] = \pm \dot{s}[\chi_\pm
\tilde{\beta}^*] \in \Secs_0(\tilde{F}^*)$. Note that
$\dot{s}'[\dot{s}_\chi[\tilde{\beta}^*]] = 0$, which by global
recognizability implies that there exists a $\tilde{\gamma}^*\in
\Secs_0(\tilde{P}^*)$ such that $\dot{s}[\tilde{\gamma}^*] =
\dot{s}_\chi[\tilde{\beta}^*]$. Note that $\tilde{\gamma}^*$ cannot be
just $\chi_+\tilde{\beta}^*$ or $-\chi_-\tilde{\beta}^*$, because their
supports need not be compact. Now, let
$\eps = \eps' - \K_\chi[\tilde{\beta}^*] + \K[\tilde{\gamma}^*]$ and
notice that $\dot{k}[\eps] = 0$ as well as
\begin{align}
	\dot{g}[\eps]
	&= \dot{g}[\eps'] - \dot{g}[\K_\chi[\tilde{\beta}^*]]
		+ \dot{g}[\K[\tilde{\gamma}^*]] \\
	&= \psi - \G[\dot{s}_\chi[\tilde{\beta}^*]] + \G[\dot{s}[\tilde{\gamma}^*]] \\
	&= \psi .
\end{align}
We have just proven
\begin{lemma}
If there exists $\psi\in \Secs_{SC}(F)$ such that $\dot{c}_g[\psi] = 0$
and $\psi\in \im \dot{g}$, there exists $\eps\in \Secs_{SC}(P)$ such
that $\dot{k}[\eps] = 0$ and $\psi = \dot{g}[\eps]$.
\end{lemma}

Equivalence with a constrained hyperbolic system now allows us to solve
the inhomogeneous problem
\begin{equation}
	\J[\psi] = \tilde{\alpha}^*, \quad \dot{c}_g[\psi] = 0 ,
\end{equation}
where the source must necessarily satisfy the gauge invariance condition
$\dot{g}^*[\tilde{\alpha}^*] = 0$. The equivalent inhomogeneous problem in
symmetric hyperbolic form is
\begin{equation}
	\dot{f}[\psi] = \bar{r}[\tilde{\alpha}^*], \quad
	\dot{c}[\psi] = \bar{r}_\J[\tilde{\alpha}^*] .
\end{equation}
Recall from Lem.~\ref{lem:inhom-constr} that this system is solvable iff
the sources satisfy the consistency identity:
\begin{equation}
	\dot{q}[\bar{r}[\tilde{\alpha}^*]]
		- \dot{h}[\bar{r}_\J[\tilde{\alpha}^*]]
	= p_\J\circ \dot{g}^*[\tilde{\alpha}^*]
	= 0 ,
\end{equation}
which is obviously satisfied, after using identity~\eqref{eq:consist1}, for any gauge
invariant source. The retarded and advanced solutions to this
inhomogeneous problem are then $\psi_\pm = \G_\pm
[\bar{r}[\tilde{\alpha}^*]]$. This means that
$\bar{r}_\J[\tilde{\alpha}^*] =
\dot{c}[\G_\pm[\bar{r}[\tilde{\alpha}^*]]$ and, in particular,
$\dot{c}_g[\psi] = 0$. Motivated by this formula, we introduce the
following retarded, advanced and causal Green functions for the gauge
fixed Jacobi system:
\begin{equation}\label{eq:E-def}
	\EE_\pm = \G_\pm \circ \bar{r}
	\quad\text{and}\quad
	\EE = \EE_+ - \EE_- = \G\circ \bar{r} .
\end{equation}
One can immediately check that $\psi = \EE[\tilde{\alpha}^*]$ satisfies
both $\dot{f}[\psi] = 0$ and $\dot{c}[\psi] = 0$, whenever
$\tilde{\alpha}^*$ is a gauge invariant dual density. By the
equivalence~\eqref{eq:Jgf-equiv}, the same solution also satisfies
$\J[\psi] = 0$ and $\dot{c}_g[\psi] = 0$.

\begin{theorem}\label{thm:Jexsplit}
Provided the gauge fixed Jacobi system $\J[\psi] = 0$, $\dot{c}_g[\psi]
= 0$ is equivalent to a constrained hyperbolic system obeying the gauge
fixing compatibility condition~\eqref{eq:gf}, as described in Sect.~\ref{sec:constr-gf},
we can define the Jacobi causal Green function $\EE$, Eq.~\eqref{eq:E-def},
such that the following diagram
\begin{equation}
\vcenter{\xymatrix{
	& \Secs_0(P)
		\ar[d]^{\dot{g}}  
	& 0
		\ar[d]            
	& \Secs_{SC}(P)
		\ar[d]^{\dot{g}}  
	& 0
		\ar[d]            
	\\
	0
		                   \ar[r] &
	\Secs_0(F)
		\ar[d]             \ar[r]^{\J} &
	\Secs_0(\tilde{F}^*)
		\ar[d]^{\dot{g}^*} \ar[r]^{\EE} &
	\Secs_{SC}(F)
		\ar[d]             \ar[r]^{\J} &
	\Secs_{SC}(\tilde{F}^*)
		\ar[d]^{\dot{g}^*} \ar[r] &
	0
	\\
	& 0
	& \Secs_0(\tilde{P}^*)
	& 0
	& \Secs_{SC}(\tilde{P}^*)
}}
\end{equation}
becomes an exact sequence after taking the vertical cohomologies.
Moreover, we have the following splittings at $\Secs_0(\tilde{F}^*)$ and
$\Secs_{SC}(F)$:
\begin{equation}
	\ker \dot{g}^* \cong \im \J \oplus \S_{SC}(F)
	\quad\text{and}\quad
	\coker \dot{g} \cong \S_{SC}(F) \oplus \ker\dot{g}^* ,
\end{equation}
where a $\phi$-Cauchy surface $\Sigma\sso M$ and a partition of unity
$\{\chi_\pm\}$ adapted to it define the splitting maps
(cf.~Def.~\ref{def:adapt-pu} and~Lem.~\ref{lem:exsplit})
\begin{align}
	\J_\chi\colon& \S_{SC}(F) \to \Secs_0(\tilde{F}^*) ,
		& \J_\chi[\psi] &= \pm \J[\chi_\pm \psi]
			= r\circ \dot{f}_\chi[\psi] + r_c\circ \dot{c}_\chi[\psi] , \\
	\EE_\chi\colon& \Secs_{SC}(\tilde{P}^{\prime*}) \to \Secs_{SC}(F) ,
		& \EE_\chi[\tilde{\alpha}^*]
			&= \sum_\pm \G_\pm[\bar{r}[\dot{g}^{\prime*}[\chi_\pm\tilde{\alpha}^*]] ,
\end{align}
where $\S_{SC}(F)$ denotes the spacelike compactly supported solutions
of $\dot{f}[\phi] = 0$, $dot{c}[\phi] = 0$, and $\dot{c}_\chi[\psi] =
\dot{c}[\pm\chi_\pm\psi]$.
\end{theorem}
\begin{proof}
The proof consists of checking the desired conclusions at each object of
horizontal sequence in the above diagram.
\begin{enumerate}
\item
	\emph{If $\psi\in \Secs_0(F)$ and $\J[\psi] = 0$, then $\psi =
	\dot{g}[\eps]$ for some $\eps\in \Secs_0(P)$.}

	Let $\tilde{\alpha}^* = \dot{f}[\psi]$. From the
	equivalence~\eqref{eq:Jgf-equiv}, we have
	\begin{equation}
		\tilde{\alpha}^*
		= \bar{r}\circ \J[\psi] + \bar{r}_c\circ \dot{c}_g[\psi]
		= \bar{r}_c \circ \dot{c}_g[\psi] .
	\end{equation}
	Then the gauge fixing compatibility condition~\eqref{eq:gf} implies
	$\dot{s}'[\tilde{\alpha}^*] = 0$.  Now, note the identity $\psi =
	\G_+[\tilde{\alpha}^*]$, which holds from the uniqueness of solutions
	with retarded support and from the fact that $\supp \psi$ is
	automatically retarded by virtue of being compact. Then
	\begin{equation}
		\dot{g}'[\psi]
		= \dot{g}'[\G_+[\tilde{\alpha}^*]]
		= \K'_+[\dot{s}'[\tilde{\alpha}^*]]
		= 0 ,
	\end{equation}
	which, together with global recognizability of the gauge
	transformations allows us to conclude that $\psi = \dot{g}[\eps]$, for
	some $\eps\in \Secs_0(P)$.
\item
	\emph{Every $\tilde{\alpha}^*\in \Secs_0(\tilde{F}^*)$, such that
	$\EE[\tilde{\alpha}^*] = \dot{g}[\eps]$, with $\eps\in \Secs_0(P)$,
	and $\dot{g}^*[\tilde{\alpha}] = 0$, can be written as
	$\tilde{\alpha}^* = \J[\psi]$ for some $\psi \in \Secs_0(F)$. We also
	have $\im \EE\circ \J \sse \im \dot{g}$.}

	The condition
	\begin{equation}
		\EE[\tilde{\alpha}^*]
		= \G[\bar{r}[\tilde{\alpha}^*]]
		= \dot{g}[\eps]
	\end{equation}
	implies that $\dot{s}\circ\dot{k}[\eps] = \dot{f}\circ\dot{g}[\eps] =
	0$, which in turn implies that $\dot{g}[\eps] =
	\dot{g}[\K[\tilde{\gamma}^*]] = \G[\dot{s}[\tilde{\gamma}^*]]$, for
	some $\tilde{\gamma}^* \in \Secs_0(\tilde{P}^*)$. We then have
	\begin{equation}
		\G[\bar{r}[\tilde{\alpha}^*]-\dot{s}[\tilde{\gamma}^*]] = 0 ,
	\end{equation}
	which implies that
	$\bar{r}[\tilde{\alpha}^*] - \dot{s}[\tilde{\gamma}^*] = \dot{f}[\psi]$, where
	$\psi\in \Secs_0(F)$ is unique and given by 
	\begin{equation}
		\psi
		= \G_+[\bar{r}[\tilde{\alpha}^*] - \dot{s}[\tilde{\gamma}^*]]
		= \G_-[\bar{r}[\tilde{\alpha}^*] - \dot{s}[\tilde{\gamma}^*]],
	\end{equation}
	since compact support is both retarded and advanced. Then the
	following calculation gives the desired conclusion:
	\begin{align}
		\J[\psi]
		&= r\circ \dot{f}[\psi] + r_c \circ \dot{c}[\psi] \\
		&= r\circ \bar{r}[\tilde{\alpha}^*]
			+ r_c \circ \dot{c}\circ \G_+ \circ \bar{r}[\tilde{\alpha}^*] \\
\notag & \quad {}
			- r\circ\dot{s}[\tilde{\gamma}^*]
			- r_c\circ\dot{c}\circ \G_+\circ \dot{s}[\tilde{\gamma}^*] \\
		&= (\id + p_\J \circ \dot{g}^* - r_c\circ \bar{r}_\J)[\tilde{\alpha}^*]
			+ r_c \circ \H_+ \circ \dot{q} \circ \bar{r}[\tilde{\alpha}^*] \\
\notag & \quad {}
			- r\circ\dot{s}[\tilde{\gamma}^*]
			- r_c\circ (\bar{r}_\J\circ \J + \bar{r}_g\circ \dot{c}_g)
					\circ \dot{g} \circ \K_+[\tilde{\gamma}^*] \\
		&= \tilde{\alpha}^* + r_c \circ
			(-\bar{r}_\J + \H_+\circ \dot{h} \circ \bar{r}_\J
				+ \H_+ \circ q_\J \circ \dot{g}^*)[\tilde{\alpha}^*] \\
\notag & \quad {}
			- r\circ \dot{s}[\tilde{\gamma}^*]
			+ r\circ (\bar{r}_c \circ \dot{c}_g \circ \dot{g})\circ
					\K_+[\tilde{\gamma}^*] \\
		&= \tilde{\alpha}^* ,
	\end{align}
	where we have used the definition~\eqref{eq:E-def} and
	identities~\eqref{eq:Jgf-equiv}, \eqref{eq:rinv1}, \eqref{eq:consist1}
	and the commutativity of diagrams~\eqref{eq:glpar} and~\eqref{eq:glrec}.
	Now, consider the following composition of operators applied to any
	$\psi \in \Secs_0(F)$:
	\begin{align}
		\dot{g}'[\EE\circ \J[\psi]]
		&= \dot{g}'\circ \G \circ \bar{r} \circ \J [\psi] \\
		&= \dot{g}'\circ (\G \circ \dot{f}) [\psi]
			- \K' \circ (\dot{s}' \circ \bar{r}_c \circ \dot{c}_g) [\psi] \\
		&= 0 ,
	\end{align}
	where the last two terms vanish by the exactness of the causal Green
	function sequence~\eqref{eq:hyp-seq} and the gauge fixing compatibility
	condition~\eqref{eq:gf}. Global recognizability then dictates that
	there exists a section $\eps\in \Secs_0(P)$ such that $\EE\circ
	\J[\psi] = \dot{g}[\eps]$, which is the desired result.
\item
	\emph{Every solution $\psi\in \Secs_{SC}(F)$ of $\J[\psi] = 0$ is of
	the form $\psi = \EE[\tilde{\alpha}^*] + \dot{g}[\eps]$ for some
	$\eps\in \Secs_{SC}(P)$ and $\tilde{\alpha}^*\in
	\Secs_0(\tilde{F}^*)$, with $\dot{g}^*[\tilde{\alpha}^*] = 0$.}

	Let $\Sigma\sso M$ be a $\phi$-Cauchy surface and $\{\chi_\pm\}$ a
	partition of unity adapted to it. Let $\tilde{\alpha}^* =
	\J_\chi[\psi]$ and $\psi' = \EE[\tilde{\alpha}^*]$, where now
	$\J[\psi'] = 0$ and $\dot{c}[\psi'] = 0$. Also, denote
	$\dot{c}_{g\chi}[\psi] = \dot{c}_g[\pm\chi_\pm\psi]$. It remains to
	show that $\psi - \psi'$ is pure gauge:
	\begin{align}
		\dot{g}'[\psi']
		&= \dot{g}'\circ \EE[\tilde{\alpha}^*]
		= \dot{g}'\circ \G \circ \bar{r}[\tilde{\alpha}^*] \\
		&= \dot{g}'\circ \G \circ \bar{r}[
			r\circ \dot{f}_\chi[\psi] + r_c\circ \dot{c}_\chi[\psi]] \\
		&= \dot{g}'\circ \G[\dot{f}_\chi[\psi]
			+ \bar{r}_c\circ r_g \circ \dot{c}_\chi[\psi]] \\
		&= \dot{g}'[\psi]
			+ \K'\circ(\dot{s}'\circ \bar{r}_c\circ \dot{c}_{g\chi})[\psi] \\
		&= \dot{g}'[\psi] ,
	\end{align}
	where we have appealed to identities~\eqref{eq:rinv1}, \eqref{eq:rinv2}
	and the gauge fixing compatibility condition~\eqref{eq:gf}. Now, we have
	$\dot{g}'[\psi-\psi'] = 0$, which implies, together with global
	recognizability, that $\psi-\psi' = \dot{g}[\eps]$ for some $\eps\in
	\Secs_{SC}(P)$.
\item \label{itm:4}
	\emph{Every $\tilde{\alpha}^*\in \Secs_{SC}(\tilde{F}^*)$, such that
	$\dot{g}^*[\tilde{\alpha}^*] = 0$, can be written as $\tilde{\alpha}^*
	= \J[\psi]$ for some $\psi\in \Secs_{SC}(F)$.}

	Global recognizability implies that there exists a
	$\tilde{\eps}^{\prime*}\in \Secs_{SC}(\tilde{P}^{\prime*})$ such that
	$\tilde{\alpha}^* = \dot{g}^{\prime*}[\tilde{\eps}^{\prime*}]$. Let
	$\psi = \EE_\chi[\tilde{\eps}^{\prime*}]$.
	Then the following calculation gives the desired conclusion:
	\begin{align}
		\J[\psi]
		&= \sum_\pm (r\circ \dot{f} + r_c\circ \dot{c}) \circ \G_\pm\circ
			\bar{r}[\dot{g}^{\prime*}[\chi_\pm \tilde{\eps}^{\prime*}]] \\
		&= r\circ \bar{r}[\tilde{\alpha}^*]
			+ r_c \circ \bar{r}_\J[\tilde{\alpha}^*]
		= (\id + p_\J\circ \dot{g}^*)[\tilde{\alpha}^*] \\
		&= \tilde{\alpha}^* .
	\end{align}
\item
	\emph{We have $\EE\circ \J_\chi = \id \pmod{\im \dot{g}}$ on $\S_{SC}(F)$.}

	Let $\psi \in \S_{SC}(F)$. Direct calculation shows
	\begin{align}
		\dot{g}'[\EE\circ \J_\chi[\psi]]
		&= \dot{g}'\circ \G \circ \bar{r} \circ \J_\chi[\psi] \\
		&= \dot{g}'[\G\circ \dot{f}_\chi[\psi]]
			+ \H \circ (\dot{s}'\circ \bar{r}_c \circ \dot{c}_g) [\psi] \\
		&= \dot{g}'[\psi]
	\end{align}
	where we have used a splitting identity from
	Lem.~\ref{lem:exsplit} and the gauge fixing compatibility condition~\eqref{eq:gf}.
	It follows that $\psi-\EE\circ \J_\chi \in \im \dot{g}$ is pure gauge.
\item
	\emph{We have $\J\circ \EE_\chi = \dot{g}^{\prime*}$ on
	$\Secs_{SC}(\tilde{P}^{\prime*})$.}

	This was already checked in item~\ref{itm:4} above.
	\qed
\end{enumerate}
\end{proof}

It is easy to see from its variational nature that the Jacobi operator
is self-adjoint $(\J)^* = \J$. If it were directly invertible, the Green
functions $\EE_\pm$ would satisfy the same relation with their adjoints
as shown in Sect.~\eqref{sec:green-adj}, making the causal Green
function anti-self-adjoint, $\EE^* = -\EE$. However, due to gauge
invariance the relation of the gauge fixed Green functions to their
adjoints is more complicated.
\begin{lemma}\label{lem:E-anti}
When restricted to act on gauge invariant dual densities, the causal
Green function of the gauge fixed Jacobi system is anti-self-adjoint up
to gauge:
\begin{equation}
	(\EE)^* = \EE \pmod{\im \dot{g}} .
\end{equation}
\end{lemma}
\begin{proof}
First, note that from identities~\eqref{eq:Jgf-equiv}
and~\eqref{eq:rinv1} we have
\begin{align}
	\J\circ\EE_\pm[\tilde{\alpha}^*]
	&= \sum_\pm (r\circ \dot{f} + r_c\circ \dot{c})\circ \G_\pm \circ
		\bar{r}[\tilde{\alpha}^*] \\
	&= r\circ \bar{r}[\tilde{\alpha}^*] + r_c\circ \bar{r}_\J[\tilde{\alpha}^*]\\
	&= (\id + p_\J\circ \dot{g}^*)[\tilde{\alpha}^*] .
\end{align}
It the follows that
\begin{align}
	(\EE_\mp)^* \circ \J \circ \EE_\pm
		&= (\EE_\mp)^* \circ (\J \circ \EE_\pm)
		= (\EE_\mp)^* + (\EE_\mp)^*\circ p_\J \circ \dot{g}^* , \\
	(\EE_\mp)^* \circ \J \circ \EE_\pm
		&= ((\EE_\mp)^*\circ \J^*) \circ \EE_\pm
		= \EE_\pm + \dot{g}\circ p_\J^* \circ \EE_\pm , \\
	\text{and hence}\quad
	\EE_\pm &= (\EE_\mp)^* + (\EE_\mp)^*\circ p_\J \circ \dot{g}^*
		- \dot{g}\circ p_\J^* \circ \EE_\pm .
\end{align}
Given that $\EE = \EE_+ - \EE_-$, we then have
\begin{equation}\label{eq:E-anti}
	\EE = -(\EE)^* - \dot{g}\circ p_\J^*\circ \EE
		- (\EE)^*\circ p_\J^* \circ \dot{g}^* ,
\end{equation}
which gives the desired conclusion. \qed
\end{proof}

As noted earlier in this section, there is considerable gauge freedom
left that is compatible with the partial gauge fixing considered above.
We call this kind of gauge fixing \emph{purely hyperbolic}. Any further
gauge fixing conditions are then called \emph{residual}. We leave the
consideration of residual gauge fixing to future work. A principal
difficulty in dealing with residual gauge fixing conditions is that the
resulting constraints are no longer parametrizable (such as operators
that are elliptic on a family of spatial slices). Thus, the kernel of
the gauge fixing conditions may contain very few, if any solutions with
spacelike compact support, which would be difficult to fit into the
current formal framework for tangent and cotangent spaces to the space
of solutions.

\subsubsection{Formal symplectic structure}\label{sec:formal-symp}
In this section we define a formal presymplectic structure on a slow
patch $\S(F,C)$ of the solution space $\S_H(F)$; using the language of
Sect.~\ref{sec:formal-forms}, it consists of a formal local differential
$2$-form $\Omega_\Sigma$ that is formally closed, $\delta\Omega_\Sigma =
0$. In a later section, Sect.~\ref{sec:formal-pois}, we shall see that
some sufficient conditions ensure that $\Omega_\Sigma$ is invertible and
hence symplectic. Thus, the solution space becomes the phase space of
classical field theory, as it can be endowed with both formal symplectic
and Poisson structures. In the presence of gauge symmetries, this
presymplectic form projects to an actual symplectic form on the space of
gauge equivalence classes of solutions, that is, the physical phase
space, or rather a slow patch $\bar{S}(F,C)$ thereof. The symplectic
forms will also be seen in a later section, Sect.~\ref{sec:glob-phsp},
to agree on overlapping slow patches and hence constitute symplectic
structure on the global phase space $\S_H(F)$.

Fix a globally hyperbolic chronal cone bundle $C\to M$, as well as a
$C$-Cauchy surface $\Sigma\sso M$.
\begin{definition}
Given a presymplectic potential current density $\omega\in
\Forms^{n-1,2}(F)$ compatible with a PDE system $\iota\colon \E\sse
J^kF$ and a closed codim-$1$ surface $\Sigma\sso M$, we can define the
following local differential $2$-form on the solution space $\S(F,C)$:
\begin{equation}
	\Omega_\Sigma(\psi,\xi)[\phi]
	= \int_\Sigma\omega(\psi,\xi)[\phi]
	= \int_\Sigma(j^\oo\phi)^* [\iota_{\hat{\xi}}\iota_{\hat{\psi}}\omega] .
\end{equation}
\end{definition}
Recall that two formal tangent vectors $\psi,\xi\in T_\phi\S$ correspond
to linearized solutions with spatially compact support, while
$\hat{\psi}$ and $\hat{\xi}$ are evolutionary vector fields on $J^\oo F$
defined by them.

Ideally, we would then prove that $\Omega_\Sigma[\phi]$ is a smooth
differential $2$-form on the infinite dimensional manifold $\S(F,C)$.
However, that would necessitate a level of functional analytic detail
going beyond the formal approach we have adopted. Instead, we note that
the dependence of $\Omega_\Sigma[\phi]$ on the solution section $\phi$
is entirely through the ordinary differential form $\omega$ on the jet
bundle $J^\oo F$ (or even one of its finite dimensional projections $J^k
F$). Since $\omega$ is smooth (by hypothesis), we also declare
$\Omega_\Sigma$ \emph{formally smooth}. In addition to $\phi$, the form
$\Omega_\Sigma[\Phi]$ depends on the surface $\Sigma$. For every
$\phi\in \S(F,C)$, the supports of the linearized solutions $\psi$ and
$\xi$ (formal tangent vectors) will have compact intersection with any
surface that is $\phi$-Cauchy. A fortiori, since $\Sigma$ is
$\phi$-Cauchy for every $\phi\in \S(F,C)$, the integral defining
$\Omega_\Sigma(\psi,\xi)[\phi]$ converges. Moreover, because the
presymplectic current density $\omega$ is horizontally closed on
solutions, the value of the integral remains well defined and does not
change as $\Sigma$ is varied within the same homology class.

It is now straightforward to prove the following
\begin{lemma}
The formal local $2$-form $\Omega_\Sigma$ is formally closed,
$\delta\Omega_\Sigma = 0$, and it depends only on the homology class of
$\Sigma$.
\end{lemma}
\begin{proof}
Fix a background solution $\phi\in \S_H(F)$ and consider linearized
solutions $\chi,\xi,\psi\in T_\phi\S$. It is straight forward to check
that $\Omega_\Sigma$ is vertically closed:
\begin{align}
	(\delta\Omega_\Sigma)(\chi,\xi,\psi)[\phi]
	&= \int_\Sigma (j^\oo\phi)^*
		[\iota_{\hat\psi}\iota_{\hat\xi}\iota_{\hat\chi}\dv\omega]
	= \int_\Sigma (j^\oo\phi)^*\iota_\oo^*
		[\iota_{\hat\psi}\iota_{\hat\xi}\iota_{\hat\chi}\dv\omega] \\
	&= \int_\Sigma (j^\oo\phi)^*
		[\iota_{\hat\psi}\iota_{\hat\xi}\iota_{\hat\chi} (\dv\iota^*\omega)]
	= 0 ,
\end{align}
where we have used the vertical closure $\iota_\oo^*\omega$, from the
presymplectic hypothesis.

Independence of the representative of the homology class of $\Sigma$
follows if we can show that the integrand of
$\Omega_\Sigma(\chi,\xi)[\phi] = \int_\Sigma (j^\oo\phi)^*
[\iota_{\hat\xi}\iota_{\hat\chi} \omega]$ is de~Rham closed on $M$. This
follows directly from the horizontal closure of $\iota_\oo^*\omega$,
again from the presymplectic hypothesis:
\begin{equation}
	\d[(j^\oo\phi)^*\iota_{\hat\xi}\iota_{\hat\chi}\omega]
	= (j^k\phi)^*\dh\iota^*[\iota_{\hat\xi}\iota_{\hat\chi}\omega]
	= (j^k\phi)^*\iota_{\hat\xi}\iota_{\hat\chi}[\dh\iota^*\omega]
	= 0 .
\end{equation}
Note that, strictly speaking, on needs to apply Stokes' theorem to
establish independence of the representative $\Sigma$, which requires
the convergence of all intermediate integrals. This property will
actually hold in the main application of this result, which is
Def.~\ref{def:symp} below, so we do not discuss it in more detail here.

In both cases, we have introduced the pullback $\iota_\oo^*$ because
$\phi$ is a solution and $\chi,\xi,\psi$ are linearized solutions, so
that the pullback $(j^\oo\phi)^*$ factors through $\iota_\oo^*$. \qed
\end{proof}

\begin{definition}\label{def:symp}
Given a (gauge fixed, if necessary) variational PDE system on the field
bundle $F\to M$, with Lagrangian density $\L$, that is equivalent to a
quasilinear, symmetric hyperbolic system on $F\to M$ with constraints
and a globally hyperbolic chronal cone bundle $C\to M$, we define the
formal \emph{variational presymplectic form} on the space $\S(F,C)$ of
$C$-slow solutions as $\Omega = \Omega_\Sigma$, where we use the
presymplectic current density $\omega$ associated to $\L$ and $\Sigma$
is $C$-Cauchy.
\end{definition}
Using the smooth splitting of $M$ induced by a globally hyperbolic cone
bundle, it is straight forward to show that all $C$-Cauchy surfaces
belong to the same homology%
	\footnote{The appropriate homology theory here should correspond to a
	variant of locally finite Borel-Moore homology, where one considers
	only chains whose intersection with every $C$-spatially compact set
	is compact. This variant does not appear to have gotten any attention
	in the literature and thus deserves further study.} %
class. That, together with the fact that a $C$-Cauchy surface will
compactly intersect the support of any $\phi$-spacelike compact
linearized solution, with $\phi\in\S(F,C)$, shows that $\Omega$ is well
defined, independently of the choice of $\Sigma$.

A bilinear form defines a linear map from a vector space to its
algebraic dual. A similar statement holds for a continuous bilinear form
and to topological dual space. However, our formal cotangent space
$T^*_\phi\bar{S}$ is neither the algebraic nor the topological dual of
the formal tangent space $T_\phi\bar{S}$. Thus we have to check this
property for $\Omega$ by hand. This will be accomplished using an analog
of a splitting map for the Jacobi system from Thm.~\ref{thm:Jexsplit},
which is analogous Lem.~\ref{lem:exsplit} for purely hyperbolic systems.
The argument in the proof was inspired by~\cite{fr-pois}.
\begin{lemma}
The presymplectic form $\Omega$ defines the following map from the
formal tangent space to the formal cotangent space:
\begin{align}
	\Omega\colon & T_\phi\S \to T^*_\phi\S \\
		& \psi \mapsto [\tilde{\alpha}^*] , \\
	\text{with} ~~
	\tilde{\alpha}^* &= \J_\chi[\psi] = \pm\J[\chi_\pm\psi] ,
\end{align}
where $\J\colon \Secs(F) \to \Secs(\tilde{F}^*)$ the Jacobi differential
operator at the dynamical linearization point $\phi\in \S_H(F)$ and
$\{\chi_\pm\}$ is a partition of unity adapted to a $C$-Cauchy surface
$\Sigma$.
\end{lemma}
\begin{proof}
Using the adapted partition of unity, we can write any solution of
$\dot{f}[\psi] = 0$ as $\psi = \psi_+ + \psi_-$, with $\psi_\pm =
\chi_\pm \psi$. If the solution also satisfies the purely hyperbolic
gauge fixing condition $\dot{s}^*\circ r^*[\psi] = 0$, by the
equivalence discussed earlier in this section it also satisfies the
Jacobi equation $\J[\psi] = 0$. Hence $\J[\psi_+ + \psi_-] = 0$ or
$\J[\psi_+] = -\J[\psi_-] = \J_\chi[\psi]$. Note that the support of
$\J[\psi_\pm]$ is compact, since $\psi_\pm$ satisfy the Jacobi equation
away from the intersection $S^+\cap S^-\cap \supp \psi$, which is by
hypothesis compact.

Next, we want to find a dual density $\tilde{\alpha}^*$ that satisfies
$\Omega[\phi](\xi,\psi) = \langle \xi,\tilde{\alpha}^*\rangle$ for any
$\xi\in T_\phi\S$, which in particular satisfies $\J[\xi]=0$. Recall
that an adapted partition of unity also depends on two additional
$C$-Cauchy surfaces $\Sigma^\pm\sso I^\pm(\Sigma)$ and the supports of
the partition are contained in $\supp\chi_\pm \sse S^\pm =
I^\pm(\Sigma^\mp)$.
\begin{align}
	\Omega(\xi,\psi)[\phi]
	&= \int_\Sigma \omega(\xi,\psi)[\phi]
	= \sum_\pm \int_\Sigma \omega(\xi,\psi_\pm)[\phi] \\
	&= \sum_\pm\int_{\Sigma^\mp}\omega[\phi](\xi,\psi_\pm)
		+ \sum_\pm \pm\int_{S^\pm\cap I^\mp(\Sigma)}\d\omega(\xi,\psi_\pm)[\phi] \\
	&= \sum_\pm \pm\int_{I^\mp(\Sigma)}(\dh\omega)(\xi,\psi_\pm)[\phi] \\
	&= \sum_\pm \pm\int_{I^\mp(\Sigma)}
		(-\dv\EL_a\wedge\dv u^a)(\xi,\psi_\pm)[\phi] \\
	&= \sum_\pm \mp\int_{I^\mp(\Sigma)}
		[(\J^I_{ab}\del_I\xi^b)\psi_\pm^a - (\J^I_{ab}\del_I\psi_\pm^b)\xi^a] \\
	&= \int_{I^-(\Sigma)} \xi\cdot\J[\psi_+]
		- \int_{I^+(\Sigma)} \xi\cdot\J[\psi_-] \\
	&= \int_{I^-(\Sigma)} \xi\cdot \J_\chi[\psi]
		- \int_{I^+(\Sigma)} \xi\cdot(-\J_\chi[\psi]) \\
	&= \int_M \xi\cdot \J_\chi[\psi]
\end{align}
After the integration by parts, the boundary integrals over $\Sigma^\pm$
were dropped since they did not intersect the support of their
integrands. Since the $\supp \psi_\pm \sse S^\pm$, the integration over
$S^\pm\cap I^\mp(\Sigma)$ could be extended to all of $I^\mp(\Sigma)$.
The term $\psi_\pm\cdot \J[\xi]$ was dropped since $\xi$ is a linearized
solution.

Recall that we are not interested in the dual density $\tilde{\alpha}^*
= \J_\chi[\psi]$ specifically, which explicitly depends on the adapted
partition of unity, but rather the equivalence class
$[\tilde{\alpha}^*]\in T^*_\phi\S$, defined modulo $\im \J$. Consider
another adapted partition of unity $\{\chi'_\pm\}$. Because they both
provide splitting maps, we have $\psi = \EE\circ\J_\chi[\psi] = \EE\circ
\J_{\chi'}[\psi]$. Then $\EE[\J_{\chi'}[\psi]-\J_\chi[\psi]] = 0$. By
exactness, $\J_\chi[\psi]$ and $\J_{\chi'}[\psi]$ must differ by an
element of $\im \J$; in other words, they represent the same equivalence
class in $T^*_\phi\S$.

To complete the proof, we use the non-degeneracy of the natural pairing
between $T_\phi\S$ and $T^*_\phi\S$ to define the operator $\Omega$ by
the formula
\begin{equation}
	\langle \xi, \Omega \psi \rangle
	= \Omega[\phi](\xi,\psi)
	= \langle \xi, \tilde{\alpha}^* \rangle
	= \langle \xi, [\tilde{\alpha}^*] \rangle ,
\end{equation}
so that $\Omega\psi = [\tilde{\alpha}^*] \in T^*_\phi\S$, with
$\tilde{\alpha}^* = \J_\chi[\psi]$. \qed
\end{proof}

Notice that in the presence of gauge symmetries (the gauge fixing is
only partial), the form $\Omega$ is degenerate, since every pure gauge
solution lies in its kernel:
\begin{equation}
	\Omega(\dot{g}[\eps],\psi)[\phi]
	= \langle \dot{g}[\eps], \pm\J[\chi_\pm\psi] \rangle
	= \langle \eps, \pm\dot{g}^*\circ\J[\chi_\pm\psi] \rangle
	= 0
\end{equation}
for any $\psi$, since N\"other's second theorem implies that
$\dot{g}^*\circ \J = 0$~\cite{lw}.
\begin{corollary}\label{cor:bOmega}
The $2$-form $\Omega$ on $T_\phi\S$ projects to a $2$-form $\bar\Omega$
on $T_\phi\bar\S$ and hence defines a map
\begin{align}
	\bar{\Omega}\colon & T_\phi\bar\S \to T^*_\phi\bar\S \\
		& [\psi] \mapsto [\tilde{\alpha}^*] , \\
	\text{with} ~~
	\tilde{\alpha}^* &= \J_\chi[\psi] .
\end{align}
\end{corollary}
In other words, formally, the quotient projection to the physical phase
space effects a presymplectic reduction $(\S_H(F),\Omega) \to
(\bar{S}_H(F),\bar\Omega)$. We shall see later on that $\bar\Omega$ is
non-degenerate and hence symplectic.

\subsubsection{Formal Poisson bivector, Peierls formula}\label{sec:formal-pois}
In this section we show that the formal symplectic form $\bar{\Omega}$
defined above is invertible and that its inverse, the \emph{formal
Poisson bivector} $\Pi$, is given by the Peierls formula
\begin{equation}
	\Pi = \EE ,
\end{equation}
where $\EE$ is the causal Green function of the Jacobi operator $\J$ as
defined in Sect.~\ref{sec:constr-gf}. To show that $\Pi$ is indeed a Poisson bivector, it
suffices to show that (a) it is an antisymmetric bilinear form on the
formal cotangent space, (b) it defines a map from the formal cotangent
space to the formal tangent space and (c) it is a two-sided inverse of
$\bar{\Omega}$ defined in Cor.~\ref{cor:bOmega}. The fact that $\Pi$ defines a Poisson
bracket, with its Leibniz and Jacobi identities, then formally follows
from standard arguments.

\begin{lemma}\label{lem:pure-hyp-peierls}
The Peierls formula specifies a map from the formal cotangent space to
the formal tangent space:
\begin{align}
	\Pi\colon & T^*_\phi\bar{\S}\to T_\phi\bar{\S} \\
		& [\tilde{\alpha}^*] \mapsto [\psi] , \quad
		\text{with}\quad \dot{g}^*[\tilde{\alpha}^*] = 0 \\
\text{and}\quad
	\psi = {}& \EE[\tilde{\alpha}^*] .
\end{align}
\end{lemma}
\begin{proof}
The challenge is to show that $\Pi$ maps equivalence classes to
equivalence classes (Def.~\ref{def:tt*-sols-gauge}). That is, that any representative
$\tilde{\alpha}^* + \dot{f}^*[\xi] + \dot{c}^*[\tilde{\gamma}^*]$ of an
equivalence class $[\tilde{\alpha}^*]\in T^*_\phi\bar\S$, with
$\dot{g}^*[\tilde{\alpha}^*] = 0$, gets mapped to the same equivalence
class in $T_\phi\bar\S$. By linearity, it suffices to check that $[0]\in
T^*_\phi\bar\S$ is mapped to $[0]\in T_\phi\bar\S$. Recall that any
solution representing $[0]\in T_\phi\bar\S$ is pure gauge
$\dot{g}[\eps]$. Note that the equivalence~\eqref{eq:Jgf-equiv} of the
$(\dot{f}\oplus \dot{c},\tilde{F}^*\oplus E)$ and $(\J\oplus \dot{c}_g,
\tilde{F}^*\oplus E_g)$ equation forms, together with the
self-adjointness of the Jacobi operator $(\J)^* = \J$, allows us to
rewrite any representative of $[0]\in T^*_\phi\bar\S$ as $\J[\xi] +
\dot{c}_g^*[\tilde{\gamma}^*]$, for some $\xi\in \Secs_0(F)$ and
$\tilde{\gamma}^*\in \Secs_0(\tilde{E}^*_g)$.  This representative will
also satisfy the identity
\begin{equation}\label{eq:g*cg*}
	\dot{g}^*\circ \dot{c}^*_g[\tilde{\gamma}^*]
	= \dot{g}^*[\J[\xi] + \dot{c}^*_g[\tilde{\gamma}^*]]
	= 0 .
\end{equation}
Direct calculation then shows that
\begin{align}
	\Pi[\J[\xi] + \dot{c}^*_g[\tilde{\gamma}^*]]
	&= \EE \circ \J[\xi] + \EE\circ \dot{c}^*_g[\tilde{\gamma}^*] \\
	&= \dot{g}[\eps] - (\dot{c}_g\circ \EE)^*[\tilde{\gamma}^*]
		- \dot{g}\circ p_\J^*\circ \EE\circ \dot{c}^*_g[\tilde{\gamma}^*] \\
\notag & \qquad {}
		- (\EE)^*\circ p_\J[\dot{g}^*\circ \dot{c}^*_g[\tilde{\gamma}^*]] \\
	&= \dot{g}[\eps
		- q_\J^* \circ (\H)^* \circ r_g^*[\tilde{\gamma}^*]
		- p_\J^*\circ \EE\circ \dot{c}^*_g[\tilde{\gamma}^*]]
\end{align}
is pure gauge. We have used the identity that $\EE\circ \J[\xi] =
\dot{g}[\eps]$ for some $\eps\in \Secs_0(P)$ (Thm.~\ref{thm:Jexsplit}),
the anti-self-adjointness identity~\eqref{eq:E-anti},
that (Eqs.~\eqref{eq:Jgf-equiv} and~\eqref{eq:consist1})
\begin{align}
	\dot{c}_g\circ \EE
	&= r_g \circ (\dot{c} \circ \G) \circ \bar{r}
	= r_g \circ (\H \circ \dot{q}\circ \bar{r}) \\
	&= (r_g \circ \H \circ q_\J) \circ \dot{g}^*
\end{align}
and the identity~\eqref{eq:g*cg*}.

Therefore, we can conclude that if $[\tilde{\alpha}^*] = [0]$, then
$[\EE[\tilde{\alpha}^*]] = [0]$. \qed
\end{proof}

\begin{lemma}\label{lem:Pi-form}
The Peierls formula defines an antisymmetric bilinear form on the formal
cotangent space:
\begin{equation}
	\Pi([\tilde{\alpha}^*],[\tilde{\beta}^*])
	= \langle \Pi[\tilde{\alpha}^*] , [\tilde{\beta}^*] \rangle
	= - \Pi([\tilde{\beta}^*],[\tilde{\alpha}^*]) ,
\end{equation}
for any $[\tilde{\alpha}^*],[\tilde{\beta}^*] \in T^*_\phi\bar{\S}$.
\end{lemma}
\begin{proof}
Recall that the representatives always satisfy
$\dot{g}^*[\tilde{\alpha}^*] = \dot{g}^*[\tilde{\beta}^*] = 0$.
Appealing directly to the anti-self-adjointness
identity~\eqref{eq:E-anti} we have
\begin{align}
	\Pi([\tilde{\alpha}^*],[\tilde{\beta}^*])
	&= \langle \Pi[\tilde{\alpha}^*], [\tilde{\beta}^*] \rangle
	= \langle \EE[\tilde{\alpha}^*] , \tilde{\beta}^* \rangle
	= \langle (\EE)^*[\tilde{\beta}^*] , \tilde{\alpha}^* \rangle \\
	&= -\langle (\EE[\tilde{\beta}^*]
			+ \dot{g}\circ p_\J^*\circ \EE[\tilde{\beta}^*]
			+ (\EE)^*\circ p_\J\circ \dot{g}^*[\tilde{\beta}^*])
			, \tilde{\alpha}^* \rangle \\
	&= -\langle [\EE[\tilde{\beta}^*]] , [\tilde{\alpha}^*] \rangle
	= - \langle \Pi[\tilde{\beta}^*] , [\tilde{\alpha}^*] \rangle \\
	&= -\Pi([\tilde{\beta}^*],[\tilde{\alpha}^*]) . \quad \qed
\end{align}
\end{proof}

\begin{theorem}
The Peierls formula gives a two-sided inverse to the formal symplectic
form, $\bar{\Omega}\Pi = \id$ on $T^*_\phi\bar{S}$ and $\Pi\bar{\Omega}
= \id$ on $T_\phi\bar{S}$.
\end{theorem}
The argument in the proof below was inspired by~\cite{fr-pois}
and~\cite[Lem.3.2.1]{wald-qft}. However, the argument has been
generalized to handle hyperbolic systems with constraints and gauge
invariance (presuming purely hyperbolic gauge fixing).
\begin{proof}
The proof uses in an essential way the splitting identities of
Thm.~\ref{thm:Jexsplit}. Consider any $[\psi]\in T_\phi\bar\S$ and
$[\tilde{\alpha}^*]\in T^*_\phi\bar\S$. To use these splitting
identities, we introduce a $\phi$-Cauchy surface $\Sigma \sso M$ and a
partition of unity $\{\chi_\pm\}$ adapted to it. Then
\begin{align}
	\langle \Pi \bar{\Omega} [\psi], [\tilde{\alpha}^*] \rangle
	&= \langle \EE\circ \J_\chi[\psi], \tilde{\alpha}^* \rangle \\
	&= \langle \psi + \dot{g}[\eps] , \tilde{\alpha}^* \rangle
			\quad \text{(for some $\eps\in \Secs_{SC}(P)$)} \\
	&= \langle [\psi] , [\tilde{\alpha}^*] \rangle .
\end{align}
Therefore, from the non-degeneracy of the natural pairing between
$T_\phi\bar\S$ and $T^*_\phi\bar\S$, we concluded that $\Pi\bar{\Omega}
= \id$. Similarly,
\begin{equation}
	\langle [\psi], \bar{\Omega} \Pi [\tilde{\alpha}^*] \rangle
	= \langle \psi , \J_\chi \circ \EE [\tilde{\alpha}^*] \rangle .
\end{equation}
But then
\begin{equation}
	\EE[\J_\chi \circ \EE[\tilde{\alpha}^*]]
	= (\EE \circ \J_\chi)\circ \EE[\tilde{\alpha}^*]
	= \EE[\tilde{\alpha}^*] + \dot{g}[\eps] ,
\end{equation}
for some $\eps\in \Secs_0(P)$. But, by the cohomological exactness of
Thm.~\ref{thm:Jexsplit}, this means that $\J_\chi\circ \EE[\tilde{\alpha}^*] -
\tilde{\alpha}^* = \J[\xi]$ for some $\xi\in \Secs_0(F)$. In other
words,
\begin{equation}
	\langle [\psi] , \bar{\Omega} \Pi [\tilde{\alpha}^*] \rangle
	= \langle \psi, \tilde{\alpha}^* + \J[\xi] \rangle
	= \langle [\psi], [\tilde{\alpha}^*] \rangle .
\end{equation}
Therefore, from the non-degeneracy of the natural pairing between
$T_\phi\bar\S$ and $T^*_\phi\bar\S$, we concluded that $\bar{\Omega}\Pi
= \id$. \qed
\end{proof}

\subsubsection{Algebra of observables, locality}\label{sec:obsv-local}
This section and the next complete the construction of a classical field
theory by constructing, at least formally, the Poisson algebra of
observables associated with it. Along the way, several important
points are discussed. (i) Since the focus of this work is more geometric
than analytical, some parts of the construction are kept formal. (ii) In
the physics literature, the given construction is often called
\emph{on-shell}. We discuss its relation to and trade offs with respect
to the \emph{off-shell} formalism, which has received considerable
attention in the recent literature on perturbative algebraic quantum
field theory (pAQFT). (iii) We define the notion of spacetime support,
corresponding to the spacetime localization of observables, and use it
to prove a classical version of \emph{microcausality} that is
generalized to field theories with dynamical causal structure.

We will use the Peierls formula for the Poisson bivector $\Pi$ to
construct a \emph{formal} Poisson bracket on the \emph{formal} algebra
of observables on the phase space $(\bar{\S}_H(F),\bar\Omega)$ of a
classical field theory. This section will actually concentrate on a
single slow patch $\S(F,C)\sse \S_H(F)$ of full phase space.
Globalization to the full phase space $\S_H(F)$ by covering it with
slow patches will be considered in the following section,
Sect.~\ref{sec:glob-phsp}.

In the last paragraph, the word \emph{formal} was used in two different
ways. The algebra that we will consider is the polynomial algebra
generated by local functionals with compact spacetime support. This
algebra is much smaller than a reasonable space of smooth functions on
the infinite dimensional manifold of solutions $\S_H(F)$, and hence is
only a formal substitute. However, this polynomial algebra is much less
complicated from the analytical point of view and is expected to be
dense in the larger, ultimately desired algebra $C^\oo(\S_H(F))$,
however it is rigorously defined. Also, the Poisson bracket of two
elements of this polynomial algebra, if defined using the bivector
$\Pi$, is not an element of the same algebra, though still a well
defined function on $\S_H(F)$. However, just as the polynomial algebra
is expected to be dense in $C^\oo(\S_H(F))$, the Poisson bracket is
expected to be continuous, so nested applications of the Poisson bracket
can be considered by replacing each evaluation thereof by an
approximating sequence within the polynomial algebra. In this sense, the
definition of the Poisson bracket is given only formally. The hypotheses
needed to justify its non-formal use are summarized in Hyp.~\ref{hyp:Coo}.

Before moving on to the technical part of this section, briefly discuss
the relation between our approach to the algebra of observables and the
one taken in the recent literature on
pAQFT~\cite{hollands-ym,bdf,fr-bv,rejzner-thesis,bfr}. That
work is in the so-called \emph{off-shell} formalism, while we work in
the \emph{on-shell} one. The space of all globally hyperbolic solution
sections $\S_H(F)\sso \Secs_H(F)$ is a submanifold%
	\footnote{This discussion should be considered formal and analogous to
	the situation where all manifolds are finite dimensional.  The
	justification of its conclusions is a matter of ongoing work in this
	field.} %
of the space of all globally hyperbolic sections and is sometimes
referred to as the \emph{shell}.  Thus, \emph{on-shell} field
configurations correspond to solutions, while the \emph{off-shell} ones
are not so restricted. Dually, on the algebraic side, the algebra of
smooth functions on $\S_H(F)$ is a quotient $C^\oo(\S_H(F)) \cong
C^\oo(\Secs_H(F)) / I_\E$ by the ideal $I_\E$ generated by the equations
of motion (an equation form of the PDE system $J^kF\supset \E\to M$).
Correspondingly, $C^\oo(\S_H(F))$ is called the on-shell algebra while
$C^\oo(\Secs_H(F))$ is called the off-shell algebra. The on-shell
formalism attempts to construct the on-shell algebra directly as the
algebra of functions on the space of solutions (or gauge equivalence
classes thereof), while the off-shell formalism rather constructs it as
the above quotient of the off-shell algebra (with a further quotient to
factor out gauge).

The off-shell formalism is advantageous in perturbative algebraic
quantum field theory (pAQFT) as it dramatically simplifies the
perturbative renormalization of interacting
theories~\cite{hollands-ym,bdf,fr-bv,rejzner-thesis}.
On the other hand, the on-shell formalism connects more directly with
the mathematical PDE literature. It is the theorems on the
well-posedness of hyperbolic PDE systems, as quoted in
Sec.~\ref{sec:pde-theory} that will allow us to establish the validity
of the (generalized) Causality and Time Slice axioms in
Sec.~\ref{sec:classcaus}. I am not aware of any methods that can be used
to establish these results directly in the off-shell formalism
(particularly for quasilinear systems) without first establishing its
equivalence with the on-shell construction. Since understanding
causality in field theories with quasilinear equations of motion is the
main motivation for this work, we feel our choice to work in the
on-shell formalism is justified. Of course, it is fully expected that
future work will realize an exact duality between the two approaches.

First suppose that that the gauge transformations are trivial. Consider
a globally hyperbolic chronal cone bundle $C\to M$ and a local
functional $A[\phi]$ (Sect.~\ref{sec:formal-forms}) on the slow patch $\Secs(F,C)$ of
the space of field configurations. Fix some $\phi\in\Secs(F,C)$ and
consider the formal exterior differential $\delta A[\phi]$ and denote
its value in $T^*_\phi\Secs$ by $\tilde{\alpha}^*$. Then $\delta
A[\phi](\psi) = \langle \psi, \tilde{\alpha}^* \rangle$ for any tangent
vector $\psi\in T_\phi\Secs$.
\begin{definition}\label{def:supp-conf}
We refer to the support of the dual density $\tilde{\alpha}^*$ as the
\emph{local spacetime support} of $A[\phi]$ at $\phi$, $\supp_{M,\phi} A
= \supp \tilde{\alpha}^*$. The \emph{global spacetime support} of
$A[\phi]$ is the closure of the union of all local spacetime supports,
$\supp_M A = \overline{\bigcup_{\phi\in\Secs(F,C)} \supp_{M,\phi} A}$.
\end{definition}
We may sometimes speak simply of spacetime support of $A$ when the
precise version of the notion is clear from context.

More than functionals on field configurations, we are interested in
functionals on the slow patch $\S(F,C)$ of the space of solutions
$\S_H(F)$. So a functional $A[\phi]$ on $\S(F,C)$ is called \emph{local}
if it is the pullback of a local functional on field configurations
along the inclusion $\S(F,C) \sso \Secs(F,C)$. The definition of
spacetime support for a functional on solutions is complicated by the
fact that many local functionals on field configurations may pullback to
the same functional $A[\phi]$ on solutions. This difficulty is resolved
by deciding that $A[\phi]$ admits multiple spacetime supports.
\begin{definition}\label{def:supp-sols}
Consider a local functional $A[\phi]$ on $\S(F,C)$ and any local
functional $A'[\phi]$ on $\Secs(F,C)$, which pulls back to $A[\phi]$ on
solutions. Then we say that $\supp_{M,\phi} A'$ is \emph{a local
spacetime support} of $A$ at $\phi$ and that $\supp_M A'$ is \emph{a
global spacetime support} of $A$. The collections of all local and
global spacetime supports of $A$ are denoted $\supp_{M,\phi} A =
[\supp_{M,\phi} A']$ and $\supp_M A = [\supp_M A']$, where the square
brackets enclose a representative.
\end{definition}
Essentially, this definition states that $\supp_{M,\phi}$ is the
collection of supports of the representatives $\tilde{\alpha}^*$ of the
value of the formal exterior differential $\delta A[\phi] =
[\tilde{\alpha}^*]$ in $T^*_\phi\S$. When referring to the properties of
the spacetime support of $A[\phi]$, we are free to refer to any
representative. For instance we say that a local functional $A[\phi]$ on
$\S(F,C)$ has \emph{a compact local spacetime support at $\phi$} if
$\supp_{M,\phi} A$ contains a compact representative. The same goes for
global spacetime support. A subtly, but importantly different notion is
of \emph{globally compact local spacetime support}, which means that the
local spacetime support $\supp_{M,\phi} A$ has a compact representative
for every $\phi\in \S(F,C)$. Also two such local functionals $A[\phi]$
and $B[\phi]$ have \emph{spacelike separated local spacetime supports at
$\phi$} if $\supp_{M,\phi} A$ and $\supp_{M,\phi} B$ have
representatives that are $\phi$-spacelike separated. If $A[\phi]$ and
$B[\phi]$ have spacelike separated local spacetime supports at every
$\phi\in \S(F,C)$, we say that $A[\phi]$ and $B[\phi]$ have
\emph{globally spacelike separated local spacetime supports}. The same
terminology also works for functionals on field configurations. On the
other hand, if we say that $A[\phi]$ and $B[\phi]$ have
\emph{$C$-spacelike separated global spacetime supports} if $\supp_M A$
and $\supp_M B$ have representatives that are $C$-spacelike separated.
Of course, $C$-spacelike separation of global spacetime supports implies
global spacelike separation of local spacetime supports, but the
converse is does not hold. In fact, one can imagine local functionals
$A[\phi]$ and $B[\phi]$ with globally spacelike separated local
spacetime supports but with $\supp_M A = \supp_M B = [M]$, because of
the way $\supp_{M,\phi} A$ and $\supp_{M,\phi} B$ vary as a function of
$\phi$.

In the presence of non-trivial gauge transformations, the gauge
invariant cotangent spaces, $T^*_\phi \bar\S$, are distinct from the
ordinary ones, $T^*_\phi \S$.
\begin{definition}
A local functional $A'[\phi]$ on $\Secs(F,C)$ is called \emph{gauge
invariant} if its formal exterior derivative $\delta A'[\phi]$ takes
values in the gauge invariant cotangent space $T^*_\phi\bar\Secs \sso
T^*_\phi\Secs$. Similarly, a local functional $A[\phi]$ on $\S(F,C)$ is
\emph{gauge invariant} if its formal exterior derivative $\delta
A[\phi]$ takes values in $T^*_\phi\bar\S \sso T^*_\phi\S$.
\end{definition}
In particular, a local functional $A[\phi]$ (whether on solutions or
field configurations) satisfies $\delta A[\phi](\dot{g}[\eps])$ for any
linearized gauge transformation at $\phi$ with gauge parameter section
$\eps\in \Secs_0(P)$. The notions of compact spacetime support and
$C$-spacelike separated supports specialize straightforwardly to gauge
invariant local functionals.

We are now ready to define the Poisson bracket of two gauge invariant
local functionals $A[\phi]$ and $B[\phi]$, which we write simply as
\begin{equation}
	\{A,B\} = \Pi(\delta A, \delta B).
\end{equation}
In more detail, let $\phi\in \S(F,C)$ and denote the values of the
formal exterior differentials at $\phi$ by $\delta A[\phi] =
[\tilde{\alpha}^*]$ and $\delta B[\phi] = [\tilde{\beta}^*]$, both in
$T^*_\phi\bar\S$. We then have
\begin{equation}\label{eq:poisbr-def}
	\{A,B\}[\phi]
	= \Pi([\tilde{\alpha}^*], [\tilde{\beta}^*])
	= \langle \EE[\tilde{\alpha}^*], \tilde{\beta}^* \rangle ,
\end{equation}
where we have invoked Lem.~\ref{lem:Pi-form} and $\EE$ is the causal
Green function of the Jacobi system at the dynamical linearization point
$\phi$. Incidentally, the last formula in Eq.~\eqref{eq:poisbr-def}
allows the extension of the Poisson bracket off-shell to gauge invariant
functionals on field configurations~\cite{marolf1,marolf2,df-peierls,bfr},
though we will not discuss this possibility at the moment.

A problem with this definition, and the reason we announced at the
beginning of this section that the Poisson bracket would be defined only
formally, is that $\{A,B\}[\phi]$ is in general not a local functional.
That is because the Green function $\EE$ is a non-local integral
operator. Consider two local coordinate charts $(x^i,u^a)$ and
$(y^j,u^b)$ on $F$ so that the dual densities $\tilde{\alpha}^*$ and
$\tilde{\beta}^*$ have the coordinate expressions $\alpha^*_a(x)\,
\d\tilde{x}$ and $\beta^*_b(y)\, \d\tilde{y}$, while $\EE$ has the
integral kernel $\EE^{ab}(x,y)$. Then
\begin{equation}
	\{A,B\}[\phi]
	= \int_{M\times M} \EE^{ab}(x,y) \, \alpha^*_a(x) \beta^*_b(y) \,
		\d\tilde{x} \wedge \d\tilde{y} .
\end{equation}
As $\EE^{ab}(x,y)$ can be seen as a vector valued distribution on
$M\times M$, the above integral is well defined, but cannot be
re-expressed as a local functional unless $\EE^{ab}(x,y)$ is supported
only on the total diagonal $x=y$, which it obviously is not. It should
be also clear that the product of two local functionals $A[\phi]B[\phi]$
is also not a local functional.

To address both of the above problems, we expand the set of functionals
under consideration to polynomials in local functionals.
\begin{definition}
Let $\Loc(F,C)$ denote the space of local functionals on $\Secs(F,C)$
with compact local spacetime support and $\PolyLoc(F,C)$ the polynomial
algebra generated by $\Loc(F,C)$. Similarly, $\Loc(F)$ and $\Loc(F)$
refer to functionals on $\Secs_H(F)$. The elements of $\PolyLoc(F,C)$
and $\PolyLoc(F)$ are referred to as \emph{polylocal} functionals.
\end{definition}
As before for local functionals, polylocal functionals may be defined on
field configurations as well as solutions. Recall that we did not give a
precise definition of infinite dimensional manifolds in our category
$\Man$ nor of the algebra of smooth functions on them. However, whatever
precise definitions that are ultimately chosen, we require the following
\begin{hypothesis}\label{hyp:Coo}
The functor of smooth functions $C^\oo\colon \Man \to \CAlg$ is such
that $C^\oo(\Secs(F,C))$ is a topological algebra containing
$\PolyLoc(F,C)$ as a dense sub-algebra, which also contains the Poisson
brackets $\{A,B\}$ of local functionals and such that $\{-,-\}$ extends
continuously to all of $C^\oo(\Secs(F,C))$. Analogous statements hold for
$C^\oo(\Secs_H(F))$ and $\PolyLoc(F)$ and upon replacement of field
configurations with solutions.
\end{hypothesis}

\begin{remark}\label{rem:polyloc-supp}
It is straight forward to generalize the formal exterior derivative from
\emph{local} functionals to \emph{polylocal} ones, as well as to more
general $C^\oo$ functionals. The same goes for the notions of
\emph{local spacetime support} and \emph{global spacetime support}. We
silently make use of this generalization from now on.
\end{remark}

So, we will content ourselves with approximating the Poisson
bracket $\{A,B\}$ by a sequence (or net or filter) of polylocal
functionals $C_i[\phi]$, with $i$ from some index set $I$. This
concludes the formal definition of the algebra of observables and the
Poisson bracket on it.
\begin{definition}
We call the algebra of gauge invariant smooth functions on the slow
patch $\S(F,C)$
together with the Poisson bracket~\eqref{eq:poisbr-def} the
\emph{Poisson algebra of observables} of the classical field theory and
denote it
\begin{equation}
	\F(F,C) \cong (C^\oo_{cst}(\bar\S(F,C)),\{\}),
\end{equation}
where the
subscript ${}_{cst}$ means that we only take functionals with compact
local spacetime support. (Note that we have identified the gauge
invariant functions on $\S(F,C)$ with functions on $\bar{S}(F,C)$, the
space of gauge equivalence classes of solutions.) Denote the underlying
commutative algebra of $\F(F,C)$ by $\A(F,C) \cong
C^\oo_{cst}(\bar\S(F,C))$, which will also be referred to as the
\emph{algebra of observables}.
\end{definition}
Of course, any results concerning $\F(F,C)$ and $\A(F,C)$ that we can
establish in this paper will only be at the level of $\PolyLoc(F,C)$,
which will be strengthened by Hyp.~\ref{hyp:Coo}.

Finally, we are ready to state and prove the main result of this
section. In quantum field theory, it usually goes under the name of
\emph{microcausality}: observables that depend on fields in spacelike
separated regions have vanishing commutators. The classical version
replaces the commutator by the Poisson bracket. Note, however, that the
result below is an important generalization of the known microcausality
property of field theories with \emph{semilinear} equations of motion,
i.e.,\ those with a non-dynamical causal structure. Usually, only the
\emph{global} spacetime support is considered and, even if the
\emph{local} spacetime support were to be considered, the notion of space-like
separation is independent of the dynamical fields and so does not change
when the observables are linearized about different solutions (for
instance, in perturbation theory). For field theories with
\emph{quasilinear} equations of motion, and hence dynamical causal
structures, the notion of spacelike separation is intrinsically changes
from solution to solution on the phase space. Thus, to even have a
suitable generalization of the spacelike separation hypothesis for
microcausality, we are forced to consider local spacetime support.
\begin{theorem}[Classical microcausality]\label{thm:microcaus}
Consider two observables $A$ and $B$ belonging to the Poisson algebra of
observables $\F(F,C)$. If their local spacetime supports are spacelike
separated, then their Poisson bracket vanishes identically,
\begin{equation}
	\{A,B\} = 0.
\end{equation}
\end{theorem}
\begin{proof}
Consider any solution $\phi\in \S(F,C)$. Let the dual densities
$\tilde{\alpha}^*,\tilde{\beta}^*\in \Secs_0(\tilde{F}^*)$ represent the
formal exterior derivatives of the observables, $\delta A[\phi] =
[\tilde{\alpha}^*]$ and $\delta B[\phi] = [\tilde{\beta}^*]$. Pick a
pair of local coordinates $(x^i,u^a)$ and $(y^j,u^b)$ on $F$ such that
the dual densities have the components $\alpha^*_a(x) \d\tilde{x}$
and $\beta^*_b(y) \d\tilde{y}$. By hypothesis, $\supp \alpha^*_a$ and
$\supp \beta^*_b$ are $\phi$-spacelike separated.

From the Peierls formula for the Poisson bracket~\eqref{eq:poisbr-def}
\begin{equation}\label{eq:poisbr-def2}
	\{A,B\}[\phi]
	= \int_{M\times M} \EE^{ab}(x,y) \,
		\alpha^*_a(x) \beta^*_b(y) \, \d\tilde{x}\wedge \d\tilde{y} ,
\end{equation}
where $\EE^{ab}(x,y)$ is the integral kernel of the causal Green
function $\EE$ of the Jacobi system linearized at $\phi$. However, by
Eq.~\eqref{eq:E-def}, $\EE = \G\circ \bar{r}$, where $\G$ is the causal Green
function of a linear symmetric hyperbolic system with adapted equation
form $(\dot{f},\tilde{F}^*)$, whose principal symbol precisely
determines the causal structure of the solution $\phi$
(Sect.~\ref{sec:lin-inhom}).
Recall that the causal Green function is the difference of the retarded
and advanced Green functions, $\G = \G_+ - \G_-$. Thus, the support of
$\G$, and hence $\EE$ and $\EE^{ab}(x,y)$, since $\bar{r}$ is a
differential operator and does not increase supports, excludes any pair
of points $(x,y)$ that are $\phi$-spacelike separated.

It is now an easy conclusion that, according to the
formula~\eqref{eq:poisbr-def2}, the Poisson bracket $\{A,B\}[\phi]$
vanishes at $\phi\in \S(F,C)$. Since the solution $\phi$ was arbitrary,
it follows that the Poisson bracket vanishes identically, $\{A,B\} = 0$.
\qed
\end{proof}

A coarser version of the above result, though stated in the more
familiar terms of global spacetime support follows immediately once the
definitions are unwound.
\begin{corollary}\label{cor:C-microcaus}
Consider two observables $A$ and $B$ belonging to the Poisson algebra of
observables $\F(F,C)$. If their global spacetime supports are
$C$-spacelike separated, then their Poisson bracket vanishes
identically, $\{A,B\} = 0$.
\end{corollary}

\begin{remark}
The notion of spacetime support in previous work on the functional
approach to classical and quantum field
theory~\cite{bdf,fr-bv,rejzner-thesis} corresponds to
what we call global spacetime support. The notion of local spacetime
support and of spacelike separation of local spacetime supports seems to
have been overlooked until now, with the notable exception
of~\cite{bfr,ribeiro}. The distinction between local and global
spacetime support is
particularly important for gauge theories and theories with dynamical
causal structures, both exemplified by GR. It is well known that there
are no gauge invariant observables in GR with compact global spacetime
support, a significant technical complication. However, the door is
still open to discover a large, technically convenient class of
observables with compact local spacetime support.
\end{remark}

Another important but simple result is
\begin{lemma}[Covariance]\label{lem:pois-covar}
Consider two globally hyperbolic chronal cone bundles $C'\to M'$ and
$C\to M$, where $M'\sse M$ and $C'$ is faster than $C$. If we restrict
our Lagrangian and PDE system to $M'$, there are induced maps of
solution spaces $\S(F,C) \to \S(F',C')$ and algebras $\A(F',C') \to
\A(F,C)$, where $F' = F|_{M'}$ is the field bundle over $M'$. These maps
induce morphisms of Poisson manifolds $(\S(F,C),\Pi) \to
(\S(F,C'),\Pi')$ and a Poisson homomorphism $(\A(F',C'),\{\}') =
\F(F',C') \to \F(F,C)$, where $\Pi'$ is also defined by the Peierls
formula and in turn defines $\{-,-\}'$.
\end{lemma}
\begin{proof}
The agreement of the Poisson bivectors $\Pi'$ and $\Pi$ follows directly
from the Peierls formula and the covariance lemma for the causal Green
function, Lem~\ref{lem:green-covar}. This agreement immediately implies
that the induced algebra homomorphism $\A(F',C') \to \A(F,C)$ is
automatically a Poisson homomorphism with respect to the brackets
$\{-,-\}'$ and $\{-,-\}$. \qed
\end{proof}

\subsubsection{The global phase space}\label{sec:glob-phsp}
Here we take the final step in the construction of the phase space of
the classical field theory and the Poisson algebra of observables on it.
We appeal to the results obtained for each slow patch $\S(F,C)$ and
recall that these patches form an open cover of the total space of
solutions $\S_H(F)$, according to Hyp.~\ref{hyp:opencover}. Moreover, each of the open
patches carries a formally smooth presymplectic form $\Omega_C[\phi]$
(which becomes a symplectic one $\bar\Omega_C[\phi]$) when projected to
$\bar\S(F,C)$, Sects.~\ref{sec:formal-symp} and~\ref{sec:formal-pois}. The symplectic structure is
translated to a Poisson structure $\Pi_C[\phi]$ by pointwise-inversion
(smoothness and Jacobi identity are presumed to follow from the
corresponding properties of the symplectic form),
Sect.~\ref{sec:formal-pois}. It remains
only to check that the symplectic and Poisson structures agree as these
patches are glued together. In a later section,
Sect.~\ref{sec:cat-glob-phsp}, this
construction via slow patches will be interpreted in categorical terms
as a \emph{colimit}, which will play an important role in the
generalized Causality property of classical field theory with
quasilinear equations of motion.

Consider two chronally comparable globally hyperbolic chronal cone
bundles $C_1\to M$ and $C_2\to M$, as well as their intersection
$C_1\cap C_2 = C_3 \to M$, also a globally hyperbolic cone bundle. The
each we can associate the corresponding patch of $C_i$-slow solutions
$\bar\S(F,C_i)$ modulo gauge transformations. On each such patch, we
have the symplectic and Poisson tensors $\bar\Omega_{C_i}[\phi]$ and
$\Pi_{C_i}[\phi]$. The slow patches intersect as $\bar\S(F,C_1)\cap
\bar\S(F,C_2) = \bar\S(F,C_3)$.  Therefore, we must check whether
$\bar\Omega_{C_i}[\phi]$ all agree when restricted to $\bar\S(F,C_3)$.
Recall that each of the symplectic forms was defined with the help of a
$C_i$-Cauchy surface $\Sigma_i\sso M$. Each of the three symplectic
forms is defined by exactly the same formula, Cor.~\ref{cor:bOmega}, with the
exception that the integration surface $\Sigma$ is replaced by the
corresponding $\Sigma_i$. By construction, $\Sigma_1$ and $\Sigma_2$ are
both $C_3$-Cauchy and so lie in the same $C_3$-spacelike homology class.
Since $\bar\Omega_{C_3}[\phi]$ depends only on that homology class, we
can directly establish that
\begin{equation}
	\bar\Omega_{C_1}[\phi]
= \bar\Omega_{C_3}[\phi]
= \bar\Omega_{C_2}[\phi]
\end{equation}
on $\bar\S(F,C_3)$. The Poisson tensors $\Pi_{C_i}[\phi]$ necessarily
satisfy the same property.

From the above discussion, we can conclude that the global space of
globally hyperbolic solutions $\S_H(F)$ is endowed with both symplectic
and Poisson structures, $\bar\Omega[\phi]$ and $\Pi[\phi]$, defined by
their agreement with the corresponding tensors on the slow patches
$\bar\S(F,C)\sse \bar\S_H(F)$. We call the symplectic ma\-ni\-fold $\PP_H(F)
= (\bar\S_H(F),\bar\Omega)$ the \emph{global phase space} of the
classical field theory under consideration. The Poisson algebra of
smooth functions on it, with Poisson structure defined by $\Pi[\phi]$,
is called the \emph{global algebra of observables} $\F_H(F) =
(\A_H(F),\{\})$, where $\A_H(F) = C_{cst}^\oo(\S_H(F))$ is the underlying
commutative algebra.

The notions of local spacetime supports of observables and their global
spacelike separation, as in Defs.~\ref{def:supp-conf}
and~\ref{def:supp-sols}, translate directly from
the slow patches to the global phase space. Therefore, as a direct
consequence of Thm.~\ref{thm:microcaus}, we have
\begin{corollary}[Global classical microcausality]\label{cor:glob-microcaus}
Consider two observables $A$ and $B$ belonging to the global Poisson
algebra of observables $\F_H(F)$. If their local spacetime supports are
globally spacelike separated, then their Poisson bracket vanishes
identically,
\begin{equation}
	\{ A, B \} = 0.
\end{equation}
\end{corollary}

\section{Natural Variational PDE Systems and Functoriality}
\label{sec:natural}
In this section we take account of the various categories and functors
introduced so far and introduce a few more. The main goal is to make
explicit the functorial nature of all the steps in the construction of a
classical field theory. For reference, some basic information about
categories and functors can be found in~\cite{cattheory,borceux}.

\subsection{Natural bundles}
So far, we have looked at a fixed spacetime manifold $M$, a fixed field
vector bundle over $F\to M$, and a fixed PDE system $\E\to M$ on it,
without concerning ourselves how the field bundle of the PDE system
should change when the spacetime manifold is changed. To make a
connection with covariant field theory, which functorially assigns
algebras to spacetime manifolds, it is helpful to look at systems of
partial differential equations that can also be assigned to spacetime
manifolds functorially. As before, we presume that all fields are
sections of vector bundles. The discussion below can be easily adapted
to arbitrary smooth bundles.  There are two issues that need to be
addressed here: vector bundles and differential operators on them. Given
an open embedding of manifolds $M\to M'$, there is a priori no natural
map from vector bundles $F\to M$ and $F'\to M$, even if the fibers of
$F$ and $F'$ are modeled on the same vector space.  The absence of such
maps is an obstruction to relating sections of $F\to M$ to sections of
$F'\to M$. The needed bundle maps must to be specified along with the
typical fiber and the total space topology, which leads to the notion of
a \emph{natural bundle}. It is common practice to use a connection to
define differential operators acting on sections of a vector bundle
(that practice was adopted in~\cite{geroch-pde}, instance). But, once again, for
an arbitrary vector bundle, there is no natural choice of connection. It
then becomes convenient to work with jet bundles instead, as we have
elected to do here. The switch between these points of view is
straightforward~\cite[\textsection 17]{kms}.

A field configuration corresponds to a section of a vector bundle $F\to
M$ over a spacetime manifold $M$. To implement covariance with respect
to arbitrary base space morphisms $\chi\colon M\to M'$ we need to be
able to turn sections of $F\to M$ into sections of the corresponding
bundle $F'\to M'$, restricted to the image $\chi(M)$. This requires a
uniquely determined corresponding bundle morphism, $F\to F'$ that fits
in the commutative diagram
\begin{equation}
\vcenter{\xymatrix{
	F \ar[r] \ar[d] & F' \ar[d] \\
	M \ar[r]^\chi   & M' ~ .
}}
\end{equation}
In essence, we cannot allow $F$ and $F'$ to be arbitrary vector bundles;
we need them to be associated functorially to $M$ and $M'$ respectively,
which motivates the following definition.
\begin{definition}
A \emph{natural vector bundle} $F$ over $n$-dimensional manifolds is a
functor from the category of manifolds to the category of vector
bundles, $F\colon \Man^n\to \VBndl$, that is right-inverse to the
forgetful base space functor, $\base\colon (A\to M) \mapsto M$, that is,
$\base\circ F = \id$.
\end{definition}

Examples of natural vector bundles are the tangent bundle $T$, the
cotangent $T^*$. Operations on natural vector bundles like direct sums,
direct products, linear duals, tensor products and jet extensions also
define natural bundles. These and other examples are discussed
extensively in~\cite{kms}. Given a natural vector bundle
$F\colon \Man \to \VBndl$, its $k$-jet bundle $M\mapsto J^kF(M)$ is also
natural.

As we have seen in Sect.~\ref{sec:jets-pdes}, to specify a $k$-th
order PDE system on a bundle $F$, we need only specify a bundle map from
$J^kF$ to an equation bundle $E$. If $F$ and $E$ are natural bundles,
then a natural transformation between $J^kF$ and $E$ defines precisely
these kinds of morphisms, $J^kF(M) \to E(M)$ in a way covariant with
diffeomorphisms of $M$. Recall that it is sometimes convenient to see
$M\stackrel{\id}{\to} M$, a trivial bundle with zero dimensional fiber,
as an object in $\Bndl$ and $\id\colon M \mapsto (M\stackrel{\id}{\to}
M)$ as the natural \emph{identity} bundle.
\begin{definition}
A \emph{natural bundle map} $f$ between natural ($\Bndl$-valued) bundles
$F$ and $E$ is a natural transformation $f\colon F\to E$.

A \emph{natural section} $f$ of a natural bundle $F$, is a natural
bundle map $f\colon M \to F$, where the domain is treated as the
identity bundle.
\end{definition}
Two examples of natural sections, given say a natural vector bundle $F$,
are the zero section and the identity section of $F^*\otimes F$. An
example of a natural bundle map, is the dual pairing
$\langle-,-\rangle\colon F^*\otimes F \to \R$, where $\R$ is the trivial
real line bundle. Other examples are provided by smooth real functions
$f\colon \R\to \R$, which can be promoted to a natural bundle map by
applying it fiberwise to the natural trivial bundle $M\times \R$.

We now have the necessary concepts to introduce natural variational
hyperbolic PDE systems.
\begin{definition}
Let $F$, $B$, $V=F\times B$ and $E$ be respectively natural
\emph{dynamical field}, \emph{background field}, \emph{total field} and
\emph{constraint} vector bundles.

A \emph{natural first order, quasilinear, symmetric hyperbolic PDE
system with constraints} $f\oplus c$ is a natural bundle map $f\oplus c
\colon J^1F\times J^\oo B \to F^*\oplus E$, such that the bundle maps
$f_M$, for any $M$ in $\Man^n$, satisfy Def.~\ref{def:symhyp}, for some
natural, non-empty subset of each fiber of $J^\oo B(M)$. 

A \emph{natural Lagrangian density} is a natural bundle map $\L\colon
J^1F\times J^\oo B\to \Lambda^nM$.

A \emph{natural symmetric hyperbolic variational PDE system} is a pair
$(f\oplus c,\L)$ where $\L$ is a natural Lagrangian density and $f$ is a
natural first order, quasilinear, symmetric hyperbolic PDE system that
is naturally equivalent (Def.~\ref{def:pde-equiv}) to gauge fixed
Euler-Lagrange equations of $\L$ (Sect.~\ref{sec:Jgreen}).
\end{definition}

The introduction of background fields is sometimes necessary to make the
Lagrangian natural. One may also need extra background fields to put the
corresponding Euler-Lagrange equations into symmetric hyperbolic form in
a natural way, as discussed for example in~\cite{geroch-pde}. The example of a
scalar wave equation on a curved background is discussed in some detail
at the end of the next section.

A field section $\psi$ on $M$ will from now on always be a section of
the total field bundle $\psi=(\phi,\beta) \colon M\to V = V\times B$.
As a consequence of the dependence of the principal symbol $\bar{f}$ of
$f$ on the value of background fields, implies that the naturally
defined chronal and spacelike cone bundles $\Gamma(M)$  and
$\Gamma^\oast(M)$ are cone bundles over $F\times J^\oo B(M)$.

An unfortunate complication in the presence of background fields is that
not all background field configurations are compatible either with the
intrinsic integrability conditions of the PDE system (like the
restriction to Einstein vacuum backgrounds in the Rarita-Schwinger
system~\cite{hm-sp32}) or with symmetric hyperbolicity (like the
restriction to Lorentzian signature of the background metric for wave
equations). With this in mind, we say that a background field section
$\beta\colon M\to B(M)$ is \emph{admissible} when (a) there exists a
total field section $\psi=(\phi,\beta)\colon M\to V$ that is globally
hyperbolic and solves the PDE system $f=0$, (b) for any $\psi$-Cauchy
surface $\Sigma \sso M$, there exists an open neighborhood of $N\supset
\Sigma$ such that the initial data $\xi|_\Sigma$ of the space of
globally hyperbolic solutions of the form $(\xi,\beta)$ on $N$ forms an
open neighborhood of $\phi|_\Sigma$ in the space of initial data allowed
by the constraints $c=0$ (without taking further integrability
conditions into account).

\subsection{Functoriality}\label{sec:funct}
In this section we (a) systematically summarize some of the concepts and
notations used throughout this paper, (b) systematically summarized
various constructions that appear as steps in the construction of a
classical field theory, and (c) remark on the categorical and functorial
properties of these objects and constructions. This section will serve
as a reference for the later discussion of how LCFT axioms, including
causality, translate to field theories with dynamical causal structure.

\begin{enumerate}
\item
	$\Man$, $\Man^n$, $\Man^\oo$: Category of smooth manifolds,
	subcategories of $n$-di\-men\-sion\-al and infinite dimensional manifolds.
	Subscript ${}_e$ stands for restriction of morphisms to open
	embeddings.
\item
	$\Bndl$: Category of smooth bundles, a subcategory of $\Man$. Fibered
	over $\Man$ with the base space functor%
	\footnote{We will re-use the symbol $\base$ to generically denote a
	forgetful functor.} %
	$\base: (E\to M) \mapsto M$.
\item
	$\VBndl$: Category of vector bundles, a subcategory of $\Bndl$.
\item
	$\CBndl$, $\ChrBndl$, $\SpBndl$: Subcategories of $\Bndl_e$ of cone
	bundles, chronal cone bundles, spacelike cone bundles. Subscripts
	${}_{sc}$ and ${}_H$ indicate stable chronality and global
	hyperbolicity and superscript ${}^n$ denotes the restriction to
	$n$-dimensional base manifolds. Morphisms in $\ChrBndl$ and $\SpBndl$
	must be chronally convex. As functors from $\Man$, the images of
	$M\mapsto \ChrBndl(M)$ and $M\mapsto \SpBndl(M)$ indicate the
	respective subcategories of cone bundles over $M$.
\item
	$J^k\colon\VBndl\to \VBndl$: Functor of jet prolongation of a vector
	bundle.
\item
	$\Secs,\Secs_0\colon \VBndl\to \Man^\oo$: Contravariant functors of
	sections, sections with compact support.
\item
	$\Secs_{SC,+,-}\colon \VBndl \times_\base \SpBndl_H \to \Man^\oo$:
	Contravariant functors of sections with spacelike compact, retarded
	($+$), or advanced ($-$) supports. The objects of the domain category
	can be identified with pairs $(F,C^\oast)$, with $F\to M$ a vector
	bundle and $C^\oast\to M$ is a spacelike cone bundle over the same
	$M$. Throughout the paper we have also used the notation
	$\Secs_{SC,+,-}(F,C)$, where $C\to M$ is a chronal cone bundle. These
	notations can be used interchangeably because of the duality between
	chronal and spacelike cone bundles. However, these sections functors
	are covariant only if applied to spacelike cone bundles.
\item
	$B,F,V,E\colon\Man_e^n\to \VBndl$: Natural bundle functors of
	background, dynamical, and total ($V=B\oplus F$) field bundles as well
	as the equation bundle.
\item
	$\L\colon J^\oo F\times J^\oo B\to \Lambda^nM$: Natural bundle section,
	defining a covariant Lagrangian density. It can also be interpreted as
	an differential form $\L\in \Omega^{n,0}(F\oplus B)$ in the
	variational bicomplex.
\item
	$f\oplus c\colon J^1F\times J^\oo B \to F^*\oplus E$: Natural bundle section
	defining a covariant first order, quasilinear, symmetric hyperbolic PDE
	system, which is equivalent to gauge fixed Euler-Lagrange equations of
	$\L$.
\item\label{itm:bkgr}
	$\Bkgr$: Category of manifolds augmented by collections of admissible
	background field configurations.  Objects are pairs $(M,\B)$ with $M$
	an $n$-manifold and $\B\sse \Secs(B(M))$. Each element $\beta\in \B$
	must satisfy all the necessary integrability conditions and be
	compatible with the existence of globally hyperbolic dynamical field
	solutions. A morphism $\chi\colon (M,\B) \to (M',\B')$ determines an
	open embedding $\chi\colon M\to M'$ compatible with background fields,
	$\chi^*(\B') \sse \B$ being a closed embedding (to play nice with
	$C^\oo$, as discussed in item~\ref{itm:Coo} below).  $\Bkgr$ is
	fibered over $\Man_e^n$ with respect to the forgetful functor to
	$\base\colon (M,\B) \mapsto M$.
\item\label{itm:spbkgr}
	$\SpBkgr = (\SpBndl^n_H \times_\base \Bkgr)_c$: Category of admissible
	background field configurations, equipped with a globally hyperbolic
	spacelike cone bundle. As functor from $\Bkgr$, the image of $M\mapsto
	\SpBkgr(M,\B)$ indicates the subcategory of different ways of
	equipping $(M,\B)$ with a spacelike cone bundle. The objects of
	$\SpBkgr$ will be denoted as $\M=(M,C^\oast,\B)$. The forgetful
	functor is again denoted $\base\colon \SpBkgr \to \Bkgr$, with
	$\base(\M) = \base(M,C^\oast,\B) = (M,\B)$. The subscript ${}_c$
	indicates that morphisms are restricted to those that are chronally
	compatible (Def.~\ref{def:chr-compat}).
\item
	$\Symp$: The generalized category of symplectic manifolds, a
	subcategory of $\Man$. It includes symplectic manifolds as well as
	manifolds foliated by symplectic leaves (which is actually the same as
	the category of regular Poisson manifolds). The morphisms are leaf
	preserving symplectomorphisms. The forgetful functor to manifolds is
	$\base\colon \Symp\to \Man$.
\item
	$\CAlg$, $\Poiss$: Categories of commutative algebras and Poisson
	algebras, with the forgetful functor $\base\colon \Poiss\to \CAlg$.
	Though, keeping away from functional analytic details, we implicitly
	treat these categories as though the objects are equipped with
	topology and the homomorphisms are continuous. Such continuity
	requirements also identify a tensor product $\otimes$ (identified with
	the categorical coproduct), which is compatible with the categorical
	product in $\Man$, as specified in Hyp.~\ref{hyp:Coo-exact} below. The
	extension of the tensor product to $\Poiss$ is the independent
	subsystems tensor product.
\item\label{itm:Coo}
	$C^\oo$, $C_{cst}^\oo$, $\POISS$: Functors of smooth functions of
	ordinary manifolds, that restricted to compact spacetime support,
	and that equipped with Poisson bracket on symplectically foliated
	manifolds, which fit in the following commutative diagram:
	\begin{equation}
	\vcenter{\xymatrix{
		\Symp \ar[d]^\base \ar[r]^\POISS & \Poiss \ar[d]^\base \\
		\Man \ar[r]^{C_{cst}^\oo} & \CAlg ~ .
	}}
	\end{equation}
	The definition of $C_{cst}^\oo(M)$ is unambiguous when $M$ is a finite
	dimensional manifold. On the other hand, its definition requires some
	functional analytical detail not tackled here. However, we have
	already postulated that, whatever ultimate definition is adopted for
	it, it will have to satisfy Hyp.~\ref{hyp:Coo}. In addition to that,
	let us make explicit additional hypotheses, of a more category
	theoretic nature, that we presume it would satisfy.
	\begin{hypothesis}\label{hyp:Coo-exact}
	The functor $C_{cst}^\oo$ gives a \emph{sheaf} of commutative
	algebras~\cite{sheaves}
	on a manifold $M$ when applied to the open subsets of
	$M$. Moreover, $C_{cst}^\oo$ maps surjections of manifolds to injective
	homomorphisms (``left exactness'') and maps closed embeddings of
	manifolds to surjective homomorphisms (``right exactness'' or smooth
	Urysohn lemma). Products of manifolds are taken to tensor products of
	algebras, $C_{cst}^\oo(M\times N) \cong C_{cst}^\oo(M)\otimes C_{cst}^\oo(N)$.
	``Transverse'' pullbacks of manifolds are taken to pushouts of
	algebras.
	\end{hypothesis}
	\begin{remark}
	The statement about \emph{transverse} pullbacks is standard in
	differential geometry of finite dimensional
	manifolds~\cite[Thm.2.8]{mr}. However, the notion of transversality
	is much more subtle for infinite dimensional
	manifolds~\cite{cbdm,lang,hamilton}. In fact, the notions of
	transversality discussed in these references may not be general enough
	to encompass the situation where it is needed in the formulation of
	the generalized Time Slice property in Sect.~\ref{sec:gen-ts}. At this
	point, we have no choice but to leave the existence of the appropriate
	notion of transversality as a conjecture.
	\end{remark}
\item
	$\S\colon\SpBkgr\to \Man$, $\PP\colon\SpBkgr_c \to \Symp$: The
	contravariant functor assigning $\S(\M)$, the set of all $C^*$-slow
	solutions of the gauge fixed Euler-Lagrange equations (equivalently of
	the corresponding constrained symmetric hyperbolic system) compatible
	with the background field configurations in $\B$, to the object
	$\M=(M,C^*,\B)$ of $\SpBkgr$. The contravariant functor $\PP$ augments
	the images of $\S$ with symplectic structure defined by the Lagrangian
	$\L$. A fixed element $\beta\in \B$ singles out a symplectic leaf in
	$\S(\M)$. It is the dependence of solutions on background fields that
	forces us to consider generalized symplectic manifolds. Note,
	however, that a symplectically foliated manifold must have at least
	dimension two. On the other hand, many of the $\S(\M)$ spaces will be
	empty sets, hence cannot be the underlying manifolds of images of
	$\PP$. For that reason, the domain category of $\PP$ is restricted to
	$\SpBkgr_c$, where the subscript ${}_c$ indicates the largest
	subcategory such that no object in the image $\S(\SpBkgr_c)$ is empty.
\item
	$\SS(M,\B)\colon \SpBkgr(M,\B) \to \Man$, $\PPP(M,\B)\colon
	\SpBkgr_c(M,\B) \to \Symp$: The contravariant functor whose image in
	$\Man$, resp.\ $\Symp$, consists of the diagram $\S(\M)\to \S(\M')$,
	resp.\ $\PP(\M) \to \PP(\M')$, of all slow patches on the total
	solution space $\S_H(M,\B)$, resp.\ phase space $\PP_H(M,\B)$, with
	the natural inclusions between them. 
\item
	$\S_H\colon \Bkgr\to \Man$, $\PP_H\colon\Bkgr_c \to \Symp$: The
	contravariant functor of all globally hyperbolic solutions that are
	compatible with specified background field configurations. The
	contravariant functor $\PP_H$ augments the images of $\S_H$ with
	symplectic structure, as above, where we have defined $\Bkgr_c =
	\base(\SpBkgr_c)$. Each $\S(\M)$ is an open slow patch of the total
	solution space $\S_H(M,\B)$. The natural isomorphisms $\S_H \cong
	\varinjlim \circ \SS$ and $\PP_H \cong \varinjlim \circ \PPP$ are
	verified in Sect.~\ref{sec:cat-glob-phsp}. These functors fit into
	the following commutative diagram:
	\begin{equation}
	\vcenter{\xymatrix{
		\SpBkgr_c \ar[dd]^\base \ar[dr]_\PP \ar[r]^\sse & \SpBkgr \ar[drr]^\S \\
		& \Symp \ar[rr]^\base & & \Man  ~ . \\
		\Bkgr_c \ar[ur]^{\PP_H} \ar[r]_\sse & \Bkgr \ar[urr]_{\S_H}
	}}
	\end{equation}
\item
	$\F=\POISS\circ \,\PP$, $\A = C_{cst}^\oo\circ \S = \base\circ \F$: Covariant
	functors of the Poisson algebra of observables on $\S(M)$ and its
	underlying commutative algebra, specified by a slow patch $\M =
	(M,C^\oast,\B)$.
\item
	$\AA(M,\B)\colon \SpBkgr(M,\B) \to \CAlg$, $\FF(M,\B)\colon
	\SpBkgr_c(M,\B) \to \Poiss$: The covariant functor whose image in
	$\CAlg$, resp.\ $\Poiss$, consists of the diagram $\A(\M)\to \A(\M')$,
	resp.\ $\F(\M) \to \F(\M')$, of all algebras of observables associated
	to slow patches in the total algebra $\A_H(M,\B)$, resp.\ Poisson
	algebra $\F_H(M,\B)$, with the natural projections between them. 
\item
	$\F_H = \POISS\circ \,\PP_H$, $\A_H=C_{cst}^\oo\circ \S_H = \base\circ \F_H$:
	The covariant functors of the Poisson algebra of observables and its
	underlying commutative algebra of the total phase space $\S_H(M,\B)$.
	The natural isomorphisms $\F_H = \varprojlim \circ \FF$ and $\A_H =
	\varprojlim \circ \AA$ are verified in Sect.~\ref{sec:cat-glob-phsp}.
	These functors fit in the following commutative diagram:
	\begin{equation}
	\vcenter{\xymatrix{
		\SpBkgr_c \ar[dd]^\base \ar[dr]_\F \ar[r]^\sse & \SpBkgr \ar[drr]^\A \\
		& \Poiss \ar[rr]^\base & & \CAlg  ~ . \\
		\Bkgr_c \ar[ur]^{\F_H} \ar[r]_\sse & \Bkgr \ar[urr]_{\A_H}
	}}
	\end{equation}
\item
	$\Bkgr_c$, $\Bkgr^*$, $\Bkgr'$: Subcategories of $\Bkgr$. As defined in a
	preceding item, $\Bkgr_c = \base(\SpBkgr_c)$. We define $\Bkgr^*$ as
	the subcategory where the collections of admissible background
	configurations are singletons, $\B = \{\beta\}$, and $\Bkgr'$ is the
	image of the category $\Bkgr^*$, where the auxiliary components of the
	background fields have been thrown away.
\end{enumerate}

With all the categorical notions summarized above, it should now be
clear how to recover the standard formulation of the LCFT axioms given
in Sect.~\ref{sec:freelcft}. Consider a semilinear wave equation on a
curved background with Lagrangian density
\begin{equation}
	\L = -\frac{1}{2}(g^{\mu\nu}\del_\mu\phi\del_\nu\phi
		+ V(\phi))\sqrt{-g}\,\d\tilde{x} .
\end{equation}
The trivial dynamical field bundle $F\colon M \mapsto \R\times M$ is
clearly natural. However, the Lagrangian density $\L$ only becomes
natural if we include the metric tensor among the background fields,
that is we set $B\colon M\mapsto S^2T^*M$, which is also clearly
natural. The Euler-Lagrange equations automatically give a natural
hyperbolic PDE system. (This system is \emph{normally hyperbolic}~\cite{bgp,bfr}.
For simplicity we do not put this system into symmetric hyperbolic form,
which would simply require extending the dynamical and background field
bundles~\cite{geroch-pde}.)

It is simple to check that restricting to the subcategory $\Bkgr^*_c$ of
$\Bkgr_c$, where objects $(M,\B)$ are equipped with a single background
field configuration $\B=\{g\}$, forces the allowed morphisms to be only
causal isometries (cf.~Def.~\ref{def:globhyp},
Def.~\ref{def:chr-compat}, Lem.~\ref{lem:chr-compat-convex} and
items~\ref{itm:bkgr}, \ref{itm:spbkgr} above). This means that we have
an equivalence of categories $\Bkgr^*_c \cong \GlobHyp_c$. By
construction then the covariant functor
\begin{equation}
	\F_H\colon \Bkgr^*_c \cong \GlobHyp_c\to \Poiss
\end{equation}
has the potential to be a classical LCFT. It remains only to check the
axioms of Def.~\ref{def:axlcft}. These checks are carried out at the end
of Sect.~\ref{sec:semilin-lcft}.

\subsection{Limits, colimits and the global phase space}
\label{sec:cat-glob-phsp}
We have seen earlier, more specifically in Sect.~\ref{sec:glob-phsp},
that the slow patches $\S(\M)$ constitute an open cover of the global
solution space $\S_H(M,\B)$. How can this relation be restated for the
algebras $\A_H(M,\B)$ and $\A(\M,\B)$, or the Poisson algebras
$\F_H(M,\B)$ and $\F(\M)$? The answer is not immediately obvious.
However, the situation becomes more clear when one realizes that the
desired constructions can be formulated in terms of the categorical
notions of $\emph{limit}$ and \emph{colimit}. Categorical limits (also
\emph{projective} or \emph{inverse} limits) and colimits (also
\emph{inductive} or \emph{direct} limits) are briefly introduced in
Sect.~\ref{sec:limits}.

A great advantage of considering the Poisson algebras $\F(\M)$ is the
simplified microcausality property, Cor.~\ref{cor:C-microcaus}, while
only the more refined microcausality property of
Cor.~\ref{cor:glob-microcaus} survives for $\F_H(M,\B)$. On the other
hand, the $\F(\M)$ for different cone bundles $C\to M$ are included as
subalgebras of $\F_H(M,\B)$, though not as independent ones. The
situation is strongly parallel to the fact that the slow patches
$\S(\M)$ are open subsets of $\S_H(M,\B)$, though with non-trivial
overlaps. The parallel is made precise by recognizing that the globally
hyperbolic slow patches $\S(\M)$, as well as their symplectic analogs
$\PP(\M)$, and the inclusions between them constitute the diagram
$\SS(M,\B)$ in $\Man$, and respectively $\PPP(M,\B)$. Similarly, the
algebras $\A(\M)$ and $\F(\M)$ and the projections between them
constitute diagrams $\AA(M,\B)$ and $\FF(M,\B)$ in $\CAlg$ and $\Poiss$,
respectively. The open cover of $\S_H(M,\B)$ by slow patches and the
agreement of symplectic structure on their overlaps,
Sect.~\ref{sec:glob-phsp}, ensures
the following limit and colimit identities hold
\begin{equation}
	\S_H(M,\B) \cong \varinjlim \SS(M,\B)
	\quad\text{and}\quad
	\PP_H(M,\B) \cong \varinjlim \PPP(M,\B) .
\end{equation}
Similarly, in the algebraic categories, the sheaf property of $C_{cst}^\oo$
ensures that the following colimit identities hold
\begin{equation}
	\A_H(M,\B) \cong \varprojlim \AA(M,\B)
	\quad\text{and}\quad
	\F_H(M,\B) \cong \varprojlim \FF(M,\B) .
\end{equation}
The same category $\SpBkgr(M,\B)$ contravariantly indexes the geometric
colimit and covariantly indexes the algebraic limit in the left column,
while the same category $\SpBkgr_c(M,\B)$ contravariantly indexes the
geometric colimit and covariantly indexes the algebraic limit in the
right column. Recall that objects $\M=(M,C^\oast,\B)$ of these indexing
categories specify a spacelike cone bundle $C^\oast\to M$. If a chronal
cone bundle were to be specified instead, the above limits and colimits
would have been swapped, since $\SpBndl(M)$ is the opposite category of
$\ChrBndl(M)$. However, it is more convenient to use spacelike cone
bundles, in light of the considerations below.

\begin{remark}
Colimits appear in the framework of LCFT already at the level of
spacetimes. Namely, consider a globally hyperbolic Lorentzian manifold
$(M,\{g\})$ in the category $\GlobHyp_c$. Let $\MM$ be the diagram whose
image in $\GlobHyp_c$ consists of all globally hyperbolic Lorentzian
submanifolds $(M_i,\{g\}_i)$ of $(M,\{g\})$, with inclusions as morphisms. The
colimit $\varinjlim \MM$ is isomorphic to $(M,\{g\})$ itself.  In fact one
recovers $(M,\{g\})$ as the colimit even if one restricts the objects of
$\MM$ to be simple in particular ways, like being topologically trivial,
having bounded radius, being the interior of a causal diamond, etc. Each
kind of simplicity translates to a correspondingly simple property of
the algebra of observables $\F(M_i,\{g\}_i)$ associated by a LCFT. Thus,
just as the non-trivial topological or geometric structure of $(M,\{g\})$ is
encoded in the diagram $\MM$ in $\GlobHyp_c$, the non-trivial algebraic
structure of $\F(M,\{g\})\cong \varinjlim \F\circ\MM$ is encoded in the
diagram $\F\circ\MM$ in $\Poiss_i$. (We use the colimit in both
categories because $\F$ is a covariant functor.) This approach underlies
the investigation of superselection sectors of LCQFT~\cite[Ch.IV]{haag},
\cite{br} and the quantization of vector field theories on manifolds of
non-trivial topology~\cite[Apx.A]{hollands-ym}, \cite{dl-maxwell}.
\end{remark}

It is sometimes convenient to simultaneously display both the limit and
colimit identities:
\begin{align}
	\S_H(M,\B) &= \varprojlim (\varinjlim \circ \SS \circ \MM) ,
	& \PP_H(M,\B) &= \varprojlim (\varinjlim \circ \PPP \circ \MM_c) , \\
	\A_H(M,\B) &= \varinjlim (\varprojlim \circ \AA \circ \MM) ,
	& \F_H(M,\B) &= \varinjlim (\varprojlim \circ \FF \circ \MM_c) ,
\end{align}
where we define $\MM$ ($\MM_c$) is the diagram in $\Bkgr$ ($\Bkgr_c$)
that consists of all sub-objects $(M',\B')\to (M,\B)$. Similar ideas
about constructing the global phase space and algebra of observables,
using an open cover of the spacetime and a cover of the space of
solutions also recently appeared in~\cite{bfr}. Though, in distinction, the
covers on $\S_H(M,\B)$ are considered there more general than ours. The
cover elements are not required to be open (only $G_\delta$) and could
be much more refined than our slow patches.

\section{Causality Axiom in Locally Covariant Classical Field Theory}
\label{sec:classcaus}
The axioms in the standard definition of LCFT, Def.~\ref{def:axlcft},
make reference to a fixed causal structure induced by the future
timelike cones of a background Lorentzian metric. It is clear that they
cannot be translated to direct conditions on the functor $\F_H$ (or
$\PP_H$), which can be constructed for theories without a fixed
background causal structure. On the other hand, these axioms can be
straightforwardly translated to conditions on the functor $\F$ (or
$\PP$), since objects in its domain category $\SpBkgr$ all have an
externally fixed causal structure given by a spacelike cone bundle.
Below, we prove theorems about the properties of the functor $\F$, which
lead naturally to generalizations of the LCFT axioms. The functor $\F_H
= \varprojlim \circ \FF$ inherits these properties through the
categorical limit construction.

\subsection{Isotony}\label{sec:gen-isot}
Given a morphism $\chi\colon \M \to \M'$ in
$\SpBkgr$, is the corresponding morphism $\chi^*\colon \S(\M)\to
\S(\M')$ surjective? The answer is not always positive; it
depends on the properties of the field theory in question. In terms of
solutions, $\chi^*$ is surjective precisely when every $C^\oast$-slow
solution on $M$ is the pullback along $\chi$ of a $C^{\prime\oast}$-slow
solution on $M'$ (equivalently, can be extended with respect to $\chi$
to a $C^{\prime\oast}$-slow solution on $M'$, where there is no presumption
of uniqueness of the extension). Broadly speaking, the extension of a
particular solution on $M$ fails if the would be extension develops
singularities in $M'$.

There are different kinds of singularities, including topological,
smooth, geometric and analytical. (T) For instance, depending on the
structure of the equations of motion, there may be a topological
obstruction in extending a solution from $M$ to $M'$, if $M$ is not
contractible~\cite{dl-maxwell}. (S) Also, it is possible that a solution
becomes unbounded or non-smooth near the boundary of $\chi(M)$ in $M'$
and hence cannot be extended continuously to $M'$. (A) The solution may
even be completely regular on $\chi(M)$ and its boundary, yet be forced
by the dynamics of the equations of motion to blow up when extended to
$M'$. The formation of shocks in fluid dynamics and of black hole
singularities in GR are prime examples of this phenomenon. (C) On the
other hand, it is possible for a solution to extend from $\chi(M)$ to
$M'$, yet the fixed cone bundle $C^{\prime*}$ is not fast enough to be
compatible with the extension or simply that the image $\chi(M)$ is not
chronally convex in $M'$ (Def.~\ref{def:chr-compat}).

It is of course desirable to formulate a set of necessary and sufficient
conditions for extensibility to hold. Unfortunately, that is in general
a very difficult problem. In particular, the study of blow ups of type
(A) constitute a major active branch of modern PDE theory.  Moreover,
the existence singularities of type (T) and (A) is determined by the
structure of the PDEs, which is fixed by the relevant physics. As such,
we may just have to learn to live with them. Inextensibility of type (S)
depends on whether we allow asymptotically irregular solutions in
$\S(\M)$.  This choice is related to the issues of the choices of the
structure of the domains $M$ and topologies on the space of solutions,
which were discussed in Sect.~\ref{sec:slow-sec} in relation to the
openness of slow patches. It is also possible that the failure of
surjectivity of $\chi^*$ due to blow up of type (S) is precisely masked
by the application of the $C_{cst}^\oo$ functor that constructs the
algebra of observables on the given slow patch of
$\S_H(M,\B)$~\cite{fr-bv,rejzner-thesis}. As
in this preceding discussion, we do not address it here directly.
Finally, blow up of type (C) is entirely within our control. Since the
problematic causal cone bundle $C^{\prime*}$ plays only an auxiliary
role in the ultimate construction of $\S_H(M',\B')$, it is no loss to
introduce a faster cone bundle, say $C^{\prime*}_1$ that might be fast
enough to be compatible with an extension of every solution in $\S(\M)$.

Unfortunately, as can be seen from the above discussion we cannot prove
a general theorem about the surjectivity of the morphism $\chi^*$.
Therefore, we introduce this property as an additional hypothesis.
\begin{definition}\label{def:ext}
A morphism $\chi$ in $\SpBkgr$, is said to be \emph{extensible} if this
morphism factors through
\begin{equation}
\vcenter{\xymatrix{
	& \M'_1 \ar[d]^{\dsse} \\
	\M \ar[ur]^{\chi_1} \ar[r]^\chi & \M' ~ ,
}}
\end{equation}
such that $\M'=(M',C^{\prime\oast},\B')$, $\M'_1 =
(M',C_1^{\prime\oast},\B')$, and the morphism $\chi_1^* = \S(\chi_1)$ is
surjective.
\end{definition}
Provided this hypothesis holds, there are no other obstacles in proving
the following generalization of the Isotony property for the functor $\F$.

\begin{theorem}[Generalized Isotony]\label{thm:gen-isot}
Consider a morphism $\chi$ in $\SpBkgr$. If $\chi$ is extensible, the
morphism $\A(\chi_1)$ is an injective homomorphism. Moreover, if
$\S(\M)$ is not empty, $\chi$ and $\chi_1$ are also morphisms in
$\SpBkgr_c$, the morphisms $\S(\chi)$ and $\S(\chi_1)$ lift to
$\PP(\chi)$ and $\PP(\chi_1)$ in $\Symp$, and $\A(\chi)$ and
$\A(\chi_1)$ lift to morphisms in $\Poiss$.
\end{theorem}
\begin{proof}
Since $\chi$ is extensible, the induced morphism $\S(\M'_1)\to \S(\M)$
is surjective. By Hyp.~\ref{hyp:Coo-exact}, $C_{cst}^\oo$ takes this
surjection to the injective homomorphism $\A(\M) \to \A(\M'_1)$. If the
corresponding solution spaces are not empty, the covariance lemma,
Lem.~\ref{lem:pois-covar}, implies that the appropriate morphisms lift
to $\Symp$ and $\Poiss$.  \qed.
\end{proof}

\subsection{Time Slice}\label{sec:gen-ts}
Given a morphism $\chi\colon \M \to \M'$ in
$\SpBkgr$, is the corresponding morphism $\S(\chi)\colon \S(\M)
\to \S(\M')$ injective? If the answer were positive, it
would mean that every $C^{\prime\oast}$-slow solution on $M'$ is uniquely
determined by its restriction to $\chi(M)\sse M'$. Clearly, there is no
hope that this is the case unless $\chi(M)$ contains a
$C^{\prime\oast}$-Cauchy surface in $M'$.
\begin{definition}\label{def:cauchy-surj}
A morphism $\chi$ in $\SpBkgr$ is called \emph{Cauchy surjective} or a
\emph{Cauchy surjection} if for the corresponding morphism in $\SpBndl^n_H$,
\begin{equation}
\vcenter{\xymatrix{
	C^\oast \ar[r] \ar[d] & C^{\prime\oast} \ar[d] \\
	M \ar[r]^\chi & M' ~ ,
}}
\end{equation}
there exists a $C^\oast$-Cauchy surface $\Sigma\sso M$ whose image $\Sigma'
= \chi(\Sigma) \sso M'$ is a $C^{\prime\oast}$-Cauchy surface. (Recall that
all cone bundles in $\SpBndl^n_H$ are globally hyperbolic.)
\end{definition}

Unfortunately, injectivity of $\S(\chi)$ does not guarantee surjectivity
of $\A(\chi)$, which would imply that every smooth function on the image
of $\S(\M')$ extends to a smooth function on $\S(\M)$. However, this
property is known to fail unless $\S(\chi)$ is a closed embedding%
	\footnote{A simple illustration is the inclusion $(0,1)\sso \R$ and
	any function that is unbounded near an end of the open interval.
	}. %
If $\S(\chi)$ is also surjective ($\chi$ is extensible), this is
automatic. However, if $\S(\chi)$ fails to be injective, then the
complement of its image (the solutions that blow up when extended from
$\chi(M)$ to $M'$) is rarely an open set.

Fortunately, a slightly more complicated diagram in $\SpBkgr$ does
guarantee a closed embedding.
\begin{definition}\label{def:cauchy-push}
Consider morphisms $\chi_1$ and $\chi_2$ in $\SpBkgr$ that fit into a
pull\-back-push\-out diagram in $\SpBkgr$,
\begin{equation}\label{eq:cauchy-push}
\vcenter{\xymatrix{
	\M_3 \ar@{}[dr]|(.3){\lrcorner}|(.7){\ulcorner} \ar[d] \ar[r]
		& \M_2 \ar[d]^{\chi_2} \\
	\M_1 \ar[r]_{\chi_1} & \M' ~ .
}}
\end{equation}
That is, the images $\chi_i(M_i)$ form an open cover of $M'$ and $M_3$
is isomorphic to their intersection. In terms of the spacelike cones, we
have $C^{\prime\oast} = C_1^\oast + C_2^\oast$ and $C_3^\oast =
C_1^\oast \cap C_2^\oast$ (cf. Sect.~\ref{sec:geom-cones}), where for
brevity we have omitted the pullbacks or pushforwards with respect to
the appropriate morphisms, while $\B_3 = \B_1 \cup \B_2$ and
$\B'=\B_1\cap \B_2$ on $M_3$, with the same shorthand. We call the
diagram~\eqref{eq:cauchy-push} a \emph{Cauchy pushout} for $\M'$ if all
morphisms are Cauchy surjections.
\end{definition}
Given a Cauchy pushout for $\M'$, the maps $\S(\M')\to \S(\M_i)$ are
injective by Cauchy surjectivity. However, they are not guaranteed to be
closed embeddings. On the other hand, consider the following diagram
of spaces of field configurations in $\Man$, where $V_i = V(M_i)$ and
$V' = V(M')$ are total field bundles naturally associated to their base
manifolds,
\begin{equation}
\vcenter{\xymatrix{
	\Secs(V_3) \ar@{}[dr]|(.77){\ulcorner} \ar@{<-}[d] \ar@{<-}[r]
		& \Secs(V_2) \ar@{<-}[d]_{\Secs(\chi_2)} \\
	\Secs(V_1) \ar@{<-}[r]^{\Secs(\chi_1)} & \Secs(V') \ar@{-->}[dr] \\
	& & \Secs(V_1)\times \Secs(V_2) \ar[uul] \ar[ull] ~ ,
}}
\end{equation}
which is most definitely a pullback diagram, with the canonical dotted
line morphism a closed embedding. The set theoretic pull back property
is assured because any pair of smooth sections on $M_1$ and $M_2$ that
agree on $M_3$ must glue together%
	\footnote{This is sometimes known as the \emph{sheaf descent property}
	of smooth sections.} %
to a smooth section on $M_3$. Now, consider a point $(\psi_1,\psi_2)$ in $\Secs(V_1)\times
\Secs(V_2)$ that is not in the image of $\Secs(V')$, that is, it
determines two sections, $\psi_1$ on $M_1$ and $\psi_2$ on $M_2$, that
do not agree on $M_3$. It is then simple to construct neighborhoods (in
either the compact open or Whitney fine topologies) of
the graphs of these sections in $V_3$ that do not intersect over at
least one point in $M_3$, which define neighborhoods of $\psi_1$ and
$\psi_2$ and hence a neighborhood of $(\psi_1,\psi_2)$ that does not
intersect the image of $\Secs(V')$. Therefore, the complement of the
image of $\Secs(V')$ is open.  In other words, the image of $\Secs(V')$
in the product space $\Secs(V_1)\times \Secs(V_2)$ is closed.

The above discussion will be sufficient to establish a closed embedding
of the solution space $\S(\M')$ in natural ambient space determined by
the Cauchy pushout. From that, we can construct a surjective
homomorphism onto the algebra $\A(\M')$ from a natural algebra also
determined by the Cauchy pushout. Unfortunately, that is not quite
satisfactory, as it is irresistible to try to formulate the generalized
Time Slice property as a stronger result, namely that the Cauchy pushout
(in the spacetime category) is taken to a pullback (in the phase space
category) and again a pushout (in the algebraic category) by the
appropriate successive contravariant functors $\P$ and $\F$ (or $\S$ and
$\A$). However, to achieve that using an appeal to
Hyp.~\ref{hyp:Coo-exact} requires an additional transversality
hypothesis on the interaction of the solution spaces $\S(\M_i)$ and
$\S(\M')$. It is not clear whether this notion has already appeared
in the standard PDE literature.
\begin{definition}\label{def:trans-descent}
A Cauchy pushout as in Eq.~\eqref{eq:cauchy-push} is said to satisfy
\emph{transverse descent} if the smooth maps $\alpha_i\colon
\S(\M_1)\times \S(\M_2) \to \S(\M_i) \to \S(\M_3)$, $i=1,2$, are
transverse to each other.
\end{definition}
\begin{remark}
Note that the images of the maps $\chi^*_i\colon \S(\M_i) \to \S(\M_3)$
intersect precisely on a subspace of $\S(\M_3)$ that can be identified
with $\S(\M')$. This property follows from the uniqueness of the Cauchy
problem and the sheaf descent property inherited from the sheaf of
smooth functions. As in the remark following Hyp.~\ref{hyp:Coo-exact},
we do not precisely specify the desired notion of transversality.
However, we conjecture that such a notion can be formulated so that the
maps $\alpha_i$ agree precisely on the subset of $\S(\M_1)\times
\S(\M_2)$ identified with $\S(\M')$ and, moreover, that a corresponding
implicit function theorem establishes that $\S(\M')$ is in fact a
submanifold. The reason the standard notions of infinite dimensional
transversality~\cite{cbdm,lang,hamilton} may not be
applicable here is that (precisely due to the possible blow up of
solutions) the images of the tangent maps $T\alpha_i$ may not be closed
subspaces of the tangent space $T\S(\M_3)$. This phenomenon appears
already for linear PDE systems, such as, for example, the Cauchy-Riemann
equations. Though, admittedly, that is not an example of a hyperbolic
PDE system.
\end{remark}

\begin{lemma}\label{lem:s-ts}
Consider a Cauchy pushout diagram satisfying transverse descent,
Defs.~\ref{def:cauchy-push} and~\ref{def:trans-descent}. The following
diagram contains a pullback square for $\S(\M')$ and the canonical
dotted line morphism is a closed embedding of manifolds:
\begin{equation}\label{eq:s-ts}
\vcenter{\xymatrix{
	\S(\M_3) \ar@{}[dr]|(.77){\ulcorner} \ar@{<-}[d] \ar@{<-}[r]
		& \S(\M_2) \ar@{<-}[d]_{\S(\chi_2)} \\
	\S(\M_1) \ar@{<-}[r]^{\S(\chi_1)} & \S(\M') \ar@{-->}[dr] \\
	& & \S(\M_1)\times \S(\M_2) \ar[uul] \ar[ull] ~ .
}}
\end{equation}
If none of $\S(\M_i)$ or $\S(M')$ are empty, then the above pullback
square lifts to the following pullback square in $\Symp$:
\begin{equation}\label{eq:p-ts}
\vcenter{\xymatrix{
	\PP(\M_3) \ar@{}[dr]|(.7){\ulcorner} \ar@{<-}[d] \ar@{<-}[r]
		& \PP(\M_2) \ar@{<-}[d]^{\PP(\chi_2)} \\
	\PP(\M_1) \ar@{<-}[r]_{\PP(\chi_1)} & \PP(\M') ~ .
}}
\end{equation}
\end{lemma}
\begin{proof}
The set theoretic pullback property holds again from the gluing property
of solution sections, with which the background fields do not
interfere. Now, recall that all section spaces are topologized as
subspaces of $\Secs(-)$ with appropriate vector bundle argument. In
particular this means that the subset $\S(\M_1)\times\S(\M_2)\cap
\Secs(V') \sse \Secs(V_1)\times \Secs(V_2)$ is closed in
$\S(\M_1)\times\S(\M_2)$. On the other hand, the gluing property ensures
that $\S(M')\cong \S(\M_1)\times \S(\M_2) \cap \Secs(V')$. (Again, the
obvious pushforwards and pullbacks have been omitted for brevity.) This
shows that the canonical dotted line morphism is a closed embedding of
topological spaces. This much we can establish without appeal to the
transverse descent property, on the other hand we must appeal to
it to establish that it is a manifold embedding, and hence a closed
embedding. The covariance lemma, Lem.~\ref{lem:pois-covar}, guarantees
that the morphisms lift to $\Symp$ whenever the corresponding solution
spaces are not empty. \qed
\end{proof}

With the above discussion in mind, we can formulate the following
generalization of the Time Slice property for the functor $\F$.
\begin{theorem}[Generalized Time Slice]\label{thm:gen-ts}
Consider a Cauchy pushout diagram that satisfies transverse descent,
Defs.~\ref{def:cauchy-push} and~\ref{def:trans-descent}. Then the
following is a pushout diagram in $\CAlg$:
\begin{equation}\label{eq:a-ts}
\vcenter{\xymatrix{
	\A(\M_3) \ar@{}[dr]|(.7){\ulcorner} \ar[d] \ar[r]
		& \A(\M_2) \ar[d]^{\A(\chi_2)} \\
	\A(\M_1) \ar[r]_{\A(\chi_1)} & \A(\M') ~ .
}}
\end{equation}
Moreover, if none of the $\S(\M_i)$ or $\S(\M')$ are empty, the above
diagram also lifts to a pushout diagram in $\Poiss$:
\begin{equation}\label{eq:f-ts}
\vcenter{\xymatrix{
	\F(\M_3) \ar@{}[dr]|(.7){\ulcorner} \ar[d] \ar[r]
		& \F(\M_2) \ar[d]^{\F(\chi_2)} \\
	\F(\M_1) \ar[r]_{\F(\chi_1)} & \F(\M') ~ .
}}
\end{equation}
\end{theorem}
\begin{proof}
According to Lem.~\ref{lem:s-ts}, we can realize $\S(\M')$ as a closed
submanifold of $\S(\M_1)\times \S(\M_2)$, with the aid of the transverse
pullback diagram~\eqref{eq:s-ts}. On the other hand,
Hyp.~\ref{hyp:Coo-exact} immediately implies that this transverse manifold
pullback is taken by $C_{cst}^\oo$ to the algebraic pushout
diagram~\eqref{eq:a-ts}. Finally, whenever the diagram~\eqref{eq:s-ts}
in $\Man$ lifts to diagram~\eqref{eq:p-ts} in $\Symp$, then the
diagram~\eqref{eq:a-ts} in $\CAlg$ automatically lifts to
diagram~\eqref{eq:f-ts} in $\Poiss$. \qed
\end{proof}

\subsection{Causality}\label{sec:gen-caus}
Given two morphisms $\chi_i\colon \M_1 \to \M$, $i=1,2$, in $\SpBkgr$,
do the corresponding morphisms \emph{surjectively} factor through the
product $\S(\M_1)\times \S(\M_2)$? Do the corresponding
morphisms factor \emph{injectively} through the tensor product
$\A(\M_1)\otimes \A(\M_2)$? Finally, when do the corresponding morphisms
factor through the \emph{independent subsystems} tensor product
$\F(\M_1)\otimes \F(\M_2)$.

There are several obstacles to positive answers to the above questions.
First, let us consider surjectivity for the spaces of solutions and,
correspondingly, injectivity for the algebras.
If we allow disconnected manifolds as base manifolds in $\Bkgr$ and
$\SpBkgr$, while making sure that morphisms in $\SpBkgr$ treat
disconnected components as spacelike separated, we can introduce a
tensor product, referred to as the \emph{disjoint union}:
\begin{equation}
	\M_1 \sqcup \M_2
		= (M_1\sqcup M_2, C_1^\oast \sqcup C_2^\oast, \B_1\times \B_2).
\end{equation}
Note that this is not a coproduct in the categorical sense, since dotted
line morphisms in the diagram
\begin{equation}
\xymatrix{
	\M_1 \ar[rd] \ar[r]^{\chi_1} & \M & \ar[l]_{\chi_2} \ar[ld] \M_2 \\
	& \M_1 \sqcup \M_2 \ar@{-->}[u]
}
\end{equation}
exists only if the images $\chi_1(M_1)$ and $\chi_2(M_2)$ are
$C^*$-spacelike separated in $M$. However, when this morphism exists, it
is canonical and is denoted by $\chi_1\sqcup \chi_2$. On the other hand,
$\Man$ does have a categorical product, which means that the following
diagram always exists:
\begin{equation}
\xymatrix{
	\M_1 \ar@{<-}[rd] \ar@{<-}[r]^{\chi_1} & \M_1\sqcup \M_2
		& \ar@{<-}[l]_{\chi_2} \ar@{<-}[ld] \M_2 \\
	& \S(\M_1) \times \S(\M_2) \ar@{<--}[u] ~ .
}
\end{equation}
Since there is no obstacle to specifying a solution independently in
each connected component, the canonical dotted line morphism is in fact
an isomorphism, $\S(\M_1\sqcup \M_2) \cong \S(\M_1) \times \S(\M_2)$. In
the $\Symp$ category, we can also define the \emph{independent subsystems
product} given by
\begin{equation}
	(N_1,\Omega_1) \times (N_2,\Omega_2) = (N_1\times N_2, \Omega_1 \oplus
	\Omega_2) .
\end{equation}
When $N_i = \S(\M_i)$, the corresponding Poisson bivector $\Pi$ on the
product is characterized by $\Pi(x,y) = \Pi_i(x,y)$ when $x,y\in M_i$
and $\Pi(x,y) = 0$ when $x$ and $y$ belong to different connected
components. Note that this is not a categorical product in $\Symp$, for
reasons very similar to why $\sqcup$ is not a coproduct in $\SpBkgr$.
However, given the above information, we do have the identity
$\PP(\M_1\sqcup \M_2) = \PP(\M_1) \times \PP(\M_2)$.

Now, given the $\SpBkgr$ diagram
\begin{equation}
\xymatrix{
	\M_1 \ar[rd] \ar[r]^{\chi_1} & \M & \ar[l]_{\chi_2} \ar[ld] \M_2 \\
	& \M_1 \sqcup \M_2 \ar@{-->}[u]
} ~ ,
\end{equation}
we always obtain the diagram in $\Man$ below
\begin{equation}
\xymatrix{
	\S(\M_1) \ar@{<-}[rd] \ar@{<-}[rd] \ar@{<-}[r]^{\S(\chi_1)} & \S(\M)
		& \ar@{<-}[l]_{\S(\chi_2)} \ar@{<-}[ld] \ar@{<-}[ld] \S(\M_2)\\
	& \S(\M_1) \times \S(\M_2) \ar@{<--}[u]
} ~ ,
\end{equation}
where dotted lines denote the canonical morphisms. Note that the
canonical morphism in $\Man$ always exists, even if $\chi_1 \sqcup
\chi_2$ does not exist in $\SpBkgr$. The surjectivity of
the canonical dotted line morphism in the last diagram is equivalent the
following: given a pair of solution sections, $\psi_1$ on $\M_1$ and
$\psi_2$ on $\M_2$, there exists a solution section $\psi$ on $\M$ that
restricts to $\psi_1$ on $\M_1$ and $\psi_2$ on $\M_2$. Already in the
earlier discussion of the Isotony property, we have noticed this
surjectivity property fails when at least one of the $\chi_i$ morphisms
fails to be extensible. However, canonical dotted line morphism may fail
to be surjective even if both $\chi_1$ and $\chi_2$ are extensible.
Namely, there may exist a pair $(\psi_1,\psi_2)\in \S(\M_1)\times
\S(M_2)$ such that both $\psi_1$ and $\psi_2$ may be extended to $\M$
individually, but not jointly.

The most common reason for that to happen is that the images
$\chi_1(M_1)$ and $\chi_2(M_2)$ are $C^\oast$-causally related in $M$,
so that, for example, the value of $\psi_1$ in $M_1$ is not consistent
with the causal influence of $\psi_2$ on $M_2$. Thus,
to even have a hope of canonical factorization, we should require the
images $\chi_1\sqcup \chi_2$ does exist in $\SpBkgr$, that is the images
$\chi_1(M_1)$ and $\chi_2(M_2)$ are $C^*$-spacelike separated in $M$.
However, even with spacelike separation, surjectivity can still fail.
Intuitively, this happens when the solution data specified by $\psi_1$
and $\psi_2$ always produces a singularity while scattering in a
joint extension to $\M$, so that the canonical morphism $\chi_1\sqcup
\chi_2$ is itself not extensible.
\begin{definition}\label{def:reg-scat}
We say that two extensible morphisms $\chi_1$ and $\chi_2$ in $\SpBkgr$
have \emph{regular scattering} if they fit in the commutative diagram
\begin{equation}
\xymatrix{
	\M_1 \ar[rd] \ar[r]^{\chi_1} & \M & \ar[l]_{\chi_2} \ar[ld] \M_2 \\
	& \M_1 \sqcup \M_2 \ar@{-->}[u]
} ~ ,
\end{equation}
such that the canonical dotted line morphism $\chi_1\sqcup \chi_2$
exists and is extensible.
\end{definition}

Once the surjectivity in $\Man$ and $\Symp$ is taken care of, the
injectivity in the algebraic categories follow in the manner indicated
in the earlier discussion of the Isotony property. With the above
discussion in mind we can formulate the following generalization of the
Causality property for the functor $\F$.
\begin{theorem}[Generalized Causality]\label{thm:gen-caus}
Consider two morphisms $\chi_i\colon\M_i \to \M$, $i=1,2$, in $\SpBkgr$
such that $\chi_1\sqcup \chi_2$ exists with regular scattering. Then
there exists an object $\M' = (M,C^{\prime\oast},\B)$ in $\SpBkgr$ such that
$\chi_i$ and $\chi_1\sqcup \chi_2$ factor according to the diagram
\begin{equation}\label{eq:m-caus}
\vcenter{\xymatrix{
	& \ar@{<-}[dl]_{\chi_1} \M \ar@{<-}[dr]^{\chi_2} \ar@{<-}[d]^{\usse} \\
	\M_1 \ar[rd] \ar[r]^{\chi'_1} & \M' & \ar[l]_{\chi'_2} \ar[ld] \M_2 ~ ,\\
	& \M_1 \sqcup \M_2 \ar@{-->}[u]
}}
\end{equation}
the canonical dotted line morphism in the corresponding diagram in
$\Man$ is surjective,
\begin{equation}\label{eq:s-caus}
\vcenter{\xymatrix{
	& \ar[dl]_{\S(\chi_1)} \S(\M) \ar[dr]^{\S(\chi_2)} \ar[d]^{\usse} \\
	\S(\M_1) \ar@{<-}[rd] \ar@{<-}[r]^{\S(\chi'_1)} & \S(\M')
		& \ar@{<-}[l]_{\S(\chi'_2)} \ar@{<-}[ld] \S(\M_2) ~ ,\\
	& \S(\M_1) \times \S(\M_2) \ar@{<--}[u]
}}
\end{equation}
and the canonical dotted line morphism in the corresponding diagram in
$\CAlg$ is injective,
\begin{equation}\label{eq:a-caus}
\vcenter{\xymatrix{
	& \ar@{<-}[dl]_{\A(\chi_1)} \A(\M) \ar@{<-}[dr]^{\A(\chi_2)} \ar@{<-}[d]^{\usse} \\
	\A(\M_1) \ar[rd] \ar[r]^{\A(\chi'_1)} & \A(\M')
		& \ar[l]_{\A(\chi'_2)} \ar[ld] \A(\M_2) ~ .\\
	& \A(\M_1) \otimes \A(\M_2) \ar@{-->}[u]
}}
\end{equation}
Moreover, if none of the objects in diagram~\eqref{eq:s-caus} are empty,
the diagram~\eqref{eq:m-caus} lifts to $\SpBkgr_c$, the
diagram~\eqref{eq:s-caus} lifts to $\Symp$ via $\PP$, and the
diagram~\eqref{eq:a-caus} lifts to $\Poiss$ via $\F$.
\end{theorem}
\begin{proof}
Regular scattering hypothesis implies that both $\chi_i$, $i=1,2$,
morphisms are extensible, using the same auxiliary object $\M'$, where
the images of $\chi'_i$ are also spacelike separated. The existence of
the canonical dotted line morphism then follows from the definition of
$\sqcup$. Applying the $\S$ functor to diagram~\eqref{eq:m-caus}
produces diagram~\eqref{eq:s-caus}, where the surjectivity of the
dotted line morphism, again, follows from the regular scattering
hypothesis. Finally, applying the $C_{cst}^\oo$ functor to
diagram~\eqref{eq:s-caus} gives diagram~\eqref{eq:a-caus}, where the
canonical dotted line morphism is injective because, by
Hyp.~\ref{hyp:Coo-exact}, $C_{cst}^\oo$ takes surjections to injective
homomorphisms.

If the diagram~\eqref{eq:m-caus} can be lifted to $\SpBkgr_c$, then it
follows from definitions and the covariance of the Poisson structures,
Lem.~\ref{lem:pois-covar}, that the diagrams~\eqref{eq:s-caus}
and~\eqref{eq:a-caus} lift to $\Symp$ and $\Poiss$, respectively. \qed
\end{proof}

\subsection{Connection with standard LCFT}\label{sec:semilin-lcft}
The above results can be applied to a semilinear, well-posed PDE system
to construct a LCFT functor satisfying the standard axioms given in
Def.~\ref{def:axlcft}.  Consider a natural, variational, semilinear PDE
system, whose background fields include a globally hyperbolic, time
oriented Lorentzian metric $g$ and such that its cone bundle of future
directed timelike vectors $\Gamma_M\to V(M)$ coincides with the bundle
of future directed timelike cones of $g$. Let $\S\colon \SpBkgr\to \Man$
and $\PP\colon\SpBkgr_c\to \Symp$ be its solutions and phase space
functors. We will allow other background fields but only such that they
influence neither the solution space nor the symplectic structure on it.
The reason is that to achieve a natural symmetric hyperbolic form extra
background fields may need to be introduced, as has already been
discussed and can be seen explicitly from the examples in~\cite{geroch-pde}.
However, since that is their only purpose, they do not influence the
dynamics. We call such background fields \emph{auxiliary}. For the
purposes of comparing with a standard LCFT, we can safely ignore them.
Recall that we introduced the category $\SpBkgr^*$, with objects that
are endowed with background field collections consisting of a single
section. Using this category to construct a standard LCFT is slightly
problematic, because of the presence of auxiliary background fields. An
easy solution is to introduce the forgetful functors $\SpBkgr^*\to
\SpBkgr'$ and $\Bkgr^*\to \Bkgr'$ that throw away the auxiliary
background fields, where $\SpBkgr'$ and $\Bkgr'$ now has only the metric
as background field.  The subscript ${}_c$ retains the same meaning as
before. By stipulation, the functors $\S$ and $\PP$ project to functors
$\S\colon\SpBkgr'\to \Man$ and $\PP\colon\SpBkgr'_c\to \Symp$, and
similarly with $\S_H$ and $\PP_H$, up to isomorphisms in the target
categories. Let a natural PDE system that satisfies all the conditions
in the preceding discussion be called a \emph{semilinear geometric wave
equation}.

We conclude this section by showing that the tools assembled so far in
this paper allow us to construct a standard classical LCFT from a
well-posed semilinear geometric wave equation. Unfortunate, a
sufficiently detailed notion of well-posedness that would allow us to
directly construct a standard LCFT functor based on the results of the
preceding section is somewhat difficult to formulate precisely. One
obstacle is the fact the properties of the $C_{cst}^\oo$ functor that
constructs the algebra of observables as the algebra of smooth functions
on an infinite dimensional phase space have so far been mostly
hypothesized (cf.~Hyp.~\ref{hyp:Coo} and Hyp.~\ref{hyp:Coo-exact})
rather than constructively proven. Another is related to the boundary
conditions that one is expected to impose on solutions in spacetimes
with non-compact Cauchy surfaces (cf.~Sect.~\ref{sec:top-choice}) and
their interplay with the standard notions of well-posedness in the PDE
literature. As such, we make no special attempt to be precise and merely
formulate the following hypothesis, whose application is made clear in
the proof of the theorem below.
\begin{hypothesis}[Well-posedness]\label{hyp:glob-wellpos}
We consider our classical field theory defined by a natural,
variational, semilinear PDE system to be \emph{globally well-posed} in
the following sense. All regular initial Cauchy data (where regularity
includes any necessary boundary conditions) extend to the entire
spacetime manifold (no finite time blow up). Moreover, in the presence
of constraints and gauge symmetry, the PDE system can be augmented with
purely hyperbolic gauge fixing (Sects.~\ref{sec:constr-gf}
and~\ref{sec:Jgreen}) such that global parametrizability
(Sect.~\ref{sec:constr}) and global recognizability
(Sect.~\ref{sec:gauge}) are satisfied. Finally, the $C_{cst}^\oo$ functor
constructs the algebras of observables in a way that is not sensitive to
the behavior of solutions near the open ends of the spacetime manifold.
\end{hypothesis}

To substantiate the above hypothesis, let us note that there exist
global well-posedness results for what we consider fundamental bosonic
fields (namely Yang-Mills~\cite{em-ym} and non-linear scalar~\cite{bfr}
fields) other than gravity (GR definitely exhibits finite time blow up,
as exemplified by the singularities of black hole and cosmological
solutions). On the other hand, the algebras of observables consisting of
microlocal functionals of compact spacetime
support~\cite{bg-lcqft,hollands-ym,fr-bv,rejzner-thesis} provide a
candidate for the functor $C_{cst}^\oo$, which seems to have the desired
properties, at least when applied to the solution spaces of linear
theories. Also, the sufficient conditions imposed on the constraints and
gauge transformations are actually very similar to those considered in
previous treatments, as exemplified in~\cite{geroch-pde}
and~\cite{hs-gauge}.

\begin{theorem}\label{thm:semilin-lcft}
Consider a semilinear geometric wave equation that is well-posed
(Hyp.~\ref{hyp:glob-wellpos}). Then we have an equivalence of
categories $\Bkgr'_c\cong \GlobHyp_c$ and the corresponding algebra of
observables functor $\F_H\colon \Bkgr'_c \to \Poiss$ satisfies the
standard axioms of classical LCFT given in Def.~\ref{def:axlcft}.
\end{theorem}
\begin{proof}
The equivalence of categories $\Bkgr'_c \cong \GlobHyp_c$ is
straightforward to establish. The objects are identical and consist of
pairs $(M,g)$ (or equivalently $(M,\{g\})$), where $g$ is a globally
hyperbolic Lorentzian metric on the manifold $M$. There exists a
morphism $(M,g)\to (M',g')$ in $\Bkgr'_c$ iff there is a morphism $\M\to
M'$ in $\SpBkgr$, between objects that project to $\base(\M) = (M,g)$
and $\base(\M') = (M',g')$ and neither solution patch $\S(\M)$ or
$\S(\M')$ is empty. For semilinear equations, a solution $\phi\in
\S(\M)$ iff the corresponding spacelike cone bundle $C^\oast\to M$ is
strictly faster than the bundle of future oriented Lorentzian cones.
Now, there exists a morphism $\M\to \M'$ only if the underlying open
embedding $\chi\colon M\to M'$ satisfies $\chi^*g' = g$ (is an isometry)
and is chronally compatible with respect to $C^\oast$ and
$C^{\prime\oast}$, which a fortiori implies that $\chi$ is also causally
compatible with respect to $g$ and $g'$. This concludes the proof of the
equivalence of the two categories.

Next, we establish a crucial property that helps us prove the remaining
conclusions. For any object $\M$ with $\base(\M) = (M,g)$ and spacelike
cone bundle $C^\oast\to M$ that is faster than the bundle of future
oriented Lorentzian cones, we have the identity $\S_H(M,g) \cong
\S(\M)$. That is, all solutions are already $C^\oast$-slow. This tells
us that the image of the functor $\SS(M,g)$ in $\Man$, consisting of all
objects and morphism $\S(\M)\to \S(\M')$ with $\base(\M) = \base(\M') =
(M,g)$, where each $\S(\M)$ is either empty (if the cone bundle
$C^\oast$ is slower than the Lorentzian cones) or $\S(\M)\cong
\S_H(M,g)$ (if the cone bundle $C^\oast$ is faster than the Lorentzian
one). The same goes for the respective objects in other categories,
$\P(\M)$, $\A(\M)$, and $\F(\M)$, except that the $\P(\M)$ are never
empty. In other words, the limit construction $\F_H(M,g) \cong \lim
\FF(M,g)$ becomes somewhat superfluous, since the image of the diagram
$\FF(M,g)$ already contains the object that is equal to the limit.

With the above simplification, which means that we can simply replace
the Poisson algebra $\F_H(M,g)$ with the algebra $\F(\M)$ for some
appropriate $\M$, the diagrams defining the generalized Isotony
(Thm.~\ref{thm:gen-isot}), Time Slice (Thm.~\ref{thm:gen-ts}) and
Causality (Thm.~\ref{thm:gen-caus}) properties collapse to the usual
notions thereof as stated in the standard definition of LCFT,
Def.~\ref{def:axlcft}. It remains only to show that the key hypotheses
needed for the above theorems, namely extensibility of morphisms in
$\SpBkgr_c$, transverse descent for Cauchy pushouts in $\SpBkgr_c$, and
regular scattering for morphisms in $\SpBkgr_c$, all follow from global
well-posedness, Hyp.~\ref{hyp:glob-wellpos}.

Given a morphism $\chi\colon \M\to \M'$, extensibility
(Def.~\ref{def:ext}) implies Isotony. Extensibility almost follows from
well-posedness. Namely, if there exist Cauchy surfaces $\Sigma\sso M$
and $\Sigma'\sso M'$ such that $\Sigma'$ extends $\Sigma$ and we have
initial data $\varphi$ on $\Sigma$ such that its pushforward
$\chi_*\varphi$ smoothly extends to some initial data $\varphi'$ on
$\Sigma'$, then the corresponding solution extends from $\chi(M)$ to
$M'$. However, even if Conj.~\ref{cnj:cauchy-ext} holds, the best we can
expect is the existence of a Cauchy surface $\Sigma'\sso M'$ that
extends a compact subset $K\sse \Sigma$. Also, even if we can extend all
of $\Sigma$ to $\Sigma'$, the initial data could behave near an open end
of $\Sigma$ in a way that is not smoothly extendible to $\Sigma'$. So to
recover Isotony at the algebraic level, we must appeal to the part of
Hyp.~\ref{hyp:glob-wellpos} according to which the algebra of
observables constructed by applying $C_{cst}^\oo$ is not sensitive to the
behavior of solutions in the neighborhood of the open ends of $\M$ and,
in particular, near the open ends of $\Sigma$. In that case, using
Conj.~\ref{cnj:cauchy-ext} and global well-posedness as above, every
smooth solution $\phi$ can be extended from an arbitrarily large compact
subset $K\sse \Sigma$ to all of $M$. The details of this argument would
have to await a more precise formulation of \ref{hyp:glob-wellpos},
and~\ref{hyp:opencover} and of the choices discussed in
Sect.~\ref{sec:top-choice}.

Given a Cauchy surjective morphism $\chi\colon \M\to \M'$
(Def.~\ref{def:cauchy-surj}), it could always be completed to a Cauchy
pushout for $\M'$ (Def.~\ref{def:cauchy-push}). Given global
well-posedness, no solution blows up in finite time. This means that all
the Cauchy surjections, including induce isomorphisms in the
diagrams~\eqref{eq:p-ts} and~\eqref{eq:f-ts}. In other words, transverse
descent holds trivially. The generalized Time Slice property then
reduces to its standard version from Def.~\ref{def:axlcft}.

Finally, since global well-posedness (in particular the absence of any
finite time blow up) also implies regular scattering,
Def.~\ref{def:reg-scat}, the generalized Causality property holds,
Thm.~\ref{thm:gen-caus}, and reduces to the corresponding standard
notion from Def.~\ref{def:axlcft}. \qed
\end{proof}

The theorems presented in this section may be seen as the translation of
the well-posedness properties of a natural variational PDE system into
the algebraic setting given by the corresponding Poisson algebras of
observables. One may even promote some of these properties to an axiom
system generalizing the existing LCFT axioms, which are only applicable
to semilinear wave systems, to the more general class of quasilinear
systems, which includes GR. We do not do so immediately because several
aspects of the current formalism, to be discussed in
Sect.~\ref{sec:discuss}, leave room for substantial improvement, which
could facilitate an optimal axiomatization.

\section{Causality Axiom in Locally Covariant Quantum Field Theory}
\label{sec:quantcaus}
This section contains some remarks on extending the results presented so
far to quantum field theory. The remarks essential consist of the
conjecture that the categorical colimit construction of the algebra of
observables will persist through perturbative deformation quantization
of the classical field theory. Since rigorous, non-perturbative
construction of quantum field theory is in general not yet possible,
that is the best we can hope for at the moment.

Recall that a classical mechanical system essentially consists of a real
Poisson algebra of observables $\F = (\A,\{\})$, where $\A = C_{cst}^\oo(\P)$
on a symplectic manifold $\P$ (the phase space), whose states are normed
positive linear functionals on $\A$, which can be expressed as
probability measures on $\P$. On the other hand, a quantum mechanical
system essentially consists of a complex, associate, non-commutative
$*$-algebra $\hat\F = (\hat\A,*,{}^*)$, whose states are normed positive
linear functionals on $\hat\A$. Strong physical and mathematical
arguments, which were already eloquently expressed in the original
papers~\cite{bffls1,bffls2,bffls3}, indicate that \emph{deformation quantization} is
the right way to define an $\hbar$-parametrized family of quantum
systems that quantize a given classical one. Briefly, a deformation
quantization of a classical mechanical system $(\A,\{\})$ is a family of
quantum mechanical systems $(\hat\A_\hbar, *_\hbar, {}^{*_\hbar})$,
where the underlying vector spaces are isomorphic to
$\mathbb{C}\otimes \A \cong \hat\A_\hbar$, the $*$-involution reduces to
standard complex conjugation, $A^{*_\hbar} \to \overline{A}$, in the
limit $\hbar\to 0$, and the non-commutative product commutator satisfies
the classical limit, $\frac{1}{i\hbar}[A,B]_\hbar \to \{A,B\}$ as
$\hbar\to 0$. The non-commutative product $*_\hbar$ is referred to as
the \emph{$*$-product}. The details of the kind of dependence on
$\hbar$ is allowed and the sense in which the limits are taken can be
found in the standard literature~\cite{landsman}.

From the constructive point of view, deformation quantization has a
number of successes to its name. The Poisson algebra of any
symplectic~\cite{dwl,fedosov,fedosov-book}
and even any Poisson manifold~\cite{kontsevich}
possesses a formal deformation quantization. The deformations are rigid
and physically inequivalent classes of deformations are controlled by
the second de~Rham cohomology group of the underlying manifold\cite{gr}.
Constructions of
formal deformation quantization commute with reduction by quotienting
out gauge symmetries, via the BV-BRST method, obstructed only by well
defined anomalies~\cite{bhw,fr-bv,rejzner-thesis}. Formal deformations
are possible also for symplectic supermanifolds and infinite dimensional
symplectic manifolds~\cite{vallejo} (and references therein). The formal
deformations of standard $\R^{2n}$
symplectic spaces can be made strict via the Wigner-Weyl-Moyal
$*$-product formula~\cite{landsman}. Similar results are available for
some other special classes of finite dimensional symplectic
manifolds~\cite{rieffel}.

As discussed in Sect.~\ref{sec:freelcft}, a classical field theory
assigns classical mechanical systems to spacetimes in a way coherent
over spacetime embeddings, tentatively summarized in the axioms for a
locally covariant field theory (LCFT) functor, Def.~\ref{def:axlcft},
and similarly for a quantum field theory (QFT or LCQFT). By a
quantization of a classical field theory, we mean an LCQFT functor that
reduces to a classical LCFT functor in the classical limit $\hbar\to 0$,
in the sense of deformation quantization, also in a way coherent with
spacetime embeddings. The recent literature on pAQFT deals specifically
with the perturbative quantization of classical field theory in the
above sense. From the relevant literature, it is clear that the formal
deformation quantization point of view is compatible with perturbative
renormalization of quantum field
theories~\cite{df-dq,hollands-ym,fr-bv,rejzner-thesis}.

Note, however, that the field theories, whose quantization has been
studied in the above way, have only been of semilinear type or whose
non-linearities have been treated perturbatively as well. Thus, their
Green functions (of appropriately linearized equations) and their
Poisson brackets (via the Peierls formula) have causal supports fixed by
external background fields (essentially a Lorentzian metric), as reflected
in the Causality axiom of Def.~\ref{def:axlcft}. The quantum version of
the Causality axiom uses $*$-product commutators, which obey a similar
causal support condition.  The precise form of this condition has been
known for a long time (in fact going back to the original axioms of Haag
and Kastler) and has been checked to hold at each perturbative
$\hbar$-order of the deformation
quantization~\cite{df-dq,rejzner-thesis}.

On the other hand, the situation for quasilinear field theories, which
have a dynamical, field-dependent causal structure (GR being a prominent
example), has been much less clear. The classical notion of causality,
relying on the causal support of the Green functions of the linearized
equations about a fixed background solution (hence, via the Peierls
formula, also of the Poisson brackets), is unproblematic and has been
understood for a long time. However, it relies the notion of a fixed
background solution or, in other words, of a fixed point in the
classical phase space. Since the standard formalism of quantum mechanics
does not make use of a phase space, the question of translation of the
notion of causality to the quantization of a quasilinear field theory
has been generally considered open, especially outside the
perturbative context, where it is sometimes referred to as \emph{quantum
fluctuation/smearing of light cones}~\cite{kh-obsv,lw-maxwell}.

In this paper, we have attempted to precisely formulate the analog of
the standard Causality axiom for quasilinear classical field theories.
This formulation has been given in terms of localizing the algebra of
observables to open subsets of the phase space (the space of all
solutions modulo gauge transformations) whose elements (individual
solutions) all define causal structures that are compatible with some
externally specified chronal cone bundle (in a way that is standard for
symmetric hyperbolic, quasilinear PDE systems, Sects.~\ref{sec:hypersys}
and~\ref{sec:chargeom}). A cover of the total phase space by such
\emph{slow} patches, can be interpreted in categorical terms as the
identification of the total phase space with the pushout of a diagram in
the category of symplectic manifolds corresponding to the slow patches
and the intersections between them.  This categorical formulation makes
it obvious that the corresponding algebraic formulation of the analog of
the Causality axiom is the identification of the total algebra of
observables with a limit of the diagram in the category of Poisson
algebras corresponding to the algebras of observables of the slow phase
space patches, with the Poisson bracket respecting the causal structure
of the algebra of observables of each individual patch. The details of
this formulation occupy Sect.~\ref{sec:gen-caus} and culminate in
Thm.~\ref{thm:gen-caus}. Similarly, an algebraic formulation is also
given to the analog of the Time Slice axiom, Thm.~\ref{thm:gen-ts} in
Sect.~\ref{sec:gen-ts}. An important check on these generalized
Causality and Time Slice properties is that they reduce to the
corresponding standard notions when specialized to semilinear field
theories, Thm.~\ref{thm:semilin-lcft}.

Finally, the above purely algebraic formulation of a generalized
Causality property of classical LCFT motivates the following conjecture
for an LCQFT deformation quantization of a classical LCFT:
\begin{conjecture}\label{cnj:quantcaus}
Let $\stAlg$ be the category of non-commutative, associative,
$*$-algebras, while $\Bkgr_c$ and $\SpBkgr_c$ are the categories of
spacetime manifolds endowed admissible background fields and,
respectively, also with a globally hyperbolic spacelike cone bundles (as
in Sect.~\ref{sec:funct}). An generalized LCQFT is a covariant functor $\hat{\F}_H \colon
\Bkgr_c \to \stAlg$, such that there exists a functor $\hat{\F}\colon
\SpBkgr_c \to \stAlg$, fitting into the following commutative diagram
\begin{equation}
	\vcenter{\xymatrix{
		\SpBkgr_c \ar[dd]^\base \ar[dr]_{\hat{\F}} \\
		& \stAlg ~ . \\
		\Bkgr_c \ar[ur]^{\hat{\F}_H}
	}}
\end{equation}
and satisfying the identity $\F_H(M,\B) = \varprojlim \FF(M,\B)$, where the
limit is indexed by the subcategory $\SpBkgr_c(M,\B)\sse \SpBkgr_c$,
whose objects are of the form $\M = (M,C^\oast,\B)$. This LCQFT also
satisfies the following Causality property. The existence of the diagram
\begin{equation}
\vcenter{\xymatrix{
	& \ar@{<-}[dl]_{\chi_1} \M \ar@{<-}[dr]^{\chi_2} \ar@{<-}[d]^{\usse} \\
	\M_1 \ar[rd] \ar[r]^{\chi'_1} & \M' & \ar[l]_{\chi'_2} \ar[ld] \M_2 ~ ,\\
	& \M_1 \sqcup \M_2 \ar@{-->}[u]
}}
\end{equation}
in $\SpBkgr_c$ implies the existence of the diagram
\begin{equation}\label{eq:q-a-caus}
\vcenter{\xymatrix{
	& \ar@{<-}[dl]_{\hat{\F}(\chi_1)} \hat{\F}(\M) \ar@{<-}[dr]^{\hat{\F}(\chi_2)} \ar@{<-}[d]^{\usse} \\
	\hat{\F}(\M_1) \ar[rd] \ar[r]^{\hat{\F}(\chi'_1)} & \hat{\F}(\M')
		& \ar[l]_{\hat{\F}(\chi'_2)} \ar[ld] \hat{\F}(\M_2) ~ .\\
	& \hat{\F}(\M_1) \otimes \hat{\F}(\M_2) \ar@{-->}[u]
}}
\end{equation}
in $\stAlg$. Recall that $\sqcup$ denotes the \emph{spacelike separated
disjoint union} (Sect.~\ref{sec:gen-caus}), while $\otimes$ denotes the
\emph{independent subsystems tensor product}
(Def.~\ref{def:indep-subsys}).
\end{conjecture}

Similarly, the Time Slice property can be straightforwardly translated
to LCQFT as well. Though the above formulation is rather abstract, we
can see how it looks in rather concrete terms in the example of GR.
Consider a gauge fixed version of GR. It may include fermionic ghost
fields, if necessary, though whose properties we do not discuss at the
moment (that is left to future work~\cite{kh-fermi}). At the moment, we also ignore
the issue of diffeomorphism invariant observables and simply consider
components of the quantized, gauge fixed metric field $\hat{g}_{ab}(x)$
as observable. That is, once smeared, they constitute elements of the
algebra $\hat{\F}_H(M,\{\beta\})$, for a spacetime manifold $M$ endowed
with admissible background field $\beta$. What can we say about the
$*$-product commutator $[\hat{g}_{ab}(x), \hat{g}_{cd}(y)]$? As is, we
cannot say very much, since we do not have complete solution for the
theory and due to dynamical causal structure, we cannot appeal to the
standard Causality axiom. However, since $\hat{\F}_H(M,\{\beta\}) =
\varprojlim \hat{\F}(\M)$, with $\M = (M,C^\oast,\{\beta\})$, there
exist canonical projections $\hat{\F}_H(M,\{\beta\}) \to \hat{\F}(\M)$
for each such $\M$. Denote the images of the metric fields under these
projections by $\hat{g}_{ab}(x) \mapsto \hat{g}_{ab}^\M(x)$. Then, if we
select a spacelike cone bundle $C^\oast\to M$ such that two given points
$x,y\in M$ are $C^\oast$-spacelike separated, our conjecture about the
Causality property in LCQFT implies that
\begin{equation}\label{eq:quant-microcaus}
	[ \hat{g}_{ab}^\M(x), \hat{g}_{cd}^\M(y) ] = 0 ,
\end{equation}
in the usual distributional sense.

As we have seen, the Causality property is a \emph{theorem} of classical
field theory. On the other hand, its version in quantum field theory is
only a \emph{conjecture}. Unfortunately, we cannot do better at the
moment, due to the general difficulty of non-perturbative quantization
of physically relevant field theories~\cite{clay-ym}. Moreover, even in
semilinear field theories, the quantum Causality axiom has only been
verified in free field theories, in four dimensions. On the other hand,
as discussed above, it has been verified perturbatively in standard
physics calculations as well as in the pAQFT literature. For interacting
field theories, the perturbative construction involves a double formal
series, in $\hbar$ and in $\lambda$ (the collection of all coupling
constants). Work is currently under way to verify it (as part of a more
general quantization program) at a level that is perturbative in $\hbar$
but non perturbative in $\lambda$, but extending Fedosov's deformation
quantization construction to LCFT. It is likely that the above conjectured
generalized Causality property for quasilinear field theories may be
verified explicitly in the same setting~\cite{hollands,ribeiro}. The
strategy of these approaches is, as in this paper, is to use PDE theory
to explicitly construct the interacting (nonlinear) classical field
theories before applying Fedosov's deformation quantization construction
to the resulting phase spaces.

\section{Discussion}
\label{sec:discuss}
In this work, we have in a sense taken up a thread left loose in the
previous work of Geroch~\cite{geroch-pde} on the application of the PDE
theory of symmetric hyperbolic systems to classical field theory. This
theory provides an unambiguous, intrinsic notion of causality not tied a
priori to a Lorentzian metric. This notion of causality is especially
useful for theories with quasilinear (rather than linear or semilinear)
equations of motions (of which General Relativity is a prominent
example), whose causal structures are field dependent. This is discussed
in detail in Sects.~\ref{sec:hypersys} and~\ref{sec:chargeom}. See
also~\cite{bannier,rainer1,rainer2,geroch-ftl} for related earlier
ideas.

Moreover, it provides the theorems cited in
Sect.~\ref{sec:pde-theory} as constructive tools to build the solution
space of the theory as an infinite dimensional manifold (the infinite
dimensional differential geometry aspect, as discussed in
Set.~\ref{sec:formal-dg}, was treated only formally here, though is
tackled in earnest in the recent works~\cite{fr-bv,rejzner-thesis,bfr}).
Geroch's casting of most classical relativistic field theories in
symmetric hyperbolic form was then supplemented by constructing (formal)
symplectic and Poisson structures (Sect.~\ref{sec:symp-pois}), turning
the space of solutions into the genuine phase space of classical field
theory. This was done by making use of the well known covariant phase
space method and the Peierls formula, that was generalized to treat
field theories that possess gauge invariance and constraints when cast
into symmetric hyperbolic form.

Together with some reasonable technical hypotheses on the structure of
these infinite dimensional phase spaces and algebras of smooth functions
on them (Hyps.\ref{hyp:opencover}, \ref{hyp:Coo}
and~\ref{hyp:Coo-exact}), the generalized Isotony, Time Slice and
Causality properties were established in Sect.~\ref{sec:classcaus}. As
an important check, these notions specialize
(Thm.~\ref{thm:semilin-lcft}) to the synonymous axioms of classical
relativistic field theory (Def.~\ref{def:axlcft}). In formulating the
generalized properties of these axioms, which are essentially algebraic
in nature, we identified strong parallels between them and aspects of
the notion of well-posedness in classical PDE theory, which are more
geometric in nature. In particular, we have highlighted the importance
of the \emph{on-shell} or phase space point of view that is dual to the
\emph{off-shell} point of view adopted in the recent
works~\cite{fr-bv,rejzner-thesis,bfr}. While the off-shell approach has
proven invaluable in the modern formulation of perturbative
renormalization of quantum field
theories~\cite{hollands-ym,bdf,fr-bv,rejzner-thesis}, the on-shell
approach together with the above mentioned results from PDE theory is
currently the only constructive method of building non-perturbatively
interacting classical field theories. The parallels between the axioms
of classical relativistic field theory and well-posedness also show that
it is sometimes too much to hope for the axioms to hold in every
physically reasonable example, especially in the presence of finite time
blow up singularities (which occur in GR, for example). In particular
the standard Isotony property fails in the presence of finite time blow
(failure of global well-posedness) up (Sect.~\ref{sec:gen-isot}). On the
other hand, local well-posedness is sufficient for the Time Slice
property to hold, unless a more technical condition of \emph{transverse
descent} is also violated (Sect.~\ref{sec:gen-ts}). On the other hand,
the Causality property holds as long as a weaker version of global
well-posedness holds that we called \emph{regular scattering}
(Sect.~\ref{sec:gen-caus}).

The formulation of the generalized Causality property relies on the key
result Thm.~\ref{thm:microcaus}, which could be called a generalized
Microcausality property and follows from the Peierls formula for the
Poisson bracket (Sect.~\ref{sec:formal-pois}) and the domain of
dependence theorem for hyperbolic PDE systems (Cor.~\ref{cor:dod}). This
connection can be seen as one of the main results of this work when
coupled with the following speculation. If a quantum field theory is
constructed from the deformation quantization of a classical field
theory, then the Poisson algebra formulation of classical causality captured
by the generalized Causality property should be replaced by a
non-commutative $*$-algebra version of the same
(Conj.~\ref{cnj:quantcaus}). If deformation quantization is carried out
using Fedosov's method, it is expected that the conjectured formulation
of quantum causality at the very least holds formally at each order in
$\hbar$~\cite{ribeiro,hollands}. Finally, if this conjecture holds, it
provides a concrete answer (Eq.~\eqref{eq:quant-microcaus}) to this old
question in quantum gravity: What can we say about the commutator of two
metric field operators in quantum GR?

One of the motivations for this work had been to translate the notion of
well-posedness from PDE theory to the algebraic setting for classical
field theory, so that this translation could serve as a basis of a more
realistic axiomatization thereof (see preceding paragraph for a
discussion of why the axioms in Def.~\ref{def:axlcft} are not completely
adequate). However, such a translation attempted in this paper has lead
to a number of necessary but somewhat tangential ideas, conjectures and
hypotheses. We believe that it is essential to address and clarify some
of them before attempting to formulate a more realistic axiomatization.

For example, we have chosen the notion of symmetric hyperbolicity
because of the large existing literature on this subject, making it easy
to use to leverage the existing results of PDE theory in the
non-perturbative construction of classical field theories
(Sec.~\ref{sec:pde-theory}). However, reducing the Euler-Lagrange
equations of a field theory to (constrained) symmetric hyperbolic form
is not always an obvious task, and, even if so, can be rather laborious
and require the introduction of auxiliary background fields to keep
everything natural. On the other hand, the notion of regular
hyperbolicity~\cite{chr-pde} seems to be much better adapted to
Lagrangian field theories, but is not yet currently general enough to
handle all cases of interest (such as for example higher order theories)
and has a substantially smaller literature devoted to it. However, since
both notions rely (underneath the hood of their respective
well-posedness theorems) on so-called energy methods, it is likely that
it would be possible to subsume both notions under a more general one
that uses intrinsic information about the characteristic cohomology of a
PDE system~\cite{bg-cohom,at,bbh} to identify approximate conservation
laws and automate their use in an energy method that would establish
local well-posedness. Pursuing such a comprehensive notion of
hyperbolicity is a promising avenue of investigation.

Another set of ideas naturally follows the notion of a conal manifold,
some of which have already been explored in Sect.~\ref{sec:chargeom}.
Conal manifolds abstract the notion of causal order from Lorentzian
geometry, but remain within the realm of differential geometry. It is
interesting to see which results generalize from Lorentzian to conal
manifolds. One important result that already exists is the splitting
theorem for globally hyperbolic cone bundles
(Prop.~\ref{prp:globhyp-lens}). This result is essentially differential
topological in nature, very reminiscent of results that have been
obtained as instances of the so-called
\emph{h-principle}~\cite{gromov-pdr,spring-convex}. On the other hand,
the known proofs of the classic~\cite{geroch-gh} and the (very recent)
generalized~\cite{fs} result use methods rather removed from differential topology
(in the former case appealing to non-continuous measures and in the
latter to Weak KAM theory). A notable variation on the classic result
uses purely order-theoretic and topological methods~\cite{minguzzi-time}. It
would be very interesting to formulate a proof using the methods of
differential topology. In particular, such methods could then be adapted
to the PL (piecewise-linear) and topological manifolds. The latter are
of interest, for example, in the study of the causal structure of rough
Lorentzian metrics~\cite{cg}. Another important set of results that
should be generalized concern the addition of boundaries consistent with
causal structure (in particular, causal compactifications)~\cite{gps}.
There is ample evidence from conformal compactifications in relativity
that understanding the structure of a causal boundary is an important
ingredient in understanding the boundary conditions and asymptotic
behavior of solutions to hyperbolic equations. On a more elementary
side, it seems fruitful to investigate a version of de~Rham cohomology
with spacelike compact supports, as well as the dual homology, whose
classes seem to be naturally represented by Cauchy surfaces. Such
theories may give us a better understanding of the topological and
order-theoretic ambiguities in the covariant phase space method
(Sect.~\ref{sec:formal-symp}) and even give us a better understanding of
the structure of the space of non-globally hyperbolic solutions of a
hyperbolic PDE system. Finally, it is important to verify the validity
of Conjs.~\ref{cnj:gh-stab} and~\ref{cnj:cauchy-ext}, given their
importance for the topological structure of the phase space
(Sect.~\ref{sec:top-choice}) and for the generalized Isotony property
(Sect.~\ref{sec:gen-isot}).

The refinement of the notion of \emph{spacetime support} to that of
\emph{local spacetime support} (Defs.~\ref{def:supp-conf}
and~\ref{def:supp-sols}) invites us to reconsider the problem of local
observables in GR. Without this refinement, it is a well known fact that
there are no gauge (diffeomorphism) invariant observables in GR with
compact spacetime support~\cite{fr-bv,rejzner-thesis}. On the other
hand, the possibility of gauge invariant observables with globally
compact local spacetime support is yet to be considered in detail.

There is still a substantial number of not-completely-resolved questions
on the topology and differential geometry of the infinite dimensional
space of solutions of a hyperbolic PDE system and the algebra of smooth
functions on it. Given the geometric, rather than the functional
analytical focus of this work, we have avoided most of such details and
instead have simply made several reasonable hypotheses
(Hyp.~\ref{hyp:opencover}, \ref{hyp:Coo} and~\ref{hyp:Coo-exact}). To
complete the program of constructive classical field theory, it is
crucial to identify the precise theorems that would replace these
hypotheses.

Finally, note that our discussion of classical field theory has been
restricted to bosonic fields. On the other hand, to take fermionic
matter into account as well as ghost fields that feature in the BV-BRST
formulation of gauge theories, we also need to consider fermi fields. We
believe that the best way to do that is shift from the setting of
manifolds to supermanifolds~\cite{schmitt}, so that fermi fields are
sections of odd vector bundles over spacetime (which is equivalent to
more common but less precise phrasing that classical fermi fields are
Grassmann valued). This approach is more geometric than and
complementary to the algebraic approach adopted
in~\cite{hollands-ym,rejzner-fermion,fr-bv,rejzner-thesis}.  The
fermionic field theories considered so far in this setting have only
been of semilinear type. An extension of this formulation to also
encompass quasilinear field theories (along the lines of the older
work~\cite{cb-sugra}) and their dynamical causal structure will be
reported elsewhere~\cite{kh-fermi}.

\begin{acknowledgements}
The author would like to thank Urs Schreiber for many interesting
discussions on the nature of classical and quantum field theory. The
author also thanks the following people for their interest in and
helpful comments on this work: Romeo Brunetti, Claudio Dappiaggi, Thomas
Hack, Klaus Fredenhagen, and Katarzyna Rejzner. The author is also
grateful to Pedro L. Ribeiro and Stefan Hollands for communicating some
of their unpublished results in the summer of 2012.
\end{acknowledgements}

\appendix

\section{Jet bundles and the variational bicomplex}\label{sec:jets}
In this section we briefly introduce jet bundles and fix the relevant
notation. For simplicity, we restrict ourselves to fields taking values
in vector bundles. However, the discussion could be straightforwardly
generalized to general smooth bundles.

We briefly introduce $k$-jets, mostly to recall some basic facts and fix
notation. More details, as well as a coordinate independent definition,
can be found in the standard
literature~\cite{olver-lie,kms,spring-convex}. Fix a vector bundle $F\to M$,
with $\dim M=n$, with fibers modeled on a vector space $U$, and consider
an adapted coordinate patch $\R^n\times U$, with coordinates
$(x^i,u^a)$.  Extend this patch to a \emph{$k$-jet patch} $\R^n\times U
\times U^{n_k}$ by adding extra copies of $U$, with new coordinates
$(x^i,u^a,u^a_i,u^a_{ij},\ldots,u^a_{i_1\cdots i_k})$, which formally
denote the derivatives of $\del_{i_1i_2\cdots}\phi^a(x)$ of a section
$\phi$ at $x$. To keep track of all the derivatives, we introduce
\emph{multi-index} notation. A multi-index $I=i_1i_2\cdots i_k$ replaces
the corresponding set of symmetric covariant coordinate indices (the
multi-index does not change when the defining $i$'s are permuted).  The
\emph{order} of this multi-index is given by $|I|=k$, with
$|\varnothing|=0$. To augment a multi-index by adding another index, we
use the notation $Ij = jI = i_1\cdots i_k j$. Thus we can write higher
order derivatives as $\del_{i_1\cdots i_k} \phi(x) = \del_I\phi(x)$, the
higher order jet coordinates as $u^a_{i_1\cdots i_k} = u^a_I$ and the
total set of coordinates on a $k$-jet patch as $(x^i,u^a_I)$, $|I|\le
k$. In particular the empty multi-index $I=\varnothing$ corresponds to
$u^a_\varnothing = u^a$.

Since the higher derivatives are symmetric in all indices, the number of
extra coordinates is given by $n_k = \sum_{l=1}^k \dim S^k \R^n$, with
$S^k$ denoting the symmetric tensor product. Given two different
coordinate patches on $F$, we define the transition maps between the
corresponding $k$-jet patches according to the usual calculus chain rule
applied to higher order derivatives. These $k$-jet patches can be glued
together into the total space of the \emph{$k$-jet bundle} $J^kF\to M$,
which includes $J^0F \cong F$.

Since $F\to M$ is a vector bundle, so is $J^kF\to M$. It is isomorphic
to $F\oplus_M (F\otimes_M S^1 T^*M) \oplus_M \cdots \oplus_M (F\otimes_M
S^k T^*M)$, but not naturally%
	\footnote{Though both the $k$-jet and the ``direct sum'' bundles can
	be constructed by applying a functor $\VBndl\to\VBndl$ to a bundle
	$F\to M$, a vector bundle automorphism $\chi\colon F \to F$ induces a
	vector bundle automorphism $J^k(\chi)\colon J^kF\to J^kF$ that need not
	be block diagonal in a basis adapted to the ``direct sum'' bundle,
	while the bundle automorphism induced by the ``direct sum'' bundle is.
	Therefore, the isomorphism of the ``direct sum'' and $k$-jet bundles
	cannot depend on the object $F\to M$ alone, and hence cannot be chosen
	naturally.}. %
Jet bundles come with natural projections $J^kF\to J^{k-1}F$, which
simply discard all derivatives of order $k$. This projection gives
$J^kF$ the structure of an affine bundle over the base $J^{k-1}F$, with
fibers modeled on the vector bundle $(F\otimes_M S^k T^*M)^{k-1}\to
J^{k-1}F$. The bundle $J^kF\to J^{k-1}F$ is affine because, in general,
bundle morphisms of $J^kF\to J^kF$ induced by vector bundle
automorphisms of $F$ are not linear, but are affine.

Given a vector bundle $E\to M$ it can be pulled back to the $k$-jet
bundle along the projection $J^kF\to M$. We introduce a convenient
notation for this pullback.
\begin{definition}
We denote by $(E)^k\to J^kF$ the pullback of $E\to M$ to $J^k F$, which
then fits into the pullback commutative square
\begin{equation}
\vcenter{\xymatrix{
	(E)^k \ar[r] \ar[d] & E \ar[d] \\
	J^kF  \ar[r]        & M ~ .
}}
\end{equation}
\end{definition}

Any smooth section $\phi\colon M\to F$ automatically gives rise to its
\emph{$k$-jet prolongation} or \emph{$k$-prolongation} $j^k\phi\colon
M\to J^kF$. Namely $j^k\phi$ is a section of the bundle $J^kF\to M$ that
is defined in a local adapted coordinate patch as
\begin{equation}
	j^k\phi(x) = (x^i,\phi^a(x),\del_i\phi^a(x),\ldots,
		\del_{i_1\cdots i_k}\phi^a(x))
	= (x^i,\del_I \phi^a(x)), ~~ |I|\le k.
\end{equation}
One can think of the $k$-prolongation symbol as a differential operator
\begin{equation}
	j^k\colon \Secs(F) \to \Secs(J^kF)
\end{equation}
of order $k$. In fact, any (not necessarily linear) differential operator of order $k$,
\begin{equation}
	f\colon \Secs(F) \to \Secs(E), ~~
	f\colon \phi \mapsto f[\phi] ,
\end{equation}
can be written as a composition of $j^k$ with an order $0$ (not
necessarily linear) operator $f\colon J^kF\to E$, such that $f[\phi] =
f(j^k\phi)$. Note that we are slightly abusing notation by denoting both
the differential operator and the bundle morphism by the same symbol
$f$.

Further, we can define an $l$-prolongation of a differential operator
$f$ of order $k$,
\begin{equation}
	p^\oo f \colon J^{k+l}F \to J^lE ,
\end{equation}
which is then a differential operator of order $k+l$, by composing with
$j^l$: $p^\oo f[\phi] = j^kf[\phi]$. Prolongation is discussed briefly
using coordinate-wise operations in Sect.~\ref{sec:integrability}. The
$k$-jet prolongation $j^k\phi$ can now be thought of as a special case
of bundle morphisms, that is, $j^k\phi = p^k \phi$, where on the right
hand side we interpret $\phi$ as the base fixing bundle morphism to
$F\to M$ from the trivial $0$-dimensional bundle $\id\colon M\to M$.
\begin{equation}
\xymatrix{
	M \ar[d]_\id \ar[r]^\phi & F \ar[d] \\
	M \ar[r] & M ~ .
}
\end{equation}

Given the sequence of projections $k$-jet bundles over $M$,
\begin{equation}
	\cdots \to J^2F \to J^1F \to J^0F\cong F ,
\end{equation}
it is convenient to introduce the \emph{infinite jet order} (or
\emph{$\oo$-jet}) bundle $J^\oo F$ defined as the projective limit over
the jet order $k$
\begin{equation}
	J^\oo F = \varprojlim J^k F .
\end{equation}
This limit implicitly defines $J^\oo F$ as an infinite dimensional
smooth manifold. The main advantage of working with $\oo$-jets is that
any function or tensor on $J^kF$ for finite $k$ can be pulled back to
$J^\oo F$. Conversely, any smooth function or tensor on $J^\oo F$
depends only on jets up to some finite order, say $k$, and can be
faithfully projected to $J^kF$. Another major convenience of working on
$J^\oo F$ is the ability to decompose the usual de~Rham differential
into its \emph{horizontal} and vertical parts
\begin{equation}
	\d = \dh + \dv.
\end{equation}
The defining property of $\dh$ is the following. Given a section
$\phi\colon M \to F$, we must have the identity
\begin{equation}
	(j^\oo\phi)^* \dh \alpha = \d (j^\oo\phi)^* \alpha,
\end{equation}
where $\alpha$ is any differential form on $J^\oo F$ and $\d$ is the
usual de~Rham differential on $M$. On the other hand, $\dv$ is
characterized by the fact that its image is annihilated by the pullback
to $M$ along any section $\phi$,
\begin{equation}
	(j^\oo \phi)^* \dv \alpha = 0.
\end{equation}
It can be checked that the horizontal
and vertical differentials anti-commute and are separately nilpotent,
\begin{equation}
	\dh^2 = 0 = \dv^2, \quad \dh\dv + \dv\dh = 0 .
\end{equation}
Note that, to apply $\dv$ or $\dh$ to forms defined on a finite order
jet bundle $J^k F$, the pullback and projection operations mentioned
above will often be applied implicitly. Thus the application of say
$\dh$ to a differential form on $J^k F$ may yield that a differential
form that projects to $J^{k+1}F$ but not to $J^k F$. In local
coordinates $(x^i,u^a)$ on $F$, and the induced coordinates
$(x^i,u^a_I)$ on $J^\oo F$, a convenient basis for differential forms is
\begin{equation}
	\dh x^i = \d x^i, \quad
	\dv u^a_I = \d u^a_I - \dh u^a_I
		= \d u^a_I - u^a_{Ii} \d x^i .
\end{equation}
We can also define two special kinds of vector fields. A vector field
$\hat{\xi}$ is \emph{horizontal} if its action in local coordinates is
\begin{equation}
	\hat{\xi}(x^i) = \xi^i, \quad
	\hat{\xi}(u^a_I) = \xi^i u^a_{iI} .
\end{equation}
for some $\xi^i=\xi^i(x,u^a_I)$. In particular, the vector field
$\hat{\del}_j$, with $\xi^i = \delta^i_j$, is horizontal. Note that
$[\hat{\del}_i,\hat{\del}_j]=0$. A vector field $\hat{\psi}$ is
\emph{evolutionary} if its action in local coordinates is
\begin{equation}
	\hat{\psi}(x^i) = 0, \quad
	\hat{\psi}(u^a_I) = \hat{\del}_I (\psi^a),
\end{equation}
for some $\psi^a = \psi^a(x,u^b_I)$, where $\hat{\del}_I(f) =
\hat{\del}_{i_1}(\hat{\del}_{i_2}(\cdots\hat{\del}_{i_k}(f)\cdots))$ for
multi-index $I=i_1i_2\cdots i_k$ (the order of application of these
vector fields does not matter since they commute). Note that the
$\psi^a$ can be seen as the fiber coordinate components of a section of
the bundle $(F)^\oo\to J^\oo F$. These definitions can be checked to be
coordinate independent.

One can show that for a horizontal vector field $\hat{\xi}$ on $J^\oo F$
there exists a vector field $\xi_\phi$ on $M$ such that their
actions on scalar functions are intertwined by the pullback along the
jet prolongation $j^\oo\phi$ of a section $\phi\colon M\to F$,
\begin{equation}
	\hat{\xi}(f)(j^\oo\phi) = \xi_\phi(f(j^\oo\phi)) ,
\end{equation}
for any scalar function $f$ on $J^\oo F$. Namely, in local coordinates,
$\xi_\phi = \xi_\phi^i\del_i$ with $\xi_\phi^i = (\iota_{\hat{\xi}} \d
x^i)(j^\oo\phi) = \hat{\xi}(x^i)(j^\oo\phi) = \xi^i(j^\oo\phi)$. On the
other hand, evolutionary vector fields $\hat{\psi}$ satisfy the identities
\begin{gather}
\label{eq:idh-comm}
	\iota_{\hat{\psi}} (\dh \alpha) + \dh (\iota_{\hat{\psi}} \alpha) = 0 , \\
\label{eq:jet-Lie}
	\Lie_\psi (j^\oo\phi)^*\alpha
		= \left.\frac{\d}{\d\eps}\right|_{\eps=0} [j^\oo(\phi+\eps\psi)]^* \alpha
		=  \iota_{\hat{\psi}} \dv \alpha = \Lie_{\hat\psi} \alpha ,
\end{gather}
for any form $\alpha\in\Forms^*(J^\oo F)$ and section $\psi\colon M\to F$.
Actually, $\psi$ could be a section of $(F)^k\to J^kF$, that is, it
could depend on $\phi^a(x)$ and its derivatives and not only on $x\in
M$. The only corresponding change in the above formula would be to
replace $\eps\psi$ by $\eps(j^k\phi)^*\psi$. Ostensibly, $\Lie_\psi$
should stand for the Lie derivative on the infinite dimensional manifold
of sections of $F\to M$, where the section $\psi$ is identified with the
vector field whose action on local coordinates is $\Lie_\psi \phi^a(x) =
\psi^a(x)$. However, since we do not delve into the differential
geometry of infinite dimensional manifolds here, we keep the symbol
$\Lie_\psi(j^\oo\phi)^*$ primitive and defined as above.

Integrations or differentiations by parts are carried out using the
following basic identity
\begin{align}
\label{eq:byparts}
	\dv u^a_{Ii}\wedge \d x^i\wedge\alpha
	&= \dv (u^a_{Ii} \d x^i) \wedge \alpha \\
	&= (\dv\dh u^a_I) \wedge \alpha \\
	&= -(\dh\dv u^a_I) \wedge \alpha \\
	&= -\dv u^a_I \wedge \dh\alpha - \dh(\dv u^a_I\wedge\alpha) .
\end{align}

This split of the de~Rham differential into horizontal and vertical
differentials also splits the de~Rham complex $\Forms^*(J^\oo F)$ of
differential forms on $J^\oo F$ into a
\emph{bicomplex}~\cite{anderson-small,anderson-big}. Since the
horizontal and vertical $1$-forms generate the graded commutative
algebra of differential forms, any form $\lambda\in \Forms^*(J^\oo F)$
can be uniquely written as
\begin{equation}
	\lambda = \sum_{h,v} \lambda_{h,v} ,
\end{equation}
where $0\le h \le n$ and $0\le v$ are respectively the horizontal and
vertical degrees form degrees. We have thus turned the differential
forms into a bigraded complex $\Forms^*(J^\oo F) = \bigoplus_{h,v}
\Forms^{h,v}(F)$, with the $\dh$ differential increasing $h$ by $1$ and
the $\dv$ differential increasing $v$ by $1$. This complex is called the
\emph{variational bicomplex}~\cite{anderson-small,anderson-big}. As with any bicomplex, we can
consider its cohomology with respect to either or any combination of the
two differentials. The \emph{horizontal cohomology} is $H^{h,v}(\dh) =
H(\Forms^*(J^\oo F),\dh)$ in degrees $(h,v)$. The \emph{vertical
cohomology} is $H^{h,v}(\dv) = H(\Forms^*(J^\oo F),\dv)$ in degrees
$(h,v)$. Both $(H^{h,*}(\dh),\dv)$ and $(H^{*,v}(\dv),\dh)$ still form
complexes, therefor we can also consider their cohomologies. The
\emph{relative cohomologies} are $H^{h,*}(\dv|\dh) =
H(H^{h,*}(\dh),\dv)$ and $H^{*,v}(\dh|\dv) = H(H^{*,v}(\dv),\dh)$.

\section{Limits and colimits in category theory}
\label{sec:limits}

Some basic information on category theory can be found in~\cite{borceux}
and more specifically about categorical limits and colimits
in~\cite{borceux,limits}.

Consider categories $\CatI$ and $\CatC$ and a functor $\DD\colon
\CatI\to \CatC$. The image $\DD(\CatI)$ (also denoted by just $\DD$)
forms a \emph{diagram} in $\CatC$ with index category $\CatI$. A
\emph{cone} to the diagram $\DD$ is an object $C$ of $\CatC$, the
\emph{vertex}, together with a morphism $c_i\colon C\to D_i = \DD(i)$
for each object $i$ of $\CatI$, such that diagram
\begin{equation}
\vcenter{\xymatrix{
	& C \ar[dl]_{c_i} \ar[dr]^{c_j} & \\
	D_i \ar[rr] & & D_j
}}
\end{equation}
is commutative for each morphism $i\to j$ in $\CatI$.  If it exists, the
\emph{limit} (also \emph{inverse} or \emph{projective limit}) is an object $\varprojlim
\DD$ of $\CatC$ that is the vertex of a cone of \emph{canonical
morphisms} $u_i\colon \varprojlim\DD \to D_i$ such that the following
universal property holds: any cone $(C,c)$ to $\DD$ factors through it
with a unique \emph{mediating morphism} (also \emph{canonical morphism})
$u_C$, which makes the following diagram commute:
\begin{equation}
\xymatrix{
	C \ar@{-->}[r]_-{u_C} \ar@/^1pc/[rr]^-{c_i}
	& \varprojlim \DD \ar[r]_-{u_i}
	& D_i
} ~ .
\end{equation}
For an illustrative example, consider the index
category $\CatI$ consisting of only the objects and morphisms $1
\to 3 \ot 2$, while $\DD$ is a functor from $\CatI$ to the category of
sets. The diagram $\DD$, the limit and the cone of canonical morphisms
fit into the following commutative diagram
\begin{equation}
\vcenter{\xymatrix{
	\varprojlim\DD \ar[d]^{u_1} \ar[r]^-{u_2} \ar[dr]^{u_3} & D_2 \ar[d] \\
	D_1 \ar[r] & D_3 ~ .
}}
\end{equation}
The universality of the limit identifies the set $\varprojlim\DD$ with
the subset of the Cartesian product $D_1\times D_2$ consisting of pairs
$(d_1,d_2)$ that get mapped to the same element in $D_3$, $d_1\mapsto
d_3 \mapsfrom d_2$. If the object $3$ and the morphisms to it were
absent from the index category $\CatI$, we would simply have
$\varprojlim\DD \cong D_1\times D_2$. In other categories, we may keep
the same intuition, but replace the Cartesian product $\times$ by the
categorical product.

The notion of a colimit is dual. As before, consider a diagram
$\DD\colon \CatI \to \CatC$. A \emph{co-cone} of the diagram $\DD$ is an
object $C$ of $\CatC$, the \emph{vertex}, together with a morphism
$c_i\colon \DD(i) = D_i\to C$ for each object $i$ of $\CatI$. If it
exists, the \emph{colimit} (also \emph{direct} or \emph{inductive limit}) is an object
$\varinjlim \DD$ of $\CatC$ that is the vertex of a cone of
\emph{canonical morphisms} $u_i\colon D_i \to \varinjlim\DD$ such that
the following universal property holds: any co-cone $(C,c)$ of $\DD$
factors through it with a unique \emph{mediating morphism} (also
\emph{canonical morphism}) $u_C$, which makes the following diagram
commute:
\begin{equation}
\xymatrix{
	D_i \ar[r]_-{u_i} \ar@/^1pc/[rr]^-{c_i}
	& \varinjlim \DD \ar@{-->}[r]_-{u_C}
	& C ~ .
}
\end{equation}
For an illustrative example, consider the index category $\CatI$
consisting of only the objects and morphisms $1 \ot 3 \to 2$, while
$\DD$ is a functor from $\CatI$ to the category of sets. The diagram
$\DD$, the colimit and the co-cone of canonical morphisms fit into the
following commutative diagram
\begin{equation}
\xymatrix{
	D_3 \ar[d] \ar[r] \ar[dr]^{u_3} & D_2 \ar[d]^{u_2} \\
	D_1 \ar[r]^{u_1} & \varinjlim\DD ~ .
}
\end{equation}
The universality of the limit identifies the set $\varinjlim\DD$ with
the quotient of the disjoint union $D_1\sqcup D_2$ by the equivalence
relation that identifies elements $d_1 \sim d_2$, respectively of $D_1$
and $D_2$, provided they are images of the same element $d_3\in D_3$,
$d_1 \mapsfrom d_3 \mapsto d_2$.  If the object $3$ and the morphisms
from it were absent from the index category $\CatI$, we would simply
have $\varprojlim\DD \cong D_1\sqcup D_2$. In other categories, we may
keep the same intuition, but replace the disjoint union $\sqcup$ by the
categorical coproduct.

\bibliographystyle{hplain}
\bibliography{paper-chargeom}

\end{document}